\documentclass[aos]{imsart}

\usepackage[numbers]{natbib}
\usepackage{xr}
\usepackage[colorlinks, bookmarks=true]{hyperref}
\usepackage{url} 
\usepackage{hyperref}
\usepackage{slashbox}
\usepackage{amsthm,amsmath,amssymb}
\usepackage{marvosym}
\usepackage[hang,small,bf]{caption}
\usepackage{mathtools}
\RequirePackage[OT1]{fontenc}
\RequirePackage{amsthm,amsmath}
\usepackage{bbm}
\usepackage{overpic}
\usepackage{mathtools}

\usepackage{graphicx}
\usepackage{float}

\usepackage{enumerate}
\usepackage{enumitem}

\usepackage{amsthm,amsmath,amssymb}
\usepackage{color}

\newtheorem{theorem}{Theorem}

\newtheorem{lemma}{Lemma}
\newtheorem{remark}{Remark}
\newtheorem{proposition}{Proposition}

\usepackage{enumerate}
\usepackage{enumitem}
\advance \topmargin by -\headheight                                           
\advance \topmargin by -\headsep                                              
\advance \topmargin .1in                                                      
\def\ben{\begin{eqnarray}}
\def\be*{\begin{eqnarray*}}
\def\non{\end{eqnarray}}
\def\no*{\end{eqnarray*}}
\begin{document}
	\begin{frontmatter}
	\title{Intelligent Sampling and Inference for Multiple Change Points in Extremely Long Data Sequences}
\begin{aug}
	\author{Zhiyuan Lu},
	\author{Moulinath Banerjee}
	\and
	\author{George Michailidis}
\end{aug}



\begin{abstract}
	Change point estimation in its offline version is traditionally performed by optimizing a data fit criterion over the data set of interest
	that considers each data point as the true location parameter. The data point that minimizes the criterion is declared as the change point estimate. For estimating multiple change points, the available procedures in the literature are analogous in spirit, but significantly more involved in execution. 
	Since change-points are local discontinuities, only data points close to the actual change point provide useful information for estimation, while data points far
	away are superfluous, to the point where using only the points close to the true
	parameter is just as precise as using the full data set. Leveraging this ``locality principle", we introduce a two-stage procedure
	for the problem at hand: the 1st stage uses a {\em sparse subsample} to obtain pilot estimates of the underlying
	change points, and the 2nd stage updates these estimates by sampling densely in appropriately defined neighborhoods around them, and computing a local change-point estimator. 
	We establish that this method achieves the same rate of convergence and even
	virtually the same asymptotic distribution as an analysis of the full data, while reducing computational complexity to
	$\sqrt{N}$ time ($N$ being the length of data set) under favorable circumstances, as opposed to at least $O(N)$ time for all current procedure.
This makes our new method promising for the analysis of exceedingly long data sets with adequately spaced out change
	points. The main results, which, in particular, lead to a prescription for constructing explicitly computable joint (asymptotic) confidence intervals for a
	{\em growing number} of change-points via the proposed procedure, are established under a signal plus noise model with independent and identically distributed error terms, but extensions to dependent data settings, as well as multiple stage ($>2$) procedures are also provided. The performance of our procedure -- which is coined ``intelligent sampling" --  is illustrated on both emulated data and a real internet data stream.
\end{abstract}
\end{frontmatter}

\section{Introduction} \label{sec:introductionsection}
Change point analysis has been extensively studied in both the statistics and econometrics literature \cite{Basseville1993detection}, due to its wide applicability in diverse fields, including economics and finance \cite{frisen2008financial}, quality control \cite{qiu2013introduction}, neuroscience \cite{koepcke2016single}, etc. A change point represents a discontinuity in the parameters of the data generating process. The literature has investigated both the {\em offline} and {\em online} versions of the problem \cite{Basseville1993detection,csorgo1997limit}. In the former case, one is given a sequence of observations and questions of interest include: (i) whether there exists a change point and (ii) if there exists one (or multiple) change point(s), identify its (their) location, as well as estimate the parameters of the data generating process to the left and right of it (them). In the latter case, one obtains new observations in a sequential manner and the main interest is in quickest detection of the change point. This paper deals with the \emph{offline} problem for very large data sequences. 

The offline analysis of \emph{extremely long} data sequences poses a body of challenges, not least owing to the fact that conventional modes of analysis based on the entire data are computationally challenging, owing to the massive size involved. The identification of change-points in a sequence, which constitute local discontinuities, requires some sort of a search procedure. A list of such methods, along with summaries of their results, can be found in \cite{niu2016multiple}. Time complexities for identification of multiple change points range between $O(N)$ and $O(N^2)$ depending on spacing conditions between change points, where $N$ is the length of the time series:\footnote{Precise identification of change-points under very relaxed spacing conditions can be accomplished, for example, by a dynamic programming based algorithm, which is of order $N^2$.}such time complexities are not attractive for very long sequences. 

Such offline scenarios arise when network engineers examine \emph{previously stored traces of traffic} in high-speed computer and communication networks and examine them in greater detail aiming to identify shifting persistent patterns followed by attribution analysis, possibly correlating the detected change points with other data that would point to either infrastructure related issues or shifts in user patterns. This should be contrasted with \emph{online monitoring} where the goal is to raise alarms \emph{in real time} whenever the underlying process is ``out-of-control" and notify network engineers, who determine the tolerance level for alarms, i.e. how large the change should be to be considered negatively impacting network operations \cite{van2014opennetmon}, \cite{ahmed2016survey}. In the \emph{offline setting}, the number of observations in network traces, corresponding to the number of packets or their payload in bytes at a granular temporal scale, is in the thousands per minute. Analogous problems come up in other systems including manufacturing processes \cite{shen2016change}, or other cyber-physical systems equipped with sensors \cite{khan2016optimal}, where data may possibly be stored in a distributed manner. 

The point of view adopted in the paper is as follows: we are given a data sequence of massive length and are interested in identifying the major changes in this series. While \emph{short-lived} changes of intensities in mean shifts may occur frequently, such transient perturbations do not affect the overall performance of most engineered or physical systems. With such applications in mind, it is very reasonable to assume that the number of truly significant changes that \emph{persist} over time, is \emph{not particularly large}. 
To communicate the core ideas effectively, we employ a canonical model: a long piece-wise constant mean model with multiple jumps -- where the number of jumps increases at a rate much slower than the length of the series ---  that are not too small relative to the fluctuations induced by noise. Also, the flat stretches between consecutive jumps are taken to be persistent. Indeed, the `piece-wise constant with jumps' strategy has been considered by a variety of authors (see \cite{niu2016multiple} and references therein) for the problem at hand and is convenient for developing the methodology. We demonstrate subsequently that our proposed strategy is robust to mis-specifications of the persistent piecewise constant model, in the sense that its detection of persistent stretches is essentially unaffected in the presence of short-lived `idiosyncratic' signal bursts.


Next, we summarize key contributions of this paper. 

{\bf 1.} We propose an effective solution to the computational problem discussed above via a strategy called ``intelligent sampling'', which proceeds by making two (or more) passes over the time-series at different levels of resolution. The first pass is carried out at a coarse time resolution and enables the analyst to obtain pilot estimates of the change-points, while the second investigates stretches of the time-series in (relatively) small neighborhoods of the initial estimates to produce updated estimates enjoying a high level of precision; in fact, essentially the same level of precision, as would be achieved by an analysis of the entire massive time-series. The core advantage of our proposed method is that it reduces computational time from \emph{linear} to \emph{sub-linear} under appropriate conditions, and, in fact close to square-root order, if the number of change-points is small relative to the length of the time-series. It is established that the computational gains (and analogously other processing gains from input-output operations) can be achieved \emph{without compromising} any statistical efficiency. From a realistic stand-point, the time resolution to be used for the first pass is difficult to determine beforehand, and to alleviate this methodological problem we propose an adaptive \emph{doubling algorithm} [Section 5.4] that samples at increasing levels of resolution till a sufficient number of initial change point estimates are found. As the emulation experiment in Section \ref{sec:realdata} demonstrates, our proposed method consistently picks up persistent structural changes in the presence of multiple spiky signals while staying agnostic to the latter. Nevertheless, if short duration shifts are also of interest, the interested researcher can further analyze the persistent segments (which are of smaller order than $N$) with any of the available procedures in the literature in a \emph{parallel} fashion, thus retaining the computational gains achieved by intelligent sampling.

{\bf 2.} Most of our results are rigorously developed in the signal plus noise model with independent and identically distributed (iid) sub-Gaussian errors, and in certain cases, normal errors. Normal errors have been widely used in the change point literature to showcase methods and establish theoretically their performance guarantees, e.g. \cite{fryzlewicz2014wild,niu2012screening,zhang2007modified}, since they provide an attractive canonical model and are more amenable to analysis. However, we also provide results and indicate extensions when the errors exhibit short or long-range dependence and also in the presence of non-stationarity, since both these features are likely to be present with very long data sequences. Some empirical evidence is provided to this effect. 
\newline
\indent Further, the focus of the presentation is on 2-stage procedures that provide all the key insights into the workings of the intelligent sampling procedure. However, for settings where the size of the data exceeds $10^{10}$, multiple stages are required to bring down the analyzed subsample to a manageable size. We therefore cover extensions to multi-stage intelligent sampling procedures as well as address how samples should be allocated at these different stages. Furthermore, such massive data sequences, often, can not be effectively stored in a single location. This does not pose a problem for intelligent sampling as it adapts well to distributed computing: it can be applied on the reduced size subsamples at the various locations where the original data are stored, followed by a single subsequent back and forth communication between the various locations and the central server, and subsequent calculations essentially carried out on the local servers. This is elaborated on in Section \ref{sec:compmethod}. 

{\bf 3.} On the inferential front, we establish asymptotic approximations to the joint distribution of the change-point estimates obtained by intelligent sampling in terms of the distribution functions of (typically asymmetrically) drifting random walks on the set of integers [Theorems \ref{thm:increasingJasymprotics} and \ref{thm:multidepend}], which can then be used to provide explicitly computable asymptotic joint confidence intervals for both finitely many and a \emph{growing number} of change points. While concentration properties of change-point estimates around the the true parameters in multiple change point problems are known, to the best of our knowledge, such results involve hard-to-pin-down constants (this is discussed in some more detail in Section \ref{sec:refitting} in the context of binary segmentation) and are therefore difficult to use in practical settings. Our prescribed methods involve estimating signal to noise ratios at the different change-points, which is easy to accomplish, and values of quantiles of drifting Gaussian random walks for different values of the drift parameter (which can be pre-generated on a computer). The arguments involved in establishing Theorem \ref{thm:increasingJasymprotics} require, among other things, careful analyses of the distribution functions of both symmetric and asymmetric Gaussian random walks (see part B of Supplement, Section \ref{sec:supplementB}) which may very well prove useful in many other contexts. To the best of our knowledge, our prescription of joint asymptotic confidence intervals for a growing number of change-points with readily computable estimates of the distributional parameters involved [see Theorem \ref{thm:increasingJasymprotics}] is the \emph{first of its kind} in the change point literature. 

{\bf 4:} We illustrate our procedure by applying it to two data-sets, the first in Section \ref{partly-emulated-data} on a partly emulated data set where the error process arises from real network traffic data and we create ground-truths by 'injecting' change-points into the error process through the addition of piecewise constant mean shifts. Section \ref{real-data-section} examines a real data set, where our intelligent sampling proposal is useful. It corresponds to network (internet) traffic destined for an autonomous system - a collection of Internet addresses under the control of a single network operator. The corresponding time series reflects aggregate incoming traffic to the autonomous system at 2 sec resolution. Short lived spikes are of less interest to network operators, since they usually correspond to transient traffic patterns. On the other hand, persisting shifts are of significant interest, indicating possible malicious activity, technical issues with the infrastructure, or at larger time scales, traffic growth that needs to be mitigated by corresponding capacity growth of the network's infrastructure. While this data is better described by a piecewise linear regime (as opposed to the piecewise constant model used for our theoretical development), our intelligent sampling technique continues to work well after modifying it to fit piecewise linear models at each stage: the procedures to obtain the first stage pilot estimates, as well as the second stage updated estimates are simply tweaked so that they can tackle piecewise linear functions. This demonstrates the scope of intelligent sampling to data-sets with more general mean structures than the one tackled in depth in this paper, and makes it attractive as a general principle.


The remainder of the paper is organized as follows. Section \ref{sec:single-CP} addresses intelligent sampling for the simpler {\em single} change point problem, which provides some fundamental insights into the nature of the procedure and its theoretical and computational properties. Section
\ref{sec:multiplechangepointsintro} deals with the main topic of this study: intelligent sampling for the multiple change point problem with a growing number of change-points, and presents the main theoretical results of the paper. Section \ref{sec:prac-imple} develops the practical methodology for intelligent sampling using \emph{binary segmentation} as the working procedure at Stage 1, and studies the computational complexity of the resulting approach. Section \ref{sec:compmethod} provides an elaborate study of the minimum subsample size required for precise inferences as a function of the length of the full data sequence and the signal-to-noise ratio using multiple stage procedures. Extensions to non-iid settings which are more pertinent for the more special case of time-series data are discussed in Section \ref{sec:dependenterrors}. The numerical performance of the procedure is calibrated via thorough simulations in Section \ref{sec:simulations}, while  applications to both emulated and real Internet data are presented in \ref{sec:realdata}. Section \ref{sec:discussion} concludes with a discussion of possible extensions of the intelligent sampling procedure, both in terms of alternative Stage 1 procedures (like wild binary segmentation and SMUCE), and also to other kinds of data (non-Gaussian data, discrete data, decaying signals). 
In the interests of space, all proofs and elaborate discussions of various facets of intelligent sampling are collected in the Supplement.

\section{Intelligent Sampling for the Single Change Point Problem}
\label{sec:single-CP}
\subsection{Single Change Point Model}\label{sec:singlemodel}
The simplest possible setting for the change point problem is the {\em stump} model, where data $(1/N,Y_1),\cdots,$ $(N/N,Y_N)$ are available
with $Y_i=f(i/N)+\varepsilon_i$ for $i=1,\cdots, N$, and where the error term $\varepsilon_i$ is independent and identically distributed (iid) following a $N(0,\sigma^2)$ distribution, while the function $f$ takes the form
\begin{eqnarray}\label{eq:stumpsignal}
f(x)=\alpha\cdot 1(x\leq \tau)+\beta\cdot 1(x>\tau),\qquad x\in (0,1),
\end{eqnarray}
for some constants $\alpha,\beta\in\mathbb{R}$, $\alpha\neq\beta$, and $\tau\in (0,1)$: the so-called `stump' model. 
For estimating the {\em change point} $\tau$ we employ a least squares criterion, given by
\begin{eqnarray}\label{eq:fullestimators}
(\hat{\alpha},\hat{\beta},\hat{\tau}):=\underset{(a,b,t)\in \mathbb{R}^2\times (0,1)}{\arg\min}\sum_{i=1}^N\left( Y_i-a\cdot 1(i/N\leq t)-b\cdot 1(i/N>t) \right)^2.
\end{eqnarray}

Using techniques similar to those in Section 14.5 of \cite{kosorok2007introduction}, 
we can establish that the estimator $\hat{\tau}$ is consistent for $\tau_N:=\lfloor N\tau\rfloor/N$, which acts as the change point among the covariates lying on the even grid. 
\begin{proposition}\label{prop:fullestimators}
For the stump model with normal errors the following hold: \\
(i) Both $(\hat{\alpha}-\alpha)$ and $(\hat{\beta}-\beta)$ converge to 0 with rate $O_p(N^{-1/2})$. \\
(ii) The change point estimate $\hat{\tau}$ satisfies
\begin{eqnarray}
\mathbb{P}\left[N(\hat{\tau}-\tau_N)=k\right]\to \mathbb{P}[L=k]\text{ for all }k\in\mathbb{Z}
\end{eqnarray}
where $L=\underset{i\in\mathbb{Z}}{\arg\min}\,X(i)$, and the random walk $\{X(i)\}_{i\in\mathbb{Z}}$ is defined as
\begin{equation}\label{Zprocess}
X(i)=\begin{cases} \Delta(\varepsilon_1^*+...+\varepsilon_{i-1}^*+\varepsilon_i^*)+i\Delta^2/2, \quad & i>0\\0, & i=0\\
-\Delta(\varepsilon_{i+1}^*+...+\varepsilon_{-1}^*+\varepsilon_{0}^*)+|i|\Delta^2/2,\quad & i<0,\end{cases}
\end{equation}
with $\varepsilon_0^*,\varepsilon_1^*,\varepsilon_2^*,\dots$ and $\varepsilon_{-1}^*,\varepsilon_{-2}^*,\dots$ being iid $N(0,\sigma^2)$ random variables and $\Delta := \beta - \alpha$.
\end{proposition}

Next, we make several notes on the random variable $L$ introduced in the above proposition, as it appears multiple times throughout the remainder of this paper. Although, at a glance, the distribution of $L$ depends on two parameters, $\Delta$ and $\sigma$, in actuality $L$ is completely determined by the signal-to-noise ratio $\Delta/\sigma$ due to the Gaussian setting. To see this, note that we can re-write $L=\underset{i\in\mathbb{Z}}{\arg\min}(Z_i/|\Delta\sigma|)$ where
\begin{eqnarray}\label{eq:Zprocess2}
\frac{X(i)}{|\Delta\sigma|}=\begin{cases}
\text{sgn}(\Delta)(\varepsilon_1^*/\sigma+\dots+\varepsilon_i^*/\sigma)+i|\Delta/\sigma|/2,\quad& i>0\\0,& i=0\\-\text{sgn}(\Delta)(\varepsilon_{i+1}^*/\sigma+\dots+\varepsilon_{0}^*/\sigma)+|i\Delta/\sigma|/2,\quad& i<0
\end{cases}
\end{eqnarray} 
Since $\{\text{sgn}(\Delta)\varepsilon_i^*/\sigma\}_{i\in\mathbb{Z}}$ are iid $N(0,1)$ random variables, invariant under $\Delta$ and $\sigma$, it follows that $L$ only depends on the single parameter $\Delta/\sigma$. Hence, from here on, denote the associated random process as
\begin{eqnarray}
\label{eq:Zprocess3}
X_\Delta(i)=\begin{cases}
\text{sgn}(\Delta)(\varepsilon_1^\diamond+\dots+\varepsilon_i^\diamond)+i|\Delta|/2,\quad& i>0\\0,& i=0\\-\text{sgn}(\Delta)(\varepsilon_{i+1}^\diamond+\dots+\varepsilon_{0}^\diamond)+|i|\cdot|\Delta|/2,\quad& i<0
\end{cases}
\end{eqnarray}
where $\varepsilon_j^\diamond$ for $j\in \mathbb{Z}$ are all iid $N(0,1)$ random variables. Denote the argmin of the random walk $X_\Delta(i)$ as $L_\Delta=\underset{i\in\mathbb{Z}}{\arg\min}X_\Delta(i)$. An immediate observable property of $L_\Delta$ is the stochastic ordering with respect to $|\Delta|$:

\begin{proposition}\label{prop:Ldist}
Suppose we have constants $\Delta_1,\Delta_2\in\mathbb{R}$ such that $0<|\Delta_1|<|\Delta_2|$, then for any positive integer $k$ 
\begin{eqnarray}
\mathbb{P}[|L_{\Delta_1}|\leq k]\leq \mathbb{P}[|L_{\Delta_2}|\leq k]
\end{eqnarray}
\end{proposition}
Practically, this stochastic ordering implies that if the $1-\alpha$ quantile $Q_{\Delta_1}(1-\alpha)$ of $|L_{\Delta_1}|$ is known, then $Q_{\Delta_1}(1-\alpha)$ can also serve as a conservative $1-\alpha$ quantile of $|L_{\Delta_2}|$ for any $|\Delta_2|\geq|\Delta_1|$. This can be useful in settings where given $J>0$ random variables $L_{\Delta_i}$ for $i=1\dots,J$, we desire positive integers $\ell_i$ for $i=1,\dots,J$ such that $\mathbb{P}[ |L_{\Delta_i}|\leq \ell_i ]\geq 1-\alpha$ for $i=1,\dots,J$. This scenario will appear in later sections where we consider models containing several change points with possibly different jump sizes. In such situations, a simple solution is to take $\ell_i=Q_{\Delta_m}(1-\alpha)$ for all $i$ where $m=\underset{1,\dots,J}{\arg\min}|\Delta_i|$, or in other words letting each $\ell_i$ be the $1-\alpha$ quantile of the $|L_{\Delta_i}|$ with the smallest parameter. Alternatively we can generate a table of quantiles for distributions $L_{\delta_1},L_{\delta_2},L_{\delta_3},\dots$ for a mesh of positive constants $\delta_1<\delta_2<\dots $ (e.g. we can let the $\delta_j=0.1j$ for $j=5,\dots,1000$), and let $\ell_i=Q_{\delta_j }(1-\alpha)$ where $\delta_j=\max\{ \delta_k:\delta_k\leq \Delta_i \}$, for $i=1,\dots,J$. 

\subsection{The Intelligent Sampling Procedure and its Properties}\label{sec:intelligentsampling}
\label{sec:singleprocedure}
\begin{enumerate}[label=(ISS\arabic*):]
\setlength{\itemindent}{.5in}
\item From the full data set of $\left(\frac{1}{N},Y_1\right),\left(\frac{2}{N},Y_2\right),...,\left(1,Y_N\right)$, take an evenly spaced subsample of approximately size $N_1=K_1N^{\gamma}$ for some $\gamma\in(0,1)$, $K_1>0$: thus, the data points are $\left(\frac{\lfloor N/N_1 \rfloor}{N},Y_{\lfloor N/N_1 \rfloor}\right)$, $\left(\frac{2\lfloor N/N_1 \rfloor}{N},Y_{2\lfloor N/N_1 \rfloor}\right)$, $\left(\frac{3\lfloor N/N_1 \rfloor}{N},Y_{3\lfloor N/N_1 \rfloor}\right)$ \dots
\item On this subsample apply least squares to obtain estimates $\left(\hat{\alpha}^{(1)},\hat{\beta}^{(1)},\hat{\tau}^{(1)}_N\right)$ for parameters $( \alpha,\beta,\tau_N)$. 
\end{enumerate}
By the results for the single change-point problem presented above, $\hat{\tau}^{(1)}_N-\tau_N$ is $O_p(N^{-\gamma})$. Therefore, if we take $w(N)=K_2N^{-\gamma+\delta}$ for some small $\delta>0$ (much smaller than $\gamma$) and any constant $K_2>0$, with probability increasing to 1, $\tau_N \in [\hat{\tau}^{(1)}_N-w(N),\hat{\tau}^{(1)}_N+w(N)]$. In other words, this provides a neighborhood around the true change point as desired; hence,
in the next stage only points within this interval will be used. 
\begin{enumerate}[label=(ISS\arabic*):]
\setlength{\itemindent}{.5in}
\setcounter{enumi}{2} 
\item Fix a small constant $\delta>0$. Consider all $i/N$ such that $i/N\in[\hat{\tau}^{(1)}_N-K_2N^{-\gamma+\delta},\hat{\tau}^{(1)}_N+K_2N^{-\gamma+\delta}]$ and $(i/N,Y_i)$ was not used in the first subsample. Denote the set of all such points as $S^{(2)}$. 
\item Fit a step function on this second subsample by minimizing
\begin{equation*}
\sum_{i/N\in S^{(2)}}\Big( Y_i-\hat{\alpha}^{(1)}1(i/N\leq d)- \hat{\beta}^{(1)}1(i/N> d)\Big)^2
\end{equation*}
with respect to $d$, and take the minimizing $d$ to be the second stage change point estimate $\hat{\tau}^{(2)}_N$.
\end{enumerate}

The next theorem establishes that the intelligent sampling estimator $\hat{\tau}^{(2)}_N$ is \emph{consistent with the same rate of convergence} as the \emph{estimator based on the full data}.

\begin{theorem}\label{thm:singlerate}
For the stump single change point model, the estimator obtained based on intelligent sampling satisfies		
$$|\hat{\tau}_N^{(2)}-\tau_N|=O_p(1/N).$$
\end{theorem}
\begin{proof}
See Section \ref{sec:proofgeneralratesingle} in Supplement Part A, where a result for a more general model is proven.
\end{proof}

To derive a clean statement of the asymptotic distribution, we introduce a slight modification to the definition of the true change point and define a new type of 'distance' function $\lambda_2:[0,1]^2\mapsto\mathbb{Z}$, as follows. First, for convenience, denote the set of $i/N$'s of the first stage subsample as 
\begin{eqnarray}
S^{(1)}:=\left\{ \frac{i}{N}:i\in\mathbb{N},\,i<N,\, i\text{ is divisible by }\lfloor N/N_1\rfloor  \right\},
\end{eqnarray}
\noindent then for any $a,b\in (0,1)$
\begin{eqnarray}
\lambda_2(a,b):=\begin{cases} \sum\limits_{i=1}^N 1\left(a<\frac{i}{N}\leq b,\, \frac{i}{N}\notin S^{(1)} \right)\qquad &\text{ if }a\leq b\,,\\
- \sum\limits_{i=1}^N 1\left(b<\frac{i}{N}\leq a,\, \frac{i}{N}\notin S^{(1)} \right) &\text{ otherwise}. 
\end{cases}
\end{eqnarray}
The modified ``distance" $\hat{\tau}_N$ is $\lambda_2(\tau_N,\hat{\tau}^{(2)}_N)$, instead of $N(\hat{\tau}^{(2)}_N-\tau_N)$, does converge weakly to a distribution as the next results establishes.
\begin{theorem}\label{thmsingledist}
For any integer $\ell$,
\begin{eqnarray}
\mathbb{P}\left[ \lambda_2\left(\tau_N,\hat{\tau}_N^{(2)}\right)=\ell \right] &\to & \mathbb{P}[L_{\Delta/\sigma}=\ell] \,.
\end{eqnarray}
\end{theorem}
\begin{proof}
See Section \ref{sec:proofsingledistgeneral} in Supplement Part A where this is established for a more general model. 
\end{proof}
{\bf Computational gains:} The results above establish that the two stage procedure can, using a subset of the full data, be asymptotically almost as precise as employing
the full data set. In practice this allows for quicker estimation of big data sets without losing precision.
The first stage uses about $N_1 \sim N^\gamma$ points to perform least squares fitting of a stump model, and this step takes $O(N^\gamma)$ computational time. The second stage applies a least-squares fit of a step function on the set $S^{(2)}$, which contains $O(N^{1-\gamma+\delta})$ points and therefore uses $O(N^{1-\gamma+\delta})$ time. 

Hence, the two stage procedure requires order $N^\gamma \vee N^{1-\gamma+\delta}$ computation time, which is minimized by setting $\gamma=1-\gamma+\delta$, or $\gamma=\frac{1+\delta}{2}$. As $\delta$ tends to 0 (any small positive value of $\delta$ yields the above asymptotic results), the optimal $\gamma$ tends to $1/2$. Therefore, one should employ $N_1=\sqrt{N}$ at the first stage and the second stage sample should be all points in the interval $[\hat{\tau}^{(1)}_N-K_2\sqrt{N},\hat{\tau}^{(1)}_N+K_2\sqrt{N}]$, minus those at the first stage, where $K_2$ ensures that this interval contains $\tau$ with an acceptable high probability $1-\alpha$.  If one knows the jump size $\Delta$, $K_2$ can be determined as the $1-\frac{\alpha}{2}$ quantile of the random variable $L_{\Delta/\sigma}$; in the realistic unknown $\Delta$ case, a lower estimate of $\Delta$ can yield a corresponding conservative value of $K_2$.

\begin{remark}\label{rem:nosingleasym}
The $\lambda_2$ distance was introduced above because it is generally not possible to derive an asymptotic distribution for $N(\hat{\tau}_N^{(2)}-\tau_N)$. 
Indeed, one can manufacture parameter settings quite easily, that produce different limit distributions along different subsequences. For a specific example, see Remark \ref{rem:nosingleasymgen} in Supplement Part A.
\end{remark}

\begin{remark}\label{rem:singlemult}
The 2-stage procedure can be extended to multiple stages. \indent In the 2-stage version, we first use a subsample of size $N^\gamma$ to find some interval $[\hat{\tau}^{(1)}-K_1N^{-\gamma+\delta_1},\hat{\tau}^{(1)}+K_1N^{-\gamma+\delta_1}]$ which contains the true value of $\tau$ with probability going to 1. However, it is possible to refrain from using all the data
points in the interval at the second stage. Instead, a 2-stage procedure can be employed as follows: take a subset of $N^\zeta$ (for some $0<\zeta< 1-\gamma+\delta_1$) points from the second stage interval and obtain an estimate $\hat{\tau}^{(2)}$ and an interval $[\hat{\tau}^{(2)}-K_2N^{-\gamma-\zeta+\delta_1 + \delta_2},\hat{\tau}^{(2)}+K_2N^{-\gamma -\zeta +\delta_1 + \delta_2}]$ (note that $\delta_1$ and $\delta_2$ can be as small as one pleases) which contains $\tau$ with probability going to 1. In the third stage, all points in the aforementioned interval (leaving aside those used in previous stages) are used to obtain the final estimate $\hat{\tau}^{(3)}$.
\newline
\newline
\indent Such a procedure will have the same rate of convergence as the one using the full data: $(\hat{\tau}^{(3)}-\tau)=O_p(1/N)$, and the same asymptotic distribution (in terms of a ``third stage distance'' similar to how $\lambda_2$ was defined) as the one and two stage procedures. In terms of computational time, the first stage takes $O(N^\gamma)$ time, the second stage $O(N^\zeta)$ time, and the final stage $O(N^{1-\gamma-\zeta+\delta_1+\delta_2})$ time, for a total of $O((N^\gamma\vee N^\zeta\vee N^{1-\gamma-\zeta+\delta_1+\delta_2}))$ time, which can reach almost $O(N^{1/3})$ time.  In general, a $K$ stage procedure, which works along the same lines can operate in almost as low as $O(N^{1/K})$ time.
\end{remark} 

\section{The Case of Multiple Change Points}\label{sec:multiplechangepointsintro}

Suppose one has access to a data set $Y_1,Y_2,\dots, Y_N$ generated according to the following model:
\begin{equation}\label{model}
Y_i=\theta_i+\varepsilon_i,\qquad i=1,2,3,...,N,
\end{equation}
where the $\theta_i$'s form a piecewise constant sequence for any fixed $N$ and the $\varepsilon_i$'s are zero-mean error terms\footnote{Specifically, we consider the triangular array of sequences $\theta_{i,N}$, which are piecewise constant in $i$. The error terms $\varepsilon_i=\varepsilon_{i,N}$ also form a triangular array, but we suppress the notation for brevity}. 
The signal is flat apart from jumps at some unknown change points $1=\tau_0<\tau_1<...<\tau_J<\tau_{J+1}=N$:  i.e. $\theta_{i_1}=\theta_{i_2}$ whenever $i_1,i_2\in (\tau_j,\tau_{j+1}]$ for some $j\in\{0,...,J\}$. The number of change points $J=J(N)$ is also unknown and needs to be estimated from the data. We impose the following basic restrictions on this model:
\begin{enumerate}[label=(M\arabic*):]
\setlength{\itemindent}{.5in}
\item there exists a constant $\bar{\theta}\in (0,\infty)$ not dependent on $N$, such that $\underset{i=1,...,N}{\max}|\theta_i|\leq \bar{\theta}$;
\item  there exists a constant $\underline{\Delta}$ not dependent on $N$, such that $\underset{i=0,...,J}{\min} \left| \theta_{\tau_{i+1}}-\theta_{\tau_i}\right|\geq \underline{\Delta}$;
\item there exists a $\Xi\in [0,1)$ and some $C>0$, such that $\delta_N:=\underset{i=0,...,J}{\min}(\tau_{i+1}-\tau_i)\geq CN^{1-\Xi}$ for all large $N$;
\item $\varepsilon_i$ for $i=1,...,N$ are independent centered subgaussian random variables, with subgaussian parameters\footnote{The subgaussian parameter of a variable $X$ is any value $\sigma>0$ such that $\mathbb{E}[\exp(sX)]\leq \exp(s^2\sigma^2/2)$ for all real values $s$} $\varepsilon_1,\dots,\varepsilon_N$ that are all bounded above by a constant $\sigma_{\max}>0$ which is not dependent on $N$.
\end{enumerate}
\begin{remark}
\label{comments-on-assumptions} 
\indent The third assumption above stipulates that the minimum gap between two consecutive stretches is bounded away from 0. This is areasonable assumption for identifying long and significantly well-separated persistent stretches in a big data setting. Condition (M4) places a restriction on the tail probability of the error terms, but still accommodates for some heteroscedastic behavior (a concern for long data sequences) by allowing noise sequence comprising of independent random variables with different distributions. 
\end{remark}

\subsection{Intelligent Sampling on Multiple Change Points}\label{sec:procedure}
The intelligent sampling procedure in the multiple change-points case works in two (or more) stages as follows: 
in the two-stage version, as in 
Section \ref{sec:single-CP}, the first stage aims to find rough estimates of the change points using a uniform subsample 
(Steps ISM1-ISM4) and the second stage produces the final estimates (Steps ISM5 and ISM6).
\begin{enumerate}[label=(ISM\arabic*):]
\setlength{\itemindent}{.5in}
\item Start with a data set $Y_1,\cdots,Y_N$ described in (\ref{model}). 
\item Take $N_1=K_1N^{\gamma}$ for some $K_1$ and $\gamma\in (\Xi,1)$ such that $N/N_1=o(\delta_N)$; for $j=1,\cdots,  N^*$ where $N^*:=\left\lfloor \frac{N}{\lfloor N/N_1\rfloor}\right\rfloor $, consider the subsample $\{Z_j\}=\{Y_{j\lfloor N/N_1 \rfloor}\}$. 
\end{enumerate}
The subsample $Z_1,Z_2,...$ can also be considered a data sequence structured as in $(\ref{model})$, and since $\delta_N>>N/N_1$, there are jumps in the signal at $\tau_j^*:=\left\lfloor \frac{\tau_j}{\lfloor N/N_1\rfloor}\right\rfloor$ for $j=1,...,J$, with corresponding minimum spacing 
\begin{equation}\label{sparsesep}
\delta_{N^*}^*:=\min_{i=1,...,J+1}|\tau_i^*-\tau_{i-1}^*|= \frac{1}{\lfloor N/N_1\rfloor} \left(\min_{i=1,\cdots,J+1}|\tau_i-\tau_{i-1}|+O(1) \right)= \left(\frac{N_1}{N}\delta_N\right)(1+o(1)) \,.
\end{equation}
\begin{enumerate}[label=(ISM\arabic*):]
\setlength{\itemindent}{.5in}
\setcounter{enumi}{2} 
\item Apply some multiple change point estimation procedure (such as binary segmentation) to the set of $Z_i$'s to obtain estimates $\hat{\tau}_1^*,...,\hat{\tau}_{\hat{J}}^*$ for the $\tau_i^*$s and $\hat{\nu}^{(1)}_0$,\dots, $\hat{\nu}^{(1)}_{\hat{J}}$ for the levels $(\nu_0,\nu_1,\dots,\nu_J)=(\theta_1,\theta_{\tau_1+1},\theta_{\tau_2+1},\dots,\theta_{\tau_J+1})$. 
\begin{itemize}
	\item the choice of the procedure does not matter so long as the estimates satisfy
	\begin{equation}\label{eq:firstconsistent}
	\mathbb{P}\left[ \hat{J}=J,\, \max_{i=1,...,J}|\hat{\tau}^*_i-\tau_i^*|\leq w^*(N^*),\,\max_{i=0,...,J}|\hat{\nu}^{(1)}_i-\nu_i|\leq \rho_N \right]\to 1
	\end{equation}
	for some sequence $w^*(N^*)$ such that $w^*(N^*)\to\infty$, $w^*(N^*)=o(\delta^*_{N^*})$ and $\rho_N\to 0$.
\end{itemize}
\end{enumerate}

\begin{enumerate}[label=(ISM\arabic*):]
\setlength{\itemindent}{.5in}
\setcounter{enumi}{3} 
\item Convert these into estimates for the $\tau_i$'s by letting $\hat{\tau}_j^{(1)}:=\hat{\tau}^*_j\lfloor N/N_1 \rfloor$ for $j=1,...,\hat{J}$. 
\begin{itemize}
	\item taking $w(N):=(w^*(N^*)+1)\lfloor N/N_1\rfloor$, expression (\ref{eq:firstconsistent}) gives
	\begin{eqnarray}\label{eq:firstconsistent2}
	\mathbb{P}\left[\hat{J}=J,\,\underset{i=1,...,J}{\max}|\hat{\tau}^{(1)}_i-\tau_i|\leq w(N),\, \max_{i=0,...,J}|\hat{\nu}^{(1)}_i-\nu_i|\leq \rho_N\right]\to 1.
	\end{eqnarray}
	\item as a consequence of conditions in (ISM3), $w(N)\to \infty$, $w(N)=o(\delta_N)$, and $w(N)>CN^{1-\gamma}$ for some constant $C$. 
\end{itemize}
\item Fix any integer $K>1$, and consider the intervals 
$\left[ \hat{\tau}_i^{(1)}-Kw(N), \hat{\tau}_i^{(1)}+Kw(N) \right]$ for $i=1,...,\hat{J}$. Denote by $S^{(2)}\left( \hat{\tau}_i^{(1)} \right)$ all integers in this interval not divisible by $\lfloor N/N_1\rfloor$.
\item For each $i=1,...,\hat{J}$, let		
\begin{equation}
\hat{\tau}^{(2)}_i=\underset{d\in S^{(2)}\left(\hat{\tau}_i^{(1)}\right)}{\arg\min} \left(\sum_{j\in S^{(2)}\left(\hat{\tau}_i^{(1)}\right)}\Big[ Y_j-(\hat{\nu}^{(1)}_{i-1}1(j<d)+\hat{\nu}^{(1)}_i1(j\geq d)) \Big]^2\right) \,. 
\end{equation}
\end{enumerate}
\begin{remark}
As with the single change point problem, a $p>2$ stage procedure can be constructed. This would involve steps (ISM1) to (ISM5), but afterwards a $p-1$ stage procedure as described in Remark \ref{rem:singlemult} for estimating single change points will be applied on every interval $\left[ \hat{\tau}_i\pm Kw(N) \right]$. 
\end{remark}
The intervals $[\hat{\tau}^{(1)}_j\pm Kw(N)]$ referred to at stage ISM4, all respectively contain $[\tau_j\pm (K-1)w(N)]$ with probability going to 1. These latter intervals have width going to $\infty$, and both of these intervals contain exactly one change point (as their widths are $O(w(N))=o(\delta_N)$). Hence, with probability $\to 1$ the multiple change point problem has simplified to $\hat{J}$ single change point problems, justifying ISM6 where a stump model is fitted inside each of $S^{(2)}(\hat{\tau}^{(1)}_j)$'s. 
\newline
\newline

\noindent 
{\bf Asymptotic behavior of the intelligent sampling based estimators:} Next, we present results on the large sample properties of the intelligent sampling estimates. 
There are several types of results we will showcase to characterize the asymptotic behavior, with the first result focusing on the rate of convergence:
\begin{theorem} \label{thm:multiorder}
	Suppose conditions (M1) to (M4) are satisfied and the first stage estimates satisfy the consistency result (\ref{eq:firstconsistent2}). Then, for any $\varepsilon>0$, there exist constants $C_1$ and $C_2$ (depending on $\epsilon$) such that
	\begin{eqnarray}
	\mathbb{P}\left[\hat{J}=J,\,\max_{k=1,...,J}|\hat{\tau}^{(2)}_k-\tau_k|\leq C_1\log(J)+C_2\right]\geq 1-\varepsilon
	\end{eqnarray}
	for all sufficiently large $N$.
\end{theorem}
\begin{proof}
	See Section \ref{sec:multiorderotherproof} of Appendix B.
\end{proof}
The rate of convergence matches that of certain estimators when applied to the full data (see the convergence properties of wild binary segmentation in Theorem 3.2 of \cite{fryzlewicz2014wild}), consistent with the proposal that the two stage estimator is on the same level of accuracy as methods using the full data set. 
\newline\newline
\indent To further characterize the asymptotic behavior of the intelligent sampling estimators for the purpose of inference, additional assumptions are required for the sake of deriving their asymptotic distribution. First, we consider a methodological issue: if the error terms around a change point are arbitrarily heteroscedastic (e.g., every error term has a distinct marginal distribution), then it is essentially impossible to estimate the distributions of the noise terms, which play a critical role in determining the asymptotic distributions of the change point estimators. Hence, we require a condition where the distribution around the change points are stable. In the sequel we assume that the error terms on either side of a change-point are i.i.d in slowly growing neighborhoods, though the distribution of errors need not be identical over the entirety of the observed data sequence. We split our results into two parts: $J$ \emph{fixed} and $J$ growing with $N$, as 
the results we establish in these two cases are somewhat differently formulated. \\\\
\noindent
{\bf $J$ fixed with $N$:} In this case, we assume the following further conditions on the model.
\begin{enumerate}[label=(M\arabic*):]
	\setlength{\itemindent}{.5in}
	\setcounter{enumi}{4}
	\item The jump sizes $\Delta_j:=\nu_j-\nu_{j-1}$ for $j=1,...,J$ are also constants not dependent on $N$.
	\newline
	\item For every $0\leq j\leq J$, the random variables $\{\varepsilon_{\tau_j+1},\dots, \varepsilon_{\tau_{j+1}}\}$ all have the same distribution as the random variable $\mathcal{E}_{j} $, where $\mathcal{E}_0,\mathcal{E}_2,\dots, \mathcal{E}_{J}$ are fixed random variables with distributions not changing with $N$.
\end{enumerate}
Under these two additional conditions, it is possible to characterize the asymptotic distribution of the $\hat{\tau}^{(2)}_j$'s. Prior to stating the result, we introduce the following distance function $\lambda_2(\cdot,\cdot)=\lambda_{2,N}(\cdot,\cdot)$, which accounts for the fact that at step (ISM5), the first stage subsample points are left out in the second subsample (and thus no first subsample point can be the final estimator):
\begin{eqnarray}\label{eq:lambda2multdef}
\lambda_2(a,b):=\begin{cases} \sum\limits_{i=1}^N 1\left(a<i\leq b)\cdot 1(i\neq k\lfloor N/N_1\rfloor\text{ for any integer } k \right)\qquad &\text{ if }a\leq b\\
- \sum\limits_{i=1}^N 1\left(b<i\leq a)\cdot 1(i\neq k\lfloor N/N_1\rfloor\text{ for any integer } k \right) &\text{ otherwise}
\end{cases}.
\end{eqnarray}
In terms of this distance function, $\hat{\tau}^{(2)}_j$'s converge as follows:
\begin{theorem}\label{thm:multidepend}
	Suppose conditions (M1) to (M6), and the consistency condition (\ref{eq:firstconsistent2}) are satisfied. Define the independent random variables $L_j$ for $1\leq j\leq J$ as
	\begin{eqnarray}
	L^*_j&:=&\underset{k\in\mathbb{Z}}{\arg\min}\; Z_j(k)\nonumber\\
	Z_j(k)&:=&\begin{cases}
	|\Delta_j|(\varepsilon^*_{j,1}+\dots+\varepsilon^*_{j,k})+k\Delta_j^2/2,\qquad & k>0\\
	0, & k=0\\
	-|\Delta_j|(\varepsilon^*_{j-1,k+1}+\dots+ \varepsilon^*_{j-1,0})+k\Delta_j^2/2, & k<0 \,,
	\end{cases}
	\end{eqnarray}
	where each $\{\varepsilon^*_{j,k}:k\in\mathbb{Z}\}$ is a set of iid random variables with the same distribution as the random variable $\mathcal{E}_j$, for $j=0,\dots,J$. 
	\newline
	\newline
	The deviations $\left\{\lambda_2\left(\tau_j,\hat{\tau}^{(2)}_j\right)\right\}_{j=1}^{\hat J}$ jointly converge to the distribution of $\left(L^*_{1},...,L^*_J\right)$. That is, for any integers $k_1,\dots,k_J$,  
	\begin{eqnarray}
	\mathbb{P}\left[\hat{J}=J,\, \lambda_2\left(\tau_j,\hat{\tau}^{(2)}_j\right)=k_j \text{ for  }1\leq j\leq J\right]\to\prod_{j=1}^J \mathbb{P}\left[L^*_j=k_j\right] \,.
	\end{eqnarray}
\end{theorem}
\begin{proof}
	See Section \ref{sec:multidependthmproof} of Supplement Part B.
\end{proof}
In practical terms the result enables statistical inference - construction of confidence regions.
\newline
\newline
\indent Note that we require a finite number of change points, which is a strong restriction from a methodological and 
applications point of view. Nevertheless, the result in Theorem \ref{thm:multidepend} is still useful, as shown through numerical experiments in Section  \ref{sec:simulations}. Next, we establish a similar result for growing number of
change points $J$.
\newline
\newline
\textbf{$J$ grows with $N$ and Gaussian Errors:} If we restrict ourselves to the case where the error terms are independent Gaussian random variables, the distribution of the intelligent sampling estimators can be characterized \emph{even if} $J\to \infty$. Specifically, assume:
\begin{enumerate}[label=(M\arabic*-Gaussian):]
	\setlength{\itemindent}{.5in}
	\setcounter{enumi}{3}
	\item $\varepsilon_1,\dots,\varepsilon_N$ are independent zero mean Gaussian random variables, with variances bounded below by $\sigma_{\min}^2$ and above by $\sigma^2_{\max}$, where $\sigma_{\min}$ and $\sigma_{\max}$ are positive constants not dependent on $N$.
\end{enumerate}
Further, also consider conditions (M1) to (M3), and (M5). Note that (M5) along with this new condition implies that we can write $\mathcal{E}_j\sim N(0,\sigma_j^2)$ for some constants $\sigma_j$s, lying between the values $\sigma_{\min}$ and $\sigma_{\max}$, a convention we will use for the remainder of the section.
\newline
\newline
\indent When the number of parameters and (the corresponding) estimates go to $\infty$ as $N$ increases, there is no fixed distribution to converge to, hence to characterize the large sample distribution, our result use probability bounds and the quantiles of a growing number of distributions. 
For all $\alpha\in(0,1)$ and positive values $\Delta,\sigma_1,\sigma_2$, let $Q_{\Delta,\sigma_1,\sigma_2}(1-\alpha)$ be the $1-\alpha$ quantile of $|L_{\Delta,\sigma_1,\sigma_2}|$, where 
\begin{eqnarray}\label{eq:gauss_walk_def}
&&L_{\Delta,\sigma_1,\sigma_2}=\underset{t\in \mathbb{Z}}{\arg\min}X_{\Delta,\sigma_1,\sigma_2}(t)\nonumber\\
&&X_{\Delta,\sigma_1,\sigma_2}(t)=\begin{cases}
t\frac{|\Delta|}{2}+\sum_{i=1}^t\varepsilon^*_{ i }\qquad & t>0\\
0 & t=0\\
\frac{|t\Delta|}{2}+\sum_{i=0}^{t-1}\varepsilon^*_{-i} \quad t<0
\end{cases}\nonumber\\
&&\text{where }\varepsilon^*_i\overset{\text{iid}}{\sim} N(0,\sigma_1^2)\text{ for }i\leq 0\text{ and } \varepsilon^*_i\overset{\text{iid}}{\sim} N(0,\sigma^2_2) \text{ for }i>0 \,.
\end{eqnarray}
One might envisage a result of the form: 
\begin{eqnarray}
\mathbb{P}\left[ \hat{J}=J;\, \left|\lambda_2\left(\tau_j,\hat{\tau}_j^{(2)}\right)\right|\leq Q_{\Delta_j,\sigma_{j-1},\sigma_j}(\sqrt[J]{1-\alpha})\text{ for all }j=1,\dots,J \right]\to 1-\alpha,
\end{eqnarray}
where $\lambda_2:=\lambda_{2,N}$ was defined in (\ref{eq:lambda2multdef}). However, as the distribution of $\underset{t\in\mathbb{Z}}{\arg\min}\;X_{\Delta,\sigma_1,\sigma_2}(t)$ is discrete, it is generally not possible to get the probabilities exactly equal to $\sqrt[J]{1-\alpha}$. Instead, we derive the following result: 
\begin{theorem}\label{thm:increasingJasymprotics}
	For any $\alpha \in (0,1)$ and positive values $\Delta,\sigma_1,\sigma_2$, define: 
	\begin{eqnarray}
	P_{\Delta,\sigma_1,\sigma_2}(1-\alpha)=\mathbb{P}\left[ \left|\underset{t\in\mathbb{Z}}{\arg\min}\;X_{\Delta,\sigma_1,\sigma_2}(t)\right|\leq Q_{\Delta,\sigma_1,\sigma_2}(1-\alpha) \right] \,.
	\end{eqnarray}
	Suppose that conditions (M1) to (M3) are satisfied, conditions (M4-Gaussian) and (M6) are satisfied, and the first stage estimates satisfy (\ref{eq:firstconsistent2}) with a $\rho_N$ such that $J\rho_N\to 0$. Then
	\begin{eqnarray}\label{eq:increasingJasymprotics1}
	&&\mathbb{P}\left[  \hat{J}=J;\, \left|\lambda_2\left( \tau_j,\hat{\tau}^{(2)}_j \right)\right|\leq Q_{|\Delta_j|-2\rho_N,\sigma_{j-1},\sigma_j}(\sqrt[J]{1-\alpha})\text{ for all }j=1\dots, J \right]\nonumber\\
	&= &\left(\prod_{j=1}^J P_{|\Delta_j|-2\rho_N,\sigma_{j-1},\sigma_j}(\sqrt[J]{1-\alpha})\right)+o(1)\,,
	\end{eqnarray}
	and 
	\begin{eqnarray}\label{eq:increasingJasymprotics2}
	&&\mathbb{P}\left[  \hat{J}=J;\, \left|\lambda_2\left( \tau_j,\hat{\tau}^{(2)}_j \right)\right|\leq Q_{|\Delta_j|+2\rho_N,\sigma_{j-1},\sigma_j}(\sqrt[J]{1-\alpha})\text{ for all }j=1\dots, J \right]\nonumber\\
	&= &\left(\prod_{j=1}^JP_{|\Delta_j|+2\rho_N,\sigma_{j-1},\sigma_j}(\sqrt[J]{1-\alpha})\right)+o(1) \,.
	\end{eqnarray}
\end{theorem}
\begin{proof}
See Section \ref{sec:increasingJasymproticsproof} for a long and detailed proof of Supplement Part B. A shorter proof sketch which highlights the key steps is also provided in Section \ref{proof-sketch-Theorem-5}. 
\end{proof}
\begin{remark}
The proof of this result relies on probability bounds for the $L_{\Delta,\sigma_1,\sigma_2}$ distributions, specifically bounds on their tail probabilities, as well as bounds on the probability that $L_{\Delta,\sigma_1,\sigma_2}\neq L_{\Delta',\sigma_1',\sigma_2'}$ when $(\Delta,\sigma_1,\sigma_2)\neq (\Delta',\sigma_1',\sigma_2')$. Using the Gaussianity assumption, such probability bounds can be methodically derived.
The additional condition requiring $J\rho_N\to 0$ stems from the details of the proof, but this is satisfied by several existing change point methods, including binary segmentation that will be showcased later on in the paper. 
\end{remark}

\section{Practical Implementation of Intelligent Sampling}
\label{sec:prac-imple} 
In the previous section, we laid out a generic scheme for intelligent sampling which requires the use of a multiple change point estimation procedure on a sparse subsample of the \emph{data-sequence}. Recall that any procedure that satisfies (\ref{eq:firstconsistent}) can be used here. A variety of such procedures have been explored by various authors (see, e.g., 
\cite{venkatraman1992consistency}, \cite{fryzlewicz2014wild}, \cite{frick2014multiscale}, and \cite{bai1998estimating}), and therefore a number of options are available. For the sake of concreteness, we pursue intelligent sampling with binary segmentation (henceforth abbreviated to ``BinSeg") employed at Step (ISM3). One main advantage of BinSeg is its computational scaling at an optimal rate of $O(N^*\log(N^*))$ when applied to a data sequence of length $N^*$, and in addition it has the upside of being easy to program, which accounts for its popularity in the change point literature. We later discuss other potential options. 
\newline
\newline
\indent However, there are some issues involved in applying the results of BinSeg to our setting. First, BinSeg does not directly provide the signal estimators that are required in (\ref{eq:firstconsistent}). We address this issue in Section \ref{sec:binsegdescription}, where we establish that given certain consistency conditions on the change points, which are satisfied by BinSeg, consistent signal estimators can be obtained by averaging the data between change point estimates. Second, there is no established method for constructing explicit confidence intervals for the actual change points using BinSeg, as existing results give orders of convergence, but no asymptotic distributions or probability bounds with \emph{explicit} constants. However, to implement intelligent sampling, one wants to have high-probability intervals around the initial change-point estimates on which to do the second round sampling, which requires calibration in terms of the coverage probability. To this end, in Section \ref{sec:refitting} we describe a procedure to be performed after applying BinSeg on the first stage subsample: the extra steps provide us with explicit confidence intervals, while not being slower than BinSeg in terms of order of computational time.
\newline
\newline
\indent Before we begin, we remind the reader that this section deals with the first stage subsample and not the whole data set. We will henceforth use the $\star$ notation in connection with the quantities involved at the first stage. 
Therefore, we let $\nu^*_j:=\nu_j$ for $j=0,\dots,J$ and $\rho^*_{N^*}:=\rho_N$, and referring back to notation used in step (ISM2) and (ISM3), we consider the sub-dataset $Z_1,\dots,Z_{N^*}$ as a multiple change point model, with change points $\tau_j^*$'s and levels $\nu^*_j$'s, following conditions (M1) to (M4) for all large $N$ (as a consequence of $Y_1,\dots, Y_N$ satisfying conditions (M1) to (M4)). Using this notation, (\ref{eq:firstconsistent}) translates to the requirement that a change point estimation scheme applied upon $Z_1,\dots,Z_{N^*}$ w procures estimates $\hat{\tau}_j^*$'s  and $ \hat{\nu}^*_j$'s (equal to $\hat{\nu}^{(1)}_j$'s in (ISM3)) such that
\begin{eqnarray}\label{eq:firstconsistentsparse}
\mathbb{P}\left[ \hat{J}=J,\,\max_{j=1,\dots J}\left| \hat{\tau}_j^*-\tau_j^* \right|\leq w^*(N^*),\, \max_{j=0,\dots,J}\left|\hat{\nu}^*_j-\nu_j^* \right|\leq \rho^*_{N^*}\right]\to 1,
\end{eqnarray}
for some sequences $w^*(N^*)$ and $\rho^*_{N^*}$ such that $w^*(N^*)\to \infty$, $w^*(N^*)=o(\delta^*_{N^*})$, and $\rho^*_{N^*}\to 0$ as $N^*\to \infty$. We, subsequently, refer to this latter condition.

\subsection{A Brief Description of Binary Segmentation}\label{sec:binsegdescription}	
Consider the model given in (\ref{model}). For any positive integers
$1\leq s\leq b<e\leq N^*$, let $n=e-s+1$ and define the Cumulative Sum (CUSUM) statistic at $b$ with endpoints $(s,e)$ as 
\begin{equation*}
\bar{Z}_{s,e}^b=\sqrt{\frac{e-b}{n(b-s+1)}}\sum_{t=s}^bZ_t-\sqrt{\frac{b-s+1}{n(e-b)}}\sum_{t=b+1}^e Z_t.
\end{equation*}
Binary segmentation is performed by iteratively maximizing the CUSUM statistics over the segment between change point estimates, 
accepting a new change point if the maximum passes a threshold parameter $\zeta_{N^*}$. Specifically,
\begin{enumerate}
\item Fix a threshold value $\zeta_{N^*}$ and initialize the segment set $SS=\{ (1,N^*) \}$ and the change point estimate set $\underline{\hat{\tau}}^*=\emptyset$.
\item Pick any ordered pair $(s,e)\in SS$, remove it from $SS$ (update $SS$ by $SS\leftarrow SS-\{ (s,e) \}$). If $s\geq e$ then skip to step 5, otherwise continue to step 3.
\item Find the argmax and max of the CUSUM statistic over the chosen $(s,e)$ from the previous step: $b_0=\underset{b\in \{s,...,e-1\}}{\arg\max}|\bar{Z}^b_{s,e}|$ and $|\bar{Z}^{b_0}_{s,e}|$
\item If $|\bar{Z}^{b_0}_{s,e}|\geq \zeta_{N^*}$, then add $b_0$ to the list of change point estimates (add $b_0$ to $\underline{\hat{\tau}}^*$), and add ordered pairs $(s,b_0)$ and $(b_0+1,e)$ to $SS$, otherwise skip to step 5.
\item Repeat steps 2-4 until $SS$ contains no elements.
\end{enumerate}
\begin{remark} 
For our model, this algorithm provides consistent estimates of {\em both} the location of the change points and the corresponding levels to the left and right of them, given certain assumptions.
Specifically, consistency results for BinSeg require settings where the minimal separation distance between change points grows faster than the length of the data sequence raised to an appropriate power. Theorem 3.1 of \cite{fryzlewicz2014wild} presents a concrete result in this direction, where this particular power -- $\Theta$ in their notation (which is $1- \Xi$ in our notation) -- is restricted to be strictly larger than $3/4$. The same spacing condition appears in later work by the same author on BinSeg in high dimensional settings \cite{cho2015multiple}. However, there appears to be a caveat. A corrigendum to \cite{cho2015multiple} released by the author \cite{chofrycorrection} 
shows that this spacing condition does not ensure consistency of BinSeg; rather, a stronger spacing condition, $\Theta > 6/7$, is needed. From our correspondence with the authors, there is strong reason to believe that the $3/4$ in \cite{fryzlewicz2014wild} \emph{should also change to} $6/7$, and accordingly, in the sequel where we focus on a BinSeg based approach, we restrict ourselves to this more stringent regime, to be conservative. 
\end{remark}  
We require the following additional assumptions:
\begin{itemize}
	\setlength{\itemindent}{.5in}
	\item[(M4 (BinSeg))] The error terms of the data sequence are i.i.d. $N(0,\sigma^2)$, where $\sigma$ is a positive constant not dependent on $N$.
	\item[(M7 (BinSeg))] $\Xi$ (from condition (M3)) is further restricted by $\Xi\in [0,1/7)\,,$
	\item[(M8 (BinSeg))] $N_1$, from step (ISM2), is chosen so that $N_1=K_1N^\gamma$ for some $K_1>0$ and $\gamma>7\Xi$.
\end{itemize}

Condition (M7 (BinSeg)) ensures that BinSeg will be consistent on some subsample of the data sequence, for if the condition was not satisfied and the minimum spacing $\delta_N$ grows slower than $N^{6/7}$, then established results on BinSeg (see Theorem \ref{frythm}) could not guarantee consistency on any subsample of the (or even the entire) data sequence. The latter of the above conditions implies selecting a large enough subsample so that Theorem \ref{frythm} becomes applicable.
When (M8 (BinSeg)) is satisfied, the first stage subsample would have size $N^*=(K_1+o(1))N^\gamma$ with minimal change point separation of $\delta^*_{N^*}=(N_1/N+o(1))\delta_N=(C+o(1))(N^{*})^{1-\Xi/\gamma}$ for some positive constant $C$. Finally,(M4 (BinSeg)) imposes a more restrictive structure on the error terms in order to satisfy error term assumptions on established results. These conditions, when taken together, lead to the following result:
\begin{theorem}\label{frythm}
Suppose that conditions (M1) to (M3), (M4 (BinSeg)), (M7 (BinSeg)), and (M8 (BinSeg)) are satisfied, with the tuning parameter $\zeta_{N^{\star}}$ chosen appropriately so that 
\begin{itemize}
	\item if $\Xi/\gamma>0$ then $\zeta_{N^*}=c_1(N^{*})^{\xi}$ where $\xi\in (\Xi/\gamma,1/2-\Xi/\gamma)$ and $c_1>0$
	\item if $\Xi/\gamma=0$ then $c_2\left(\log(N^*)\right)^p\leq \zeta_{N^*}\leq c_3(N^*)^\xi$ where $p>1/2, \xi<1/2$, and $c_2,c_3>0\,.$
\end{itemize}
Define $E_{N^*}=\left(\frac{N^*}{\delta^*_{N^*}}\right)^2\log(N^*)$. Then, there exist positive constants $C, C_1$ such that 
\begin{equation}\label{eq:binsegprob}
\mathbb{P}\Big[ \hat{J}=J;\quad \max_{i=1,...,J} |\hat{\tau}^*_i-\tau^*_i|\leq CE_{N^*} \Big]\geq 1-C_1/N^* \,.
\end{equation}
\end{theorem}
\begin{remark}
\label{rem:frythm}
The above theorem is adapted from Theorem 3.1 of \cite{fryzlewicz2014wild} which applies to the more general setting where $\underline{\Delta}$, the minimum signal jump, can decrease to 0 as $N\to\infty$. There are some methodological issues with the above result. First, there is no easy recipe for determining $C$ explicitly, which is addressed in the subsequent section. Second, this result does not give an explicit value for the tuning parameter $\zeta_N$, but there do exist methods such as model selection using a Schwarz information criterion (see \cite{fryzlewicz2014wild}), plus software to perform this within the R \texttt{changepoint} package.
\end{remark}

Next, we introduce estimators for the signals $\hat{\nu}_j^{\star} :=\mathbb{E}[Z_{\tau_j^{\star}+1}]$, for $j=0,\dots,J$. Intuitively, they can be estimated by the average of data points between change points estimates:
\begin{equation}\label{eq:signalestdef}
\hat{\nu}^*_j=\frac{1}{\hat{\tau}^*_{j+1}-\hat{\tau}^*_j}\left(\sum_{\hat{\tau}^*_j<i\leq \hat{\tau}^*_{j+1}}Z_i\right)\qquad \text{ for }j=0,...,\hat{J}
\end{equation}
with the convention of $\hat{\tau}^*_0:=0$ and $\hat{\tau}^*_{\hat{J}+1}:=N^*$. These estimators are consistent: 
\begin{lemma}\label{lem:binsegsignalconsistent}
Suppose conditions (M1) to (M3), (M4 (BinSeg)), (M6 (BinSeg)), and (M7 (BinSeg)) are satisfied, the $\hat{\tau}^*_i$'s are the BinSeg estimators, and $\hat{\nu}^*_i$'s are the signal estimators defined in (\ref{eq:signalestdef}). Then there exists a sequence $\rho^*_{N^*}\to 0$ such that $J\rho^*_{N^*}\to 0$ and 
\begin{eqnarray}\label{eq:binseg1stconsistent}
\mathbb{P}\left[\hat{J}=J;\quad \max_{i=1,...,J}|\hat{\tau}^*_i-\tau^*_i|\leq C E_{N^*};\quad \max_{i=0,...,J}|\hat{\nu}^*_i-\nu^*_i|\leq \rho^*_{N^*}\right]\to 1
\end{eqnarray}
as $N^*\to\infty$. 
\end{lemma}
\begin{proof}
See Section \ref{sec:binsegsignalconsistentproof} in Supplement Part B.
\end{proof}
By setting $\rho^*_{N^*}=\rho_N$ and $CE_{N^*}=w^*(N^*)$, condition (\ref{eq:firstconsistent}), will be satisfied for a $\rho_N$ satisfying $J\rho_N\to 0$, which meets all requirement of Theorems \ref{thm:increasingJasymprotics}.

\subsection{Calibration of intervals used in Stage 2 of Intelligent Sampling}\label{sec:refitting}

\indent Constructing confidence intervals based on Theorem \ref{frythm} would require putting a value on $CE_{N^*}=C(N^*/\delta^*_{N^*})^2\log(N^*)$ from (\ref{eq:binsegprob}). An estimate of $\delta^*_{N^*}$ can be obtained from the minimum difference of consecutive $\hat{\tau}^*_j$'s, but an explicit expression
for $C$ is unavailable, and the existing literature on binary segmentation does not appear to provide such an explicit expression. To address this issue, we introduce a calibration method which allows the construction of confidence intervals with explicitly calculable width around the first stage estimates $\hat{\tau}^{(1)}_j$'s: the idea is to fit stump models on data with indices $[\hat{\tau}^{(1)}_{j-1}+1,\hat{\tau}^{(1)}_{j+1}]$, as each of these stretches forms a stump model with probability going to 1. 
\newline
\newline
Consider starting right after step (ISM4) (e.g., Figure \ref{fig:calibration1}), where we have rough estimates $\hat{\tau}_i^*$'s of the change points (with respect to the $\{Z_i\}$ sequence) and $\hat{\nu}^{(1)}_i$'s of the signals, obtained from the $N^*$ sized subsample $\{Z_i\}$.

\begin{minipage}[t]{.48\textwidth}
\begin{center}
	\begin{overpic}[scale=0.27]{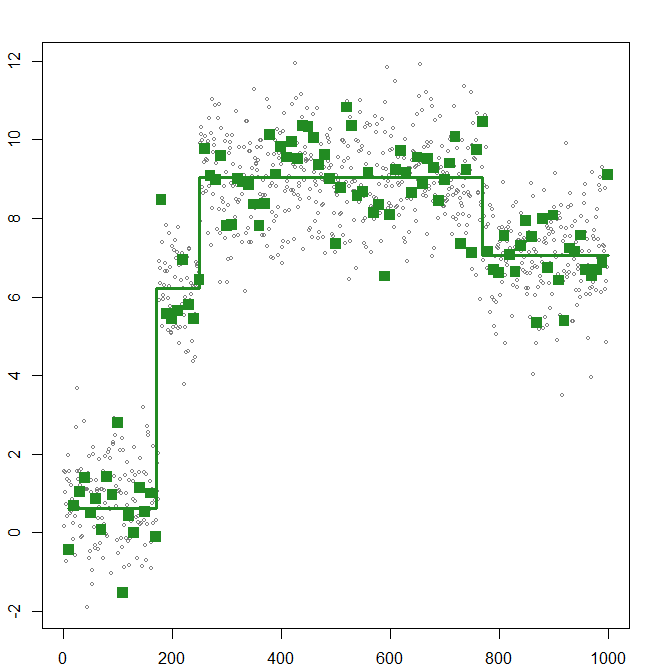}			\end{overpic}
\end{center}
\begin{figure}[H]
	\caption{Green points are $Z_i$'s, solid green line is the BinSeg estimate.}
	\label{fig:calibration1}
\end{figure}
\end{minipage}
\begin{minipage}[t]{0.04\textwidth}
~~	
\end{minipage}
\begin{minipage}[t]{0.48\textwidth}
\begin{center}
	\begin{overpic}[scale=0.27]{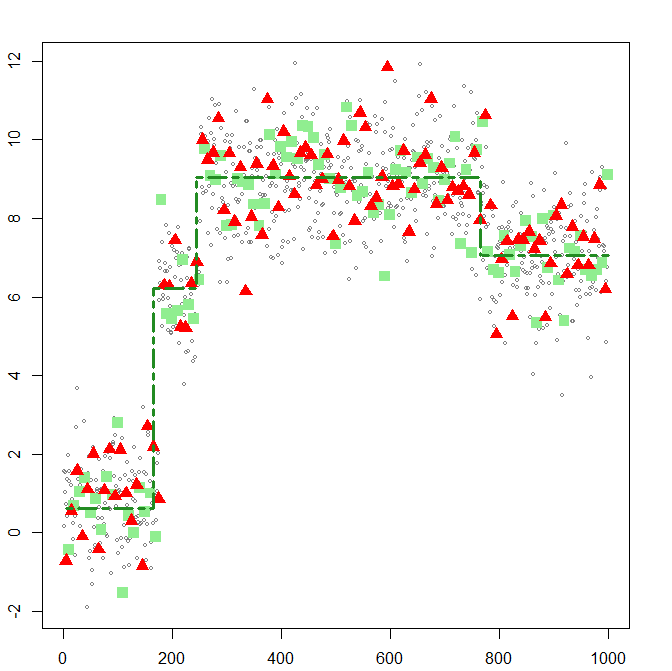}			\end{overpic}
\end{center}
\begin{figure}[H]
	\caption{$Z_i$'s are light green points, BinSeg estimates as dashed green line, $V_i$'s as red points.}
	\label{fig:calibration2}
\end{figure}
\end{minipage}

\begin{minipage}[t]{.48\textwidth}
\centering
\begin{overpic}[scale=0.27]{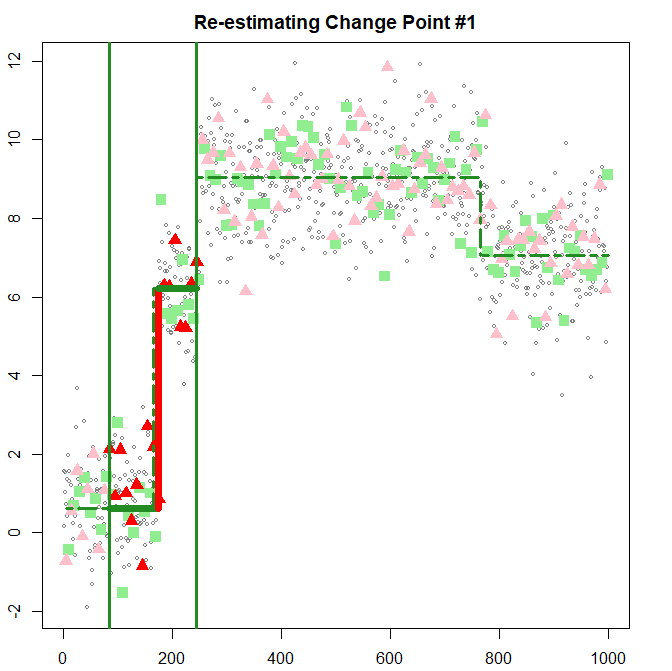}			\end{overpic}
\end{minipage}
\begin{minipage}[t]{0.04\textwidth}
~~	
\end{minipage}
\begin{minipage}[t]{0.48\textwidth}
\centering
\begin{overpic}[scale=0.27]{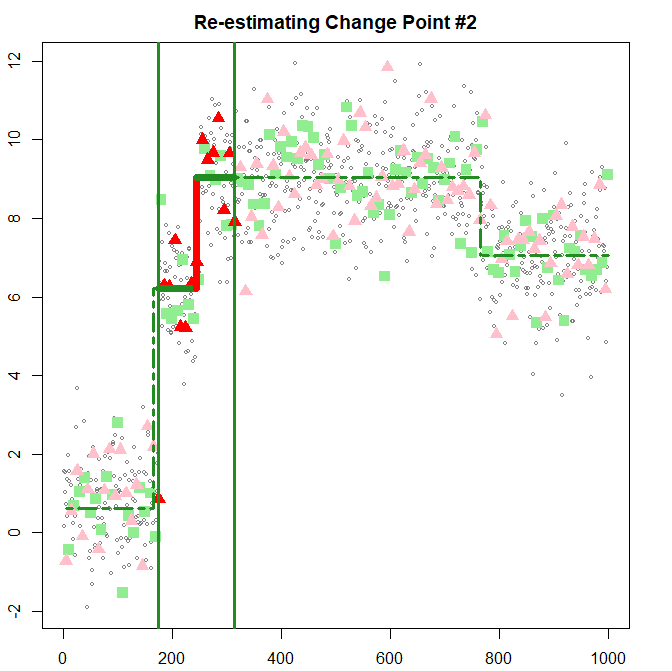}			\end{overpic}
\end{minipage}
\begin{figure}[H]
\caption{Re-estimation of the first (\textbf{Left Panel}) and second (\textbf{Right Panel}) detected change points (similar procedure for third estimated change point not shown). Solid green and solid red lines denote stump estimates using $V_i$'s from $\{V_k:\hat{\tau}^*_{j}-\hat{d}_j\leq k\leq \hat{\tau}^*_{j}+ \hat{d}_j\}$ intervals.}
\label{fig:calibration3}
\end{figure}

We then pick a different subsample $\{V_i\}$ of equal size to the $\{Z_i\}$ subsample and consider the $\hat{\tau}^*_i$'s and $\hat{\nu}^{(1)}_i$'s as estimates for the parameters of this data sequence (e.g. Figure \ref{fig:calibration2}).
For each $j$, fit a one-parameter stump model $f_j^{(d)}(k)=\hat{\nu}^{(1)}_{j-1}1(k\leq d)+\hat{\nu}^{(1)}_{j}1(k>d)$ (here $d$ is the discontinuity parameter) to the subset of $\{V_k:\hat{\tau}^*_{j-1}+1\leq k\leq \hat{\tau}^*_{j+1}\}$ given by $\{V_k:\hat{\tau}^*_{j}-\hat{d}_j\leq k\leq \hat{\tau}^*_{j}+\hat{d}_j\}$ where $\hat{d}_j=\min\{ (\hat{\tau}^*_j-\hat{\tau}^*_{j-1}),(\hat{\tau}^*_{j+1}-\hat{\tau}^*_{j}) \}$, to get an updated least squares estimate of $\tau_j$ (e.g. Figure \ref{fig:calibration3})\footnote{This subset is used instead of the full interval to avoid situations where $\frac{\tau_j-\tau_{j-1}}{\tau_{j+1}-\tau_j}\to$ 0 or $\infty$, which makes matters easier for theoretical derivations.}.
\newline
\newline
Formally, the calibration steps are:
\begin{enumerate}[label=(ISM4-\arabic*):]
\setlength{\itemindent}{.5in}
\item pick a positive integer $k_N$ less than $\lfloor N/N_1\rfloor$ \footnote{A good pick is $k_N=\left\lfloor\frac{\lfloor N/N_1\rfloor}{2}\right\rfloor$}, take a subsample $\{V_i\}$ from the dataset of $\{Y_i\}$ which is the same size as the $\{Z_i\}$ subsample, by letting $V_i=Y_{i\lfloor N/N_1\rfloor -k_N}$ for all $i$. 
\end{enumerate}
The $\{V_i\}$ subsample also conforms to the model given in (\ref{model}), with change points $\tau^{**}_i=\max\{j\in\mathbb{N}: j\lfloor N/N_1\rfloor -k_N\leq \tau_i\}$ and minimum spacing $\delta_{N^*}^{**}:= \min_k(\tau_{k+1}^{**}-\tau_k^{**})$ which satisfies $|\delta_{N^*}^{**}-\delta_{N^*}^*|\leq 1$.
\begin{enumerate}[label=(ISM4-\arabic*):]
\setlength{\itemindent}{.5in}
\setcounter{enumi}{1}
\item For each $i=1,...,J$, consider the estimates $\hat{\tau}_i^*$ (obtained from the $\{Z_j\}$ subsample at step (ISM3)) as estimators for $\tau^{**}_i$. From (\ref{eq:firstconsistent}) it is possible to derive that 
\begin{equation}\label{eq:reconsistent}
\mathbb{P}\left[ \hat{J}=J,\, \max_{i=1,...,J}|\hat{\tau}^*_i-\tau_i^{**}|\leq w^*(N^*)+1,\,\max_{i=0,...,J}|\hat{\nu}^{(1)}_i-\nu_i|\leq \rho_N  \right]\to 1
\end{equation}
where $w^*(N)+1\to\infty$ and $(w^*(N^*)+1)/\delta^{**}_{N^*}\to 0 \,.$
\item For each $i=1,...,\hat{J}$, define $\hat{d}_i=\min\{ (\hat{\tau}_{i+1}^*-\hat{\tau}_{i}^*),(\hat{\tau}_{i}^*-\hat{\tau}_{i-1}^*) \}$ (with $\hat{\tau}_{0}^*=0$ and $\hat{\tau}_{\hat{J}+1}^*=\lfloor N/\lfloor N/N_1\rfloor\rfloor$), and re-estimate the change points by letting
\begin{eqnarray}
\hat{\tau}^{re}_i:=\underset{t:|t-\hat{\tau}_{i}^*|< \hat{d}_i}{\arg\min}\left[\sum_{j: |j-\hat{\tau}_{i}^*|<\hat{d}_i}\big( V_i-\hat{\nu}^{(1)}_{i-1}1(j\leq t)-\hat{\nu}^{(1)}_i1(j>t) \big)^2\right]
\end{eqnarray}
for $i=1,...,\hat{J}\,.$
\item To translate the $\hat{\tau}^{re}_i$'s (change point estimates for the subsample for the $\{V_j\}$'s) into estimates for $\tau_1,\dots,\tau_J$ (change points for the full data set), set the first stage change point estimators as $\hat{\tau}^{(1)}_i:=\hat{\tau}^{re}_i\lfloor N/N_1\rfloor-k_N\,.$ 
\end{enumerate}
\begin{remark}
It is important to note that although the above steps are presented in the context of using BinSeg in the first stage, in practice they can be used with many other procedures. Note that any change point detection procedure that provides consistent estimates can be used in the first stage. More broadly, these steps can be used outside of the intelligent sampling framework. Given a consistent change point estimation scheme \emph{for which a method to construct explicit confidence intervals is not known}, one could split the data in two subsamples, the odd points (first, third, fifth, etc data points) and the even points. The aforementioned estimation scheme could be applied to the odd points, and the subsequent steps (ISM4-1) to (ISM4-4) could be applied  to the even points. A result similar to that of Theorem \ref{thm:reconsistentineq}, presented below, could then be used to construct confidence intervals.
\end{remark} 

Using this procedure, we can express the deviations of the $\hat{\tau}^{re}$ estimators using the quantiles $Q_{\Delta,\sigma_1,\sigma_2}$ that were defined in the paragraph before (\ref{eq:gauss_walk_def}).
\begin{theorem}\label{thm:reconsistentineq}
Suppose conditions (M1) to (M3), and (M4 (BinSeg)) are satisfied, and the estimation method used in step (ISM3) satisfies (\ref{eq:firstconsistent}), and the pertinent $\rho_N$ appearing in (\ref{eq:firstconsistent}) also satisfies $J\rho_N\to 0$. For any sequence\footnote{The sequence could be a sequence that converges to 0 or not converge to 0, it only needs to stay between 0 and 1, and satisfy the bound $\alpha\geq cN^{-\eta}$.} $\alpha_N$ between 0 and 1 such that $\alpha_N\geq CN^{-\eta}$ for some positive $C$ and $\eta$, we have
\begin{eqnarray}\label{eq:reconsistentineq}
\mathbb{P}\left[ \hat{J}=J,|\hat{\tau}_j^{re}-\tau_j^{**}|\leq Q_{\Delta_j,\sigma,\sigma}\left(1-\frac{\alpha_N}{J}\right)  \text{ for }j=1,\dots,J \right]\geq 1-\alpha_N+o(\alpha_N)\nonumber\\
\end{eqnarray}
\end{theorem}
\begin{proof}
See Section \ref{sec:reconsistentproof} of Supplement Part B.
\end{proof}

\begin{remark}
Similar to Theorem \ref{thm:increasingJasymprotics}, the condition $J\rho_N\to 0 $ is required here. This is because the proofs of both results are similar in structure. For intelligent sampling with BinSeg used in stage 1, we re-iterate that $J\rho_N\to 0$ is automatically satisfied.
\end{remark}
The practical implication of this result for implementing intelligent sampling is that we can obtain explicitly calculable simultaneous confidence intervals for the change-points. 
The intervals $[\hat{\tau}^{re}_j\pm Q_{\Delta_j,\sigma,\sigma}\left(1-\frac{\alpha_N}{J}\right)]$ for $j=1,\cdots,J$  capture the sparse scale change points $\tau^{**}_j$ for $j=1,\cdots,J$ with probability approaching 1 if we choose some $\alpha_N\to 0$. Converting back to the original scale, the intervals $[\hat{\tau}^{(1)}_j\pm (Q_{\Delta_j,\sigma,\sigma}\left(1-\frac{\alpha_N}{J}\right)+1)\lfloor N/N_1\rfloor]$ for $j=1,\cdots,\hat{J}$ have the properties that they simultaneously capture $\tau_1,\cdots,\tau_J$ with probability approaching 1.\footnote{In practice, the $\Delta_j/\sigma$'s are estimated from data.} The second stage samples are then picked as points within each of these intervals that are not the form $i\lfloor N/N_1 \rfloor$ or $i\lfloor N/N_1 \rfloor-k$ for some integer $i$ (since the first stage points cannot be used at steps (ISM5) and (ISM6)).

\subsection{Computational Considerations}\label{sec:BSWBS1} 
The order of computational time is higher than in the single change point case, owing to the fact that the procedure can involve a growing number of data intervals at the second stage. For the sake of our analysis, we make the simplifying assumption that $\delta_N/N^{1-\Xi}\to C_1$ and $J(N)/N^\Lambda\to C_2$ for some $\Lambda\in[0,\Xi]$ and some positive constants $C_1,C_2$. As a reminder, for intelligent sampling with BinSeg used in stage 1, conditions (M6 (BinSeg)) and (M7 (BinSeg)) automatically impose the condition that $\Lambda\leq\Xi<1/7$. 
\newline
\newline
\indent Under the previous assumptions, we consider performing our two stage procedure, with the first stage being:
\begin{enumerate}
	\item apply BinSeg on an evenly spaced subsample with $N_1=CN^\gamma$ for some $\gamma \in (0,1]$, $C>0$,
	\item use the BinSeg results to apply steps (ISM4-1) to (ISM4-4), obtaining confidence intervals,
\end{enumerate}
and then the second stage of the procedure following through steps (ISM5) and (ISM6). We state here (full details in Section \ref{sec:time_order_longer} of the supplement) that the total computational time of these two stages will be $O( N^{\gamma_{\min}\vee (1-\gamma_{\min}+\Lambda)  }\log(N) )$ where
\begin{eqnarray}
\gamma_{\min} = \max\left\{ \frac{1+\Lambda}{2},7\Xi+\eta \right\}
\end{eqnarray}
where $\eta$ is any small positive value. We note that in the special case where $\Lambda=0$ and $\Xi=0$, which translates to finitely many ($J\nrightarrow \infty$) change points that are roughly spaced evenly apart, we would have $\gamma_{\min} =1/2$ and the total computational time would be $O(\sqrt{N}\log(N))$. This time scaling would be the shortest time order for the two stage procedure.
\newline
\newline
\indent To achieve lower computational time order, we can also consider multiple ($>2$) stage procedures. A three stage procedure would have the same first stage (BinSeg and re-estimate on separate subsamples); the second stage estimate would estimate a stump model within each confidence interval from stage 1, but using a subsample from each interval instead of the whole interval of data, to create another set of confidence intervals; at stage 3 stump models would be fitted with all the points within each confidence interval created in stage 2. A four and five stage procedure can be similarly designed. In terms of running time, a $k$-stage procedure can, in principle, run in $O(\sqrt[k]{N}\log(N))$ time in certain cases, but for methodological concerns and for illustrative purposes we focused on the two stage procedure. Further details and discussions regarding more stages can be found in Section \ref{sec:compmethod} of the Supplement.

\section{Performance Evaluation of Intelligent Sampling: Simulation Results}\label{sec:simulations}

We next present a series of simulation results. Part of our computational work illustrates the validity of our theoretical work, such as the rate of convergence, the lower than $O(N\log(N))$ computational running time, and the asymptotic distributions. We also present some other simulations to analyze the practical benefits of our procedure. All simulations in this section were performed on a server with 24 Intel Xeon ES-2620 2.00 GHz CPUs, with a total RAM of 64 GB.

\subsection{Rate of Convergence and Time}\label{sec:ratesimu}
  Our first showcased set of simulation aims to confirm the predicted scaling of the maximal deviation (no faster than $\log(J)$) and computational time ($O(\sqrt{N}\log(N))$ in the best cases). We simulate a sequence of data sequences with increasing $N$, and to show that the $\sqrt{N}\log(N)$ time scaling is reachable, we set the number of change points to be growing slowly and the separation between consecutive change points to be somewhat close to being equal.
\newline
\newline
We generated sequences of length $N$ varying from $10^5$ to $10^{7.5}$, evenly on the log scale, with the number of change points being $J\approx\log_{10}(N)^2$. The change point location and the signal levels were randomly generated:
\begin{itemize}
\item The spacings $(\tau_1,\tau_2-\tau_1,\tau_3-\tau_2,....,N-\tau_J)$ were generated as the sum of the constant $\frac{N}{1.5J}$ and the rounded values of 
$\left(N-\frac{(J+1)N}{1.5J}\right)\cdot (U_{(1)},U_{(2)}-U_{(1)},\dots ,U_{(J)}-U_{(J-1)},1-U_{(J)})  $ where $U_{(1)},\dots, U_{(J)}$ are the ordered variables from $J$ iid uniform $[0,1]$ variables.
\item The signals were generated as a Markov chain with $\nu_0$ initialized as 0, and iteratively, given $\nu_i$, $\nu_{i+1}$ is generated from a density proportional to $f(x)=\exp(-0.3(x-\nu_i-\underline{\Delta}))1(\nu_i+\underline{\Delta}\leq x\leq M)+\exp(0.3(x+\nu_i-\underline{\Delta}))1(\nu_i-\underline{\Delta}\geq x\geq -M)$, where $M$ was taken to be 10 and $\underline{\Delta}$ was taken to be 1.
\end{itemize}
For each of 10 values of $N$, 50 configurations of change points and signals were generated, and on each of those configurations 40 datasets with iid $N(0,1)$ error terms simulated. At stage 1 we performed BinSeg and the steps outlined in Section \ref{sec:refitting}, along with some ad-hoc additional steps to increase accuracy. In detail, these steps are
\begin{enumerate}
	\item an even subsample was taken with $N_1=50\sqrt{N}$, and BinSeg was performed on this dataset with tuning parameter $\zeta_{N}=N_1^{0.2}$ (which is within the range described in Theorem \ref{frythm})
	\item remove points too close together by removing every $\hat{\tau}_i^*$ where $\hat{\tau}_i^*-\hat{\tau}_{i-1}^*\leq 15$
	\item estimate for the signals and remove every $\hat{\tau}_i^{*}$ where the estimated signal to the left and the right have a difference less than 0.5
	\item perform steps (ISM4-1) to (ISM4-4)
	\item repeat steps 2 and 3
	\item construct confidence intervals using the $L_{\hat{\Delta}_j,\hat{\sigma},\hat{\sigma}}$ distributions and then proceed with steps (ISM5) and (ISM6)
\end{enumerate}
To gauge the rate of convergence, the maximum deviations $\underset{j=1,...,J}{\max}|\lambda_2(\tau_j,\hat{\tau}^{(2)}_j)|$ were recorded for instances where $\hat{J}=J$ (which occurs over 99.5\% of the time for all $N$). We also compared the mean computational times of intelligent sampling vs the mean time using BinSeg on the whole data sequence, where the latter was averaged over 100 iterations with similarly generated data (50 configurations of randomly generated change points and signals, 2 runs each configuration). The results are depicted in Figure \ref{fig:convergence-rate-mcp} and are in
accordance with the theoretical investigations: the quantiles of $\max\,\lambda_2\left(\tau_j,\hat{\tau}_j^{(2)}\right)$ scale sub-linearly with $\log(J)$ (which is $\sim \log(\log(N))$ in this setup where $J\sim \log^2(N)$) as predicted in Theorem \ref{thm:multiorder}, and the computational time of intelligent sampling scales in the order of $\sqrt{N}$ compared to the order $N$ computational time of using BinSeg on the entire dataset. 
\begin{figure}[H]
	\begin{center}
		\includegraphics[scale=0.37]{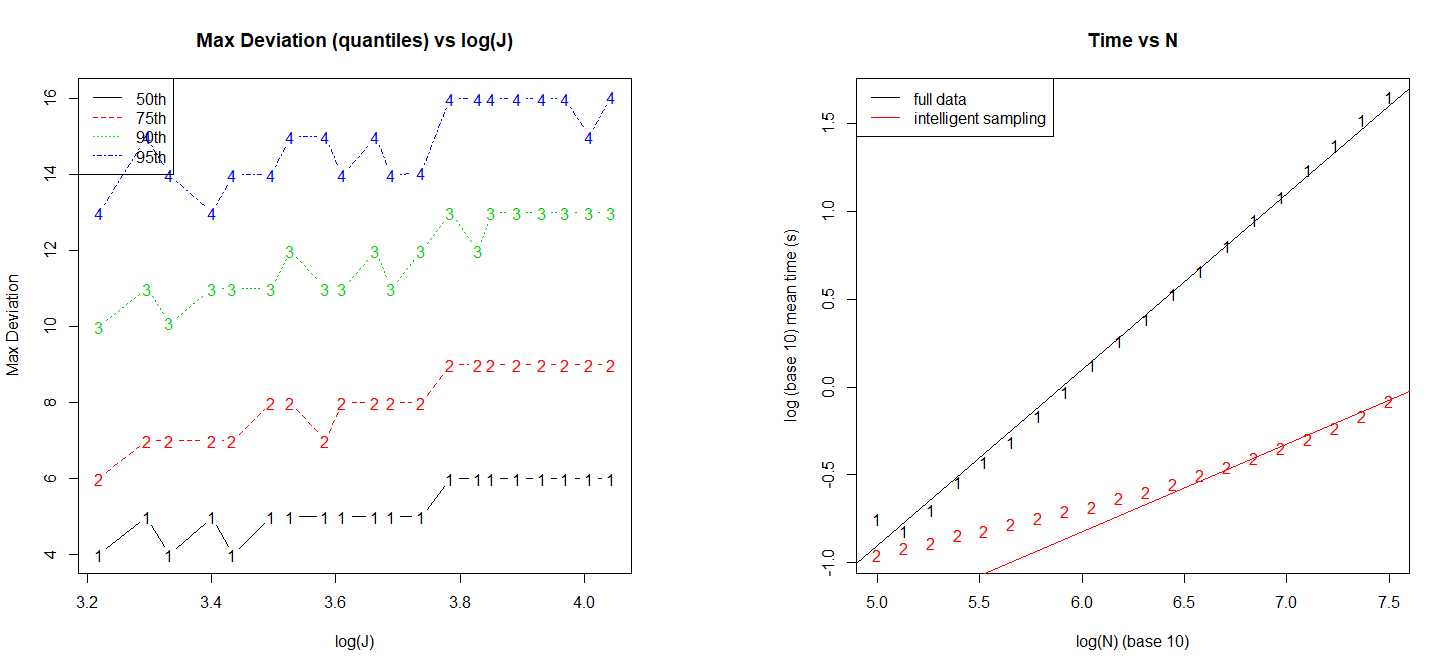}
	\end{center}
	\caption{\textbf{Left:} Quantiles of the max deviations versus $\log(\log(N))$, which is the same order as $\log(J)$. Over the observed regime of parameters, the maximal deviation scales  with $J$, as was predicted by Theorem \ref{thm:multiorder}. \textbf{Right:} Log-log plot of mean computational time when using intelligent sampling to obtain the final change-point estimates at stage two, and using BinSeg on the full data to construct change-point estimates, with reference lines of slope 1 (black) and 0.5 (red) respectively. To give some sense of the actual values, 
		for $N=10^{7.5}$ the average time for intelligent sampling vs full data were, respectively, 0.644 and 31.805 seconds. \label{fig:convergence-rate-mcp}} 
\end{figure}

\subsection{Asymptotic Distribution}\label{sec:simu_asym_dist}
A second set of simulation experiments was used to illustrate the asymptotic distribution of the change point deviations. We considered a setting with $N=10^7$ and 55 change points, which was the maximum number of change points used in the last simulation setting. However, we consider different error term distributions and placement of change points in 4 different settings:
\begin{enumerate}[label=(Setup \arabic*):]
\setlength{\itemindent}{.5in}
\item one set of signal and change point locations generated as in the previous set of simulations, with i.i.d. $N(0,1)$ error terms
\item change points evenly spaced apart with signals 0,1,0,1,.., repeating, and i.i.d. $N(0,1)$ error terms
\item change points evenly spaced part with 0,1,0,1,... repeating signals, and error terms generated as 
$\varepsilon_i=\frac{ \varepsilon^*_i+\varepsilon^*_{i+1}+\varepsilon^*_{i+2}+\varepsilon^*_{i+3}+\varepsilon^*_{i+4}}{\sqrt{5}}$ for all $i=1,...,N$, where the $\varepsilon^*_i$'s are generated as i.i.d. $N(0,1)$;
\item change points evenly spaced with signals 0,1,0,1,..., and error terms generated from an AR(1) series with parameter 0.8, and each marginally $N(0,1)$.
\end{enumerate}
The first and second setups will demonstrate the result of Theorem \ref{thm:multidepend}. As for the other two setups with dependent error terms, we claim that the asymptotic distribution will follow a similar form as the asymptotic distribution in the iid case, with more specific details in Section \ref{sec:dependenterrors} of the supplement. For all 4 cases the first stage of intelligent sampling was performed identically as for the previous set of simulations, and with the same tuning parameters. At the second stage, first stage subsample points were omitted for data with setups 1 and 2, but not for setups 3 and 4.    
From 2000 iterations on each of the 4 simulation setups, the distributions of the maximum deviations (maximum of $|\lambda_2(\tau_j,\hat{\tau}_j^{(2)})|$ for the first two setups and $|\hat{\tau}^{(2)}_j-\tau_j|$ for the other two setups) are seen to match well with their predicted asymptotic distributions. To illustrate the convergence of the individual change point estimates, we also show that the distribution of the 27th change point matches with the $L$-type distributions appearing in Proposition \ref{thm:notiidconsistent} of the supplement. 

\begin{figure}[H]
\begin{center}
	\includegraphics[scale=0.29]{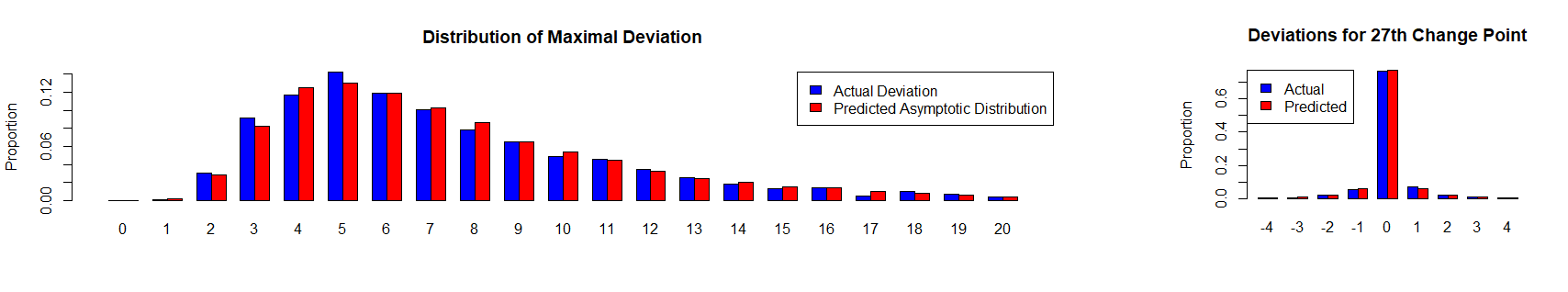}
\end{center}
\caption{Distributions of $\max_{1\leq j\leq 55}\lambda_2\left(\tau_j,\hat{\tau}_j^{(2)}\right)$ (\textbf{Left}) and $\lambda_2\left(\tau_{27},\hat{\tau}_{27}^{(2)}\right)$ (\textbf{Right}) from simulations of setup 1. 
}\label{fig:setup1}
\end{figure}
\begin{figure}[H]
\begin{center}
	\includegraphics[scale=0.29]{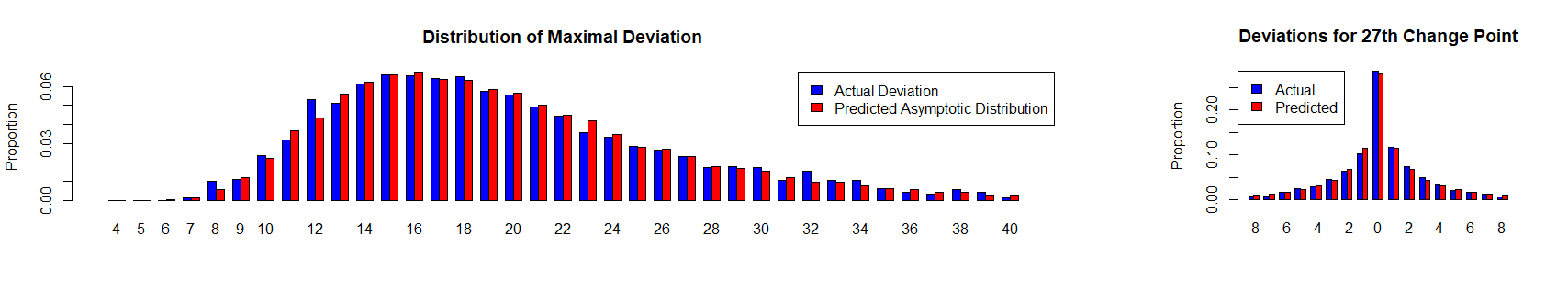}
\end{center}
\caption{Distributions of $\max_{1\leq j\leq 55}\lambda_2\left(\tau_j,\hat{\tau}_j^{(2)}\right)$ (\textbf{Left}) and $\lambda_2\left(\tau_{27},\hat{\tau}_{27}^{(2)}\right)$ (\textbf{Right}) from setup 2. 
}\label{fig:setup2}
\end{figure}
\begin{figure}[H]
\begin{center}
	\includegraphics[scale=0.29]{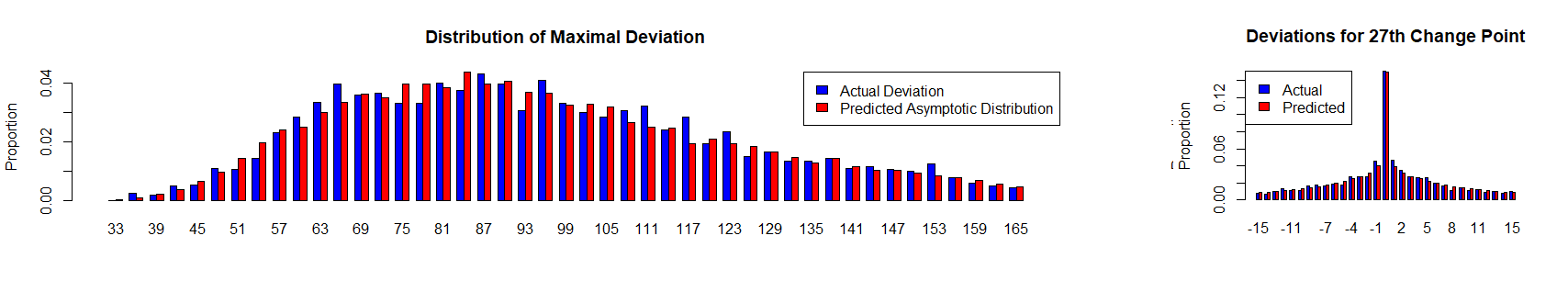}
\end{center}
\caption{Distributions of $\max_{1\leq j\leq 55}\left|\tau_j-\hat{\tau}_j^{(2)}\right|$ (\textbf{Left}) and $\left|\tau_{27}-\hat{\tau}_{27}^{(2)}\right|$ (\textbf{Right}) from setup 3
}\label{fig:setup3}
\end{figure}
\begin{figure}[H]
\begin{center}
	\includegraphics[scale=0.29]{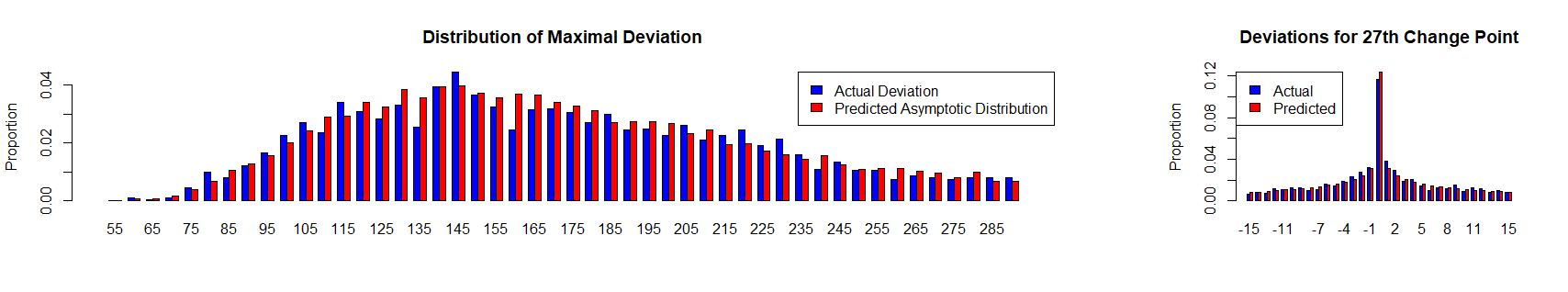}
\end{center}
\caption{Distributions of $\max_{1\leq j\leq 55}\left|\tau_j-\hat{\tau}_j^{(2)}\right|$ (\textbf{Left}) and $\left|\tau_{27}-\hat{\tau}_{27}^{(2)}\right|$ (\textbf{Right}) from setup 4.
}\label{fig:setup4}
\end{figure}
The distribution of the deviations for setup 1 is the least spread out, with the primary reason that the jump between signals is randomly generated, but lower bounded by 1, while the other 3 setups have signal jumps all fixed at 1. Setups 3 and 4 have the most spread out distributions, as the dependence among the error terms causes the estimation to be less accurate, but only up to a constant rather than an order of magnitude. Nevertheless, in all 4 setups, the change point estimates behave very closely to what Theorems \ref{thm:increasingJasymprotics} and \ref{thm:multidepend} (for the first two scenarios) and Proposition \ref{thm:notiidconsistent} (for the third and fourth scenarios) predict.
\newline
\newline
\subsection{Asymptotic Distribution for the Heteroscedastic Case} We further explored the validity of Proposition \ref{thm:notiidconsistent} by examing a setting involving heteroscedatic errors. A data sequence of length $N=10^7$ with 55 evenly spaced changed points and signals of magnitude $0,1,0,1,0,1,\dots,1$ was generated. Instead of generating error terms from a stationary series, we, instead, generated them as independent segments of Gaussian processes as follows. For $j=1,2,3,\dots$, 
\begin{enumerate}
\item from $\tau_{4j}$ to $\frac{\tau_{4j+2}+\tau_{4j+1}}{2}$ the errors are iid $N(0,1)$; 
\item from  $\frac{\tau_{4j+2}+\tau_{4j+1}}{2}$ to $\frac{\tau_{4j+3}+\tau_{4j+2}}{2}$ the errors are $\varepsilon_i=\frac{\varepsilon_i^*+0.5\varepsilon_{i+1}^*+0.25\varepsilon_{i+2}^*}{\sqrt{1^2+0.5^2+0.25^2}}$ where the $\varepsilon_i^*$'s are iid $N(0,1)$ (and will be treated as a generic iid $N(0,1)$ sequence from here on;)
\item from $\frac{\tau_{4j+3}+\tau_{4j+2}}{2}$ to $\tau_{4j+3}$, error terms are $\varepsilon_i=0.5\cdot\frac{\varepsilon_i^*+\varepsilon_{i+1}^*+\varepsilon_{i+2}^*+\varepsilon_{i+3}^*}{\sqrt{4}}$; 
\item from $\tau_{4j+3}$ to $\tau_{4j+4}$ the error terms are $\varepsilon_i=0.7\cdot\frac{\varepsilon_i^*+\varepsilon_{i+3}^*}{\sqrt{2}}$;  
\end{enumerate}
and the error terms generated in each stretch are independent of those in any other stretch. This creates a situation where around $\tau_{4j+1}$ the error terms are iid $N(0,1)$, around $\tau_{4j+2}$ the error terms are stationary, and around $\tau_{4j+3}$ and $\tau_{4j+4}$ the error terms are stationary to the left and to the right, but their autocorrelation and marginal variances change at the change points. With the same intelligent sampling procedure as setups 3 and 4, and the same tuning parameters, we ran 2000 replicates of this setup and recorded the $\hat{\tau}^{(2)}_j-\tau_j$ values for $j=1,\dots,55$.
\begin{figure}[H]
\begin{center}
	\includegraphics[scale=0.3]{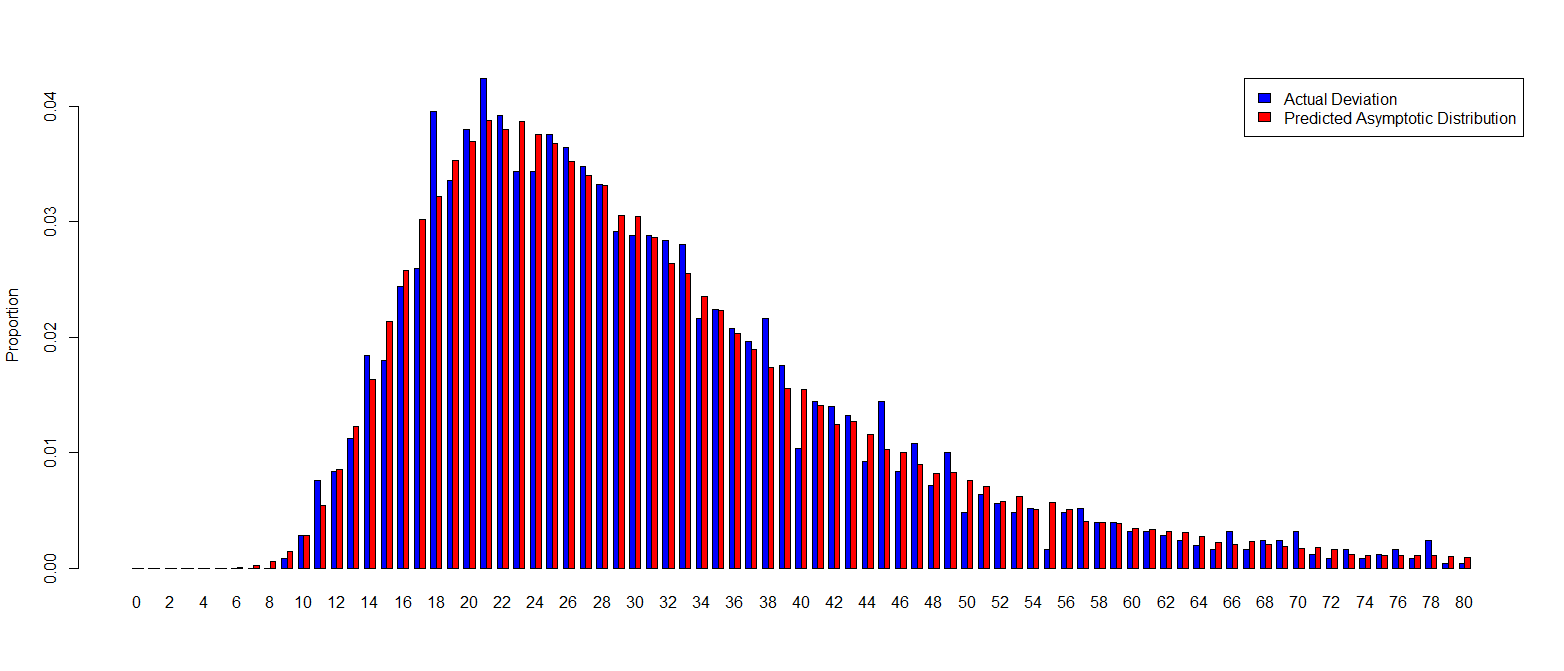}
\end{center}
\end{figure}
\begin{figure}[H]
\begin{center}
	\includegraphics[scale=0.27]{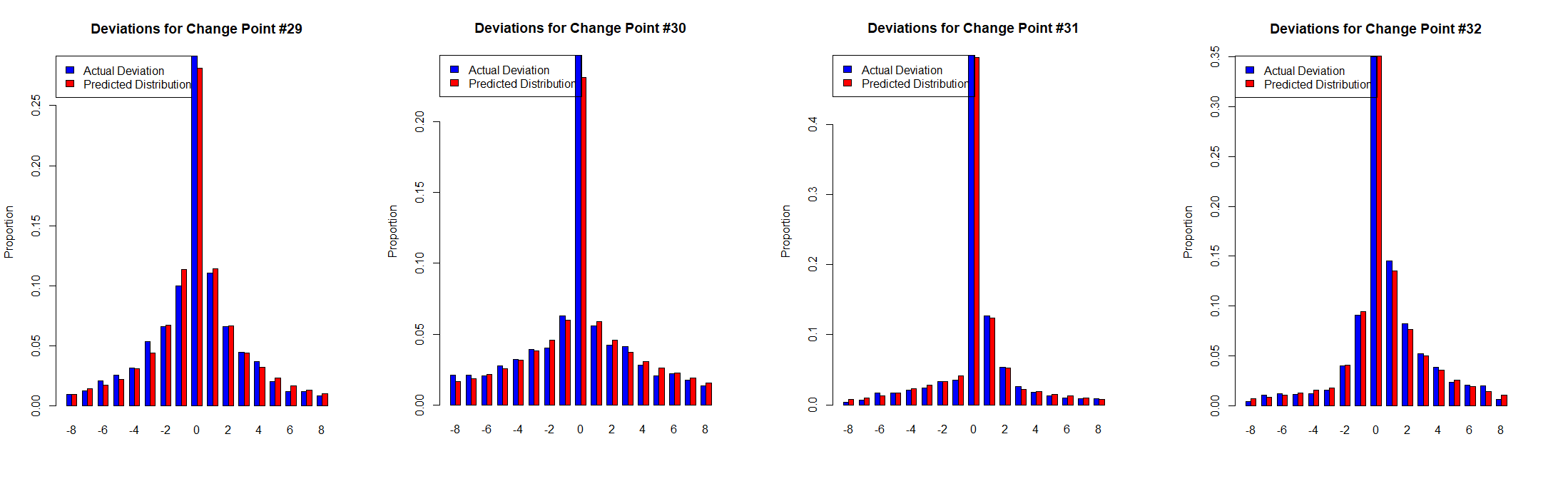}
\end{center}
\caption{\textbf{Top:} Predicted and actual distribution of the maximal deviation $\max_{1\leq j\leq 55}|\hat{\tau}^{(2)}_j-\tau_j|$. \textbf{Bottom, Left to Right:} Predicted and actual distributions of the individual deviations $|\hat{\tau}^{(2)}_j-\tau_j|$ for $j=29, 30, 31,$ and $32$. }
\end{figure}
\noindent		
These numerical results are consistent with Proposition \ref{thm:notiidconsistent} and even for change points whose error distributions are different to the left and the right of the corresponding change point, the deviations match up very closely with the stated asymptotic distributions.

\subsection{Choice of Subsample Size}\label{sec:simu_binseg_double}
In order to showcase the veracity of our asymptotic results, for the previous simulation settings we tuned the first stage subsample size to ensure that the exact number of change points was detected with high accuracy. In practice, it is difficult to know how to choose the best size of the first stage subsample, since one does not know the number of change points or how far apart they are. In this section we propose a systematic \emph{data-driven} way of determining a correct subsample size at the first stage: draw a (small) subsample of size $N_1$, perform BinSeg, record the number of change-points detected; then perform BinSeg on a fresh subsample whose size is enhanced by a constant factor and re-record the number of change-points. Keep increasing the subsample size and re-estimating until some stopping criterion related to the growth in the number of detected change-points is reached. Formally, we propose the following procedure, coined {\bf ``BinSeg-Double''}:
\begin{enumerate}
	\item perform BinSeg on an even subsample with $N_1=2\sqrt{N}$, obtaining an estimate $\hat{J}_1$ for $J$ along with change point estimates;
	\item now generate a new subsample of size $N_1$ that is double the previous value and obtain an estimate $\hat{J}_2$ along with change point estimates;
	\item continue the previous step until we reach a $\hat{J}_i$ such that $\hat{J}_i>\hat{J}_{i-2}+5$ and $|\hat{J}_i-\hat{J}_{i-1}|<5$, or until $\left|\min\{ \hat{J}_{i-3},\dots,\hat{J}_i \}-\max\{ \hat{J}_{i-3},\dots,\hat{J}_i \}\right|<5$.
\end{enumerate}
With regard to the stopping criteria at the last step, we reason that generally, the number of change points detected will go up at every iteration until nearly all change points are detected, or the subsample size has increased to the size of the full data set. Hence, we stop doubling the subsample once the number of detected change points essentially stabilizes. This stopping criterion has good results in several settings as we shall see in subsequent simulations, but other scenarios might require alternative stopping rules for better accuracy. In practical settings the practitioner may use a different last step to adapt to the difficulty of the problem.
\newline
\newline
\indent We gauged the effectiveness of this method as follows. We generated data sequences of length $N=10^6$ and iid $N(0,1)$ error terms, with evenly spaced change points, and the signals alternating between $0,\Delta, 0,\Delta,\dots$, where $\Delta$ is the SNR parameter, allowed to vary aross simulation settings in addition to the number of change-points. For each configuration, we invoked intelligent sampling at stage 1 using the BinSeg-Double procedure followed by curtailing consecutive change point estimates that are within 5 points of each other. As an added detail, for the sake of methodology, we used the R  changepoint package functions in this and all subsequent sections, instead of the manually coded BinSeg algorithm with pre-specified tuning parameter as in the previous sections. The second stage calculations remain unaltered. We recorded the computational time and confidence level coverage for each simulation configuration, and compared them with the computational time of running BinSeg on the entire dataset. 
\begin{table}[H]
	\caption{Two Stage Running Time}
\begin{tabular}{|c||c|c|c|c|c|c|c|c|c|}
	\hline 
	\backslashbox{SNR~~}{$J$~~} & 50 & 100 & 150 & 200 & 250 & 300 & 350 & 400 & 450\\ 
	\hline\hline 
	1 & 1.67 & 4.6 & 11.84 & 15.1 & 39.58 & 49.8 & 53.95 & 80.41 & 132.22\\ \hline 
	1.5 & 1.08 & 2.85 & 9.04 & 7.42 & 20.12 & 20.84 & 27.34 & 39.3 & 74.16\\ \hline 
	2 & 0.71 & 1.68 & 3.34 & 7.63 & 11.7 & 20.37 & 24.48 & 45.79 & 57.06\\ \hline 
	
\end{tabular}\\
Average running time in seconds across 200 iterations of every configuration.

\caption{Full Data Running Time}
\begin{tabular}{|c||c|c|c|c|c|c|c|c|c|}
	\hline 
	\backslashbox{SNR~~}{$J$~~} & 50 & 100 & 150 & 200 & 250 & 300 & 350 & 400 & 450\\ 
	\hline\hline 
	1 & 17.63 & 26.66 & 35.77 & 45.14 & 54.56 & 64.2 & 73.77 & 83.72 & 93.47\\ \hline 
	1.5 & 17.39 & 26.55 & 35.92 & 44.59 & 53.75 & 63.34 & 72.75 & 82.71 & 92.4\\ \hline 
	2 & 17.13 & 26.01 & 35.17 & 44.44 & 53.63 & 63.39 & 72.61 & 82.37 & 91.86\\ \hline 
	
\end{tabular}\\
Average running time in seconds across 50 iterations of every configuration.
\end{table}
As is evident from these tables, the running time of intelligent sampling relative to the full data procedure is lowest when the number of change points is low and the SNR at these changes are high. Conversely, when the number of change points is large and/or the SNR values are low, it becomes necessary to take a bigger subsample at the first stage, increasing the relative computation time.

We also show that our proposed first stage procedure leads to accurate results. In each iteration, after the second stage change point estimates were obtained and the timer stopped, we constructed $(1-\alpha)$ level (not simultaneous) confidence intervals around each estimate, and recorded the proportion of true change points captured in the union of these intervals. Specifically, at each iteration $k$ we constructed a confidence region $\hat{C}_k$, and recorded 
\begin{eqnarray}
\frac{1}{J\cdot\text{(\# of iterations)}}\sum_{ k=1 }^{\text{(\# of iterations)}}\sum_{j=1}^J 1( \tau_j\in \hat{C}_k ).
\end{eqnarray}
\begin{table}[H]
	\caption{Change point Coverage for Two Stages}
\begin{tabular}{|c||c|c|c|c|c|c|c|c|c|}
	\hline 
	\backslashbox{SNR~~}{$J$~~} & 50 & 100 & 150 & 200 & 250 & 300 & 350 & 400 & 450\\ 
	\hline\hline 
	1 & 0.966 & 0.971 & 0.945 & 0.944 & 0.968 & 0.966 & 0.935 & 0.967 & 0.969\\ \hline 
	1.5 & 0.978 & 0.985 & 0.958 & 0.958 & 0.983 & 0.981 & 0.954 & 0.982 & 0.982\\ \hline 
	2 & 0.993 & 0.997 & 0.972 & 0.971 & 0.996 & 0.995 & 0.971 & 0.995 & 0.996\\ \hline 
	
\end{tabular}\\
Proportion of change points covered over all 200 iterations with $\alpha=0.01$. As all proportions are close to 0.99, this indicates accurate estimates.
\end{table}
We next perform a similar set of simulations, with $N$ and the error structure staying the same but with an added element of randomness from the placement of the change points and values of the signals. The change points $(\tau_1,\dots,\tau_J)$ are generated as $N/(4(J+1))$ plus the random variables $\frac{3N}{4(J+1)}(U_{(1)}, U_{(2)}-U_{(1)},U_{(3)}-U_{(2)},\dots,1-U_{(J)} )$, where $U_{(1)},\dots, U_{(J)}$ are the ordered variables from $J$ iid uniform (0,1) random variables. The changes in signal at the different change point are randomly generated as independent $(2S-1)*V$ variables, where $S$ is Bernoulli with $p=0.5$ and $V$ is uniform on [1,4].

\begin{table}[H]
	\caption{Running Time for Random Design}
\begin{tabular}{|c||c|c|c|c|c|c|c|c|c|}
	\hline 
	\backslashbox{Method~~}{$J$~~} & 50 & 100 & 150 & 200 & 250 & 300 & 350 & 400 & 450\\ 
	\hline\hline 
	Two Stage & 1.02 & 3.3 & 9.84 & 21.67 & 32.06 & 44.7 & 64.09 & 86.89 & 109.84\\ \hline 
	Full Data & 18.07 & 27.36 & 36.81 & 46.38 & 56.04 & 65.77 & 75.64 & 85.32 & 95.16\\ \hline 
	
\end{tabular}
\\
Average running time in seconds across 200 iterations of the random design setup.
\end{table}
The coverage proportion of actual change points are
\begin{table}[H]
	\caption{Coverage Proportion for Random Design}
\begin{tabular}{|c||c|c|c|c|c|c|c|c|c|}
	\hline 
	$J$~~ & 50 & 100 & 150 & 200 & 250 & 300 & 350 & 400 & 450\\ 
	\hline\hline 
	Coverage & 0.985 & 0.986 & 0.993 & 0.994 & 0.995 & 0.996 & 0.996 & 0.995 & 0.996\\ \hline 
	
\end{tabular}
\\
Proportion of change points covered by $\alpha=0.01$ level (non-simultaneous) confidence sets across 200 iterations.  
\end{table}
~\indent We note that in both simulation setups, intelligent sampling gains the most computational time savings over the full data analysis when the separation between change points is long (corresponding to smaller $J$) and when the SNR is high. Such setups require only a small subsample at stage 1 to capture all the change points. 
In certain cases where $J$ is high enough and/or the SNR low enough, \emph{it would be impossible to detect all the change points at stage 1 with any strictly smaller subsample, meaning BinSeg-Double would continue until the full dataset is utilized}. In this respect, the doubling procedure is \emph{automatically calibrated} to the intrinsic difficulty of the problem. 

Our findings are perfectly consistent with our previous claims that intelligent sampling is most useful for detecting sparsely placed, high SNR change points in long data sequences; in settings with densely placed, lower SNR change points (or a large number of change points), one can not obtain accurate results without utilizing large subsamples (or even the entire data set), thus mitigating, or even nullifying, the computational time savings of intelligent sampling.

\section{Data Applications}\label{sec:realdata}
We apply the intelligent sampling framework to two network traffic data sets. 
In the first one that exhibits a stationary behavior, we artificially inject change points, by shifting the local mean of
the data across different stretches. Hence, there is a known ground truth to calibrate the performance of intelligent sampling in a setting where other features of the data (error distributions, presence of temporal dependence) may not fully adhere to
the assumptions used in establishing the theoretical properties of the procedure. The second data set is a fully observed data set where the ground truths are unknown. 
\subsection{Partly Emulated Data} 
\label{partly-emulated-data} 
The effectiveness of the proposed intelligent sampling procedure is illustrated on an Internet traffic data set, obtained from the public CAIDA repository
\texttt{http://data.caida.org/datasets/passive/passive-oc48/20020814-160000.UTC/pcap/} that contains traffic traces from an OC48 (2.5 Gbits/sec)
capacity link. The trace under consideration contains traffic for a two hour period from a large west coast Internet service provider back in 2002.
The original trace contains all packets that went through the link in an approximately 2 hour interval, but after some aggregation into bins of length 300 microseconds, the resulting data sequence comprises $N=1.5\times 10^7$ observations. After applying a square-root transformation,
a small segment of this sequence is depicted in Figure \ref{fig:raw-data} and some of its statistical characteristics in Figure \ref{fig:realnormal}, respectively.
\begin{figure}[H]
\begin{center}
	\includegraphics[scale=0.37]{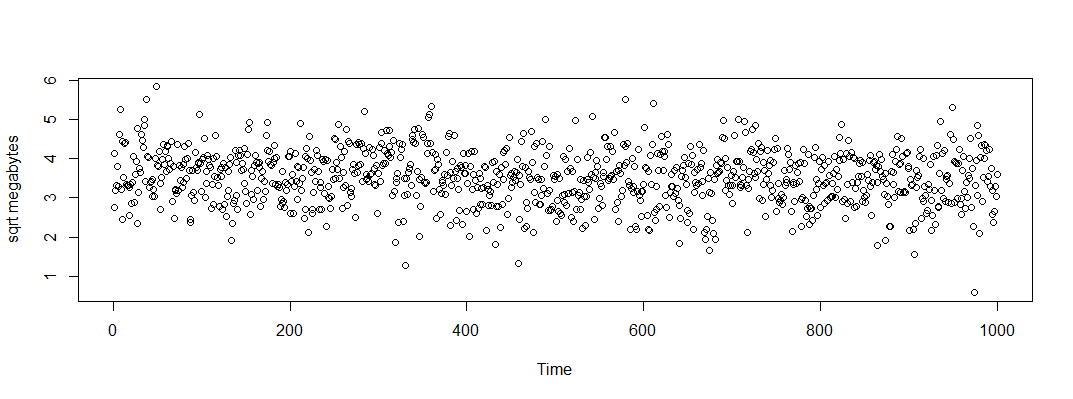}
\end{center}
\caption{First 5000 time points of the data after a square root transformation.}\label{fig:raw-data}
\end{figure}
\begin{figure}[H]
\begin{center}
	\includegraphics[scale=0.37]{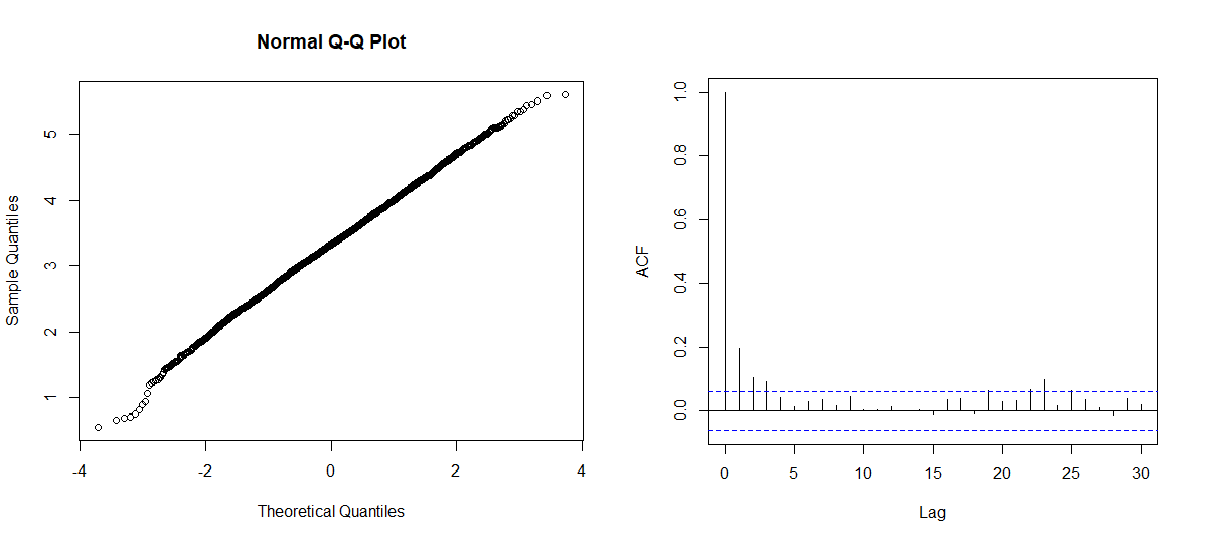}
\end{center}
\caption{QQ plot and estimated ACF of first 5000 points of data set after square root transformation. Note the normality of the transformed data.}\label{fig:realnormal}
\end{figure}
It can be seen that the data are close to marginally normally distributed, while their autocorrelation decays rapidly and essentially disappears after 
a lag of 10. Similar exploratory analyses performed for multiple stretches of the data lead to similar conclusions. Hence, for the remainder of the analysis,
we work with the square-root transformed data and model them as a short range dependent sequence.

To illustrate the methodology, we used an {\em emulation} setting, where we injected various mean shifts to the mean level of the data of random durations,
as described next. This allows us to test the proposed intelligent sampling procedure, while at the same time retaining all features in the original data.

In our emulation experiments, we posit that there are two types of disruptions, short term spikes that may be the result of specific events (release of a software upgrade, a new product or a highly anticipated broadcast) and longer duration disruptions that may be the result of malicious activity \cite{carl2006denial,kallitsis2016amon}. 
We first take the post square root transformation data $Y_1,\dots,Y_N$ and add long, persistent change points by taking $Y_i\leftarrow Y_i+\theta_i$, where $\theta_i$ is piecewise constant with change points $(\tau_1,\dots,\tau_J)$ ($J=200$) generated as $\frac{N}{2.5(J+1)}+\frac{1.5N}{2.5(J+1)}\cdot (U_{(1)},U_{(2)},\dots, U_{(J)})$, where $(U_{(1)},\dots, U_{(J)})$ are the order statistics of $J$ iid uniform (0,1) variables, and signals $(0,S_1,0,S_2,\dots)$ where $S_1,S_2,\dots$ are iid uniform (0.7,2) variables. Next we emulate a larger number of \emph{spikes} in the data, or short term large increases in the signal, by randomly selecting 400 locations for these spikes as $(V_1,\dots,V_{400})=(N-50)\cdot(U_1,\dots,U_{500})$ (where $U_1,U_2,\dots$ are iid uniform (0,1)), and setting $(Y_{V_j},\dots,Y_{V_j+50})\leftarrow  (Y_{V_j},\dots,Y_{V_j+50})+W_j$ for all $j$, where $W_j$ is an iid uniform (10,15) variable. A depiction of a segment of the data with the emulated signal is given in Figure \ref{fig:emulated-traffic}.
\begin{figure}[!h]
\begin{center}
	\begin{overpic}[scale=0.35,tics=10]{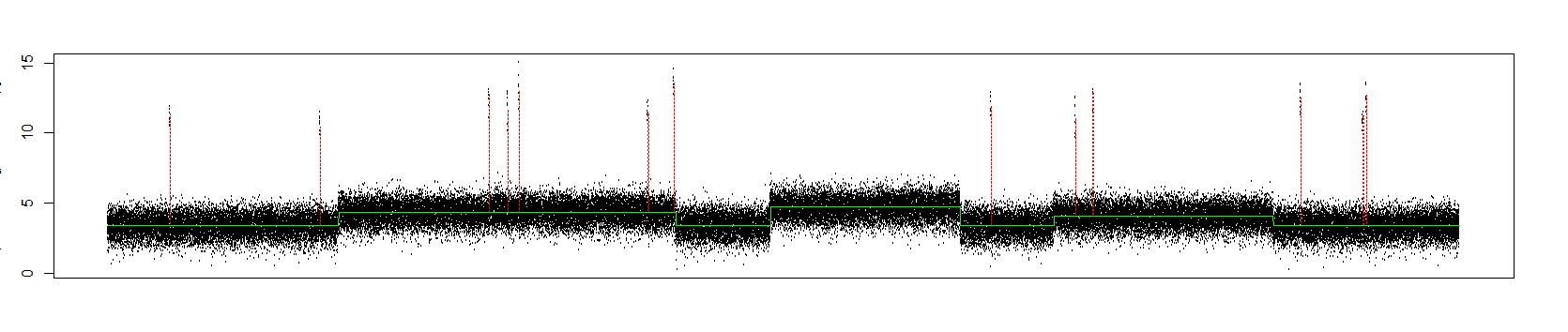}
	\end{overpic}
\end{center}
\caption{Example of the first half million points of an emulated data example. 
} \label{fig:emulated-traffic}
\end{figure}

As mentioned in the introduction, the main objective of the proposed methodology is to identify {\em long duration, persistent shifts} in the data sequence using a limited number of data points; in the emulation scenario used, this corresponds to change points $(\tau_1,\tau_J)$, while we remain indifferent to
the points $V_j, V_j+50$ for $j=1,\dots,400$, which as previously described are change points corresponding to the spikes\footnote{We note that the theoretical development does not include spiky signals. Nonetheless, we included spiky signals in our emulation to mimic the pattern of internet traffic data. As will be seen later, our method is quite robust to the presence of this added feature.}

The two-stage intelligent sampling procedure was implemented as follows:  (i) the first stage used BinSeg-Double , followed by removal of estimates less than 4 points apart, followed by steps (ISM4-1) to steps (ISM4-4), and then followed by removal of estimates less than 16 apart.  (ii) For each $j$, the second stage interval surrounding $\hat{\tau}_j^{(1)}$ was chosen to have half width  $Q_{\hat{\Delta}_j,\hat{\sigma},\hat{\sigma}}\left(1-\frac{0.01}{\hat{J}}\right)$ where $\tilde{\Delta}_j$'s and $\hat{\sigma}$ are estimates of the jump and the common standard deviation. A stump model was then fitted to the data in each second stage interval to obtain the final estimates of the change-points and the final (2nd stage) CIs were constructed. 
\newline
\newline
\indent To assess the accuracy, we calculated the coverage proportion of the 90\%, 95\%, and 98\% level confidence intervals over different emulation settings. To construct these confidence intervals we had to randomly generate data sequences with identical distribution structure as the data (which would give us a random sample of $L$-type distributions and their quantiles). We generated these sequences as marginally normal random variables, with marginal standard deviation the same as the sample sd of the first 50,000 points of the $1.5\times 10^7$ length data. Finally, the ACF of the generated series was matched with the sample ACF of the first 50,000 points up to a lag of 20: we first generated vectors of iid normal variables, then multiplied them with the Cholesky square root of the band matrix created with the sample ACF (the bandwidth of this matrix is 20, and non-zero entries are taken from the first 20 values of the sample ACF).

Intelligent sampling exhibits highly satisfactory performance: among all 200 change points corresponding to persistent changes, the average coverage probability for the 90\%, 95\%, and 98\% nominal confidence intervals were 0.874, 0.914, and 0.939 respectively. On the other hand, for change points induced by the spikes, the average coverage probability was lower than 0.035 even for the 98th confidence interval. However, since the focus of intelligent sampling is on long duration persistent signals, missing the spiky signals is of no great consequence. In terms of computational burden, the average emulation setting utilized 17.1\% of the full dataset, requiring an average time of 29.7 seconds to perform the estimation. 

\begin{figure}[H]
\begin{center}
	\includegraphics[scale=0.35]{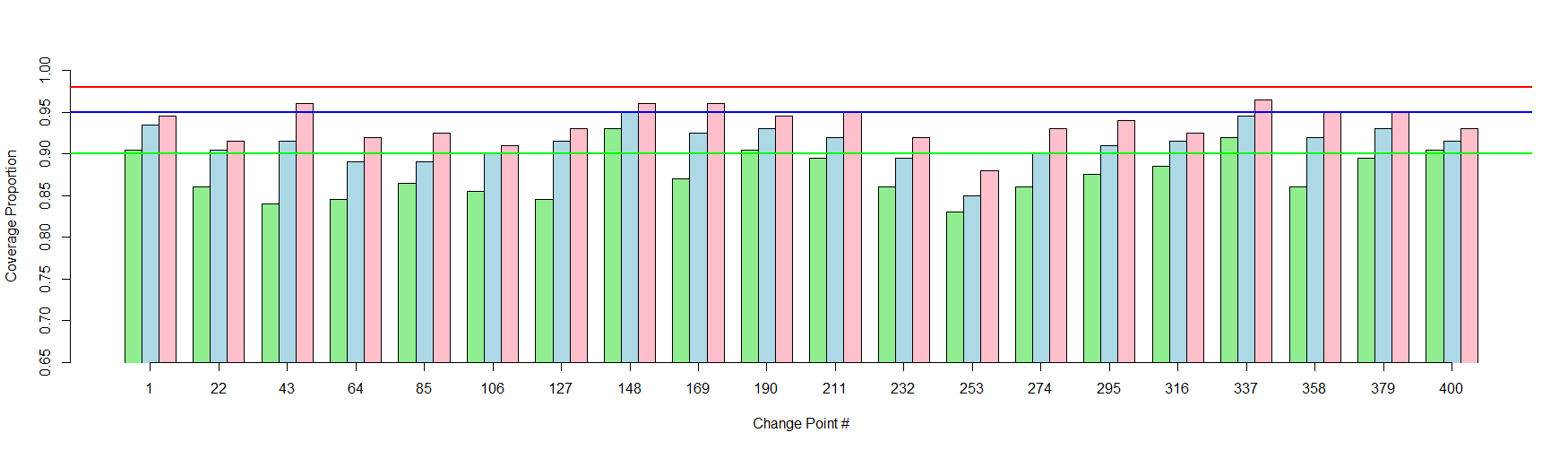}
\end{center}
\caption{Coverage proportions, the proportion of time when the change point was covered by some confidence interval, for the 90\% level (green bars), 95\% level (blue bars), and 98\% level (red bars) within the 500 iterations, for a select number of 20 change points (change point \# 2 is always the second one in order, \# 3 is the third in order, etc). Horizontal reference lines are at 0.9 (green), 0.95 (blue), and 0.98 (red).}
\end{figure}	

\subsection{Real Data Example}
\label{real-data-section} 
	We further gauged the practical performance of intelligent sampling by evaluating its performance when applied to a different data set with \emph{naturally occurring change points} of varying durations. The growing capacity of Internet links
	(routinely at 10 Gbps) together with proliferation of such links have rendered network monitoring a more challenging 
	but at the same time critical task, given increased malicious activities. However, most monitoring tools rely on sampling packets at rates of 1:1000 or 1:10000 in order to minimize interference with network infrastructure. In an experimental project,
an alternative that aggregates traffic at 2 sec intervals and then records it was developed at Merit Network (see details on
the technology in \cite{kallitsis2016amon}). The data analyzed next correspond to a trace of all incoming traffic to the autonomous system administered by Merit Network and corresponds to a data sequence obtained during the spring of 2019,  comprising of 2080768 observations. The objective is to identify persistent change points in the traffic pattern. However, 
since the data have not been previously analyzed, no ground truth is known.

Unlike the 2002 trace previously analyzed, the current trace (segments in logarithmic scale shown in 
Figure \ref{fig:dataportion1}) exhibits a richer set of patterns, but broadly resembles a piecewise linear signal plus noise, with multiple discontinuities and abrupt changes in slopes.
\begin{figure}[H]
	\begin{center}
		\includegraphics[scale=0.3]{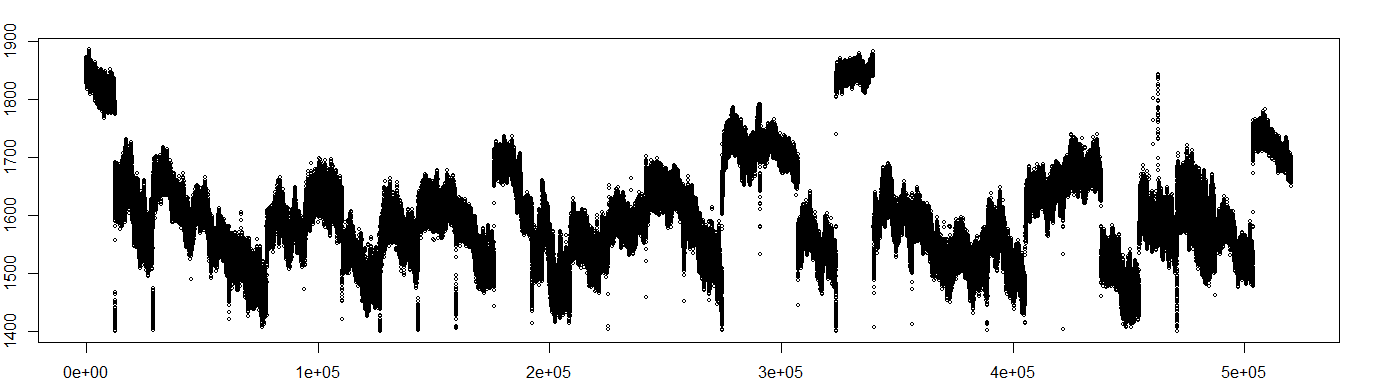}
	\end{center}
	\caption{Plot of the last quarter of the data (i.e., points $0.75N$ to $N$ of the data).}\label{fig:dataportion1}
\end{figure}

\begin{figure}[H]
	\begin{center}
		\includegraphics[scale=0.3]{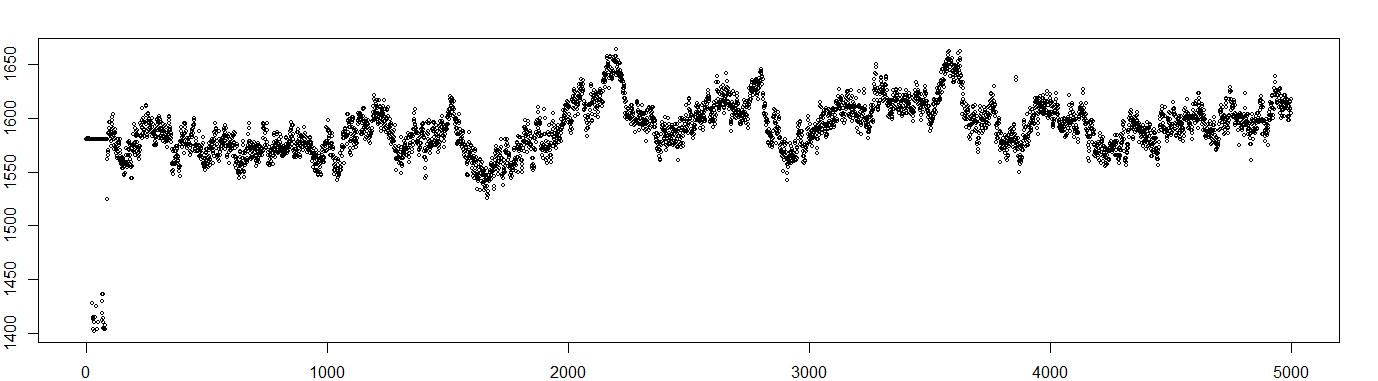}
	\end{center}
	\caption{Plot of the first 5000 points of the data.}\label{fig:dataportion2}
\end{figure}
~\indent The data set contains several particularly prominent change points exhibiting very large jumps, but zooming in to short intervals reveals that the signal oscillates up and down with a wavelength of a few hundred points, as shown in Figure \ref{fig:dataportion2}. Therefore, the data also contain numerous less prominent change points in signal slope at the crest and troughs of the signal vibrations. For our analysis, we are more interested in the locations where large jumps are located. Detection of \emph{all minor signal changes} requires a thorough analysis of the entire data set.
\newline
\newline
\indent Analysis of this data set via intelligent sampling is conceptually similar to data sets with piecewise constant stretches that have been the primary focus of this paper. The multiple change point algorithm in Stage 1 must now be able to deal with piecewise linear signals. Once the change-points from Stage 1 are localized, a piecewise linear function with a single jump can be fitted at Stage 2 in each local neighborhood to get updated estimates of the change points and corresponding confidence intervals by simulating from an estimated error distribution. Specifically, instead of the BinSeg-Double procedure described in Section \ref{sec:simu_binseg_double} in the first stage, we use a doubling procedure that performs Narrowest-Over-Threshold (denoted as NOT henceforth) detection (introduced in \cite{baranowski2019narrowest})  after each doubling. The NOT procedure functions similarly to BinSeg. At each step, a CUSUM-like expression is maximized over sub-intervals of the data\footnote{To be very specific, the maximization is performed over random intervals contained in the segment.}. A change point is declared if the maximum exceeds a specified threshold value, after which the segment is bisected by the new change point estimate and the process repeats over the two components. Because we are trying to detect only the biggest changes in signals instead of all of them, we can use the value of the threshold tuning parameter to limit the number of change points detected by the algorithm. In the associated \texttt{R} package \texttt{not}, we performed this tuning by varying the \texttt{q.max} parameter of the \texttt{features} function, but we set all other tuning parameters to their default values. The stopping criteria for the first stage was slightly altered from before.\footnote{Another stopping criterion was added: if the number of change points ever decreases, stop and keep the previous results.} 
\begin{figure}[H]
	\begin{center}
		\includegraphics[scale=0.23]{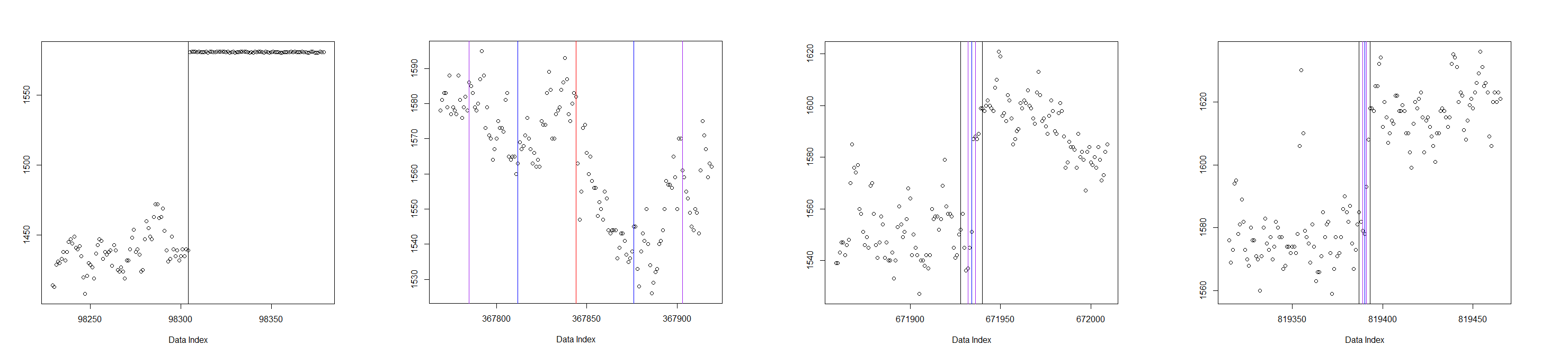}
	\end{center}
	\caption{A few selected estimates (red) and confidence intervals (blue denotes 90\% level, purple for 95\%, and black for 99\%). To construct the confidence levels, for each change point estimate we fit 2 AR time series models using 75 points on the left and 75 points on the right. }\label{fig:CFI_examples}
\end{figure}
We are interested in comparing the performance of intelligent sampling with NOT versus applying NOT over the whole data. Thus, over several values of the \texttt{q.max} tuning parameter, we performed estimation 40 times each for intelligent sampling and full data estimation. All computing was performed on a desktop with an Intel Core i7-8700K CPU, with a script that utilized the \texttt{parallel} package in \texttt{R}. We remind the reader that because of the randomized nature of the NOT algorithm, results can vary between the different runs of estimation.
\newline
\newline

\begin{table}[H]

	\begin{center}
		\caption{Average running time over 40 iterations.\label{runtime}}

		\begin{tabular}{|c||c|c|c|c|c|c|c|}
			\hline 
			\backslashbox{Method~~}{\texttt{q.max}~~} & 50 & 75 & 100 & 125 & 150 & 175 & 200\\ 
			\hline\hline 
			Two Stage & 12.03 & 4.81 & 8.57 & 9.94 & 21.88 & 8.52 & 11.73\\ \hline 
			Full Data & 463.3 & 475.94 & 494.75 & 517.46 & 503.08 & 500.86 & 507.12\\ \hline 
			
		\end{tabular}
	\end{center}
	\begin{center}
	
		\caption{Average number of change points detected over 40 iterations.\label{numCP}}
		
		\begin{tabular}{|c||c|c|c|c|c|c|c|}
			\hline 
			\backslashbox{Method~~}{\texttt{q.max}~~} & 50 & 75 & 100 & 125 & 150 & 175 & 200\\ 
			\hline\hline 
			Two Stage & 49.8 & 72.8 & 77.3 & 80.83 & 81.2 & 77.55 & 80.42\\ \hline 
			Full Data & 49.85 & 70.12 & 76.33 & 78.75 & 83.1 & 80.62 & 79.92\\ \hline 
			
		\end{tabular}
	\end{center}
\end{table}

\noindent In terms of the average number of change points detected, intelligent sampling and the full data analysis detected around the same number of change points for every scenario. When \texttt{q.max} is low, both methods detected close to the maximum number of change points allowed by the package function, but as is clear from Table \ref{runtime}, the number of detected change points stabilized at around 80 for both procedures, irrespective of the value of \texttt{q.max}. Table \ref{numCP} gives a snapshot of the time comparison of the full data analysis with intelligent sampling, with the latter generally seen to be between 40 to 60 times faster. \\\\
To see whether these detected change points correspond to the larger change points, we also show the distribution of the estimated jump sizes.
\begin{figure}[H]
	\begin{center}
		\includegraphics[scale=0.3]{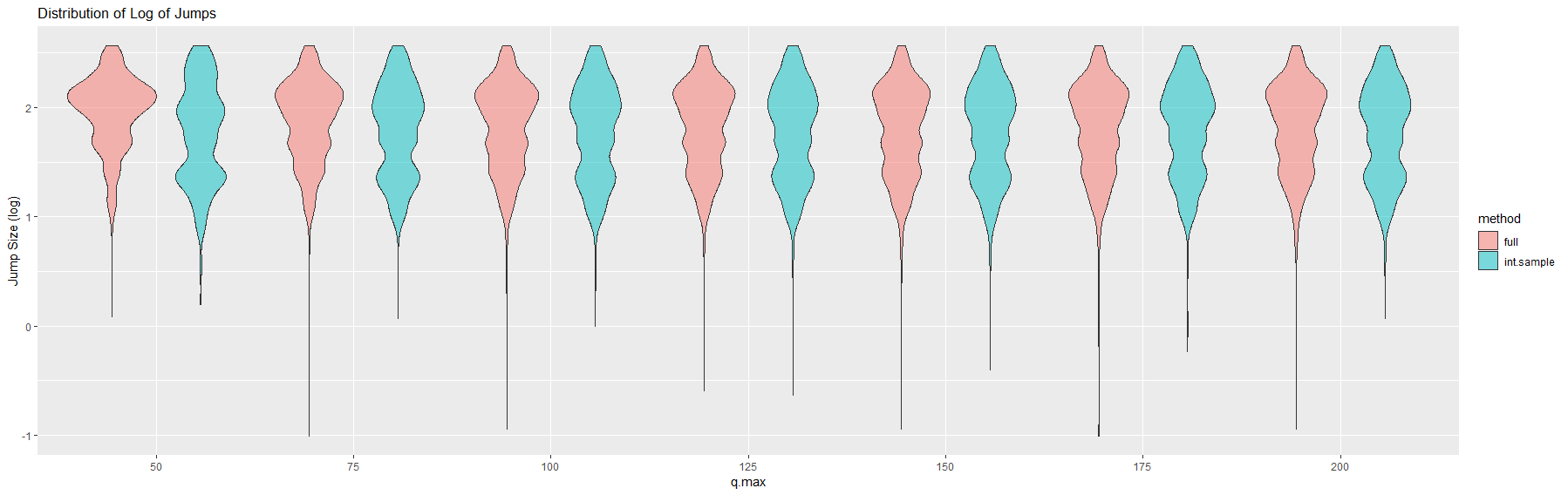}
	\end{center}
	\caption{Distribution of the log (base 10) of the absolute values of jump sizes.}\label{fig:jumpdist}
\end{figure}

\begin{minipage}{0.48\hsize}\centering
	\begin{figure}[H]
		\begin{center}
			\includegraphics[scale=0.25]{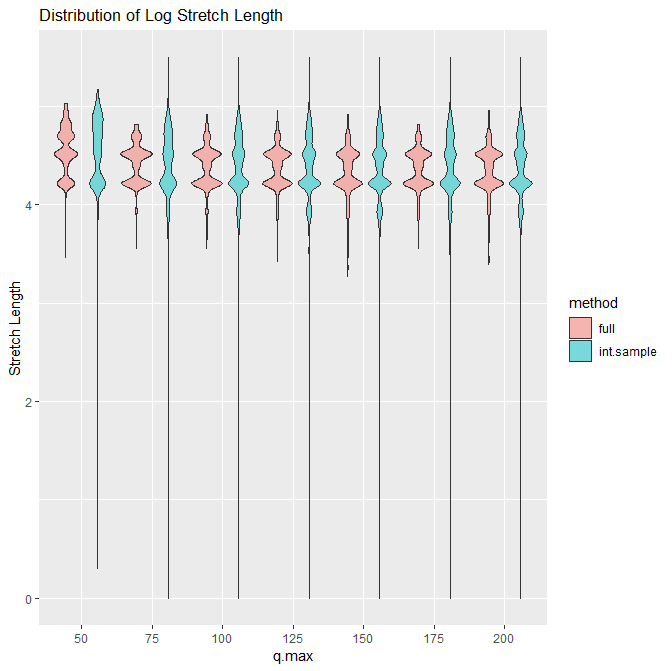}
		\end{center}
		\caption{Distribution of the log (base 10) of the distances between consecutive change point estimates.}\label{fig:stretchdist}
	\end{figure}
\end{minipage}
\begin{minipage}{0.04\hsize}
	~~
\end{minipage}
\begin{minipage}{0.48\hsize}\centering
	
	\begin{table}[H]
		\centering
		\caption{Proportion of Large Jumps}
		\begin{tabular}{|c||c|c|}
			\hline 
			q.max & Two Stage & Full Data\\ 
			\hline\hline 
			50 & 0.55 & 0.78  \\ \hline 
			75 & 0.59 & 0.67  \\ \hline 
			100 & 0.57 & 0.62  \\ \hline 
			125 & 0.56 & 0.61  \\ \hline 
			150 & 0.57 & 0.58  \\ \hline 
			175 & 0.6 & 0.61  \\ \hline 
			200 & 0.57 & 0.61  \\ \hline 
			\hline 
		\end{tabular}\\
		~\newline
		\newline
		Proportion of jumps that are greater than 50. The threshold of 50 was chosen because the data in Figure \ref{fig:dataportion1} seems to oscillate with an amplitude of no more than 50 when far from the change points.
	\end{table}
\end{minipage}
\begin{figure}[H]
	\centering
	\includegraphics[scale=0.25]{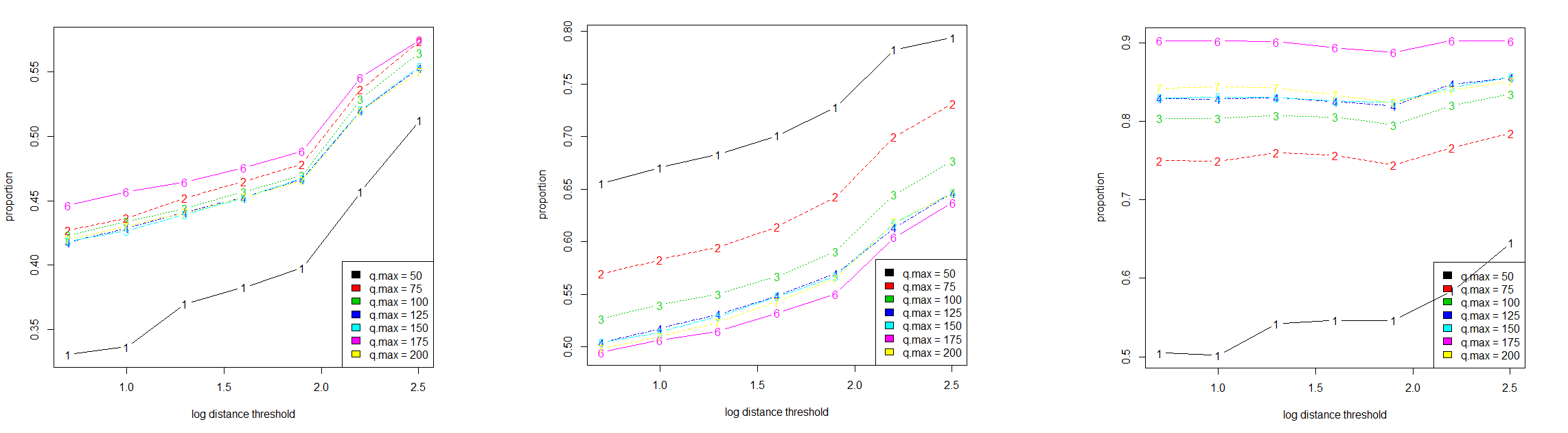}
	\caption{
		\textbf{Left:} Proportion of estimates from intelligent sampling that are close to the estimates from full data estimation, calculated using the expression in (\ref{eq:proportion_expression}). \textbf{ Center: }A comparison between half of the full data estimates versus the other half, using an expression similar to (\ref{eq:proportion_expression}). \textbf{Right: }The ratio of the values in the left graph divided by the values in the center graph.}\label{fig:closeproportions}
\end{figure}
~\newline
\noindent For each \texttt{q.max} value, the distribution of the logarithm of the jump sizes are similar, with the jump distribution from the intelligent sampling procedure being skewed more towards the lower values. The proportion of detected jumps which exceed 50 (a value that appears to be a reasonable upper bound on the noise amplitude from inspection of Figure 14) is also similar between the two estimation methods. It is clear that most of the detected jumps are large.  Figure \ref{fig:stretchdist} demonstrates the distributions of distances between consecutive change points, and it is clear from the plots that the vast majority of inter change point distances obtained by intelligent sampling are between 10000 and 10000 points, which shows that our approach is picking up long and persistent changes, exactly what it is designed for. 
\newline
\newline 
\indent Another way of comparing the change point estimates procured from intelligent sampling is to quantify the proportion of intelligent sampling estimates that are close to some full data estimate. 
For each fixed \texttt{q.max} value, we have the full data estimates $\{ \hat{\tau}^{(j)}_{i,f}:i=1,\dots,\hat{J}_{j,f} \}$ and the intelligent sampling estimates $\{ \hat{\tau}^{(j)}_{i,s}:i=1,\dots,\hat{J}_{j,s} \}$ for iterations $j=1,\dots, 40$. We fix a bandwidth parameter $\delta\in \{5,10,20,\dots,320\}$, and calculate the following proportions to gauge how similar the two sets of estimates are:
\begin{eqnarray}\label{eq:proportion_expression}
\frac{1}{40}  \sum_{k=1}^{40}  \left[ \frac{1}{\sum_{j=1}^{40}\hat{J}_{j,s}}\sum_{\substack{j=1,\dots,40\\i=1,\dots,\hat{J}}_{j,s}}  1\left(\min_{\ell=1,\dots,\hat{J}_{k,f}}\left| \hat{\tau}_{\ell,f}^{(k)}-\hat{\tau}_{i,s}^{(j)} \right|<\delta\right) \right].
\end{eqnarray}
These values can be seen in the left panel of Figure \ref{fig:closeproportions} for a variety of parameter values, and they come fairly close to 60\%. Although this proportion is well under 100\%, this is primarily due to the randomness of performing NOT estimation on this specific data set. Indeed, if we split the full data estimates with the first 20 iterations in one group and the last 20 iterations in the other group, then comparison of the two groups of the full data estimates with an expression similar to (\ref{eq:proportion_expression}) will still result in values not exceeding 60\%. Taking the ratios of the proportions for each \texttt{q.max} and $\delta$ values, we see that for large \texttt{q.max} values the ratios are very close to 1. In other words, the comparison of intelligent sampling estimates to the full data estimates show no more dis-similarity than when comparing the full data estimates across disjoint runs. 

For drawing downstream conclusions from an analysis of this type, one should next flag the change-point locations that are persistent across the 
different estimation runs, and subject them to a more refined local analysis. 
We however do not proceed in that direction, primarily because the point 
of this analysis is somewhat different -- namely, the effectiveness of intelligent sampling compared to full data analysis on real data. 
\newline
\newline
\indent In conclusion, the above analysis amply demonstrates that intelligent sampling was able to extract a similar number of change points as a full data analysis on the entire data, including finding a large proportion of estimators which also shows up in a full data estimate. It was able to do this while using much less time, as the running time data demonstrates.

\section{Concluding Remarks and Discussion}\label{sec:discussion}
This paper introduced sampling methodology that reduces significantly the computational requirements in multi-change point problems, while
not compromising on the statistical accuracy of the resulting estimates. It leverages the {\em locality principle}, which is obviously at work in the
context of the classical signal-plus-noise model employed throughout this study. Intelligent sampling is devised specifically for detecting major but relatively infrequent changes in a data stream, what one can think of as 
significant 'regime changes', and should \emph{not} be relied upon to identify short-lived disturbances or small perturbations to the mean level of a data stream. 
While the paper has dealt with a one-dimensional data sequence, extensions to problems involving multiple (potentially high-dimensional) data [e.g., sequences produced by cyber-physical systems equipped with a multitude of sensors monitoring physical or man-made phenomena] are of obvious interest. Also, while our theoretical development studies the canonical model with piecewise flat stretches between change-points, extensions to allow variations in the mean level between change-points, e.g. piecewise polynomial or nonparametric specifications of the mean function, while requiring more careful estimation of mean levels, pose no added conceptual difficulties: indeed, our second data set uses a piecewise linear model. Finally, the focus in this paper has primarily been on a two-stage procedure, which is easiest to implement in practice and suitable for many applications. Nevertheless, as illustrated in Section
\ref{sec:compmethod}, in specific settings involving data sets of length exceeding $10^{10}$ points, a multi-stage procedure may be advantageous.  

A key technical requirement for intelligent sampling is that the procedure used to obtain the 1st stage estimates needs to exhibit consistency properties, e.g. $(\ref{eq:firstconsistent2})$. The choice of binary segmentation in our exposition, or its wild binary segmentation variant (which modifies BinSeg by computing the cusum statistics on an additional number of random intervals) presented in detail in Section \ref{WBINSEG-SUPP} of the Supplement, is due to their computational attractiveness and the fact that they readily provide consistent estimates of the number of change points and their locations. Yet another variant of binary segmentation that we just became aware of, named `seeded binary segmentation' \cite{kovacs2020seeded}, shows promise as a first stage procedure for intelligent sampling, and deserves future exploration. 

We now briefly turn our attention to two other popular multiple change point methods and their potential relevance to the intelligent sampling problem. 
Two popular for models defined as in $(\ref{model})$ are the estimation of multiple structural breakpoints introduced in \cite{bai1998estimating} and PELT as described in \cite{killick2012optimal}. The method described in \cite{bai1998estimating} does give consistent estimates, but only under the much stricter condition that $J$ is a constant and there exists values $\beta_1,\dots,\beta_J\in (0,1)$ such that $\tau_j=[\beta_jN]$ for all $j=1,\dots,J$ and $N$. Further, to run the actual procedure would require the use of dynamic programming which is computationally expensive ($O(N^2)$ time). With the PELT procedure, the implementation itself runs in a more manageable $O(N)$ time; however, this works under the very different Bayesian setting where the spacings $\tau_{j+1}-\tau_j$ are iid generated from some distribution. Further, PELT was built upon a procedure described in \cite{yao1984estimation}, which examines another Bayesian model where every point $\{1,\dots,N\}$ has a probability $p$ of being a change point, and the development did not go into details regarding rates of convergence of the change point estimates. Due to the theoretical and computational restrictions of the multiple structural breakpoints method and the differing framework under which PELT works, we focused our analysis on binary segmentation.

We also mention the SMUCE procedure, introduced in \cite{frick2014multiscale}, where lower probability bounds for the events $\mathbb{P}[\hat{J}\neq J]$ and $\mathbb{P}[\underset{j=1,\dots,J}{\max}\underset{i=1,\dots,\hat{J}}{\min}|\hat{\tau}_i-\tau_j|\leq c_N]$, for any sequence $c_N$, were derived.  These results can be combined to yield  $\mathbb{P}[\hat{J}=J;\quad \underset{j=1,\dots,J}{\max}|\hat{\tau}_j-\tau_j|\leq c_N]\to 1 $ under certain restrictions and for some sequences $c_N $ that are $o(\delta_N)$, and therefore could be used in the first stage of intelligent sampling. SMUCE has the flexibility of working for a broader class of error terms \footnote{Some results apply when the errors are iid from a general exponential family.} but as was stated in \cite{frick2014multiscale}, the procedure involves dynamic programming which runs in $O(N^2)$ time. This last point is less of an issue for a modified version of SMUCE designed for iid Gaussian errors with heterogeneous variances. H-SMUCE, in \cite{pein2016heterogeneous}, could run the procedure in as low as $O(N)$ time in some cases. Overall, SMUCE could be used as the first part of intelligent sampling, and the regimes of $\delta_N$ and restrictions on the subsample size $N_1$ needed for intelligent sampling to be consistent could be fleshed out in a similar manner as in this paper. However, as BinSeg and WBinSeg are somewhat easier to implement computationally, we chose to perform our analysis with them instead. 

Our simulation results indicate that for non-Gaussian errors, as well as dependent error terms (with or without Gaussian distributions), the deviations of our estimators behave like the $L$-type distributions, even in $J$ growing with $N$ settings. This suggests that our results could be extended to broader classes of error terms, and future work could consider models incorporating errors with dependence structures and/or with non-Gaussian distributions. Extending Theorem \ref{thm:increasingJasymprotics} to these settings would require an in-depth investigation into probability bounds on the argmin of a drifted random walk with non-i.i.d and/or non-Gaussian random components along the lines of the derivations in Section \ref{sec:supplementpartc} of the Supplement for random walks with Gaussian innovations. We speculate that this work would be more amenable to rigorous analysis when the tail probabilities of the error terms decay exponentially (similar to Gaussian distributions) and when the dependence is local , e.g. $m$-dependence, where each error term is only correlated with the $m$ neighbors to its left and right. 

In conclusion, any procedure used at stage 1 of intelligent sampling puts restrictions on the model specifications, as consistent second stage estimators cannot be obtained if the first stage procedure is not consistent. Established results for BinSeg, as in \cite{venkatraman1992consistency} and \cite{fryzlewicz2014wild}, consider only the i.i.d Gaussian framework. Extending the BinSeg based approach 
to a more flexible class of error terms therefore requires theoretical exploration of BinSeg's properties beyond Gaussian errors, or using alternative methods at stage 1 which do not need the Gaussian error framework, e.g. \cite{bai1998estimating} and \cite{frick2014multiscale}. 
%

\section{Supplement Part A (Single Change Point Problem)}
\subsection{Problem Setup}\label{sec:singlenonparametric}
Instead of proving Theorems \ref{thm:singlerate} and \ref{thmsingledist} directly, we shall consider a more general nonparametric result from which the two theorems will follow as a special case. As before suppose the time series data is $(x_1,Y_1),...,(x_N,Y_N)$, where $x_i=i/N$ and $Y_i=f(x_i)+\varepsilon_i$  for $i=1,...,N$. We will make the weaker assumptions that
\begin{itemize}
	\item $f$ is a right continuous function in $[0,1]$ with a single left discontinuity at some point $\tau_0\in (0,1)$, with jump size $f(\tau_0+)-f(\tau_0-)=\Delta$ 
	\item there exists a $\beta_f>0$ where $|f(x)-f(y)|\leq \beta_f|x-y|$ whenever $(x-\tau_0)(y-\tau_0)>0$ 
	\item the errors $\varepsilon_i$'s are iid $N(0,\sigma^2)$ error terms 
\end{itemize}
The main difference between this model and the model presented in section \ref{sec:singlemodel} is the looser restriction on the signal $f$: here $f$ could be any Lipschitz continuous function with a single discontinuity and not be constrained to the family of piecewise constant functions. 
\newline
\newline
\indent We will first remark on some background regarding this this model before moving on to proving some results. Estimation procedures for such a dataset can be found in \cite{loader1996change}, where one-sided polynomial fitting was used to obtain an estimate $\hat{\tau}_N$ for $\tau_N:=\lfloor N\tau_0\rfloor/N$. In summary, fix a sequence of bandwidth $h=h_N$, a non-negative integer $p$, and a kernel function $K$ with support on $[-1,1]$. Next, for all $x_m\in (h,1-h)$, consider the signal estimates

\begin{eqnarray}
&&\hat{f}_-(x_m):=\pi_1\left(\underset{(a_0,\dots,a_p)\in\mathbb{R}^{p+1}}{\arg\min}\left(\sum_{j=0}^{Nh}K\left(\frac{j}{Nh}\right)\big (Y_{m-j-1}-a_0-a_1j-...-a_pj^p\big)^2 \right)\right)\nonumber\\
&&\hat{f}_+(x_m):=\pi_1\left(\underset{(a_0,\dots,a_p)\in\mathbb{R}^{p+1}}{\arg\min}\left(\sum_{j=0}^{Nh}K\left(\frac{j}{nh}\right)\big (Y_{m+j}-a_0-a_1j-...-a_pj^p\big)^2\right)\right),
\end{eqnarray}
where $\pi_1$ is the projection functions such that $\pi_1(a_0,\dots,a_p)=a_0$. The change point estimate is
\begin{eqnarray}
\hat{\tau}_N:=\underset{x_i\in (h,1-h)}{\arg\max}|\hat{f}_+(x_i)-\hat{f}_-(x_i)|.
\end{eqnarray}
This estimator is consistent under a few regularity conditions on the kernel $K$ and conditions on how fast $h$ converges to 0. For the sake of brevity we will not mention all those conditions here, but we will note that under said conditions, $\mathbb{P}[N(\hat{\tau}_N-\tau_N)=k]\to\mathbb{P}[L(\Delta/\sigma)=k]$ for all $k\in\mathbb{Z}$, which is the exactly asymptotic result for $\hat{\tau}_N$ obtained by least squares in a stump model setting. Finally, for our purposes we propose estimators $\hat{\alpha}$ and $\hat{\beta}$ for $f(\tau_0-)$ and $f(\tau_0+)$, respectively, defining them as
\begin{eqnarray}
\hat{\alpha}:=\frac{\sum_{j=0}^{Nh} K\left(\frac{j}{Nh}\right)Y_{N\hat{\tau}_N-j-1} }{\sum_{j=1}^{Nh}K\left( \frac{j}{Nh} \right)}\nonumber\\
\hat{\beta}:=\frac{\sum_{j=0}^{Nh} K\left(\frac{j}{Nh}\right)Y_{N\hat{\tau}_N+j} }{\sum_{j=1}^{Nh}K\left( \frac{j}{Nh} \right)}.
\end{eqnarray}
These two estimators are consistent:
\begin{lemma}
	\label{cor:frateconv} $|\hat{\alpha}-f(\tau_0+)|$ and $|\hat{\beta}-f(\tau_0-)|$ are $O_p(h\vee (Nh)^{-1/2})$.
\end{lemma}
~\newline
\indent It is possible to perform intelligent sampling to this nonparametric setting as in steps (ISS1)-(ISS4), though with a slight adjustment. Instead of fitting a stump function at step (ISS2), use one-sided local polynomial fitting with bandwidth $h$ on the first stage subsample to obtain estimates $(\hat{\alpha}^{(1)},\hat{\beta}^{(1)},\hat{\tau}_N)$ for the parameters $(f(\tau_0-),f(\tau_0+),\tau_N)$. These first stage estimators satisfy the following consistency result: 
\begin{eqnarray}\label{eq:singleconsistentcond}
\mathbb{P}\left[ |\hat{\tau}_N^{(1)}-\tau_N|\leq w(N);\quad |\hat{\alpha}^{(1)}-f(\tau_0-)|\vee|\hat{\beta}^{(1)}-f(\tau_0+)|\leq \rho_N \right]\to 1
\end{eqnarray}
for the sequence $w(N)=CN^{1-\gamma+\delta}$ where $\delta$ and $C$ can be any positive constants, and some sequence $\rho_N\to 0$ (an explicit sequence can be derived by Lemma \ref{cor:frateconv}). The consistency condition in (\ref{eq:singleconsistentcond}) is sufficient for a generalized versions of Theorems \ref{thm:singlerate} and \ref{thmsingledist}:
\begin{theorem}\label{thm:singlerategeneral}
	\begin{eqnarray}
	\hat{\tau}^{(2)}-\tau_N=O_p(N^{-1})
	\end{eqnarray}
\end{theorem}
\begin{theorem}\label{thm:singledistgeneral}
	Suppose the conditions of Theorem \ref{thm:singledistgeneral} are satisfied, then for all integers $k\in\mathbb{Z}$ we have
	\begin{eqnarray}
	\mathbb{P}\left[ \lambda_2\left( \tau_N,\hat{\tau}_N^{(2)} \right)=k \right]\to\mathbb{P}\left[ L_{\Delta/\sigma}=k \right]
	\end{eqnarray}
\end{theorem}
These results can hold under a general nonparametric setting, but they still hold for the stump model from Section \ref{sec:single-CP}. Since consistency condition $(\ref{eq:singleconsistentcond})$ is if $f$ is a stump function and least square fitting was used at step (ISS2) as it was written in Section \ref{sec:singleprocedure}, Theorems \ref{thm:singlerategeneral} and \ref{thm:singledistgeneral} do imply Theorems \ref{thm:singlerate} and \ref{thmsingledist}. The proof of Theorem \ref{thm:singlerategeneral} will be covered in Section \ref{sec:proofgeneralratesingle}, while the proof of Theorem \ref{thm:singledistgeneral} will be covered in Appendix B at Section \ref{sec:proofsingledistgeneral}.
\begin{remark}
	We note that not only do the consistency results of intelligent sampling for stump models generalize to this nonparametric setting, the computational time aspects translates also does not change. Local polynomial fitting on $n$ data points takes $O(N)$ computational time, see e.g. \cite{seifert1994fast}. Therefore the computational time analysis in Section \ref{sec:singlemodel} still holds for this nonparametric case.
\end{remark}
\begin{remark}\label{rem:nosingleasymgen}
	It is not possible derive asymptotic distribution results for $N(\hat{\tau}_N^{(2)}-\tau_N)$. Consider the case where $\tau=0.5$, $N_1=\sqrt{N}$ (or $\gamma=0.5$), and the two subsequences $N=2^{2j}$ or $N=3^{2j}$ for some large integer $j$. In such cases the first stage subsample would choose points that have integer multiples of $1/2^j$ or $1/3^j$ as their x-coordinate.
	
	\begin{itemize}
		\item If $N=2^{2j}$, $\tau_N=\frac{\lfloor 2^{2j}\cdot 0.5\rfloor}{2^{2j}}=0.5$ is an integer multiple of $1/2^j$, and hence $\tau_N$ is an x-coordinate used in the first stage.
		\item If $N=3^{2j}$, then $\tau_N=\frac{\lfloor 3^{2j}\cdot 0.5\rfloor}{3^{2j}}$, and it can be checked that $\lfloor 3^{2j}\cdot 0.5\rfloor$ is an even integer not divisible by 3. Since the x-coordinate of every first stage data point takes the form $\frac{k}{3^j}=\frac{3^jk}{3^{2j}}$ for some integer $k$, this means $\tau_N$ is not used in the first stage.
	\end{itemize}
	Hence, in the former case we cannot ever have $\hat{\tau}^{(2)}_N=\tau_N$, while in the latter case, we have $\tau^{(2)}_N=\tau_N$ and Theorem \ref{thmsingledist} tells us that $\mathbb{P}[\hat{\tau}^{(2)}_N=\tau_N]$ converges to the nonzero $\mathbb{P}[L(\Delta/\sigma)=0]$ as $j$ increases. Clearly, we have two subsequences for which $\mathbb{P}[\hat{\tau}^{(2)}_N=\tau_N]$ converges to different values.
\end{remark}
\indent To further validate the extension to this nonparametric setting, we also ran a set of simulations for when $Y_i=2\sin(4\pi x_i)+2\cdot 1(x_i>0.5)+\varepsilon_i$, where $\varepsilon_i$ are iid $N(0,1)$. We took 15 values of $N$ between 2500 and $10^6$, chosen evenly on the log scale, and applied intelligent sampling on 1000 replicates. For each of these values of $N$. First stage used roughly $N_1=\sqrt{N}$ points, which were subjected to one sided local polynomial fitting with a parabolic kernel and bandwidth $h=N_1^{-0.3}$, while the second stage interval had half-width $8/\sqrt{N}$. Figures \ref{fig:singleratenonparametric} and \ref{fig:singledistnonparametric} show results consistent with Theorems \ref{thm:singlerategeneral} and \ref{thm:singledistgeneral}.

\begin{figure}[H]
	\begin{center}
		\includegraphics[scale=0.37]{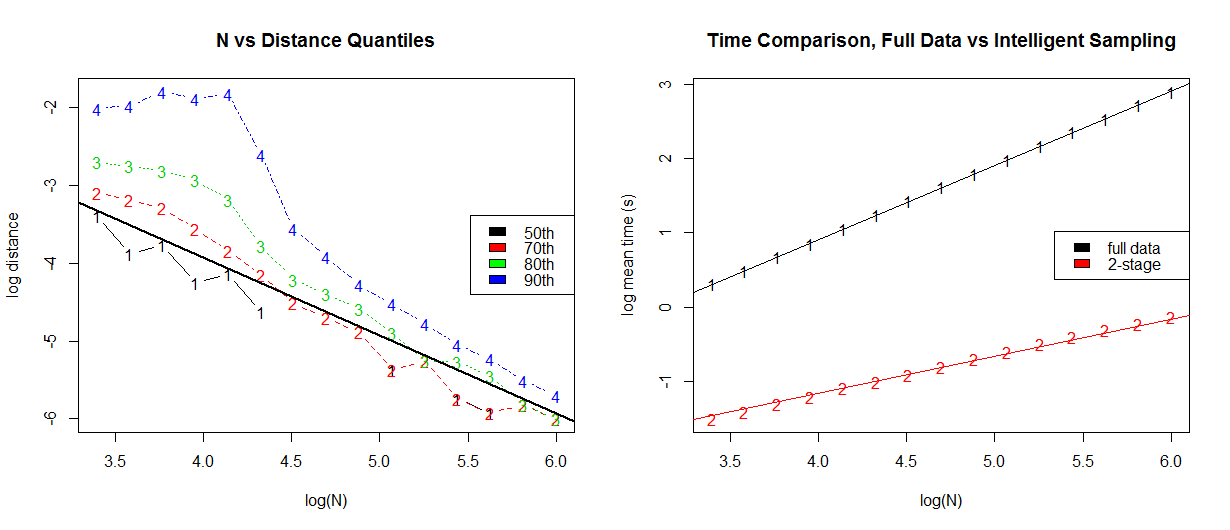}
		\caption{left graph shows log-log plot of the quantiles of $|\hat{\tau}^{(2)}_N-\tau_N|$ versus $N$, with the solid black line having a slope of exactly -1. Some datapoints for the quantiles of the 50th quantiles do not appear since for some $N$, the median of $|\hat{\tau}^{(2)}_N-\tau_N|$ was 0. Right graph is a log-log plot of the mean computational time of using all datapoints (black) and intelligent sampling (red), with the solid black line having a slope of exactly 1 and the solid red a slope of exactly 0.5.  }
		\label{fig:singleratenonparametric}
	\end{center}
\end{figure}

\begin{figure}[H]
	\begin{center}
		\includegraphics[scale=0.4]{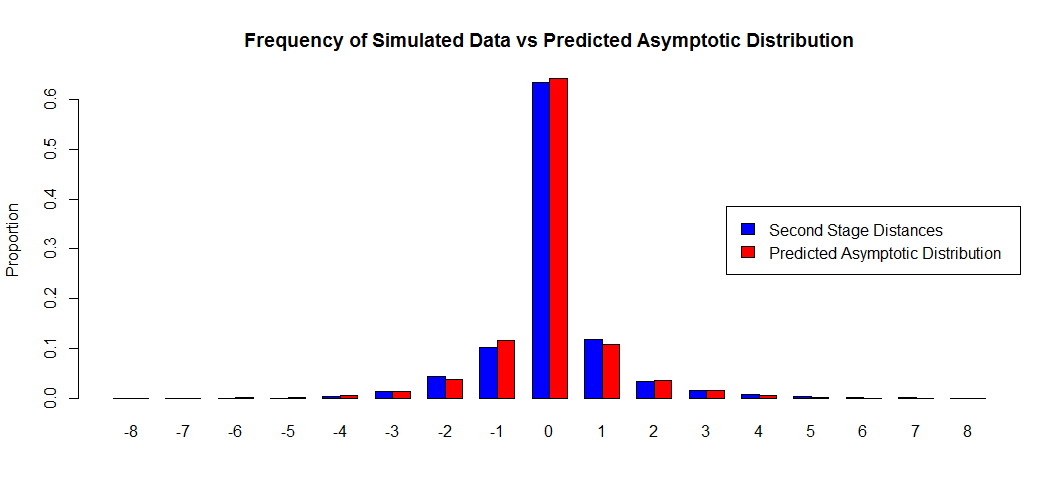}
		\caption{ distribution of $\lambda_2\left(  \tau_N,\hat{\tau}^{(2)}\right)$ values (blue) compared with the distribution of $L$ from Theorem \ref{thmsingledist}. }
		\label{fig:singledistnonparametric}
	\end{center}
\end{figure}

\subsection{Proof of Corollary \ref{cor:frateconv}}
\begin{proof}
	We will show that for any $\epsilon>0$, 
	\begin{eqnarray}
	\mathbb{P}\left[  |\hat{\beta}-f(\tau+)|>C_0(h\vee (Nh)^{-1/2})\right]\leq \epsilon.
	\end{eqnarray}
	We start off by utilizing the rate of convergence of the change point estimator: there is a constant $C_1>0$ such that 
	\begin{equation*}
	\mathbb{P}\left[|\hat{\tau}_N-\tau|> \frac{C_1}{N}\right]< \frac{\epsilon}{2}
	\end{equation*}
	for all sufficiently large $N$. Hence, for any $C>0$ we have, 
	\begin{eqnarray}
	&&\mathbb{P}\Big[ |\hat{\beta}_N)-f(\tau+)|>C \Big]\leq\nonumber\\
	&&\mathbb{P}\left[ |\hat{\beta}_N-f(\tau+)|>C \text{ and }|\hat{\tau}_N-\tau|\leq \frac{C_1}{N} \right]+\mathbb{P}\left[ |\hat{\tau}_N-\tau|> \frac{C_1}{N} \right]\leq \nonumber\\
	&&\mathbb{P}\left[ |\hat{\beta}_N(t)-f(\tau+)|> C\text{ for some }|t-\tau|\leq \frac{C_1}{N}\right]+\frac{\epsilon}{2}\label{p1}
	\end{eqnarray}
	where 
	\begin{eqnarray}
	\hat{\beta}_N(t):=\frac{\sum_{j=0}^{Nh}K\left( \frac{j}{Nh} \right)Y_{Nt+j}}{\sum_{j=1}^{Nh}K\left( \frac{j}{Nh}\right)}
	\end{eqnarray}
	Next, we bound the first term above, so consider only the case where $|t-\tau|\leq \frac{C_1}{N}$. By expanding we have
	\begin{eqnarray}
	|\hat{\beta}(t)-f(\tau+)|&=& \left| \frac{\sum_{j=0}^{Nh}K\left(\frac{j}{Nh}\right)[f(t+j/N)+\varepsilon_{Nt+j}]}{\sum_{j=1}^{Nh}K\left(\frac{j}{Nh}\right)}-f(\tau+)\right|\nonumber\\
	&\leq&  \frac{\sum_{j=0}^{Nh}K\left(\frac{j}{Nh}\right)|f(t+j/N)-f(\tau+)|}{\sum_{j=1}^{Nh}K\left(\frac{j}{Nh}\right)}+\left| \frac{\sum_{j=0}^{Nh}K\left(\frac{j}{Nh}\right)\varepsilon_{Nt+j}}{\sum_{j=1}^{Nh}K\left(\frac{j}{Nh}\right)}\right|\nonumber\\
	&:=& A(t)+|B(t)|
	\end{eqnarray}
	First, we derive a bound for $A(t)$. If $t\geq \tau$ then we have 
	\begin{eqnarray}
	A(t) &=& \frac{\sum_{j=0}^{Nh}K\left(\frac{j}{Nh}\right)|f(t+j/N)-f(\tau+)|}{\sum_{j=0}^{Nh}K\left(\frac{j}{Nh}\right)}\nonumber\\
	&\leq & \frac{\sum_{j=0}^{Nh}K\left(\frac{j}{Nh}\right)\beta_f|t+j/N-\tau|}{\sum_{j=0}^{Nh}K\left(\frac{j}{Nh}\right)}\nonumber\\
	&\leq &\frac{\beta_f\sum_{j=0}^{Nh}K\left(\frac{j}{Nh}\right)\left(\frac{C_1+j}{N}\right)}{\sum_{j=0}^{Nh}K\left(\frac{j}{Nh}\right)}\nonumber\\
	&=& \beta_f h\frac{ \frac{1}{Nh}\sum_{j=0}^{Nh}K\left(\frac{j}{Nh}\right)\left(\frac{j}{Nh}\right) }{\frac{1}{Nh} \sum_{j=0}^{Nh}K\left(\frac{j}{Nh}\right)}+\frac{C_1}{N}
	\end{eqnarray}
	Note that since $Nh\to\infty$ as $N\to \infty$, we have $\frac{1}{Nh}\sum_{j=0}^{Nh}K\left(\frac{j}{Nh}\right)\frac{j}{Nh}\to \int_0^1xK(x)\,dx$ (which exists) and $\sum_{j=1}^{Nh}K\left(\frac{j}{Nh}\right)\to \int_0^1K(x)\,dx=1$, as $N\to\infty$, hence we can find a constant $M>0$ such that
	\begin{equation}
	A(t)\leq \beta_f M h+\frac{C_1}{N}
	\end{equation} 
	for all sufficiently large $N$. On the other hand, suppose $t<\tau$. For sufficiently large $N$ we would have $N(\tau-t)\leq C_1<Nh$ and so
	\begin{eqnarray}
	A(t) &=& \frac{\sum_{j=0}^{N(\tau-t)-1}K\left(\frac{j}{Nh}\right)|f(t+j/N)-f(\tau+)|}{\sum_{j=0}^{Nh}K\left(\frac{j}{Nh}\right)}+ \frac{\sum_{j=N(\tau-t)}^{Nh}K\left(\frac{j}{Nh}\right)|f(t+j/N)-f(\tau+)|}{\sum_{j=0}^{Nh}K\left(\frac{j}{Nh}\right)}\nonumber\\
	&\leq & \frac{\sum_{j=0}^{N(\tau-t)-1}K\left(\frac{j}{Nh}\right)(\Delta+\beta_f(\tau-t-j/N))}{\sum_{j=0}^{Nh}K\left(\frac{j}{Nh}\right)}+ \frac{\sum_{j=N(\tau-t)}^{Nh}K\left(\frac{j}{Nh}\right)\beta_f(t+j/N-\tau)}{\sum_{j=0}^{Nh}K\left(\frac{j}{Nh}\right)}\nonumber\\
	&\leq & \frac{K^\uparrow N(\tau-t)\Delta }{\sum_{j=1}^{Nh}K\left(\frac{j}{Nh}\right)}+\frac{\beta_f\sum_{j=0}^{Nh}K\left(\frac{j}{Nh}\right)\left(\frac{C_1+j}{N}\right)}{\sum_{j=0}^{Nh}K\left(\frac{j}{Nh}\right)}\nonumber\\
	&&(K^\uparrow \text{ is any constants that uniformly bounds the function }K\text{ from above on}[0,1])\nonumber\\
	&\leq & (K^\uparrow \Delta)\frac{\frac{1}{Nh}C_1}{\frac{1}{Nh}\sum_{j=0}^{Nh}K\left(\frac{j}{Nh}\right)}+\beta_f h\frac{ \frac{1}{Nh}\sum_{j=0}^{Nh}K\left(\frac{j}{Nh}\right)\left(\frac{j}{Nh}\right) }{\frac{1}{Nh} \sum_{j=0}^{Nh}K\left(\frac{j}{Nh}\right)}+\frac{C_1}{N}
	\end{eqnarray}
	which, for sufficiently large $N$, can be bounded by $\frac{M_1}{Nh}+\beta_f Mh+\frac{C_1}{N}$ for some constants $M,M_1>0$. Hence, this shows that $A(t)$ itself is $O(h\vee(Nh)^{-1})$ for all $t$ where $|t-\tau|\leq \frac{C_1}{N}$.
	\newline
	\newline
	Next, we consider the random term
	\begin{equation}
	B(t)= \frac{\sum_{j=0}^{Nh}K\left(\frac{j}{Nh}\right)\varepsilon_{Nt+j}}{\sum_{j=0}^{Nh}K\left(\frac{j}{Nh}\right)}
	\end{equation}
	which satisfies
	\begin{eqnarray}
	\mathbb{E}[B(t)]&=& 0\nonumber\\
	\text{var}(B(t))&=& \frac{\sum_{j=0}^{Nh}K\left(\frac{j}{Nh}\right)^2}{\left(\sum_{j=0}^{Nh}K\left(\frac{j}{Nh}\right)\right)^2}\nonumber\\
	&=& (Nh)^{-1}\frac{\frac{1}{Nh}\sum_{j=0}^{Nh}K\left(\frac{j}{Nh}\right)^2}{\left(\frac{1}{Nh}\sum_{j=0}^{Nh}K\left(\frac{j}{Nh}\right)\right)^2}\nonumber\\
	&\leq &(Nh)^{-1}2\int_0^1K(x)^2\,dx,\qquad\text{ for all sufficiently large }N.
	\end{eqnarray}
	Thus $B(t)=O_p((Nh)^{-1/2})$ by Chebychev's inequality.
	\newline
	\newline
	Combining these results on $A(t)$ and $B(t)$ derived above, one can find constants $C_2,C_3>0$ such that for all $N>N_2$ for some integer $N_2$ we have 
	\begin{eqnarray}
	A(t)&\leq & C_2[h\vee(Nh)^{-1}]\nonumber\\
	\mathbb{P}[|B(t)|>C_3(h\vee(Nh)^{-1/2})]&\leq& \frac{\epsilon}{2(2C_1+3)}
	\end{eqnarray}
	for all $|t-\tau|\leq \frac{C_1}{N}$, to get from (\ref{p1}):
	\begin{eqnarray}
	&&\mathbb{P}\Big[ |\hat{\beta}_N-f(\tau+)|>(C_1+C_2)(h\vee(Nh)^{-1/2}) \Big]\leq\nonumber\\
	&&\mathbb{P}\left[ |\hat{\beta}_N(t)-f(\tau+)|> (C_1+C_2)(h\vee(Nh)^{-1/2})\text{ for some }|t-\tau|\leq \frac{C_1}{N}\right]+\frac{\epsilon}{2}\leq\nonumber\\
	&& \mathbb{P}\left[ A(t)+|B(t)|> (C_1+C_2)(h\vee(Nh)^{-1/2})\text{ for some }|t-\tau|\leq \frac{C_1}{N}\right]+\frac{\epsilon}{2}\leq\nonumber\\
	&& \sum_{t:|t-\tau|\leq C_1/N}\mathbb{P}\left[ A(t)+|B(t)|> (C_1+C_2)(h\vee(Nh)^{-1/2})\right] +\frac{\epsilon}{2}\leq\nonumber\\
	&& \sum_{t:|t-\tau|\leq C_1/n}\mathbb{P}\left[ |B(t)|> C_2(h\vee(Nh)^{-1/2})\right] +\frac{\epsilon}{2}\leq\nonumber\\
	&&\sum_{t:|t-\tau|\leq C_1/N}\frac{\epsilon}{2(2C_1+3)}+\frac{\epsilon}{2}\leq \epsilon
	\end{eqnarray}
	for all $N\geq N_1\vee N_2$. This establishes that $|\hat{\beta}_N-f(\tau+)|$ is $O_p(h\vee(Nh)^{-1/2})$, and the proof for $|\hat{\alpha}_N-f(\tau-)|$ proceeds similarly.
\end{proof}

\subsubsection{Proof of Theorem \ref{thm:singlerategeneral}}\label{sec:proofgeneralratesingle}
The structure of this proof will be similar to the rate of convergence proof found in Lan et al (2007). We will initially set some notations: let $\tau_N:=\lfloor N\tau\rfloor /N$, and define 
\begin{equation}
\tau^{(2)}_N:=\begin{cases}
\tau_N\qquad &\text{if }\tau_N\text{ is not in first subsample }\\
\tau_N-1/N &\text{if }\tau_N\text{ is in first subsample }
\end{cases}
\end{equation}
We will show that $\left(\hat{\tau}^{(2)}_N-\tau^{(2)}_N\right)$ is $O_p(1/N)$, which will also demonstrate the same rate of convergence for $\left(\hat{\tau}^{(2)}_N-\tau_N\right)$. An additional property of $\tau^{(2)}_N$, used later on, is the fact that $\lambda_2\left( \tau^{(2)}_N,\hat{\tau}^{(2)}_N \right)=\lambda_2\left( \tau_N,\hat{\tau}^{(2)}_N \right)$. This will be utilized in the proof of Theorems 
\ref{thmsingledist} and \ref{thm:multidepend}.

\begin{proof} 

	Denote $G_N$ as the joint distribution of $(\hat{\alpha}^{(1)},\hat{\beta}^{(1)},\hat{\tau}^{(1)}_N)$. Given any constant $\epsilon>0$, there is a positive constant $C_\epsilon$ such that for all sufficiently large $N$ we have
	\begin{eqnarray}
	(\hat{\alpha}^{(1)},\hat{\beta}^{(1)},\hat{\tau}^{(1)}_N)&\in& \left[ f(\tau-)- \rho_N,f(\tau-)+ \rho_N \right]\nonumber\\
	&\times& \left[ f(\tau+)- \rho_N,f(\tau+)+ \rho_N \right]\times \left[ \tau-C_\epsilon/N^\gamma,\tau+ C_\epsilon/N^\gamma \right]
	\end{eqnarray}
	with probability at least $1-\epsilon$. Denote this event as $R_N$. It follows that for any sequence $\{a_N\}$,
	\begin{eqnarray}
	&&\mathbb{P}\Big[ N|\hat{\tau}^{(2)}_N-\tau^{(2)}_N|>a_N \Big]\leq \nonumber\\
	&&\int_{R_N}\mathbb{P}\Big[N|\hat{\tau}^{(2)}_N-\tau^{(2)}_N|>a_N\Big|
	(\hat{\alpha}^{(1)},\hat{\beta}^{(1)},\hat{\tau}^{(1)}_N)=(\alpha,\beta,t) \Big]\,dG_N(\alpha,\beta,t)+\epsilon\leq\nonumber\\
	&&\sup_{(\alpha,\beta,t)\in R_N}\mathbb{P}\Big[N|\hat{\tau}^{(2)}_N-\tau^{(2)}_N|>a_N\Big|
	(\hat{\alpha}^{(1)},\hat{\beta}^{(1)},\hat{\tau}^{(1)}_N)=(\alpha,\beta,t) \Big]+\epsilon
	\end{eqnarray}
	Next, we show that this first term is smaller than any $\epsilon>0$ for a sequence $a_N=O(1/N)$ and all sufficiently large $N$, by bounding the probability that 
	$$\mathbb{P}_{\alpha,\beta,t}\left[N|\hat{\tau}^{(2)}_N-\tau|>a_N\right]:=\mathbb{P}\left[N|\hat{\tau}^{(2)}_N-\tau|>a_N\Big|
	(\hat{\alpha}^{(1)},\hat{\beta}^{(1)},\hat{\tau}^{(1)}_N)=(\alpha,\beta,t)\right]$$
	for any given $(\alpha,\beta,t)\in R_N$.
	\newline
	\newline
	Conditional on the first stage estimates equaling $(\alpha,\beta,t)$, we can rewrite $\hat{\tau}^{(2)}_N$ and $\tau$ as maximizers of:
	\begin{eqnarray}
	\hat{\tau}^{(2)}_N&=&
	\underset{d\in S^{(2)}}{\arg\min}\left(\frac{1}{\lambda_2(S^{(2)}(t))}\sum_{i:i/N\in S^{(2)}(t)} \left(Y_i-\frac{\alpha+\beta}{2}\right)(1(i/N\leq d)-1(i/N\leq \tau))\right) \nonumber\\
	&:=& \underset{d\in S^{(2)}}{\arg\min}\;\mathbb{M}_n(d)\nonumber\\
	\tau^{(2)}_N&=&\underset{d\in S^{(2)}}{\arg\min}\left(\frac{1}{\lambda_2(S^{(2)}(t))} \sum_{i:i/N\in S^{(2)}(t)}\left(Y_i-\frac{\alpha+\beta}{2}\right)(1(i/N\leq d)-1(i/N\leq \tau))\right)\nonumber\\
	&:=&\underset{d\in S^{(2)}}{\arg\min}\; M_n(d)
	\end{eqnarray}
	Since $\tau^{(2)}_N-C_\epsilon/N^\gamma\leq t \leq \tau^{(2)}_N+C_\epsilon/N^\gamma+2/N$, for $N$ large enough so that $KN^{\delta}/2>(C_\epsilon+2)$, we have $t-KN^{-\gamma+\delta}<\tau^{(2)}_N-KN^{-\gamma+\delta}/2<\tau^{(2)}_N+KN^{-\gamma+\delta}/2<t+KN^{-\gamma+\delta}$. This enables us to define the function $A(r)$ in the domain where $8N^{-1+\gamma-\delta}<r<K/2$, such that 
	\begin{eqnarray}\label{aexp}
	a(r)&:=&\min\{ M_n(d):|d-\tau^{(2)}_N |\geq  rN^{-\gamma+\delta}\}\nonumber\\
	&=& \min_{|d-\tau^{(2)}|\geq rN^{-\gamma+\delta}}\frac{\sum_{i:i/N\in S^{(2)}}\left(f(i/N)-\frac{\alpha+\beta}{2}\right)(1(i/N\leq d)-1(i/N\leq \tau))}{\lambda_2[t-KN^{-\gamma+\delta},t+KN^{-\gamma+\delta}]}
	\end{eqnarray}
	To make $a(r)$ simpler to work with, we show that for sufficiently large $N$, there exists a constant $A>0$ such that $a(r)\geq Ar$. First, because
	\begin{eqnarray}
	S^{(2)}(t)&:=&\left\{i/N:i\quad\in\mathbb{N},\quad i/N\in [t\pm KN^{-\gamma+\delta}],\quad \lfloor N/N_1\rfloor\text{ does not evenly divide }i \right\}\nonumber\\	
	&\subset& [t-KN^{-\gamma+\delta},t+KN^{-\gamma+\delta}]\nonumber\\
	&\subset & [\tau^{(2)}-2KN^{-\gamma+\delta},\tau^{(2)}+2KN^{-\gamma+\delta}]
	\end{eqnarray}
	this implies 
	\begin{eqnarray}
	|f(i/N)-f(\tau+)|\leq 2\beta_fKN^{-\gamma+\delta}&\qquad &\text{ for all }i/N\in S^{(2)}, i/n>\tau\nonumber\\
	|f(i/n)-f(\tau-)|\leq 2\beta_fKN^{-\gamma+\delta}&\qquad &\text{ for all }i/n\in S^{(2)}, i/n\leq\tau\nonumber
	\end{eqnarray}
	\noindent Combine this with the fact that $|\alpha- f(\tau-)|$ and $|\beta- f(\tau+)|$ are $o(1)$, which implies for sufficiently large $N$, and for any $i/n\in S^{(2)}(t)$,
	\begin{eqnarray}
	f(i/n)-\frac{\alpha+\beta}{2}>\frac{\Delta}{4}\qquad &\text{ if }i/n>\tau\nonumber\\
	f(i/n)-\frac{\alpha+\beta}{2}<-\frac{\Delta}{4}\qquad &\text{ if }i/n\leq\tau\nonumber
	\end{eqnarray}
	The preceding fact implies that every term in the summand of (\ref{aexp}) is positive, and therefore the minimizing $d$ for (\ref{aexp}) would be either $\tau^{(2)}\pm rN^{-\gamma+\delta}$:
	\begin{eqnarray}
	a(r)&=&\left(\frac{\sum_{i:i/n\in S^{(2)}(t)}\left(f(i/n)-\frac{\alpha+\beta}{2}\right)(1(i/n\leq \tau^{(2)}+rN^{-\gamma+\delta})-1(i/n\leq \tau))}{\lambda_2[t-KN^{-\gamma+\delta},t+KN^{-\gamma+\delta}]}\right)\wedge \nonumber\\
	&&\left(\frac{\sum_{i:i/n\in S^{(2)}(t)}\left(f(i/n)-\frac{\alpha+\beta}{2}\right)(1(i/n\leq \tau^{(2)}-rN^{-\gamma+\delta})-1(i/n\leq \tau))}{\lambda_2[t-KN^{-\gamma+\delta},t+KN^{-\gamma+\delta}]}\right)\nonumber\\
	&\geq& \frac{\Delta}{4}\cdot \frac{\lambda_2(\tau^{(2)},\tau^{(2)}+rN^{-\gamma+\delta}]\wedge\lambda_2(\tau^{(2)}-rN^{-\gamma+\delta},\tau^{(2)}]}{\lambda_2[t-KN^{-\gamma+\delta},t+KN^{-\gamma+\delta}]}
	\end{eqnarray}
	It can also be shown that for $N$ large enough (specifically $\lfloor N^{1-\gamma}\rfloor \geq 2$) and any $d_1,d_2\in [t-KN^{-\gamma+\delta},t+KN^{-\gamma+\delta}]$ such that $d_2-d_1 \geq 8/N$, we have
	\begin{eqnarray}
	\lambda_2(d_1,d_2]\geq [N(d_2-d_1)-2]-\left[\frac{N(d_2-d_1)}{\lfloor N^{1-\gamma} \rfloor}+1\right]\geq \frac{N(d_2-d_1)}{8}\nonumber
	\end{eqnarray}
	In a slightly similar fashion, it can be argued that for all large $N$, $\lambda_2[t-KN^{-\gamma+\delta},t+KN^{-\gamma+\delta}]\leq 3KN^{1-\gamma+\delta}$. Since we restricted $r$ to be greater than $8N^{-1+\gamma-\delta}$, this means 
	\begin{eqnarray}\label{eq:armorelinear}
	a(r)&\geq& \frac{\Delta}{4}\cdot\frac{NrN^{-\gamma+\delta}}{8\cdot 3KN^{1-\gamma+\delta}}\nonumber\\
	&\geq & \frac{\Delta}{96K}r
	\end{eqnarray}
	Hence, this shows that $a(r)$ is greater than some linear function with 0 intercept. 
	\newline
	\newline
	\indent Now define $b(r)=(a(r)-M_n(\tau^{(2)}_N))/3=a(r)/3$, then we have the following relation:
	\begin{eqnarray}
	\sup_{d\in S^{(2)}}|\mathbb{M}_n(d)-M_n(d)|\leq b(r)\quad\Rightarrow\quad |\hat{\tau}^{(2)}_N-\tau^{(2)}_N|\leq rN^{-\gamma+\delta}
	\end{eqnarray}
	To show the above is true, suppose $d\in [t-KN^{-\gamma+\delta},t+KN^{-\gamma+\delta}]$ and $|d-\tau^{(2)}_N|>rN^{-\gamma+\delta}$. If, in addition, the left expression above holds, then
	\begin{eqnarray}
	\mathbb{M}_n(d)\geq M_n(d)-b(r)\geq a(r)-b(r)\quad\Rightarrow\nonumber\\
	\mathbb{M}_n(d)-\mathbb{M}_n(\tau^{(2)})\geq a(r)-b(r)-M_n(\tau^{(2)})-b(r)=b(r)>0
	\end{eqnarray}
	Since $\mathbb{M}_n(d)>\mathbb{M}_n(\tau^{(2)}_N)$ and $\hat{\tau}^{(2)}_N$ minimizes $\mathbb{M}_n$ among all points in $S^{(2)}(t)$, this implies $d$ could not equal $\tau^{(2)}_N$, showing that $|\hat{\tau}^{(2)}_N-\tau^{(2)}_N|\leq rN^{-\gamma}$.
	\newline
	\newline
	\indent Next, we bound $\mathbb{P}_{\alpha,\beta,t}\Big[ |\hat{\tau}^{(2)}_N-\tau^{(2)}_N|\leq rN^{-\gamma+\delta}  \Big]$. First, we  split it into the two parts: 
	\begin{eqnarray}
	&&\mathbb{P}_{\alpha,\beta,t}\Big[|\hat{\tau}^{(2)}_N-\tau^{(2)}_N|>rN^{-\gamma+\delta}\Big]\leq \nonumber\\
	&&\mathbb{P}_{\alpha,\beta,t}\left[rN^{-\gamma+\delta}< |\hat{\tau}^{(2)}_N-\tau^{(2)}_N|\leq \eta N^{-\gamma+\delta} \right]+\mathbb{P}_{\alpha,\beta,t}\left[ |\hat{\tau}^{(2)}_N-\tau^{(2)}_N|> \eta N^{-\gamma+\delta} \right]:=\nonumber\\
	&&P_N(\alpha,\beta,t)+Q_N(\alpha,\beta,t)
	\end{eqnarray}
	where $\eta=K/3$. We first consider the term $P_n(\alpha,\beta,t)$. Because
	\begin{eqnarray}
	rN^{-\gamma+\delta}<|\hat{\tau}^{(2)}_N-\tau^{(2)}_N|\leq \eta N^{-\gamma+\delta}\quad&\Rightarrow&\quad \inf_{\tau^{(2)}_N+rN^{-\gamma+\delta}<d\leq \tau^{(2)}_N+\eta N^{-\gamma+\delta}}\mathbb{M}_n(d)\leq \mathbb{M}_n(\tau)
	\nonumber\\
	&\text{or }&\inf_{\tau^{(2)}_N-\eta N^{-\gamma+\delta}<d\leq \tau^{(2)}_N-rN^{-\gamma+\delta}}\mathbb{M}_n(d)\leq \mathbb{M}_n(\tau),
	\end{eqnarray}
	we can first split $P_n(\alpha,\beta,t)$ into the two terms
	\begin{eqnarray}
	P_N(\alpha,\beta,t)&\leq& P_{N,1}(\alpha,\beta,t)+P_{N,2}(\alpha,\beta,t)\nonumber\\
	&=:& \mathbb{P}_{\alpha,\beta,t}\left[ \sup_{\tau^{(2)}_N+rN^{-\gamma+\delta}<d\leq \tau^{(2)}_N+\eta N^{-\gamma+\delta}}(\mathbb{M}_n(\tau)-\mathbb{M}_n(d))\geq 0 \right]+\nonumber\\
	&&\mathbb{P}_{\alpha,\beta,t}\left[ \sup_{\tau^{(2)}_N-\eta N^{-\gamma+\delta}<d\leq \tau^{(2)}_N-rN^{-\gamma+\delta}}(\mathbb{M}_n(\tau)-\mathbb{M}_n(d))\geq 0 \right]
	\end{eqnarray}
	We first form an upper bound for $P_{n,1}(\alpha,\beta,t)$ for all $(\alpha,\beta,t)\in R_n$. Note that 
	\begin{eqnarray}
	&&\mathbb{M}_n(\tau^{(2)}_N)-\mathbb{M}_n(d)\nonumber\\&=& -(\mathbb{M}_n(d)-M_n(d))-M_n(d)\nonumber\\
	&=& -\frac{\sum\limits_{i:\;i/N\in S^{(2)}(t)}\left[\left(Y_i-\frac{\alpha+\beta}{2}\right)-\left(f(i/N)-\frac{\alpha+\beta}{2}\right)\right](1(i/N\leq d)-1(i/N\leq \tau))}{\lambda_2[t-KN^{-\gamma+\delta},t+kN^{-\gamma+\delta}]}-M_n(d)\nonumber\\
	&=& -\frac{ \sum\limits_{i:\;i/N\in S^{(2)}(t)\cap (\tau^{(2)},d]}\varepsilon_i }{\lambda_2[t-KN^{-\gamma+\delta},t+kN^{-\gamma+\delta}]}-\frac{\sum\limits_{i:\;i/N\in S^{(2)}(t)\cap(\tau^{(2)},d]}\left(f(i/N)-\frac{\alpha+\beta}{2}\right)}{\lambda_2[t-KN^{-\gamma+\delta},t+kN^{-\gamma+\delta}]}
	\end{eqnarray}
	As previously explained, the $\left(f(i/N)-\frac{\alpha+\beta}{2}\right)$ term in the second summand can be bounded below by $\Delta/4$ for all sufficiently large $N$, and hence this leads to:
	\begin{eqnarray}
	\mathbb{M}_n(\tau)-\mathbb{M}_n(d)\geq 0\quad\Rightarrow\nonumber\\
	-\sum_{i:\;i/N\in S^{(2)}\cap (\tau^{(2)},d]}\varepsilon_i\geq \frac{\Delta}{4}\lambda_2(\tau^{(2)},d]
	\end{eqnarray}
	It thus follows that
	\begin{eqnarray}
	P_{N,1}(\alpha,\beta,t)\leq \mathbb{P}_{\alpha,\beta,t}\left[\sup_{\substack{\tau^{(2)}_N+rN^{-\gamma+\delta}<d\\\leq \tau^{(2)}_N+\eta N^{-\gamma+\delta}}} \left(\frac{1}{\lambda_2(\tau^{(2)},d]}\right)\left|\sum_{{i:\;i/N\in S^{(2)}(t)\cap (\tau^{(2)},d]}}\varepsilon_i\right|\geq \frac{\Delta}{4}\right]
	\end{eqnarray}
	and by the Hajek-Renyi inequality, we get
	\begin{eqnarray}
	&&\mathbb{P}_{\alpha,\beta,t}\left[\sup_{\substack{\tau^{(2)}_N+rN^{-\gamma+\delta}<d\\\leq \tau^{(2)}_N+\eta N^{-\gamma+\delta}}} \left(\frac{1}{\lambda_2(\tau,d]}\right)\left|\sum_{i:\;i/N\in S^{(2)}(t)\cap (\tau^{(2)},d]}\varepsilon_i\right|\geq \frac{\Delta}{4}\right]\nonumber\\
	&\leq& \frac{16}{\Delta^2}\left(\frac{1}{\lambda_2(\tau^{(2)},\tau^{(2)}+rN^{-\gamma+\delta}]}+\sum_{j=\lambda_2(\tau^{(2)},\tau^{(2)}+rN^{-\gamma+\delta}]}^{\lambda_2(\tau^{(2)},\tau^{(2)}+\eta N^{-\gamma+\delta}]}\frac{1}{j^2}\right)\nonumber\\
	&\leq &\frac{32}{\Delta^2}\cdot\frac{1}{\lambda_2(\tau^{(2)},\tau^{(2)}+rN^{-\gamma+\delta}]}
	\end{eqnarray}
	We argued earlier that $\lambda_2 (\tau^{(2)},\tau^{(2)}+rN^{-\gamma+\delta}]\geq rN^{1-\gamma+\delta}/8$ for $N$ sufficiently large enough, thus 
	\begin{equation}
	P_{N,1}(\alpha,\beta,t)\leq \frac{8B}{rN^{1-\gamma+\delta}}
	\end{equation}
	where  $B=32/\Delta^2$. From this expression we arrive at $P_{N,1}(\alpha,\beta,t)\leq \epsilon$ (for any $\epsilon>0$) eventually, by setting $r=CN^{-1+\gamma-\delta}$ where $C$ is any constant satisfying $C>8$ and $8B/C\leq \epsilon$.
	\newline
	\newline
	To bound $Q_N(\alpha,\beta,t)$, from (\ref{eq:armorelinear}) we've argued that $a(r)$ is eventually greater than a multiple of $r$ when $r> 8N^{-1+\gamma-\delta}$. Since we've defined $b(r)=a(r)/3$, we can find some positive constant $B'$ where $b(r)\geq B'r$ when $r> 8N^{-1+\gamma-\delta}$ (and for all large $N$). Since $\eta=K/3> 8N^{-1+\gamma-\delta}$, eventually, this leads to
	\begin{eqnarray}
	&&\mathbb{P}_{\alpha,\beta,t}\Big[ |\hat{\tau}^{(2)}-d|>\eta N^{-\gamma+\delta} \Big]\nonumber\\
	&\leq& \mathbb{P}_{\alpha,\beta,t}\left[ \sup_{d\in S^{(2)}(t)}|\mathbb{M}_n(d)-M_n(d)|>b(\eta)\right]\nonumber\\
	&\leq& \mathbb{P}_{\alpha,\beta,t}\left[ \sup_{d\in S^{(2)}(t)}|\mathbb{M}_n(d)-M_n(d)|>B'\eta\right]\nonumber\\
	&=&\mathbb{P}_{\alpha,\beta,t}\left[ \sup_{d\in S^{(2)}(t)}\frac{\left|\sum_{i:\; i/N\in S^{(2)}(t)}{\epsilon_i}(\mathbbm{1}(i/N\leq d)-\mathbbm{1}(i/N\leq \tau))\right|}{\lambda_2(S^{(2)}(t))}>B'\eta \right]
	\end{eqnarray}
	Using Corollary 8.8 from \cite{geer2000empirical}, the latter expression is bounded by $C_1\exp(-C_2\eta^2\lambda_2(S^{(2)}(t)))$ for some positive constants $C_1, C_2$, which converges to 0.
\end{proof}

\subsubsection{Proof of Theorem \ref{thm:singledistgeneral}}\label{sec:proofsingledistgeneral}
\begin{proof} Let $\{x_1^{(2)},x_2^{(2)},...\}$ be the x-coordinates of the data, not used in the first stage, with corresponding response variable $(Y_1^{(2)},Y_2^{(2)},...)$ and error terms $(\varepsilon^{(2)}_1,\varepsilon_2^{(2)},...)$. As a set,  $\{x_1^{(2)},x_2^{(2)},...\}$ equals $\{x_1  ,...,x_N\}-\left\{ \frac{\lfloor N/N_1\rfloor}{N},\frac{2\lfloor N/N_1\rfloor}{N},... \right\}$. Note that we do not have $x_{j}^{(2)}=j/N$ for every integer $j$, and additionally we can write $\tau^{(2)}=x_m^{(2)}$ for some integer $m$. Since our estimate will also be one of the $x^{(2)}_i$'s, we can then denote $\hat{m}$ be the integer such that $\hat{\tau}^{(2)}=x^{(2)}_{\hat{m}}$. Note that we have the following relation between $\hat{m}-m$ and the $\lambda_2$ function on intervals:
	\begin{eqnarray}
	\hat{m}-m=\begin{cases}
	\lambda_2(\tau^{(2)},\hat{\tau}^{(2)}]\qquad& \text{ when }\hat{\tau}^{(2)}>\tau\\
	-\lambda_2(\hat{\tau}^{(2)},\tau^{(2)}]&\text{ when }\hat{\tau}^{(2)}\leq\tau
	\end{cases}
	\end{eqnarray}
	Hence we can write results on $\lambda_2(\tau^{(2)},\hat{\tau}^{(2)})$ in terms of $\hat{m}-m$.
	\newline
	\newline
	\indent After taking a subset $S^{(2)}$ of $\{x_1^{(2)},x_2^{(2)},...,x_{N-N_1}^{(2)}\}$ (specifically $S^{(2)}$ are those within $KN^{-\gamma+\delta}$ of the pilot estimate $\hat{\tau}^{(1)}$), we minimize
	\begin{eqnarray}
	\hat{\Delta}^{(2)}(t)&:=&\sum_{i:x_i\in S^{(2)}} \left(Y_i-\frac{\hat{\alpha}_N^{(1)}+\hat{\beta}_N^{(1)}}{2}\right)(1(x_i\leq t)-1(x_i\leq \tau))\nonumber\\
	&=& \sum_{i:x_i^{(2)}\in S^{(2)}} \left(Y_{i}^{(2)}-\frac{\hat{\alpha}_N^{(1)}+\hat{\beta}_N^{(1)}}{2}\right)(1(x_i^{(2)}\leq t)-1(x_i^{(2)}\leq \tau^{(2)}))
	\end{eqnarray}
	over all points $t\in S^{(2)}$ to obtain the estimate for the change point. Equivalently the domain of $\hat{\Delta}^{(2)}(t)$ can be extended to all $t\in \{x^{(2)}_1,x^{(2)}_2,...\}$, letting
	$$\hat{\Delta}^{(2)}(t)=\max \left\{ \hat{\Delta}^{(2)}(r):r\in S^{(2)} \right\}+1 \qquad\text{ for }t\notin S^{{(2)}}$$
	The argmin of this extension is the argmin of the function restricted to $S^{(2)}$. This extended definition will be used for the next result:
	
	\begin{lemma}\label{lem:randomwalkconv} For any fixed positive integer $j_0>0$,
		\begin{eqnarray}
		\hat{\Delta}^{(2)}\left(x^{(2)}_{m+j} \right)&=&\frac{j\Delta}{2}+\epsilon^{(2)}_{m+1}+...+\epsilon^{(2)}_{m+j}+o_p(1)\qquad\text{ for }1\leq j\leq j_0\nonumber\\
		\hat{\Delta}^{(2)}\left(x^{(2)}_{m} \right)&=& 0+o_p(1)\nonumber\\
		\hat{\Delta}^{(2)}\left(x^{(2)}_{m-j} \right) &=& \frac{j\Delta}{2}-\epsilon_m^{(2)}-...-\epsilon_{m-j+1}^{(2)}+o_p(1)\qquad \text{ for } 1\leq j\leq j_0
		\end{eqnarray}
	\end{lemma}
	From this lemma it is straightforward to show the asymptotic distribution of $\lambda_2\left( \tau_N,\hat{\tau}^{(2)}_N \right)$ is the distribution of $L_{\Delta/\sigma}$, the argmax of the random process 
	\begin{eqnarray}
	X_{\Delta/\sigma}(j)=\begin{cases}
	\frac{|j|\Delta}{2}-\varepsilon_{-1}^*-...-\varepsilon_{j}^*\qquad &\text{, for }j<0\\
	0&\text{, for }j=0\\
	\frac{j\Delta}{2}+\varepsilon_1^*+...+\varepsilon_j^*&\text{, for }j>0
	\end{cases}
	\end{eqnarray}
	where the $\{ \varepsilon_j \}_{j\in\mathbb{Z}}$ are iid $N(0,\sigma^2)$ random variables.
	
	For any fixed $\epsilon>0$ and integer $j$, we will show that $\left|\mathbb{P}\left[\hat{m}-m=j\right]- \mathbb{P}\left[L_{\Delta/\sigma}=j\right]\right|\leq \epsilon$ for all sufficiently large $N$. To do this we will first establish 3 probability bounds.
	\newline
	\newline
	\textbf{First Bound}: First we will show that with high probability we can approximate the stochastic process $L_{\Delta/\sigma}$, which has support $\mathbb{Z}$, with a stochastic process $L_{\Delta/\sigma}^{(k)}$, which has a finite support $\mathbb{Z}\cap [-k,k]$.
	\newline
	\newline
	\indent We note that there exists an integer $j_1>|j|$, such that $|L_{\Delta/\sigma}|>j_1$ with probability less than $\epsilon/3$. For any integer $k$ with $k\geq j_1$, define $L_{\Delta/\sigma}^{(k)}:=\underset{|i|\leq k}{\arg\min}\{X_{\Delta/\sigma}(i)\}$. In the case that $|L_{\Delta/\sigma}|\leq k$, we have $L_{\Delta/\sigma}^{(k)}=L_{\Delta/\sigma}$, and using this we can show that $\mathbb{P}[L_{\Delta/\sigma}=j]$ is within $\epsilon/3$ of $\mathbb{P}[L_{\Delta/\sigma}^{(k)}=j]$:
	\begin{eqnarray}\label{eqrep2}
	&&\left|  \mathbb{P}[L_{\Delta/\sigma}=j]-\mathbb{P}[L_{\Delta/\sigma}^{(k)}=j] \right| \nonumber\\
	&=& \left|  \mathbb{P}\bigg[L_{\Delta/\sigma}=j,|L_{\Delta/\sigma}|\leq k\bigg]-\mathbb{P}\bigg[L_{\Delta/\sigma}^{(k)}=j,|L_{\Delta/\sigma}|\leq k\bigg]-\mathbb{P}\bigg[L_{\Delta/\sigma}^{(k)}=j,|L_{\Delta/\sigma}|> k\bigg] \right|\nonumber\\
	&\leq& \left|  \mathbb{P}\bigg[L_{\Delta/\sigma}=j,|L_{\Delta/\sigma}|\leq k\bigg]-\mathbb{P}\bigg[L_{\Delta/\sigma}^{(k)}=j,|L_{\Delta/\sigma}|\leq k\bigg]\right|+\mathbb{P}[|L_{\Delta/\sigma}|>k]\nonumber\\
	&=&\left|  \mathbb{P}\bigg[L_{\Delta/\sigma}^{(k)}=j,|L_{\Delta/\sigma}|\leq k\bigg]-\mathbb{P}\bigg[L_{\Delta/\sigma}^{(k)}=j,|L_{\Delta/\sigma}|\leq k\bigg]\right|+\mathbb{P}[|L_{\Delta/\sigma}|>k]\nonumber\\
	&\leq & 0+\frac{\epsilon}{3}
	\end{eqnarray}
	\newline
	\newline
	\textbf{Second Bound}: We will show that there exists an integer $j_0>j_1$ such that $|\hat{m}-m|\leq j_0$ with probability greater than $1-\frac{\epsilon}{3}$. From our theorem on the rate of convergence, we can find some integer $j_0>j_1$ such that for all sufficiently large $N$,
	\begin{eqnarray}
	\mathbb{P}\left[|\hat{\tau}^{(2)}-\tau|\leq \frac{j_0-2}{N}\right]>1-\frac{\epsilon}{3}.
	\end{eqnarray}
	
	\noindent When $|\hat{\tau}^{(2)}-\tau|\leq \frac{j_0-2}{N}$, we have $|\hat{m}-m|\leq j_0$; first we can show
	\begin{eqnarray}\label{ineqminor}
	\left|x_{\hat{m}}^{(2)}-x_{m}^{(2)}\right|&\leq & \left|\hat{\tau}^{(2)}-\tau\right|+\left|\tau-\tau^{(2)}\right|\nonumber\\
	&\leq & \frac{j_0-2}{N}+\frac{2}{N}\nonumber\\&=&\frac{j_0}{N},
	\end{eqnarray}
	and second, because the $\left\{ x^{(2)}_1,x^{(2)}_2,... \right\}$ grid is just the equally spaced $\left\{ 1/N,2/N,...,N/N \right\}$ with some points taken out, the result of (\ref{ineqminor}) implies $|\hat{m}-m|\leq j_0$. Hence
	\begin{eqnarray}
	\mathbb{P}\left[ |\hat{m}-m|\leq j_0\right]&\geq& \mathbb{P}\left[|\hat{\tau}^{(2)}-\tau|\leq \frac{j_0-2}{N}\right]\nonumber\\
	&>& 1-\frac{\epsilon}{3}
	\end{eqnarray}
	\textbf{Third Inequality}: Define $\hat{\tau}^{(2)}_{j_0}$ to be the minimizer of $\hat{\Delta}^{(2)}(\cdot)$ on the set $\left\{x_{m-j_0}^{(2)},x_{m-j_0+1}^{(2)},...,x_{m+j_0}^{(2)}\right\}$, and let $\hat{m}_{j_0}$ be its corresponding index such that $\hat{\tau}^{(2)}_{j_0}=x^{(2)}_{\hat{m}_{j_0}}$. In the case when $|\hat{m}-m|\leq j_0$, then $\hat{\tau}^{(2)}_{j_0}$ would be equal to $\hat{\tau}^{(2)}$, and $\hat{m}=\hat{m}_{j_0}$. Using this notation we can obtain the following bound:
	\begin{align}\label{eqrep1}
	&\left|\mathbb{P}\left[\hat{m}-m=j\right]-\mathbb{P}\left[\hat{m}_{j_0}-m=j\right]\right|\nonumber\\
	&=  \left|\mathbb{P}\bigg[ \hat{m}-m=j,|\hat{m}-m|\leq j_0 \bigg]-\mathbb{P}\bigg[\hat{m}_{j_0}-m=j,|\hat{m}-m|\leq j_0\bigg]-\mathbb{P}\bigg[\hat{m}_{j_0}-m=j,|\hat{m}-m|> j_0\bigg]\right|\nonumber\\
	&=  \left|\mathbb{P}\bigg[ \hat{m}_{j_0}-m=j,|\hat{m}-m|\leq j_0 \bigg]-\mathbb{P}\bigg[\hat{m}_{j_0}-m=j,|\hat{m}-m|\leq j_0\bigg]-\mathbb{P}\bigg[\hat{m}_{j_0}-m=j,|\hat{m}-m|> j_0\bigg]\right|\nonumber\\
	&\leq\mathbb{P}\Bigg[ |\hat{m}-m|>j_0 \Bigg]\nonumber\\
	&\leq \epsilon/3\nonumber\\
	\end{align}
	
	\indent Consider the stochastic process $\hat{\Delta}^{(2)}(x^{(2)}_{m+i})$ for $i\in \{-j_0,...,0,...,j_0\}$. The previous lemma showed that, as a random variable in $\mathbb{R}^{2j_0+1}$, $\left(\hat{\Delta}^{(2)}(x^{(2)}_{m-j_0}),...,\hat{\Delta}^{(2)}(x^{(2)}_{m+j_0})\right)$  converges in distribution to $$(X_{\Delta/\sigma}(-j_0),...,X_{\Delta/\sigma}(j_0)).$$
	Also consider the function $\text{Ind}_{min}:\mathbb{R}^{2j_0+1}\to \mathbb{Z}$, defined as 
	\begin{eqnarray}
	\text{Ind}_{min}(a_1,...,a_{2j_0+1})=\left(\underset{i=1,...,2j_0+1}{\arg\min}(a_i)\right)-(j_0+1).
	\end{eqnarray}
	It can be easily checked that $\text{Ind}_{min}$ is a continuous function, and by definition, we also have
	\begin{eqnarray}
	L^{j_0}_{\Delta/\sigma}&=& \text{Ind}_{min}(X_{\Delta/\sigma}(-j_0),...,X_{\Delta/\sigma}(j_0))\nonumber\\
	\hat{m}_{j_0}-m&=&\text{Ind}_{min} \left(\hat{\Delta}^{(2)}(x^{(2)}_{m-j_0}),...,\hat{\Delta}^{(2)}(x^{(2)}_{m+j_0})\right).
	\end{eqnarray}
	Hence, by the continuous mapping theorem we have $\hat{m}_{j_0}-m$ converging to $L^{j_0}_{\Delta/\sigma}$ in distribution. For sufficiently large $N$, the absolute difference between $\mathbb{P}[L^{j_0}_{\Delta/\sigma}=j]$ and $\mathbb{P}[\hat{m}_{j_0}-m=j]$ will be less than $\epsilon/3$.
	\newline
	\newline
	\indent Combining what we have just shown, for sufficiently large $N$ we have
	\begin{eqnarray}
	&&\left|\mathbb{P}\left[\hat{m}-m=j\right]- \mathbb{P}\left[L_{\Delta/\sigma}=j\right]\right|\nonumber\\
	&\leq & \big|\mathbb{P}\left[\hat{m}-m=j\right]-\mathbb{P}\left[\hat{m}_{j_0}-m=j\right]\big|+\big|\mathbb{P}\left[\hat{m}_{j_0}-m=j\right]-\mathbb{P}\left[L^{j_0}_{\Delta/\sigma}=j\right]\big|\nonumber\\
	&&+\big|\mathbb{P}\left[L^{j_0}_{\Delta/\sigma}=j\right]-\mathbb{P}[L_{\Delta/\sigma}=j]\big|\nonumber\\
	&\leq& \frac{\epsilon}{3}+\frac{\epsilon}{3}+\frac{\epsilon}{3}
	\end{eqnarray}
\end{proof}

\noindent \textbf{Proof of Lemma \ref{lem:randomwalkconv}}
\begin{proof} First note that with probability increasing to 1, $x_{m-j_0}^{(2)},x^{(2)}_{m-j_0+1},...,x^{(2)}_{m+j_0}$ are all contained inside $S^{(2)}$, and this fact will be shown first. Since $\hat{\tau}^{(1)}-\tau=O_p(N^{-\gamma})$, for any $\epsilon>0$ it is possible to find a constant $C>0$ such that
	\begin{eqnarray}
	\mathbb{P}\left[  \hat{\tau}^{(1)}-CN^{-\gamma}\leq \tau\leq \hat{\tau}^{(1)}+CN^{-\gamma}\right]>1-\epsilon
	\end{eqnarray}
	for all sufficiently large $N$. Additionally, for all sufficiently large $N$ we have $\frac{4+j_0}{N}\leq (KN^\delta-C)N^{-\gamma}$, and which means that if $|\tau^{(1)}-\tau|\leq CN^{-\gamma}$ then
	\begin{eqnarray}
	\hat{\tau}^{(1)}-KN^{-\gamma+\delta}&=&\hat{\tau}^{(1)}-(KN^{-\gamma}-C)N^{-\gamma}-CN^{-\gamma}\nonumber\\
	&\leq& \tau-\frac{4+2j_0}{N}\nonumber\\
	\hat{\tau}^{(1)}+KN^{-\gamma+\delta}&=&\hat{\tau}^{(1)}+(KN^{-\gamma}-C)N^{-\gamma}+CN^{-\gamma}\nonumber\\
	&\geq&\tau+\frac{4+2j_0}{N}
	\end{eqnarray}
	Finally, for all sufficiently large $N$, we have $\lfloor N^{1-\gamma}\rfloor > 2$, i.e. the first stage subsample chooses points which are spaced more than $2/N$ points apart. Hence,
	\begin{eqnarray}
	x^{(2)}_{m-j_0}&\geq & x^{(2)}_m-2\left(\frac{j_0+2}{N}\right)=\tau^{(2)}-2\left(\frac{j_0+2}{N}\right)\nonumber\\&\geq &\tau-\frac{2j_0+4}{N}\nonumber\\
	x^{(2)}_{m+j_0}&\leq & x^{(2)}_m+2\left(\frac{j_0+2}{N}\right)\nonumber\\
	&\leq & \tau+\frac{2j_0+4}{N},
	\end{eqnarray}
	which leads to the conclusion that for all $N$ large enough, we have
	\begin{eqnarray}
	1-\epsilon&< &\mathbb{P}\left[  \hat{\tau}^{(1)}-CN^{-\gamma}\leq \tau\leq \hat{\tau}^{(1)}+CN^{-\gamma}\right]\nonumber\\
	&\leq & \mathbb{P}\left[  \hat{\tau}^{(1)}-KN^{-\gamma+\delta}\leq \tau-\frac{2j_0+4}{N}< \tau+\frac{2j_0+4}{N}\leq  \hat{\tau}^{(1)}+KN^{-\gamma+\delta}\right]\nonumber\\
	&\leq &\mathbb{P}\left[ x^{(2)}_{m-j_0}\geq \hat{\tau}^{(1)}-KN^{-\gamma+\delta}\text{ and } x^{(2)}_{m+j_0}\leq \hat{\tau}^{(1)}+KN^{-\gamma+\delta}\right]\nonumber\\
	&=&\mathbb{P}\left[ x^{(2)}_{m-j_0}\text{ and } x^{(2)}_{m+j_0}\text{ are in }S^{(2)}\right]
	\end{eqnarray}
	Therefore, consider the case for which $x^{(2)}_{m-j_0}$ through $x^{(2)}_{m+j_0}$ are contained in $S^{(2)}$. Under this condition we have $\hat{\Delta }^{(2)}\left(x^{(2)}_m\right)=0$ by simple calculation, and for any $0<j\leq j_0$,
	\begin{eqnarray}
	&&\left|\hat{\Delta}^{(2)}(x^{(2)}_{m+j})-\left(\frac{j\Delta}{2}+\sum_{i=1}^j\epsilon^{(2)}_{m+i}\right)\right|\nonumber\\
	&=&\left| \sum_{i:x^{(2)}_i\in S^{(2)}\cap \big(x^{(2)}_m,x^{(2)}_{m+j}\big]}\left( Y_i^{(2)}-\frac{\hat{\alpha}_N+\hat{\beta}_N}{2}\right)-\left(\frac{j\Delta}{2}+\sum_{i=1}^j\epsilon^{(2)}_{m+i}\right)\right|\nonumber\\
	&=&\left| \sum_{i=1}^j\left( f(x_{m+i}^{(2)})-\frac{\hat{\alpha}_N^{(1)}+\hat{\beta}_N^{(1)}}{2}\right)-\frac{j\Delta}{2}\right|\nonumber\\
	&=& \Bigg|  \sum_{i=1}^j\Bigg[\left( f(x_{m+i}^{(2)})-f(\tau+)\right)+\left(\frac{f(\tau+)+f(\tau-)}{2}-\frac{\hat{\alpha}_N^{(1)}+\hat{\beta}_N^{(1)}}{2}\right)+\nonumber\\
	&&\left( \frac{f(\tau+)-f(\tau-)}{2}-\frac{\Delta}{2} \right)\Bigg]\Bigg|\nonumber\\
	&\leq& \left|   \sum_{i=1}^j\left( f(x_{m+i}^{(2)})-f(\tau+)\right)\right| +\frac{j_0}{2}\left(\left| \hat{\alpha}_N^{(1)}-f(\tau-)\right|+\left|\hat{\beta}_N^{(1)}-f(\tau+)\right|\right)
	\end{eqnarray}
	For the first term above, using an earlier argument we make the case that for sufficiently large $N$ we have $\lfloor N/N_1\rfloor>2$ and  $x^{(2)}_{m+i}\leq \tau+\frac{2i+4}{N}$, hence
	\begin{eqnarray}
	\left|   \sum_{i=1}^j\left( f(x_{m+i}^{(2)})-f(\tau+)\right)\right|&\leq &   \sum_{i=1}^j\beta_f \left| x_{m+i}^{(2)}-\tau\right|\nonumber\\
	&\leq& \beta_f\sum_{i=1}^j \frac{2i+4}{N}\nonumber\\
	&\leq & \frac{\beta_f}{N}(2j_0^2+4j_0)
	\end{eqnarray}
	For the second term, it was shown earlier that both $\left| \hat{\alpha}_N^{(1)}-f(\tau-)\right|$ and $\left|\hat{\beta}_N^{(1)}-f(\tau+)\right|$ are $O_p\left(h\vee \sqrt{\frac{1}{N_1 h}} \right)$, and hence, so is their sum. Overall, this shows that for sufficiently large $N$, $\left|\hat{\Delta}^{(2)}(x^{(2)}_{m+j})-\left(\frac{j\Delta}{2}+\sum_{i=1}^j\epsilon^{(2)}_{m+i}\right)\right|$ is (uniformly for all $0\leq j\leq j_0$) bounded above by the random variable 
	\begin{equation}\label{diffbound}
	\frac{\beta_f}{N}(2j_0^2+4j_0)+\frac{j_0}{2}\left(\left| \hat{\alpha}_N^{(1)}-f(\tau-)\right|+\left|\hat{\beta}_N^{(1)}-f(\tau+)\right|\right),
	\end{equation}
	which is $o_p(1)$. Similarly, again for any $0<j\leq j_0$,
	\begin{eqnarray}
	&&\left|\hat{\Delta}^{(2)}(x^{(2)}_{m-j})-\left(\frac{j\Delta}{2}-\sum_{i=1}^j\epsilon^{(2)}_{m-i+1}\right)\right|\nonumber\\
	&\leq& \left|   \sum_{i=0}^{j-1}\left( f(x_{m-i}^{(2)})-f(\tau)\right)\right| +\frac{j_0}{2}\left(\left| \hat{\alpha}_N^{(1)}-f(\tau-)\right|+\left|\hat{\beta}_N^{(1)}-f(\tau+)\right|\right)
	\end{eqnarray}
	which is again uniformly bounded by the expression in (\ref{diffbound}).
	\newline
	\newline
	Therefore, given any $\epsilon>0$, we have
	\begin{eqnarray}
	&&\mathbb{P}\left[ \left|\hat{\Delta}^{(2)}(x^{(2)}_{m+j})-\left(\frac{j\Delta}{2}+\sum_{i=1}^j\epsilon^{(2)}_{m+i}\right)\right|\geq \epsilon\right]\nonumber\\
	&\leq & \mathbb{P}[x^{(2)}_{m-j_0}\notin S^{(2)}\text{ and/or }x^{(2)}_{m+j_0}\notin S^{(2)}]+\nonumber\\
	&&\mathbb{P}\left[ x^{(2)}_{m-j_0},x^{(2)}_{m+j_0}\in S^{(2)}\text{, and }\left|\hat{\Delta}^{(2)}(x^{(2)}_{m+j})-\left(\frac{j\Delta}{2}+\sum_{i=1}^j\epsilon^{(2)}_{m+i}\right)\right|\geq \epsilon\right]\nonumber\\
	&\leq &\mathbb{P}[x^{(2)}_{m-j_0}\notin S^{(2)}\text{ and/or }x^{(2)}_{m+j_0}\notin S^{(2)}]+\nonumber\\
	&&\mathbb{P}\left[\frac{\beta_f}{N}(2j_0^2+4j_0)+\frac{j_0}{2}\left(\left| \hat{\alpha}_N^{(1)}-f(\tau-)\right|+\left|\hat{\beta}_N^{(1)}-f(\tau+)\right|\right)\geq \epsilon\right]\nonumber\\
	&\to & 0+0\qquad\text{ for all }0<j\leq j_0
	\end{eqnarray}
	and similarly, 
	\begin{eqnarray}
	&&\mathbb{P}\left[ \left|\hat{\Delta}^{(2)}(x^{(2)}_{m})\right|\geq \epsilon\right]\to 0\nonumber\\&&
	\mathbb{P}\left[ \left|\hat{\Delta}^{(2)}(x^{(2)}_{m-j})-\left(\frac{j\Delta}{2}+\sum_{i=1}^j\epsilon^{(2)}_{m-i+1}\right)\right|\geq \epsilon\right]\to 0
	\end{eqnarray}
\end{proof}

\section{Supplement Part B (Multiple Change Point Problem)}\label{sec:supplementB}
Here we will provide proofs for the results presented in Section \ref{sec:multiplechangepointsintro} of the main paper. The model setup will be the same as that section.

\subsection{Detailed Computational Time Analysis}\label{sec:time_order_longer}
In this section we will give a more detailed analysis on the computational time that was briefly touched upon in Section \ref{sec:BSWBS1} of the main paper. We will make the same assumption that $\delta_N/N^{1-\Xi}\to C_1$ and $J(N)/N^\Lambda\to C_2$ for some $\Lambda\in[0,\Xi]$ and some positive constants $C_1,C_2$. As a reminder, for intelligent sampling with BinSeg at stage 1, conditions (M6 (BinSeg)) and (M7 (BinSeg)) automatically impose the condition that $\Lambda\leq\Xi<1/7$.
\newline
\newline 
\indent The BinSeg procedure, when applied to a data sequence of $n$ points, takes $O(N\log (N))$ time to compute (see \cite{fryzlewicz2014wild}). Since first stage of intelligent sampling involves applying BinSeg to $O(N^\gamma)$ points, it therefore takes $O(N^\gamma\log(N))$ time to obtain the first stage estimators. After the BinSeg estimates are obtained, we use the method described in Section \ref{sec:refitting} to upgrade them to ones whose 
asymptotic distributions are known, this subsequent step only involving least squares fitting upon $O(N^\gamma)$ points and therefore requiring only $O(N^\gamma)$ computational time, leaving the total time as $O(N^\gamma\log(N))$ up to this point. 
\newline
\newline
\indent From here on, we  use Theorem \ref{thm:reconsistentineq} and construct confidence intervals $\left[ \hat{\tau}^{(1)}_j\pm\left(Q_{\Delta_j,\sigma,\sigma}(\sqrt[\hat{J}]{1-\alpha})+1\right)\lfloor \frac{N}{N_1}\rfloor \right]$ for $j=1,\dots,\hat{J}$. Lemma \ref{lem:quantbound} (see Section \ref{sec:proofquantbound}, Supplement Part C) tells us that $ Q_{\Delta_j,\sigma,\sigma}(\sqrt[\hat{J}]{1-\alpha})$ can be bounded by a multiple of $\log(\hat{J})$, and therefore conditional on the value of $\hat{J}$, the second stage of intelligent sampling will involve least squares fitting on $O\left( \hat{J}N^{1-\gamma}\log(\hat{J}) \right)$ points, taking $O\left( \hat{J}N^{1-\gamma}\log(\hat{J}) \right)$ time to compute. Although the distribution of the $\hat{J}$ obtained from BinSeg is not fully known, a consequence of Theorem \ref{frythm} is that $\mathbb{P}\left[ \hat{J}=J \right]\geq 1-CN^{-1}$ for some constant $C$, and therefore
\begin{eqnarray}
\mathbb{E}\left[ \hat{J}\log(\hat{J}) \right]\leq J\log(J)+C\frac{N\log(N)}{N}=O(J\log(N)).
\end{eqnarray}
This leads to the conclusion that the second stage has a computational time that is on average $O\left( JN^{1-\gamma}\log(N) \right)=O\left( N^{1-\gamma+\Lambda}\log(N) \right)$, and the entire procedure takes $O\left( N^{\gamma\vee(1-\gamma+\Lambda)}\log(N) \right)$ time.
\newline
\newline
\indent Using this result we could choose an optimal $\gamma $ and obtain the optimal computational time for each value of $\Xi\in [0,1/7)$ and $\Lambda\in [0,\Xi]$. This can be done by setting the order of the first stage ($O(N^\gamma\log(N))$) to equal the order of time for the second stage ($O(N^{1-\gamma+\Lambda}\log(N))$) which would be $\gamma=\frac{1+\Lambda}{2}$. However Condition (M7 (BinSeg)) prevents this from being done everywhere by placing the restriction that $\gamma>7\Xi$. Thus $\gamma_{min}$ would be the maximum of $\frac{1+\Lambda}{2}$ and $7\Xi+\eta$ ($\eta$ any tiny positive value), resulting in order $ N^{\gamma_{min}\vee(1-\gamma_{min}+\Lambda)}\log(N)$ computational time. 
\begin{itemize}
	\item For $\Xi\in [0,1/14)$, we have $\frac{1+\Lambda}{2} <4\Xi$ and hence $\gamma_{min}=\frac{1+\Lambda}{2}$ and the computational time is order $N^{(1+\Lambda)/2}\log(N)$.
	\item For $\Xi\in [1/13,1/7)$, we have $\frac{1+\Lambda}{2} >7\Xi$, hence $\gamma_{min}=\frac{1+\Lambda}{2}$ and the computational time is order $N^{4\Xi+\eta}\log(N)$.
	\item For $\Xi \in [1/14,1/13)$, $\gamma_{min}$ can be either $\frac{1+\Lambda}{2}$ or $7\Xi+\eta$, whichever is greater, and the computational time would be either $N^{(1+\Lambda)/2}\log(N)$ or $N^{7\Xi+\eta}\log(N)$. respectively.
\end{itemize}
\begin{table}[H]
	\caption{\label{table-time-binseg}Table of $\gamma_{min}$ and computational times for various values of $\Xi$. Also shown are their values for extreme value of $\Lambda$ ($\Lambda=0$ and $\Lambda=\Xi$). For $\Xi\geq 1/7$ no values of $\gamma$ will allow us to obtain consistency from Theorem \ref{frythm}}
\end{table}
\begin{table}[H]
	\begin{minipage}{0.99\textwidth}
		\centering
		\begin{tabular}{|c|c|c|c|c|}
			\hline
			$\Xi$ & $[0,1/14)$ & $[1/14,1/13)$ & [1/13,1/7) & [1/7,1] \\
			\hline
			$\gamma_{min}$ & $\frac{1+\Lambda}{2}$ & $\max\left\{\frac{1+\Lambda}{2},7\Xi+\eta \right\}$ & $7\Xi+\eta$ & N/A\\
			\hline
			Order of Time & $N^{(1+\Lambda)/2}\log(N)$ & $\max\{ N^{(1+\Lambda)/2},N^{7\Xi+\eta} \}\cdot\log(N$)  & $N^{7\Xi+\eta}\log(N)$ & N/A \\
			\hline
			\hline
			$\gamma_{min}$ ($\Lambda=0$) & $\frac{1}{2}$ & $7\Xi+\eta$ & $7\Xi+\eta$ & N/A\\
			\hline
			Time ($\Lambda=0$) & $N^{1/2}\log(N)$ & $N^{7\Xi+\eta} \log(N$)  & $N^{7\Xi+\eta}\log(N)$ & N/A \\
			\hline
			\hline
			$\gamma_{min}$ ($\Lambda=\Xi$) & $\frac{1+\Xi}{2}$ & $\frac{1+\Xi}{2}$ & $7\Xi+\eta$ & N/A\\
			\hline
			Time ($\Lambda=\Xi$) & $N^{(1+\Xi)/2}\log(N)$ & $N^{(1+\Xi)/2} \log(N)$  & $N^{7\Xi+\eta}\log(N)$ & N/A \\
			\hline
		\end{tabular}
	\end{minipage}
\end{table}

\begin{figure}[!h]
	
	\begin{center}
		\begin{overpic}[scale=0.32,tics=10]{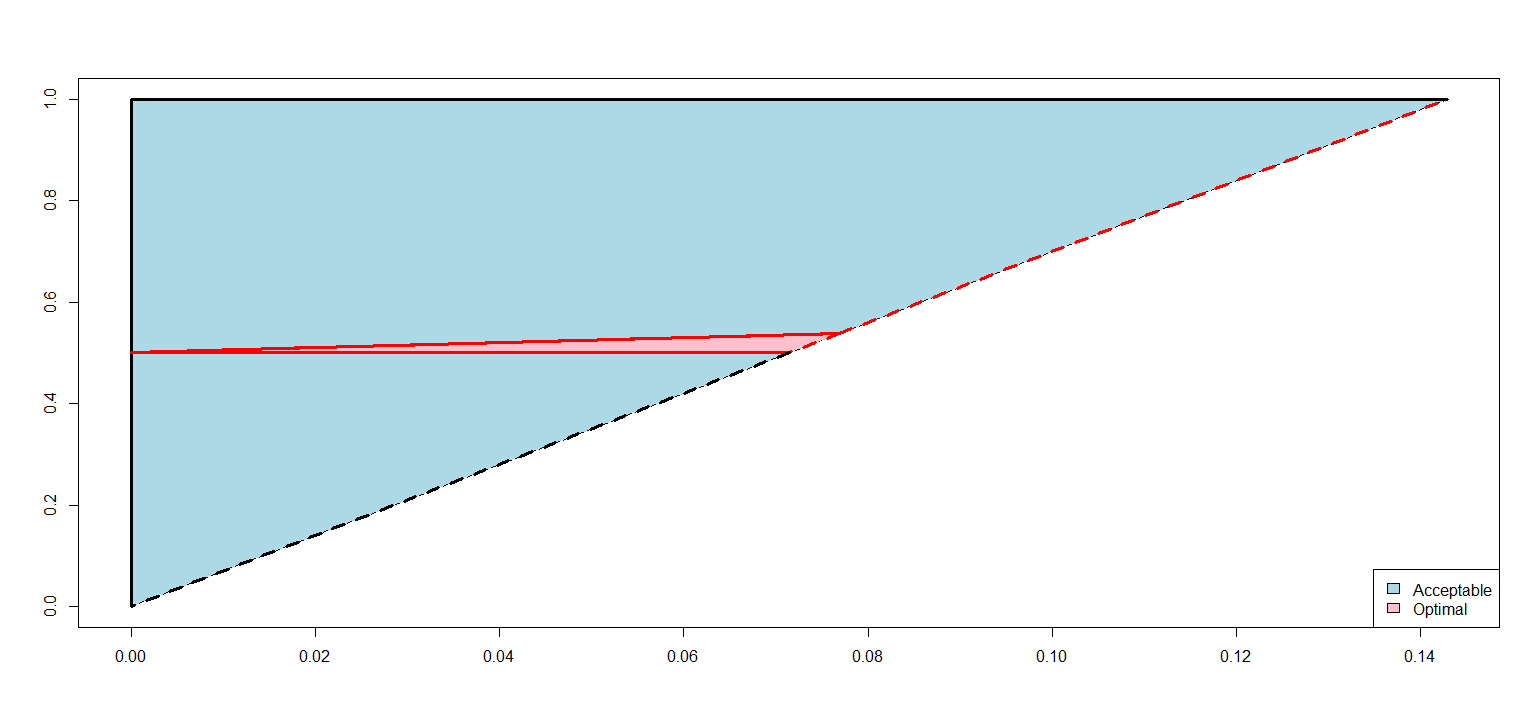}
			\put (55,-1){\Large{$\Xi$}}
			\put (0,24){\Large$\gamma$}
			\put (35,28){$ \color{red} \gamma=\frac{1+\Xi}{2} $}
			\put (20,21){$ \color{red} \gamma=\frac{1}{2} $}
			\put (55,22){$\gamma=7\Xi $}
			\put (44,45){$\gamma$ vs $\Xi$ for BinSeg}
		\end{overpic}
	\end{center}
	\caption{Blue triangle encompasses all valid values of $\gamma$ vs $\Xi$ as set by (M7 (BinSeg)). Pink region, solid red lines, and dotted red lines denotes $\gamma_{min}$ for each $\Xi$ ($\gamma_{min}$ can vary for different values of $\Lambda$ even when $\Xi$ is fixed, hence the red region).}\label{fig-binseg-time}  
\end{figure}

It can be seen that the biggest decrease in order of average computational time is for small values of $\Xi$ and $\Lambda$, and in fact for $\Xi<1/14$ and $\Lambda=0$ it is $O(\sqrt{N}\log(N))$, which is marginally slower than intelligent sampling on a single change point. For larger values of $\Xi$, there is less than a square root drop in $N\log(N)$ (order of using BinSeg on the whole data) to $N^{\gamma_{min}}\log(N)$ (intelligent sampling), to the point where as $\Xi\to 1/7$, both procedures take near the same order of time.

\begin{remark}
	Note that when implementing the intelligent sampling strategy knowledge of $\Xi$ is desirable, but in practice, its value is unknown. If one is willing to impose an upper bound on $\Xi$, intelligent sampling can be implemented with this (conservative) upper-bound.
\end{remark}

\begin{remark} 
	{\bf Multistage Intelligent Sampling in the multiple change-point problem:} We can also consider intelligent sampling with multiple $(>2)$ stages of estimation for model (\ref{model}). An $m$-stage intelligent sampling procedure would entail:
	\newline
	{\bf a.}  Take a uniform subsample $Y_{\lfloor N/N_1\rfloor},Y_{2\lfloor N/N_1\rfloor},Y_{3\lfloor N/N_1\rfloor},\dots$, where $N_1=KN^{\gamma}$ for some $K>1,\gamma\in (0,1)$, to obtain estimates $\hat{J}$, $\hat{\tau}^{(1)}_1,\dots,\hat{\tau}_{\hat{J}}^{(1)}$, and confidence intervals $\left[ \hat{\tau}^{(1)}_j-w(N),\hat{\tau}_j^{(1)}+w(N) \right]$, $1\leq j\leq\hat{J}$, for the change points.
	\newline
	{\bf b.} On each interval $\left[ \hat{\tau}^{(1)}_j-w(N),\hat{\tau}_j^{(1)}+w(N) \right]$ for $1\leq j\leq \hat{J}$ perform the $(m-1)$ stage intelligent sampling procedure for the single change point (as described in Remark \ref{rem:singlemult}). 
\end{remark}

\subsection{Sample Size Considerations from a Methodological Angle}\label{sec:compmethod}

In the asymptotic setting of Section \ref{sec:time_order_longer}, we were concerned about minimizing the \emph{order} of computational time required 
for locating the change points through intelligent sampling, assuming that certain important quantities were known. The focus in this section is on
obtaining explicit expressions for the minimum sample size that the procedure requires to correctly identify the underlying change points. Obviously,
the minimum sample size is the key driver in the computational time formlas provided, albeit not the single one, and also addresses computer memory
usage issues. In order to develop explicit expressions for the total computational time, one would need to know exactly how fast BinSeg runs versus data size, in terms of its model parameters\footnote{For example, BinSeg will generally terminate in fewer steps on a dataset with fewer change points than on another dataset of the same length but more change points.} and this is unavailable as an exact expression. Therefore, we look at minimizing the subsample utilized as a proxy, with the added benefit of deriving the least amount of data that must be held in memory at a single time.

We have already investigated the optimal order of the first stage subsample, denoted $N_1$, and demonstrated in Section \ref{sec:time_order_longer} that in the best cases the size of both the first and second stage subsamples scales as $\sqrt{N}\log(N)$. 
Although valid, these previous analyses only apply to an abstract asymptotic setting. In practice, given a data set with fixed (large) $N$, a different approach is needed to determine the optimal number of points to use at each different stage. 

Given the number of change points and their associated SNR's,  we show below how to optimally allocate samples in order to minimize the total number used, for intelligent sampling. For simplicity we assume the error terms satisfy condition (M4 (BinSeg)). We start with the two-stage intelligent sampling procedure and assume that in stage 1, roughly $N_1$ points are used for BinSeg and another $N_1$ points, for the calibration steps described in Section \ref{sec:refitting}. 
At stage 2, we work with $\hat{J}$ (which is $\approx J$) intervals. Using Theorem \ref{thm:reconsistentineq}, setting the width of the second stage intervals to be $\left(Q_{\Delta,\sigma,\sigma}\left( 1-\frac{\alpha}{J} \right)+1\right)\left\lfloor\frac{N}{N_1}\right\rfloor$ for a small $\alpha$ will ensure that they cover the true change points with high probability (close to $1-\alpha$ if not greater).
Assuming $N_1$ is large enough so that the first stage is accurate (ie $\hat{J}=J$ and $\max_j|\hat{\tau}^{(1)}_j-\tau_j|$ is small with high probability)
, the number of points used in the two stages, combined, is approximately
\begin{eqnarray}\label{eq:totalpoints}
2N_1+\frac{2\left(\sum_{j=1}^J \left(Q_{\Delta_j,\sigma,\sigma}\left(1-\frac{\alpha}{J}\right)+1\right)\right)N}{N_1}.
\end{eqnarray}
This presents a trade-off, e.g. if we decrease $N_1$ by a factor of 2, the second term in (\ref{eq:totalpoints}) increases by a factor of 2. To use a minimal number of points in both stages, we need to set $N_1=\sqrt{N\sum_{j=1}^J \left(Q_{\Delta_j,\sigma,\sigma}\left(1-\frac{\alpha}{J}\right)+1\right)}$. In turn this yields a minimum of $4\sqrt{N\sum_{j=1}^J \left(Q_{\Delta_j,\sigma,\sigma}\left(1-\frac{\alpha}{J}\right)+1\right)}$ when plugged into (\ref{eq:totalpoints}). For any given values of $N$, $J$, and SNR, this provides a lower bound on the minimum number of points that intelligent sampling must utilize, and Tables \ref{tab:percentsmallN} and \ref{tab:percentmiddleN} depict some of these lower bounds for a select number of these parameters. 

\begin{minipage}[t]{.45\linewidth}
	\centering
	\begin{figure}[H]
		\caption{\label{tab:percentsmallN}For $N=1.5\times 10^7$, the minimal percentage of data that must be used for various values of $J$ and SNR, assuming all jumps have equal SNR and $\alpha=0.01$. 
		}
		~\newline
		\centering
		\includegraphics[scale=0.4]{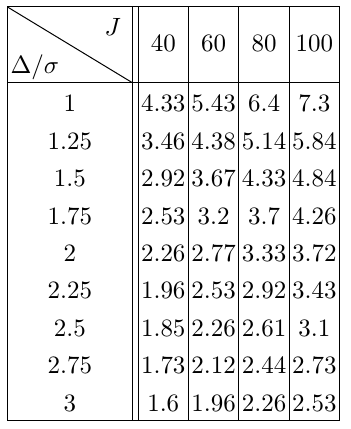}
	\end{figure}
	
	%
\end{minipage}
\begin{minipage}[t]{0.04\linewidth}
	~~	
\end{minipage}
\begin{minipage}[t]{0.45\linewidth}
	\begin{figure}[H]
		\caption{\label{tab:percentmiddleN}For $N=1.5\times 10^{10}$, the minimal percentage of data that must be used for various values of $J$ and SNR, assuming all jumps have equal SNR and $\alpha=0.01$. 
		}
		\centering
		~\newline
		\includegraphics[scale=0.4]{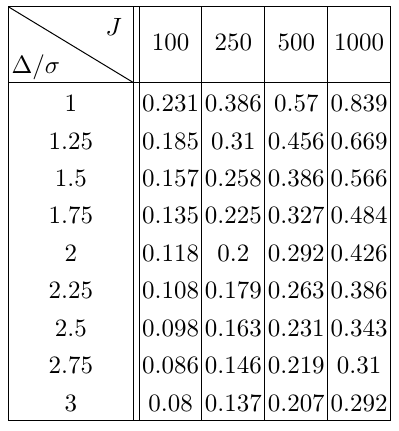}
	\end{figure}
\end{minipage}


Note that while the fraction of points used for the larger $N$ above is smaller, in absolute terms, this still translates to very large subsamples: even just 0.57\% of $1.5\times 10^{10}$ (for SNR 1 and $J=500$ on Table \ref{tab:percentmiddleN}) is a very large dataset of $8.6\times 10^7$, which almost requires server type
computer capabilities. The situation becomes more tenuous for larger values of $N$. This suggests that \emph{a larger number of stages} is in order for sample sizes of $N$ exceeding $10^{10}$. 
\newline
\newline
For a three-stage implementation, suppose $\approx N_1$ points are utilized at stage 1, letting us form simultaneous confidence intervals that are (approximately) of the form
\begin{eqnarray}
\left[ \hat{\tau}^{(1)}_j- \sum_{j=1}^J \left(Q_{\Delta_j,\sigma,\sigma}\left(1-\frac{\alpha}{J}\right)+1\right)\frac{N}{N_1},\hat{\tau}^{(1)}_j+ \sum_{j=1}^J \left(Q_{\Delta_j,\sigma,\sigma}\left(1-\frac{\alpha}{J}\right)+1\right)\frac{N}{N_1}\right]\qquad\text{for }j=1,\dots,J\nonumber
\end{eqnarray}
(assuming $\hat{J}=J$ for simplification). At stage 2, suppose at the $j$'th confidence interval we subsample roughly $N_2^{(j)}$ points, giving us a subsample which skips approximately every $2Q_{\Delta_j,\sigma,\sigma}\left(1-\frac{\alpha}{J}\right)\frac{N}{N_1N_2^{(j)}}$ points. Hence, at stage 3 we work with confidence intervals that are (approximately) of the form
\begin{eqnarray}
\left[ \hat{\tau}^{(2)}_j\pm\left(Q_{\Delta_j,\sigma,\sigma}\left(1-\frac{\alpha}{J}\right)+1\right)\left( 2\left(Q_{\Delta_j,\sigma,\sigma}\left(1-\frac{\alpha}{J}\right)+1\right)\frac{N}{N_1N_2^{(j)}} \right)\right]
\end{eqnarray}
for $j=1,\dots,J$. In total all three stages use around a total of 
\begin{eqnarray}\label{eq:threestageN}
2N_1+\sum_{j=1}^J N_2^{(j)}+\frac{4N}{N_1}\left( \sum_{j=1}^J \frac{\left(Q_{\Delta_j,\sigma,\sigma}\left(1-\frac{\alpha}{J}\right)+1\right)^2}{N_2^{(j)}} \right)
\end{eqnarray}
points. This expression is minimized by setting 
\begin{eqnarray}N_1=N^{1/3}\left(\sum_{k=1}^J \left(Q_{\Delta_j,\sigma,\sigma}\left(1-\frac{\alpha}{J}\right)+1\right)\right)^{2/3}\qquad\text{ and }\nonumber\\ N_2^{(j)}=2N^{1/3}\frac{Q_{\Delta_j,\sigma,\sigma}\left(1-\frac{\alpha}{J}\right)+1}{\left(\sum_{k=1}^J \left(Q_{\Delta_k,\sigma,\sigma}(\alpha,J)+1\right)\right)^{1/3}}
\end{eqnarray}
for $j=1,\dots,J$, which in turn gives a minimum of $6N^{1/3}\left(\sum_1^J \left(Q_{\Delta_j,\sigma,\sigma}\left(1-\frac{\alpha}{J}\right)+1\right)\right)^{2/3}$ for (\ref{eq:threestageN}). A similar analysis on a four-stage procedure would have the optimal subsample allocation as 
\begin{eqnarray}N_1=N^{1/4}\left(\sum_{k=1}^J \left(Q_{\Delta_j,\sigma,\sigma}\left(1-\frac{\alpha}{J}\right)+1\right) \right)^{3/4}\qquad\text{ and }\nonumber\\ N_2^{(j)}=N_3^{(j)}=2\left(Q_{\Delta_j,\sigma,\sigma}\left(1-\frac{\alpha}{J}\right)+1\right)N^{1/4}\left(\sum_{k=1}^J \left(Q_{\Delta_k,\sigma,\sigma}(\alpha,J)+1\right) \right)^{-1/4}\nonumber\\\text{ for }j=1,\dots,J\end{eqnarray} 
which yields a total of $8N^{1/4}\left( \sum_{k=1}^J \left(Q_{\Delta_j,\sigma,\sigma}\left(1-\frac{\alpha}{J}\right)+1\right) \right)^{3/4}$ points utilized.
~\newline
\begin{minipage}[t]{0.45\linewidth}
	\centering
	\begin{figure}[H]
		\centering
		\caption{\label{tab:threestagemiddleeN} For $N=1.5\times 10^{10}$,
			minimal percentage of the data that must be used for a three stage
			procedure, assuming all jumps have equal SNR and $\alpha=0.01$. 
		}
		~\newline 
		\includegraphics[scale=0.4]{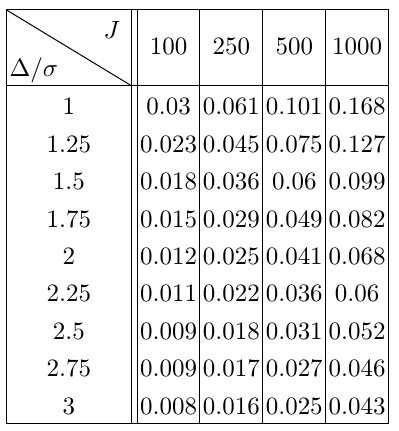}
	\end{figure}
\end{minipage}
\begin{minipage}[t]{0.04\linewidth}
	~~
\end{minipage}
\begin{minipage}[t]{0.45\linewidth}
	\begin{figure}[H]
		\centering
		\caption{\label{tab:fourstagelargeN} For $N=1.5\times 10^{12}$, minimal percentage of the data 
			that must be used for a four stage procedure,
			assuming all jumps have equal SNR and $\alpha=0.01$. 
		}
		\includegraphics[scale=0.4]{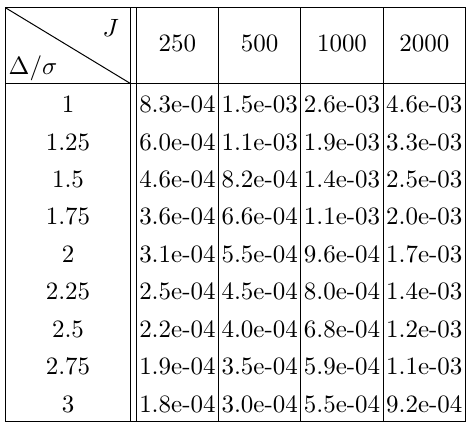}
	\end{figure}
	
\end{minipage}
~\newline
Comparing Figures \ref{tab:percentmiddleN} and \ref{tab:threestagemiddleeN}, we focus on the case of 1000 change points with SNR 1.5: using three stages allows us to decrease the minimal required points by a factor of around five. The ease on computations is greater when looking at the largest amount of data the computer must handle at a time:\setcounter{footnote}{0}\footnote{For intelligent sampling the largest data subset the computer has to work with and hold in memory at any moment, under these optimal allocations and when all change points have equal SNR, is the roughly $N_1$ sized data set used at the initial step for BinSeg. All subsequent steps can work with sub-intervals of data less than $N_1$ in size.} this is $N_1\approx 2.1\times 10^7$ for two stages and $N_1\approx2.5\times 10^6$ for three stages, a decrease by a factor of 9. Meanwhile for a dataset of size 1.5 trillion, using four stages allows us to work with subsamples of size at most $N_1\approx 4.7\times 10^6$ for the more demanding scenario of SNR 1.5 and 2000 change points, a very manageable dataset for most computers. 

We note here that these optimal allocations are valid assuming that BinSeg is able to pin down $\hat{J}$ and the change points with the initial subsample. In general, this will be the case provided the SNR is reasonable, and the initial subsample is large enough so that the change-points are adequately spaced apart. For example, in the context of the above tables, one can ask whether BinSeg will accurately estimate the parameters on a 2.4 million length dataset with 1000 evenly spaced change points, or 2000 change points on a 4.7 million length data with 2000 evenly spaced change points, under a constant SNR of 1.5 (which is of modest intensity). To this end, we ran a set of simulations 
and concluded
that if there are over 1000 data points between consecutive change points of SNR 1.5, based on these two settings and for appropriate tuning parameters, BinSeg's estimators satisfy $\hat{J}=J$ and $\max|\hat{\tau}_j-\tau_j|\leq 150$ with probability over 99\%.

Observe also that the formulas provided depend on the values of the SNRs at the change points and the actual number of change points ($J$). 
In practice, neither will be known, and the the practitioner will not be able to determine the derived allocations exactly. In such situations, conservative lower bounds on the SNRs and a conservative higher bound on $J$, can yield valid (but conservative) sampling allocations when plugged in to the expressions derived through this section. Such bounds can be obtained if background information about the problem and the data are available, or via rough pilot estimates on an appropriately sparse subsample.

It is also worth pointing out that the intelligent sampling procedure is readily adaptable to a distributed computing environment, which can come into play, especially with data sets of length exceeding $10^{12}$ that are stored sequentially across several storage disks. In such cases, the two sparse subsamples at the first stage, which are of much smaller order, can be transferred over to a central server (a much easier exercise than transferring all the data on to one server), where the first is analyzed via binary segmentation to determine the initial change-points, and the other used for the re-estimation procedure and associated confidence intervals as described in Section \ref{sec:refitting}. As the number of disks on which the data are stored is of a much smaller order than the length of the data, each re-estimated change-point and its associated confidence interval will typically belong to a stretch of data completely contained within one storage disk, and the subsequent resampling and estimation steps can be performed on the local processor, after the information on the confidence interval has been transferred back from the central server. An occasional communication between two machines may be necessary.

\subsection{Dependent Errors}\label{sec:dependenterrors}

The proposed intelligent sampling procedure for multiple change point problems has so far been presented in the setting of i.i.d. data for a signal-plus-noise model. However, many data sequences (such as time series) usually exhibit temporal correlation. Hence, it is of interest to examine the properties of the procedure under a non-i.i.d. data generating mechanism.

While we believe that results akin to Theorem 3 (growing number of change-points in the i.i.d. error regime) should go through under various forms of dependence among errors, a theoretical treatment of this would require a full investigation of the tail properties of random walks under dependent increments and is outside the scope of this paper. An asymptotic distributional result, analogous to Theorem 5, under finitely many change points in the dependent regime is also expected to hold. We present below a proposition for the finite $J$ case under a set of high-level assumptions. 

Suppose that the data sequence is in the form (\ref{model}) and satisfies conditions (M1) to (M3) and (M6). Upon the error terms, we impose the assumption that they have an autocorrelation structure which dies out at a polynomial rate or faster, and locally around the change points assume that the joint distributions of the errors are fixed [i.e. invariant to $N$]:
\begin{enumerate}[label=(M4-alt\arabic*):]
	\setlength{\itemindent}{.5in}
	\item $\varepsilon_j$'s are each marginally $N(0,\sigma_j^2)$, and there exist positive constants $\sigma_{max}$, $B$ and $\alpha$, independent of $N$, such that $\sigma_j\leq \sigma_{max}$ and $\text{cor}(\varepsilon_j,\varepsilon_{j+k})\leq Bk^{-\alpha}$ for any $j$ and $j+k$ from 1 to $N$.
	\item there exists a sequence $w_e(N)\to\infty$ and Gaussian sequences $\{\epsilon_{i,j} \}_{i\in\mathbb{Z}}$ (not required to be stationary) for $j=1,\dots,J$, such that for all $j=1,\dots,J$ and all sufficiently large $N$, $\{ \varepsilon_{\tau_j-w_e(N)},\dots,\varepsilon_{\tau_j+w_e(N)} \}$ has the same joint distribution as $\{ \epsilon_{-w_e(N),j},\dots,$ $\epsilon_{w_e(N),j} \}$.
\end{enumerate}

On a set of data where (M4-alt1) and (M4-alt2) hold (along with assumptions (M1) to (M3), and (M6)), we want steps (ISM1) and (ISM4) to go through with some procedure that ensures 
\begin{eqnarray}\label{eq:firstconsistent3}
\mathbb{P}\left[\hat{J}=J,\,\underset{i=1,...,J}{\max}|\hat{\tau}^{(1)}_i-\tau_i|\leq w(N),\, \max_{i=0,...,J}|\hat{\nu}^{(1)}_i-\nu_i|\leq \rho_N\right]\to 1.
\end{eqnarray}
for some sequence $w(N)\to\infty$ and $w(N)=o(\delta_N)$.

Next, we desire for the final estimators $\hat{\tau}^{(2)}$ to be $O_p(1)$ consistent and have the property that for each $\epsilon>0$ there exists a constant $C$ such that
\begin{eqnarray}\label{eq:secondnotiidconsistent}
\mathbb{P}\left[\hat{J}=J ;\quad\max_{i=1,\dots,J}\left|\hat{\tau}^{(2)}_i-\tau_i\right|\leq C \right]\geq 1-\epsilon \,.
\end{eqnarray}
for all sufficiently large $N$. 
\begin{proposition}\label{thm:notiidconsistent}
	Suppose conditions (M1) to (M3), (M4-alt1), (M4-alt2), and (M6) are satisfied. Next, suppose the first stage estimators satisfy (\ref{eq:firstconsistent3}) and the second stage estimators, constructed as in the i.i.d. setting but with a minor modification\footnote{See the remark right after the proposition.} satisfy (\ref{eq:secondnotiidconsistent}). Define the random walks
	\begin{eqnarray}
	Z_{i,j}=\begin{cases}
	\Delta_j(\epsilon_{1,j}+\dots+\epsilon_{i,j})-i\Delta_j^2/2,\qquad & i>0\\
	0, & i=0\\
	\Delta_j(\epsilon_{i+1,j}+\dots+ \epsilon_{0,j})-i\Delta_j^2/2, & i<0 \,,
	\end{cases}
	\end{eqnarray}
	with the $\epsilon_{i,j}$'s from condition (M4-alt2), for $j = 1, 2, \ldots, J$, and 
	denote $\tilde{L}_j:=\underset{i\in \mathbb{Z}}{\arg\min} \,Z_{i,j}$.
	Then $|\hat{\tau}^{(2)}_j-\tau_j|$'s for $j=1,\dots,J$ jointly converge to the distribution of $(\tilde{L}_1,...,\tilde{L}_J)$: for any integers $k_1,\cdot,k_J$, 
	\begin{eqnarray}
	\mathbb{P}\left[\hat{J}=J,\, |\hat{\tau}^{(2)}_j-\tau_j|=k_j \text{ for  }1\leq j\leq J\right]\to\prod_{j=1}^J \mathbb{P}[\tilde{L}_j=k_j]
	\end{eqnarray}
\end{proposition}

\begin{remark}
	As in the i.i.d. case, the intervals $[\hat{\tau}^{(1)}_j-Kw(N),\hat{\tau}^{(1)}_j+Kw(N)]$ for $j=1,\dots,J$ [obtained at step (ISM5)] would each contain only one change point with probability approaching one. We are therefore still justified in fitting stump models on each interval, although with a slight modification. Unlike the i.i.d. error terms scenario, the joint distribution of the error terms at the second stage does change when we condition on the estimators $\hat{\tau}^{(1)}_j$'s, regardless if we leave out the $\{Z_j\}$ subsample at the second stage. We thus make the following modification to (ISM5) in \ref{sec:procedure} (and assume this altered procedure is used from here on in this section):
	\begin{itemize}
		[label=(ISM5-alt):]
		\setlength{\itemindent}{.5in}
		\item take $S^{(2)}\left( \hat{\tau}_i^{(1)} \right)$ as all integers in $[\hat{\tau}_i^{(1)}-Kw(N),\hat{\tau}_i^{(1)}+Kw(N)]$ without any points omitted.
	\end{itemize}
	Step (ISM-6) then proceeds as before. 
\end{remark}

\begin{remark}
	We note that the asymptotic distribution given above and in Theorem \ref{thm:multidepend} have the same form, since as before, conditional on (\ref{eq:firstconsistent3}) being true, intelligent sampling simplifies the problem into multiple single change point problems. Using Proposition \ref{thm:notiidconsistent} to construct confidence intervals in a practical setting requires an idea of the joint distribution of the 
	$\{\epsilon_{ij}\}$'s. In practice, one would have to impose some structural conditions, e.g. assuming an ARMA or ARIMA structure on long stretches of the errors to the left and right of the change points.
\end{remark}

\begin{remark}
	The hard work lies in the verification of the high-level conditions (\ref{eq:firstconsistent3}) and (\ref{eq:secondnotiidconsistent}) in different dependent error settings, and as mentioned previously, is not dealt with in this paper but should constitute interesting agenda for follow-up research on this problem. We will, however, use Proposition 3 in our simulation and data analysis sections to construct confidence intervals for various dependent error scenarios. 
\end{remark}

\subsection{Proof of Theorem \ref{thm:multiorder}}\label{sec:multiorderotherproof}
As a reminder, this theorem only required the error terms to be independent and zero-mean subgaussian, with subguassian parameters $\sigma_i$ for $1\leq i\leq N$, which means
\begin{eqnarray}
\mathbb{E}[\exp(s\varepsilon_i)]\leq \exp\left( \frac{\sigma_i^2 s^2}{2} \right)\nonumber\\
\mathbb{P}[|\varepsilon_i|\geq s]\leq 2\exp\left( -\frac{s^2}{2\sigma_i^2} \right)
\end{eqnarray}
for all $s\in \mathbb{R}$. Since condition (M4) also states that the subgaussian parameters ($\sigma_i$'s) are bounded above by a constant $\sigma_{\max}$ not dependent on $N$, this means that 
\begin{eqnarray}
\mathbb{E}[\exp(s\varepsilon_i)]\leq \exp\left( \frac{\sigma_{\max}^2 s^2}{2} \right)\nonumber\\
\mathbb{P}[|\varepsilon_i|\geq s]\leq 2\exp\left( -\frac{s^2}{2\sigma_{\max}^2} \right)
\end{eqnarray}
for all $i=1,\dots,N$ and $s\in \mathbb{R}$.
\newline
\newline
\indent Much of the theory used in this proof and subsequent proof will extensively deal with the points used at the second stage, where points used at stage 1 will be skipped. To make these arguments, we will denote some notation to clearly , define the points $\tau^{(2)}_j$'s as 
\begin{eqnarray}
\tau^{(2)}_j:=\begin{cases}
\tau_j-1\qquad &\text{if }\tau_j\text{ was a first stage subsmaple point}\\
\tau_j &\text{otherwise}
\end{cases}
\end{eqnarray}
for $j=1,\dots,J$ define
\begin{eqnarray}
S^{(2)}(t):=\left\{ i\in\mathbb{N}:\; |i-t|\leq Kw(N),\quad Y_i\text{ not used in 1st stage subsample} \right\}.
\end{eqnarray}

\begin{proof}
	Define the event 
	\begin{eqnarray}\label{eventAN}
	\mathcal{R}_N:=\left\{ \hat{J}=J;\quad \max_{i=1,...,J}\left|\hat{\tau}^{(1)}_i-\tau_i\right|\leq w(N);\quad \max_{i=0,...,J}|\hat{\nu}_i^{(1)}-\nu_i|\leq \rho_N \right\},
	\end{eqnarray}
	Denote $G_N$ as the joint distribution of $J,\hat{\tau}^{(1)}_1,...,\hat{\tau}^{(1)}_J,\hat{\nu}^{(1)}_0,...,\hat{\nu}^{(1)}_J$; the domain of $G_N$ would be $\bigcup_{k=0}^{N-1}\mathbb{N}^{k+1}\times \mathbb{R}^{k+1}$.
	
	Then, for any sequence $\{a_N\}$, we can bound $\mathbb{P}\left[ \hat{J}=J;\quad \max_{i=1,...,J}\left| \hat{\tau}_i^{(2)}-\tau_i^{(2)}\right|\leq a_N\right]$ from below by:
	\begin{eqnarray}\label{eq:deriv1}
	&&\mathbb{P}\left[ \hat{J}=J;\quad \max_{i=1,...,J}\left| \hat{\tau}_i^{(2)}-\tau_i^{(2)}\right|\leq a_N\right]\nonumber\\
	&\geq& \sum_{k=0}^{N-1}\underset{  \substack{ 0<t_1<t_2<....<t_k<N\\v_1,...,v_k\in\mathbb{R}} }{\int}\mathbb{P}\left[\hat{J}=J;\,  \underset{i=1,...,J}{\max}\left| \hat{\tau}^{(2)}_i-\tau_i^{(2)} \right|\leq a_N \Bigr| \hat{J}=k,\hat{\tau}^{(1)}_j=t_j,\hat{\nu}^{(1)}_j=v_j\text{ for }j\leq k\right] \nonumber\\
	&& dG_N(k,t_1,...,t_k,v_0,...,v_k)\nonumber\\
	&\geq& \underset{\substack{ |t_i-\tau_i|\leq Kw(N)\\ |v_i-\nu_i|\leq \rho_{N}\\\text{for all }i }}{\int} \mathbb{P}\left[\hat{J}=J;\, \underset{i=1,...,J}{\max}\left| \hat{\tau}^{(2)}_i-\tau_i^{(2)} \right|\leq a_N \Bigr| \hat{J}=J,\hat{\tau}^{(1)}_j=t_j,\hat{\nu}^{(1)}_j=v_j\text{ for }j\leq J \right]\nonumber\\
	&& dG_N(J,t_1,...,t_J,v_0,...,v_J)\nonumber\\
	&\geq & \left(\underset{\substack{ |t_i-\tau_i|\leq Kw(N)\\ |v_i-\nu_i|\leq \rho_{N}\\\text{for all }i }}{\inf} \mathbb{P}\left[\hat{J}=J;\, \underset{i=1,...,J}{\max}\left| \hat{\tau}^{(2)}_i-\tau_i^{(2)} \right|\leq a_N \Bigr| \hat{J}=J,\hat{\tau}^{(1)}_j=t_j,\hat{\nu}^{(1)}_j=v_j\text{ for }j\leq J \right]\right)\cdot\nonumber\\
	&&\mathbb{P}\left[ \hat{J}=J;\quad \max_{i=1,...,J}\left|\hat{\tau}^{(1)}_i-\tau_i\right|\leq Kw(N);\quad \max_{i=0,...,J}|\hat{\nu}_i^{(1)}-\nu_i|\leq \rho_N\right]\nonumber\\
	&\geq & \left(\underset{\substack{ |t_i-\tau_i|\leq Kw(N)\\ |v_i-\nu_i|\leq \rho_{N}\\\text{for all }i }}{\inf} \mathbb{P}\left[ \underset{i=1,...,J}{\max}\left| \hat{\tau}^{(2)}_i-\tau_i^{(2)} \right|\leq a_N \Bigr| \hat{J}=J,\hat{\tau}^{(1)}_j=t_j,\hat{\nu}^{(1)}_j=v_j\text{ for }j\leq J \right]\right)-\nonumber\\
	&&\qquad\quad\mathbb{P}[\mathcal{R}_N\text{ is false}]
	\end{eqnarray}
	We wish to show that for all $\epsilon>0$, there exists a sequence $a_N=O(\log(J(N)))$ such that 
	$$\mathbb{P}\left[ \hat{J}=J;\quad \max_{i=1,...,J}\left| \hat{\tau}_i^{(2)}-\tau_i^{(2)}\right|\leq a_N\right]>1-\epsilon$$
	for all large $N$. It is sufficient to show this is satisfied by the second to last line of (\ref{eq:deriv1}), as $\mathcal{R}_N$ is true with probability increasing to 1. Henceforth, we will work with the probability 
	$$\mathbb{P}\left[\underset{i=1,...,J}{\max}\left| \hat{\tau}^{(2)}_i-\tau_i^{(2)} \right|\leq a_N \Bigr| \hat{J}=J,\hat{\tau}^{(1)}_j=t_j,\hat{\nu}^{(1)}_j=v_j\text{ for }j\leq J \right]\nonumber $$
	and, in the domain $|t_i-\tau_i|\leq Kw(N)$  and $|v_i-\nu_i|\leq \rho_{N}$ for all $i$, we show that it is greater than $1-\epsilon$ for all sufficiently large $N$ and $a_N=C_1\log J+C_2$ for some $C_1,C_2>0$. In the remainder of the proof, assume that all $t_i$'s and $v_i$'s fall within this domain.
	\newline
	\newline
	\indent For sufficiently large $N$, we have $Kw(N) \leq \delta_N/4$, and therefore no two of the second stage intervals ($[ t_i-Kw(N),$ $t_i+Kw(N)]$ for $i=1,...,J$) intersect.
	Because each $\hat{\tau}^{(2)}_j$ is a function of all $Y_i$'s in the disjoint index sets $S^{(2)}(t_j)\subset [ t_j-Kw(N),$ $t_j+Kw(N)]$ and the two level estimates $\hat{\nu}^{(1)}_{j-1}$ and $\hat{\nu}^{(1)}_j$, conditional independence holds:
	\begin{eqnarray}
	&&\mathbb{P}\left[\quad \max_{i=1,...,J}\left| \hat{\tau}_i^{(2)}-\tau_i^{(2)}\right| \leq a_N\Bigr| \hat{J}=J,\hat{\tau}^{(1)}_j=t_j,\hat{\nu}^{(1)}_j=v_j\text{ for }j\leq J \right]\nonumber\\
	&=& 
	\prod_{i=1}^J \mathbb{P}\left[ \left| \hat{\tau}^{(2)}_i-\tau_i^{(2)} \right|\leq a_N \Bigr| \hat{J}=J,\hat{\tau}^{(1)}_j=t_j,\hat{\nu}^{(1)}_j=v_j\text{ for all }j \right]
	\nonumber
	\end{eqnarray}
	To show the above product is eventually greater than some $1-\varepsilon$, it would suffice to show that, for all $1\leq k\leq J$ and sufficiently large $N$, 
	\begin{eqnarray}\label{mainineq}
	&&\mathbb{P}_{J,\boldsymbol{v},\boldsymbol{t}}\left[ \left| \hat{\tau}^{(2)}_k-\tau_k^{(2)} \right|>a_N \right]\nonumber\\
	&:=&\mathbb{P}\left[ \left| \hat{\tau}^{(2)}_k-\tau_k^{(2)} \right|> a_N \Bigr| \hat{J}=J,\hat{\tau}^{(1)}_j=t_j,\hat{\nu}^{(1)}_j=v_j\text{ for all }j \right]\nonumber\\
	&\leq& C_\epsilon/J
	\end{eqnarray}
	for some $C_\epsilon<-\log(1-\varepsilon)$. 
	\newline
	\newline
	\indent For any $k$ between 1 to $J$ inclusive, we can write explicit expressions for $\tau_k^{(2)}$ and $\hat{\tau}_k^{(2)}$. 

	\begin{eqnarray}
	\hat{\tau}^{(2)}_k&=&\underset{d\in S^{(2)}(t_k)}{\arg\min}\left(\text{sgn}(v_k-v_{k-1})\sum\limits_{i\in S^{(2)}(t_k)}\left( Y_i-\frac{v_{k-1}+v_k}{2} \right)\left[1(i\leq d)-1\left(i\leq \tau_k^{(2)}\right)\right]\right)\nonumber\\
	&:=& \underset{d\in  S^{(2)}(t_k)}{\arg\min}\mathbb{M}_k(d)
	\end{eqnarray} 
	Next, since, $t_k\in [\tau_k-w(N),\tau_k+w(N)]$, for $N$ large enough so that $\frac{K-1}{2}w(N)>1$ we have
	\begin{eqnarray}
	t_k-Kw(N)\leq \tau_k-(K-1)w(N)< \tau_k^{(2)}-\frac{K-1}{2}w(N)<\nonumber\\
	\tau_k^{(2)}+\frac{K-1}{2}w(N)< \tau_k+(K-1)w(N)\leq t_k+Kw(N)
	\end{eqnarray}
	That is to say, the set $S^{(2)}(t_k)$ includes the interval $\left[ \tau^{(2)}_k-\frac{K-1}{2}w(N), \tau^{(2)}_k+\frac{K-1}{2} \right]$, minus the first stage subsample points, regardless of which $t_k$ was used among those permissible in $R_N$. Therefore, for any even integer $a_N$ such that $1<a_N<\frac{K+1}{2}w(N)$ (which would be satisfied for all large $N$ if $a_N=O(\log(N))=o(w(N))$), 
	\begin{eqnarray}
	&&\mathbb{P}_{J,\boldsymbol{v},\boldsymbol{t}}\left[ \left| \hat{\tau}^{(2)}_k-\tau^{(2)}_k \right|>a_N \right]\nonumber\\
	&\leq &\mathbb{P}_{J,\boldsymbol{v},\boldsymbol{t}}\left[ \left|\lambda_2\left( \tau^{(2)}_k,\hat{\tau}^{(2)}_k \right)\right|>\frac{a_N}{2} \right]
	\end{eqnarray}
	for all sufficiently large $N$ such that the first stage subsample samples more sparsely than taking every other point. In order for $\left|\lambda_2\left( \hat{\tau}^{(2)}_k,\tau^{(2)}_k \right)\right|>a_N/2$, there must exists a $d$ such that 
	\begin{eqnarray}
	\left|\lambda_2(\tau^{(2)}_k,d)\right|>a_N/2\qquad\text{ and }\qquad
	\mathbb{M}_k(d)\leq \min_{\substack{\ell\in S^{(2)}(t_k)\\ |\lambda(\tau_k^{(2)},\ell)|\leq \frac{a_N}{2} }}\mathbb{M}_k(\ell)< \mathbb{M}_k(\tau^{(2)}_k)=0.
	\end{eqnarray} 
	Therefore, 
	\begin{eqnarray}\label{eq:condprobsumbound}
	&&\mathbb{P}_{J,\boldsymbol{v},\boldsymbol{t}}\left[ \left| \hat{\tau}^{(2)}_k-\tau^{(2)}_k \right|>a_N \right]\nonumber\\
	&\leq &\mathbb{P}_{J,\boldsymbol{v},\boldsymbol{t}}\left[ \left|\lambda_2\left( \tau^{(2)}_k,\hat{\tau}^{(2)}_k \right)\right|>\frac{a_N}{2} \right]\nonumber\\
	&\leq &\mathbb{P}_{J,\boldsymbol{v},\boldsymbol{t}}\left[ \exists d\in S^{(2)}(t_k)\text{ where }\left|\lambda_2(\tau^{(2)},d)\right|>a_N/2\text{ and } \mathbb{M}_k(d)<0\right]\nonumber\\
	&\leq &\sum_{\substack{d\in S^{(2)}(t_k)\\ \left|\lambda_2(\tau^{(2)},d)\right|>a_N/2}} \mathbb{P}_{J,\boldsymbol{v},\boldsymbol{t}}\left[ \mathbb{M}_k(d)\leq 0 \right]
	\end{eqnarray}
	Each $\mathbb{P}_{J,\boldsymbol{v},\boldsymbol{t}}\left[ \mathbb{M}_k(d)\leq 0 \right]$ can be bounded as follows by recognizing that each $\mathbb{M}_k(d)$ has subGaussian distribution: for every $d\in S^{(2)}(t_k)$ and sufficiently large $N$ so that $\rho_N\leq |\underline{\Delta}|/2$ (and hence $\text{sgn}(v_{k-v_{k-1}})=\text{sgn}(\nu_k-\nu_{k-1})$ for all $k$)
	\begin{eqnarray}
	&&\mathbb{M}_k(d)\nonumber\\
	&=& \text{sgn}(v_k-v_{k-1})\sum\limits_{i\in S^{(2)}(t_k)}\left( Y_i-\frac{v_{k-1}+v_k}{2} \right)\left[1(i\leq d)-1\left(i\leq \tau_k^{(2)}\right)\right]\nonumber\\
	&=&\begin{cases}
	\text{sgn}(\nu_k-\nu_{k-1})\left(\lambda_2( \tau^{(2)}_k,d )\left( \nu_{k}-\frac{v_{k-1}+v_k}{2} \right)+\underset{\substack{\tau^{(2)}_k<\ell\leq d\\ \ell\in S^{(2)}(t_k)}}{\sum}\varepsilon_\ell\right) \quad &\text{for }d>\tau_k^{(2)}\\
	0 &\text{for }d=\tau_k^{(2)}\\
	\text{sgn}(\nu_k-\nu_{k-1})\left(\lambda_2( \tau^{(2)}_k,d )\left( \nu_{k-1}-\frac{v_{k-1}+v_k}{2} \right)-\underset{\substack{d<\ell\leq\tau^{(2)}_k\\ \ell\in S^{(2)}(t_k)}}{\sum}\varepsilon_\ell\right) \quad &\text{for }d<\tau_k^{(2)}		
	\end{cases}\nonumber\\
	&=&\begin{cases}
	\left|\lambda_2( \tau^{(2)}_k,d )\right|\left( \left|\frac{\nu_{k+1}-\nu_k}{2}\right|+\text{sgn}(\nu_k-\nu_{k-1})\hat{D}_k \right)+\text{sgn}(\nu_k-\nu_{k-1})\underset{\substack{\tau^{(2)}_k<\ell\leq d\\ \ell\in S^{(2)}(t_k)}}{\sum}\varepsilon_\ell \quad &\text{for }d>\tau_k^{(2)}\\
	0 &\text{for }d=\tau_k^{(2)}\\
	\left|\lambda_2( \tau^{(2)}_k,d )\right|\left( \left|\frac{\nu_{k}-\nu_{k-1}}{2}\right|-\text{sgn}(\nu_k-\nu_{k-1})\hat{D}_k \right)-\text{sgn}(\nu_k-\nu_{k-1})\underset{\substack{d<\ell\leq\tau^{(2)}_k\\ \ell\in S^{(2)}(t_k)}}{\sum}\varepsilon_\ell \quad &\text{for }d<\tau_k^{(2)}		
	\end{cases}\nonumber\\
	\end{eqnarray}
	where 
	\begin{eqnarray}
	\hat{D}_k:=\frac{(\nu_k-v_k)+(\nu_{k-1}-v_{k-1})}{2}
	\end{eqnarray}
	The maximal deviation of the signal estimates are at most $\rho_N\to 0$ away from the true signal estimates, and hence for all sufficiently large $N$, $\left|\hat{D}_k\right|\leq |\rho_N|\leq\frac{\underline{\Delta}}{2}$ (where $\underline{\Delta}$ bounds the minimum absolute jump from below). In this case
	\begin{eqnarray}
	\left|\frac{\nu_{k+1}-\nu_k}{2}+\hat{D}_k \right|\geq \frac{\underline{\Delta}}{2}\quad\text{ and }\quad\left| \frac{\nu_{k+1}-\nu_k}{2}-\hat{D}_k \right|\geq\frac{\underline{\Delta}}{2}.
	\end{eqnarray}
	Therefore the mean of the $\mathbb{M}_k(d)$'s satisfy
	$$\mathbb{E}[\mathbb{M}_k(d)] = \left|\lambda_2( \tau^{(2)}_k,d )\right|\left( \left|\frac{\nu_{k+1}-\nu_k}{2}\right|\pm\text{sgn}(\nu_k-\nu_{k-1})\hat{D}_k \right)\geq   \left|\lambda_2( \tau^{(2)}_k,d )\right|\frac{\underline{\Delta}}{2}.$$
	At the same time, $\mathbb{M}_k(d)-\mathbb{E}[\mathbb{M}_k(d)]$ is a zero mean subGaussian random variable with variance parameter $\sigma_{\max}^2\left|\lambda_2( \tau^{(2)}_k,d )\right|$, hence

	\begin{eqnarray}
	&&\mathbb{P}_{J,\boldsymbol{v},\boldsymbol{t}}\left[ \mathbb{M}_k(d)\leq 0 \right]\nonumber\\
	&=&\mathbb{P}_{J,\boldsymbol{v},\boldsymbol{t}}\left[ \mathbb{M}_k(d)-\mathbb{E}[\mathbb{M}_k(d)]\leq -\mathbb{E}[\mathbb{M}_k(d)] \right]\nonumber\\
	&\leq &\mathbb{P}_{J,\boldsymbol{v},\boldsymbol{t}}\left[ \mathbb{M}_k(d)-\mathbb{E}[\mathbb{M}_k(d)]\leq -\left|\lambda_2( \tau^{(2)}_k,d )\right|\frac{\underline{\Delta}}{2} \right]\nonumber\\
	&\leq& \frac{1}{2}\exp\left( -\frac{\underline{\Delta}^2}{8\sigma_{\max}^2}\left|\lambda_2(\tau^{(2)}_k,d)\right| \right)
	\end{eqnarray}
	Going back to (\ref{eq:condprobsumbound}), this gives:
	\begin{eqnarray}
	&&\mathbb{P}_{J,\boldsymbol{v},\boldsymbol{t}}\left[ \left| \hat{\tau}^{(2)}_k-\tau^{(2)}_k \right|>a_N \right]\nonumber\\
	&\leq &\sum_{\substack{d\in S^{(2)}(t_k)\\ \left|\lambda_2(\tau^{(2)},d)\right|>a_N/2}} \mathbb{P}_{J,\boldsymbol{v},\boldsymbol{t}}\left[ \mathbb{M}_k(d)\leq 0 \right]\nonumber\\
	&\leq &\frac{1}{2}\sum_{\substack{d\in S^{(2)}(t_k)\\ \left|\lambda_2(\tau^{(2)},d)\right|>a_N/2}}\exp\left( -\frac{\underline{\Delta}^2}{8\sigma_{\max}^2}\left|\lambda_2(\tau^{(2)}_k,d)\right| \right)\nonumber\\
	&\leq &\frac{1}{2}\left[ \sum_{k=\frac{a_N}{2}+1}^\infty\exp\left( -\frac{\underline{\Delta}^2}{8\sigma_{\max}^2}k \right) +\sum_{k=-\frac{a_N}{2}-1}^{-\infty}\exp\left( -\frac{\underline{\Delta}^2}{8\sigma_{\max}^2}|k| \right)  \right]\nonumber\\
	&=&\left(1-\exp\left( -\frac{\underline{\Delta}^2}{8\sigma_{\max}^2} \right)\right)^{-1}\exp\left( -\frac{\underline{\Delta}^2}{8\sigma_{\max}^2}\left(\frac{a_N}{2}+1\right) \right)
	\end{eqnarray}
	Therefore in order to bound $\mathbb{P}_{J,\boldsymbol{v},\boldsymbol{t}}\left[ \left| \hat{\tau}^{(2)}_k-\tau^{(2)}_k \right|>a_N \right]$ by $C_\varepsilon/J$ as we stated back in (\ref{mainineq}), a legitimate choice of $a_N$ which will satisfy this bound can be found by solving
	\begin{eqnarray}
	\frac{\exp\left( -\frac{\underline{\Delta}^2}{8\sigma_{\max}^2} \right)}{1-\exp\left( -\frac{\underline{\Delta}^2}{8\sigma_{\max}^2} \right)}\exp\left( -\frac{\underline{\Delta}^2}{8\sigma_{\max}^2}\cdot\frac{a_N}{2} \right)=\frac{C_\epsilon}{J}.
	\end{eqnarray}
	The solution to this will be of the form $a_N=C_1\log(J)+C_2$, where $C_1$ and $C_2$ are constants not dependent on $N$, and because $\log(J)=O(\log(N))$, this is a valid solution.
\end{proof}

\subsection{Proof sketch of Theorem \ref{thm:increasingJasymprotics}}\label{sec:increasingJasymproticsproofshort}
\label{proof-sketch-Theorem-5} 
\begin{proof}
	Let $\mathcal{R}_N$ be the event
	\begin{eqnarray}
	\mathcal{R}_N:= \left\{ \hat{J}=J;\,\max_{j=1,\dots,J}\left| \hat{\tau}^{(1)}_j-\tau_j \right|\leq w(N);\, \max_{j=0,\dots,J}\left| \hat{\nu}^{(1)}_j-\nu_j \right|\leq \rho_N \right\}
	\end{eqnarray}
	Define $G_N()$ as the joint distribution of the first stage estimates $(\hat{J},\boldsymbol{\hat{\tau}^{(1)}},\boldsymbol{\hat{\nu}^{(1)}})=(\hat{J},\hat{\tau}^{(1)}_1,\dots,\hat{\tau}^{(1)}_{\hat{J}},\hat{\nu}^{(1)}_0,\dots,\hat{\nu}^{(1)}_{\hat{J}})$. Because $\mathbb{P}[\mathcal{R}_N]\to 1$, it is possible to show that the difference between
	\begin{eqnarray}\label{eq:condinte2short}
	\mathbb{P}\left[  \hat{J}=J;\, \left|\lambda_2\left( \tau_j^{(2)},\hat{\tau}_j^{(2)} \right)\right|
	\leq Q_{(|\Delta_j|-2\rho_N)/\sigma}(\sqrt[J]{1-\alpha})\text{ for }j=1,\dots,J  \right]
	\end{eqnarray}
	and $\prod_{j=1}^J P_{(|\Delta_j|-2\rho_N)/\sigma}(\sqrt[J]{1-\alpha})$, equals the expression
	\begin{eqnarray}
	&&\underset{ (k,\boldsymbol{t},\boldsymbol{\alpha})\in\mathcal{R}_N }{\int}\Bigg(\prod_{j=1}^J\mathbb{P}\left[ \left|\lambda_2\left( \tau_j^{(2)},\hat{\tau}_j^{(2)} \right)\right|\leq Q_{(|\Delta_j|-2\rho_N)/\sigma}(\sqrt[J]{1-\alpha})\Big| \hat{J}=k;\,\boldsymbol{\hat{\tau}^{(1)} }=\boldsymbol{t};\, \boldsymbol{\hat{\nu}^{(1)} }=\boldsymbol{v} \right]\nonumber\\
	&& -\prod_{j=1}^J P_{(|\Delta_j|-2\rho_N)/\sigma}(\sqrt[J]{1-\alpha})\Bigg) dG_N(k,\boldsymbol{t},\boldsymbol{v})+o(1)
	\end{eqnarray}
	It is therefore sufficient to show that the difference inside the integral is, for all $(k,\boldsymbol{t},\boldsymbol{\alpha})\in\mathcal{R}_N$, uniformly bounded in absolute value by a $o(1)$ term. Henceforth, consider only such admissible $k$'s, $\boldsymbol{t}$'s, $\boldsymbol{v}$'s (which, of course, restricts $k$ to $J$), and additionally that $N$ is large enough so that $\rho_N\leq |\underline{\Delta}|/8$ and distance between consecutive points in $S^{(1)}$ is more than 2 (i.e. $\min\limits_{i,j\in S^{(1)},i\neq j}|i-j|>2$). We will proceed to show (\ref{eq:increasingJasymprotics1}) by obtaining an upper bound for the following absolute difference, for every $j=1,\dots, J$:
	\begin{eqnarray}
	\Bigg| \mathbb{P}\left[ \left|\hat{\tau}^{(2)}_j-\tau_j\right|\leq Q_{(|\Delta_j|-2\rho_N)/\sigma}(\sqrt[J]{1-\alpha})\Big| \hat{J}=J;\,\boldsymbol{\hat{\tau}^{(1)} }=\boldsymbol{t};\, \boldsymbol{\hat{\nu}^{(1)} }=\boldsymbol{v} \right]\nonumber\\-P_{(|\Delta_j|-2\rho_N)/\sigma}(\sqrt[J]{1-\alpha})\Bigg| \,.
	\end{eqnarray}
	This upper bound will be derived in several components.
	\newline
	\newline
	\noindent \textbf{First Component:} A more explicit expression for the change point estimates is, for every $j=1,\dots,J$,
	\begin{eqnarray}
	\hat{\tau}^{(2)}_j&=&\underset{t\in S^{(2)}(t_i)}{\arg\min}(\text{sgn}(v_j-v_{j-1}))\sum_{i\in S^{(2)}(t_i)}\left( Y_i-\frac{v_j+v_{j-1}}{2} \right)( 1(i\leq t)-1(i\leq \tau_j^{(2)} ))\nonumber\\
	&=:& \underset{t\in S^{(2)}(t_i)}{\arg\min} \hat{X}_j^{(2)}(t)
	\end{eqnarray}
	$\hat{X}^{(2)}(t)$ is a random walk over the set $t\in S^{(2)}(t_j)$, which is a set of integers skipping over the points of $S^{(1)}$. To express this as a random process over a set of integers without skipped points (a requirement to apply some probability bounds), 
	define the one-to-one function $\pi_{2,j}(s):=\lambda_2(\tau_j^{(2)},s)$ for $s \in \{1,\dots,N\}-S^{(1)}$. This allows us to consider the random walk $\hat{X}^{(2)}_j( \pi_{2,j}^{-1}(t) )$ over the set $t\in \pi_{2,j}(S^{(2)}(t_j))$. This is a set of integers containing no skipped points, and furthermore it can be shown that
	\begin{eqnarray}\label{eq:setcontain}
	\left[ \pm  Q_{(|\Delta_j-2\rho_N)/\sigma}(\sqrt[J]{1-\alpha}) \right]\subset  \pi_{2,j}(S^{(2)}(t_j))
	\end{eqnarray}
	(could be done by showing $[\pm ((K-1)/2)w(N)]\subset \pi_{2,j}(S^{(2)}(t_j))$ and using Lemma \ref{lem:quantbound} to argue that $Q_{(|\Delta_j-2\rho_N)/\sigma}(\sqrt[J]{1-\alpha})=O(\log(N))=o(w(N))$).
	\newline
	\newline
	From here, $\hat{X}^{(2)}_j(\pi_{2,j}^{-1}(t))$ is a random walk with non-symmetrical linear drifts on either side of $ t=0 $. We wish to approximate $\hat{X}^{(2)}(t)$ with a random walk with symmetrical drift, so consider $X'_{|\Delta_j|-2|\hat{D}_j|}(t)$, defined as
	\begin{eqnarray}
	X'_{|\Delta_j|-2|\hat{D}_j|}(t):=\begin{cases}
	t\frac{|\Delta_j|-2|\hat{D}_j|}{2}+\text{sgn}(\Delta_j)\sum\limits_{i=1}^{t}\varepsilon_{ \pi_2^{-1}\left(\pi_2(\tau_j^{(2)})+i  \right) }\qquad & \text{for }t>0\\
	0 &t=0\\
	|t|\frac{|\Delta_j|-2|\hat{D}_j|}{2}-\text{sgn}(\Delta_j)\sum\limits_{i=t}^{-1}\varepsilon_{\pi_{2}^{-1}\left(\pi_2(\tau_j^{(2)})+i+1\right)}\qquad & \text{for }t<0
	\end{cases}.
	\end{eqnarray}
	where 
	\begin{eqnarray}
	\hat{D}_j=\frac{\nu_j-v_j}{2}+\frac{\nu_{j-1}-v_{j-1}}{2}.
	\end{eqnarray}
	This random walk is very close to $\hat{X}^{(2)}(\pi_{2,j}^{-1}(t))$, since it equals either 
	$X'_{|\Delta_j|-2|\hat{D}_j|}(t)+2|\hat{D}_jt|1(t>0)$ or $X'_{|\Delta_j|-2|\hat{D}_j|}(t)+2|\hat{D}_jt|1(t<0)$ for $t\in \pi_{2,j}^{-1}(S^{(2)}(t_j)$, and $|\hat{D}_j|<\rho_N\to 0$. In either case, this begs the use of Lemma \ref{lem:randomwalkcompoppo} and Lemma \ref{lem:randomwalkcomp}, two results which allows us to compare the argmin of such similar random walks. The only important condition they require is satisfied by (\ref{eq:setcontain}). Their application gives
	\begin{eqnarray}\label{eq:firstcompabsshort}
	&&\Bigg|\mathbb{P}\left[ \left|\lambda_2\left(\tau_j^{(2)},\hat{\tau}^{(2)}_j\right)\right|\leq Q_{(|\Delta_j|-2\rho_N)/\sigma}(\sqrt[J]{1-\alpha})\Bigg|\hat{J}=J,\boldsymbol{\hat{\tau}^{(1)}}=\boldsymbol{t},\, \boldsymbol{\hat{\nu}^{(1)} }=\boldsymbol{v}\right]\nonumber\\
	&&-\mathbb{P}\left[ \left|\underset{t\in \pi^{-1}_{2,j}(S^{(2)}(t_j))}{\arg\min}X'_{|\Delta_j|-2|\hat{D}_j|}(t)\right|\leq Q_{(|\Delta_j|-2\rho_N)/\sigma}(\sqrt[J]{1-\alpha})\Bigg|\hat{J}=J,\boldsymbol{\hat{\tau}^{(1)}}=\boldsymbol{t},\, \boldsymbol{\hat{\nu}^{(1)} }=\boldsymbol{v} \right]\Bigg|  \nonumber\\
	&=&\Bigg|\mathbb{P}\left[ \left|\underset{t\in \pi^{-1}_{2,j}(S^{(2)}(t_j))}{\arg\min}X^{(2)}_j(t)\right|\leq Q_{(|\Delta_j|-2\rho_N)/\sigma}(\sqrt[J]{1-\alpha})\Bigg|\hat{J}=J,\boldsymbol{\hat{\tau}^{(1)}}=\boldsymbol{t},\, \boldsymbol{\hat{\nu}^{(1)} }=\boldsymbol{v}\right]\nonumber\\
	&&-\mathbb{P}\left[ \left|\underset{t\in \pi^{-1}_{2,j}(S^{(2)}(t_j))}{\arg\min}X'_{|\Delta_j|-2|\hat{D}_j|}(t)\right|\leq Q_{(|\Delta_j|-2\rho_N)/\sigma}(\sqrt[J]{1-\alpha})\Bigg|\hat{J}=J,\boldsymbol{\hat{\tau}^{(1)}}=\boldsymbol{t},\, \boldsymbol{\hat{\nu}^{(1)} }=\boldsymbol{v} \right]\Bigg|  \nonumber\\
	&\leq &C^*_1\rho_N
	\end{eqnarray}
	For some constant $C^*_1>0$, not dependent on $j$ or $N$.
	~\newline
	\newline
	\textbf{Second Component: }Now $X'_{|\Delta|-2|\hat{D}_j|}(t)/\sigma$ has the same exact distribution as $X_{(|\Delta|-2|\hat{D}_j|)/\sigma}(t)$, for all integers $t\in \pi_{2,j}^{-1}(S^{(2)}(t_j)$. It was also shown in the previous section that the set $\pi_{2,j}^{-1}(S^{(2)}(t_j)$ contains the interval of integers $\left[ \pm Q_{(|\Delta_j|-2\rho_N)/\sigma}(\sqrt[J]{1-\alpha}) \right]$ for all large $N$. Therefore apply Lemma \ref{lem:twoXrandwalkcomp} to obtain
	\begin{eqnarray}\label{eq:secondcompabsshort}
	&&\Biggr|\mathbb{P}\left[ \left|\underset{ t\in \pi^{-1}_{2,j}(S^{(2)}(t_j))}{\arg\min}X'_{|\Delta_j|-2|\hat{D}_j|}(t)\right|\leq Q_{(|\Delta_j|-2\rho_N)/\sigma}(\sqrt[J]{1-\alpha})\Bigg|\hat{J}=J,\boldsymbol{\hat{\tau}^{(1)}}=\boldsymbol{t},\, \boldsymbol{\hat{\nu}^{(1)} }=\boldsymbol{v} \right]\nonumber\\
	&&-\mathbb{P}\left[ \left|\underset{t\in \pi^{-1}_{2,j}(S^{(2)}(t_j))}{\arg\min}X_{(|\Delta_j|-2\rho_N)/\sigma}(t)\right|\leq Q_{(|\Delta_j|-2\rho_N)/\sigma}(\sqrt[J]{1-\alpha})\right]\Biggr|\leq C^*_2\rho_N
	\end{eqnarray}
	for some $C^*_2>0$ not dependent on $j$ or $N$.
	~\newline
	\newline
	\textbf{Third Component: }The set $\pi_{2,j}^{-1}(S^{(2)}(t_j))$ contains the set $\left[\pm \left(\frac{(K-1)w(N)}{2}-1\right) \right]$. This allows an application of Lemma \ref{lem:eqprob} to obtain
	\begin{eqnarray}\label{eq:underdiff2short}
	&&\Biggr|\mathbb{P}\left[ \left|\underset{t\in \pi^{-1}_{2,j}(S^{(2)}(t_j))}{\arg\min}X_{(|\Delta_j|-2\rho_N)/\sigma}(t)\right|\leq Q_{(|\Delta_j|-2\rho_N)/\sigma}(\sqrt[J]{1-\alpha})\right]\nonumber\\
	&&-\mathbb{P}\left[ \left|\underset{|t|\leq \frac{(K-1)w(N)}{2}-1}{\arg\min}X_{(|\Delta_j|-2\rho_N)/\sigma}(t)\right|\leq Q_{(|\Delta_j|-2\rho_N)/\sigma}(\sqrt[J]{1-\alpha})\right]\Biggr|\nonumber\\
	&\leq & A_1\exp\left(-B_1\left(\frac{(K-1)w(N)}{2}-1\right)\right)
	\end{eqnarray}
	for some constant $A_1>0$, $B_1>0$ not dependent on $j$ or $N$.	The same application of the lemma can also yield
	\begin{eqnarray}\label{eq:underdiff3short}
	&&\Bigg|\mathbb{P}\left[ \left|\underset{|t|\leq \frac{(K-1)w(N)}{2}-1}{\arg\min}X_{(|\Delta_j|-2\rho_N)/\sigma}(t)\right|\leq Q_{(|\Delta_j|-2\rho_N)/\sigma}(\sqrt[J]{1-\alpha})\right]
	\nonumber\\
	&&-\mathbb{P}\left[ \left|\underset{t\in\mathbb{Z}}{\arg\min}X_{(|\Delta_j|-2\rho_N)/\sigma}(t)\right|\leq Q_{(|\Delta_j|-2\rho_N)/\sigma}(\sqrt[J]{1-\alpha})\right]\Bigg|\nonumber\\
	&\leq &A_2\exp\left(-B_2\left(\frac{(K-1)w(N)}{2}-1\right)\right)
	\end{eqnarray}
	for some $A_2>0$ and $B_2>0$ not dependent on $j$ or $N$. Adding up these two upper bounds imply
	\begin{eqnarray}\label{eq:thirdcompabsshort}
	&&\Biggr|\mathbb{P}\left[ \left|\underset{t\in \pi^{-1}_{2,j}(S^{(2)}(t_j))}{\arg\min}X_{(|\Delta_j|-2\rho_N)/\sigma}(t)\right|\leq Q_{(|\Delta_j|-2\rho_N)/\sigma}(\sqrt[J]{1-\alpha})\right]\nonumber\\
	&&-\mathbb{P}\left[ \left|\underset{t\in\mathbb{Z}}{\arg\min}X_{(|\Delta_j|-2\rho_N)/\sigma}(t)\right|\leq Q_{(|\Delta_j|-2\rho_N)/\sigma}(\sqrt[J]{1-\alpha})\right]\Bigg|\nonumber\\
	&\leq &(A_1\vee A_2)\exp\left(-(B_1\wedge B_2)\left(\frac{(K-1)w(N)}{2}-1\right)\right)
	\end{eqnarray}
	~\newline
	\noindent\textbf{Sum of the Components:} Adding up the differences in (\ref{eq:firstcompabs}), (\ref{eq:secondcompabs}), and (\ref{eq:thirdcompabs}):
	\begin{eqnarray}\label{eq:individualboundshort}
	&&\Bigg|\mathbb{P}\left[ \left|\lambda_2\left(\tau_j^{(2)},\hat{\tau}^{(2)}_j\right)\right|\leq Q_{(|\Delta_j|-2\rho_N)/\sigma}(\sqrt[J]{1-\alpha})\Bigg|\hat{J}=J,\boldsymbol{\hat{\tau}^{(1)}}=\boldsymbol{t},\, \boldsymbol{\hat{\nu}^{(1)} }=\boldsymbol{v}\right]\nonumber\\
	&&-\mathbb{P}\left[ \left|\underset{t\in\mathbb{Z}}{\arg\min}X_{(|\Delta_j|-2\rho_N)/\sigma}(t)\right|\leq Q_{(|\Delta_j|-2\rho_N)/\sigma}(\sqrt[J]{1-\alpha})\right]\Bigg|\nonumber\\
	&\leq& C_4\rho_N+C_5\exp\left[-C_6(K-1)w(N)\right]
	\end{eqnarray}
	for some constants $C_4$, $C_5$, and $C_6$. This allows us to bound
	\begin{eqnarray}
	&&\prod_{j=1}^J\mathbb{P}\left[ \left|\lambda_2\left( \tau_j^{(2)},\hat{\tau}_j^{(2)} \right)\right|\leq Q_{(|\Delta_j|-2\rho_N)/\sigma}(\sqrt[J]{1-\alpha})\Big| \hat{J}=J;\,\boldsymbol{\hat{\tau}^{(1)} }=\boldsymbol{t};\, \boldsymbol{\hat{\nu}^{(1)} }=\boldsymbol{v} \right]\nonumber\\
	&& -\prod_{j=1}^J P_{(|\Delta_j|-2\rho_N)/\sigma}(\sqrt[J]{1-\alpha})
	\end{eqnarray}
	with
	\begin{eqnarray}
	&&J\Bigg|\mathbb{P}\left[ \left|\lambda_2\left(\tau_j^{(2)},\hat{\tau}^{(2)}_j\right)\right|\leq Q_{(|\Delta_j|-2\rho_N)/\sigma}(\sqrt[J]{1-\alpha})\Bigg|\hat{J}=J,\boldsymbol{\hat{\tau}^{(1)}}=\boldsymbol{t},\, \boldsymbol{\hat{\nu}^{(1)} }=\boldsymbol{v}\right]\nonumber\\
	&&-\mathbb{P}_{|\Delta_j|-2\rho_N}(\sqrt[J]{1-\alpha})\Bigg|\left(\mathbb{P}_{|\Delta_j|-2\rho_N}(\sqrt[J]{1-\alpha})\right)^{-1}\nonumber\\
	&\leq&J\left(C_4\rho_N+C_5\exp\left[-C_6(K-1)w(N)\right]\right)(1-\alpha)^{-1/J}\nonumber\\
	&\leq&C_4(1-\alpha)^{-1}J\rho_N+C_5(1-\alpha)^{-1}N\exp\left[-C_6(K-1)w(N)\right]
	,\end{eqnarray}
	which goes to 0 since $J\rho_N\to 0$ and $w(N)\geq C N^{1-\gamma}$ for some constant $C$.
\end{proof}

\subsection{Proof of Theorem \ref{thm:increasingJasymprotics}}\label{sec:increasingJasymproticsproof}
In this section we perform the proof in the special situation where all the error terms of the data sequence are iid $N(0,\sigma^2)$ for some positive $\sigma$. In the more general case where the error terms are independent Gaussian with variances bounded above and below by $\sigma^2_{\max}$ and $\sigma_{\min}^2$, with identical distributions between change points, similar results from Section \ref{sec:supplementpartc} and methods from the following proof will yield the statement of Theorem \ref{thm:multidepend}. Because we are dealing with iid error terms, we could simplify the notations presented with the theorem in the main draft. First, we could simplify the random walks
\begin{eqnarray}
X_{\Delta}(t):=X_{\Delta,1,1}(t)=\begin{cases}
t\frac{|\Delta|}{2}+\sum_{i=1}^t\varepsilon^*_{ i }\qquad & t>0\\
0 & t=0\\
\frac{|t\Delta|}{2}+\sum_{i=0}^{t-1}\varepsilon^*_{-i} \quad t<0
\end{cases}\\
\text{where }\varepsilon^*_i\overset{\text{iid}}{\sim} N(0,1)\text{ for all }i,
\end{eqnarray}
denote the argmin of this random walk as $L_{\Delta}:=\underset{t\in\mathbb{Z}}{\arg\min}\; X_\Delta(t)$, simplify the quantiles
\begin{eqnarray}
Q_{\Delta}(1-\alpha):=Q_{\Delta,1,1}(1-\alpha)=\min\;\left\{t\geq 0:\mathbb{P}\left[\left|L_\Delta\right|\leq t  \right]\geq 1-\alpha\right\}
\end{eqnarray}
for any $\alpha\in [0,1]$, and their associated probabilities 
\begin{equation}P_{\Delta}(1-\alpha)=P_{\Delta,1,1}(1-\alpha)\left[ |L_\Delta|\leq Q_{\Delta}(1-\alpha) \right] .
\end{equation}
We will use these notations for the remainder of this section and the entirety of Section \ref{sec:supplementpartc}.
\newline
\newline
Using this notation and working within this simplified framework of iid errors, the results of the theorem translates to:
\begin{eqnarray}\label{eq:increasingJasymproticsineqs}  
&&\mathbb{P}\left[ \hat{J}=J;\, \left|\lambda_2\left( \tau_j^{(2)},\hat{\tau}_j^{(2)} \right)\right|
\leq Q_{(|\Delta_j|-2\rho_N)/\sigma}(\sqrt[J]{1-\alpha})\text{ for }j=1,\dots,J  \right]\nonumber\\
&=& \left(\prod_{j=1}^JP_{(|\Delta_j|-2\rho_N)/\sigma}(\sqrt[J]{1-\alpha})\right)+o(1),
\end{eqnarray}  
and the second inequality
\begin{eqnarray}\label{eq:increasingJasymproticsineqs2} 
&&\mathbb{P}\left[ \hat{J}=J;\, \left|\lambda_2\left( \tau_j^{(2)},\hat{\tau}_j^{(2)} \right)\right|\leq Q_{(|\Delta_j|+2\rho_N)/\sigma}(\sqrt[J]{1-\alpha})\text{ for }j=1,\dots,J  \right]\nonumber\\ &=& \left(\prod_{j=1}^JP_{(|\Delta_j|+2\rho_N)/\sigma}(\sqrt[J]{1-\alpha})\right)+o(1),
\end{eqnarray}
Both these inequalities can be shown utilizing various probability bounds regarding the argmins of random walks of the form $X_\Delta(\cdot)$, and other similar random walks. To utilize such result, the estimators $\hat{\tau}^{(2)}_j$ must be expressed as the argmin of random walks with positive linear drifts. We thus introduce some new notation and new insight to facilitate this task.
\newline
\newline
\indent First we define the indexing function $\pi_2=\pi_{2,N}$, which maps the set of integers $\{1,\dots,N\}-S^{(1)}$ to the set of integers $\{1,\dots, N-|S^{(1)}| \}$:
\begin{eqnarray}\label{def:pi2func}
\pi_2(k)=\sum_{j=1}^k1(j\notin S^{(1)}).
\end{eqnarray}
For every $N$, $\pi_2$ is a strictly increasing bijection, and it also has the property such that for any two values $i$ and $j$ in the domain of $\pi_2$, $\lambda_2(i,j)=\pi_2(j)-\pi_2(i)$. Additionally, consider the following subset of the full data:
\begin{eqnarray}\label{set:secstagedata}
Y_{\pi_2^{-1}(1)},Y_{\pi_2^{-1}(2)},\dots, Y_{\pi_2^{-1}(N-|S^{(1)}|)}.
\end{eqnarray}
These all the datapoints which were not used in the first stage subsample. This subset of the full dataset is also a change point model following conditions (M1) to (M4) with the same signal upper bound $\bar{\theta}$, the signal jump lower bound $\underline{\Delta}$, and the same error distribution. The change points for this subset is $\pi_2(\tau_j^{(2)})$ for $j=1,\dots,J$. In fact, (\ref{set:secstagedata}) and the first stage subsample $\{ Y_i \}_{i\in S^{(1)}}$ can be considered statistically independent change point models, resulting in the first stage estimates obtained using $\{ Y_i \}_{i\in S^{(1)}}$ being independent of (\ref{set:secstagedata}). This means that conditional on the first stage estimates, the distribution of (\ref{set:secstagedata}) does not change from their marginal distributions.
\newline
\newline
\indent A property of the $\pi_2$ function that will be used is a relationship between the $\pi_2^{-1}$ function and the "normal" subtraction, namely, we want to be able to compare $\pi_2^{-1}(a,b)$ to $b-a$ for $a,b\in \{1,\dots, N-|S^{(1)}|\}$. For all $N$ large enough such that the distance between consecutive points in $S^{(1)}$ is more than 2 (i.e. $\min\limits_{i,j\in S^{(1)},i\neq j}|i-j|>2$), we have the following property: for any integers $a$ and $b$ such that $a$ and $a+b$ are in $\{1,\dots, N-|S^{(1)}|\}$,
\begin{eqnarray}\label{eq:compdists}
\left|\pi_2^{-1}(a+b)-\pi_2^{-1}(a)\right|\leq 2|b|.
\end{eqnarray}
One way to see this is the following: for any $c,d\in \{ 1,\dots,N \}-S^{(1)}$, $|\lambda_2(c,d)|$ counts the number of points in either $(c,d]$ (if $c\leq d$) or $(d,c]$ (if otherwise), and therefore $|\lambda_2(c,d)|\geq |d-c|/2$ due to the $S^{(1)}$ containing no two points that are less than 2 apart. Knowing this, the fact that
\begin{eqnarray}
\left|\lambda_2\left( \pi_2^{-1}(a),\pi_2^{-1}(a+b) \right)\right|=\left|\pi_2\left(\pi_2^{-1}(a+b)\right)-\pi_2\left( \pi_2^{-1}(a) \right)\right|= |b|,
\end{eqnarray}
means that
\begin{eqnarray}
2|b|=  2\left|\lambda_2\left( \pi_2^{-1}(a),\pi_2^{-1}(a+b) \right)\right|\geq
\left|\pi_2^{-1}(a+b)-\pi_2^{-1}(a)\right|
\end{eqnarray}

\indent We will offer a proof of the inequality given in (\ref{eq:increasingJasymproticsineqs}). The steps for verifying (\ref{eq:increasingJasymproticsineqs2}) would only require a few modifications.
\begin{proof}
	Let $\mathcal{R}_N$ be the event
	\begin{eqnarray}
	\mathcal{R}_N:= \left\{ \hat{J}=J;\,\max_{j=1,\dots,J}\left| \hat{\tau}^{(1)}_j-\tau_j \right|\leq w(N);\, \max_{j=0,\dots,J}\left| \hat{\nu}^{(1)}_j-\nu_j \right|\leq \rho_N \right\}
	\end{eqnarray}
	Again define $G_N()$ joint distribution of the first stage estimates $(\hat{J},\boldsymbol{\hat{\tau}^{(1)}},\boldsymbol{\hat{\nu}^{(1)}})=(\hat{J},\hat{\tau}^{(1)}_1,\dots,\hat{\tau}^{(1)}_{\hat{J}},\hat{\nu}^{(1)}_0,\dots,\hat{\nu}^{(1)}_{\hat{J}})$.
	\begin{eqnarray}\label{eq:condinte2}
	&&\mathbb{P}\left[  \hat{J}=J;\, \left|\lambda_2\left( \tau_j^{(2)},\hat{\tau}_j^{(2)} \right)\right|
	\leq Q_{(|\Delta_j|-2\rho_N)/\sigma}(\sqrt[J]{1-\alpha})\text{ for }j=1,\dots,J  \right]\nonumber\\
	&= & \mathbb{P}\left[  \hat{J}=J;\, \left|\lambda_2\left( \tau_j^{(2)},\hat{\tau}_j^{(2)} \right)\right|
	\leq Q_{(|\Delta_j|-2\rho_N)/\sigma}(\sqrt[J]{1-\alpha})\,\forall j;\, \mathcal{R}_N  \right]\nonumber\\
	&&+\mathbb{P}\left[ \hat{J}=J;\, \left|\lambda_2\left( \tau_j^{(2)},\hat{\tau}_j^{(2)} \right)\right|
	\leq Q_{(|\Delta_j|-2\rho_N)/\sigma}(\sqrt[J]{1-\alpha})\,\forall j;\,\text{not }\mathcal{R}_N \right]\nonumber\\
	&=&\underset{ (k,\boldsymbol{t},\boldsymbol{\alpha})\in\mathcal{R}_N }{\int}\mathbb{P}\left[ \max_{j=1,\dots,J}\left|\lambda_2\left( \tau_j^{(2)},\hat{\tau}_j^{(2)} \right)\right|\leq Q_{(|\Delta_j|-2\rho_N)/\sigma}(\sqrt[J]{1-\alpha})\Big| \hat{J}=k;\,\boldsymbol{\hat{\tau}^{(1)} }=\boldsymbol{t};\, \boldsymbol{\hat{\nu}^{(1)} }=\boldsymbol{v} \right]dG_N(k,\boldsymbol{t},\boldsymbol{v})\nonumber\\
	&&+\mathbb{P}\left[\hat{J}=J;\, \left|\lambda_2\left( \tau_j^{(2)},\hat{\tau}_j^{(2)} \right)\right|
	\leq Q_{(|\Delta_j|-2\rho_N)/\sigma}(\sqrt[J]{1-\alpha})\,\forall j;\, \text{not }\mathcal{R}_N \right]\nonumber
	\end{eqnarray}
	Because the probability of $\mathcal{R}_N$ goes to 1, the difference between this probability and $\prod_{j=1}^J P_{(|\Delta_j|-2\rho_N)/\sigma}(\sqrt[J]{1-\alpha})$ is
	\begin{eqnarray}
	&&\underset{ (k,\boldsymbol{t},\boldsymbol{\alpha})\in\mathcal{R}_N }{\int}\Bigg(\mathbb{P}\left[ \max_{j=1,\dots,J}\left|\lambda_2\left( \tau_j^{(2)},\hat{\tau}_j^{(2)} \right)\right|\leq Q_{(|\Delta_j|-2\rho_N)/\sigma}(\sqrt[J]{1-\alpha})\Big| \hat{J}=k;\,\boldsymbol{\hat{\tau}^{(1)} }=\boldsymbol{t};\, \boldsymbol{\hat{\nu}^{(1)} }=\boldsymbol{v} \right]\nonumber\\
	&& -\prod_{j=1}^J P_{(|\Delta_j|-2\rho_N)/\sigma}(\sqrt[J]{1-\alpha})\Bigg) dG_N(k,\boldsymbol{t},\boldsymbol{v})+o(1)
	\end{eqnarray}
	It is therefore sufficient to show that the difference inside the integral is, for all $(k,\boldsymbol{t},\boldsymbol{\alpha})\in\mathcal{R}_N$, uniformly bounded in absolute value by a $o(1)$ term. In other words, show there is a sequence $C_{N,\alpha}=o(1)$ such that
	\begin{eqnarray}\label{eq:condprodprobdev}
	&&\Bigg| \mathbb{P}\left[ \max_{j=1,\dots,J}\left|\lambda_2\left( \tau_j^{(2)},\hat{\tau}_j^{(2)} \right)\right|\leq Q_{(|\Delta_j|-2\rho_N)/\sigma}(\sqrt[J]{1-\alpha})\Big| \hat{J}=k;\,\boldsymbol{\hat{\tau}^{(1)} }=\boldsymbol{t};\, \boldsymbol{\hat{\nu}^{(1)} }=\boldsymbol{v} \right]\nonumber\\
	&&-\prod_{j=1}^J P_{(|\Delta_j|-2\rho_N)/\sigma}(\sqrt[J]{1-\alpha})\Bigg|\leq C_{N,\alpha}
	\end{eqnarray}
	for all admissible $(k,\boldsymbol{t},\boldsymbol{v})\in\mathcal{R}_N$.
	%
	\newline
	\newline
	\indent Henceforth, consider only such admissible $k$ (which restricts to $k=J$), $\boldsymbol{t}$'s and $\boldsymbol{v}$'s, and additionally suppose that $N$ is at least large enough so that $\rho_N\leq |\underline{\Delta}|/16$ and distance between consecutive points in $S^{(1)}$ is more than 2 (i.e. $\min\limits_{i,j\in S^{(1)},i\neq j}|i-j|>2$). We will proceed to show (\ref{eq:condprodprobdev}), by obtaining an upper bound for the following absolute difference, for every $j=1,\dots, J$:
	\begin{eqnarray}
	&&\Bigg| \mathbb{P}\left[ \left|\hat{\tau}^{(2)}_j-\tau_j\right|\leq Q_{(|\Delta_j|-2\rho_N)/\sigma}(\sqrt[J]{1-\alpha})\Big| \hat{J}=J;\,\boldsymbol{\hat{\tau}^{(1)} }=\boldsymbol{t};\, \boldsymbol{\hat{\nu}^{(1)} }=\boldsymbol{v} \right]\nonumber\\
	&&-P_{(|\Delta_j|-2\rho_N)/\sigma}(\sqrt[J]{1-\alpha})\Bigg|
	\end{eqnarray}
	This upper bound will be derived in several components.
	\newline
	\newline
	\noindent \textbf{First Component:} A more explicit expression for the change point estimates is
	\begin{eqnarray}
	\hat{\tau}^{(2)}_j&=&\underset{t\in S^{(2)}(t_i)}{\arg\min}(\text{sgn}(v_j-v_{j-1}))\sum_{i\in S^{(2)}(t_i)}\left( Y_i-\frac{v_j+v_{j-1}}{2} \right)( 1(i\leq t)-1(i\leq \tau_j^{(2)} ))\nonumber\\
	&=:& \underset{t\in S^{(2)}(t_i)}{\arg\min} \hat{X}_j^{(2)}(t)
	\end{eqnarray}
	Since $N$ was assumed to be large enough so that $\rho_N< |\underline{\Delta}|/8\leq (\nu_j-\nu_{j-1})/8$, the sign of $v_j-v_{j-1}$ is the same as the sign of $\Delta_j:=\nu_j-\nu_{j-1}$, making  the optimized expression above equal
	\begin{align}
	\hat{X}^{(2)}_j(t) &=\text{sgn}(\Delta_j)\sum_{i\in S^{(2)}(t_i)}\left( Y_i-\frac{v_j+v_{j-1}}{2} \right)( 1(i\leq t)-1(i\leq \tau_j^{(2)} ))\nonumber\\
	&=\begin{cases}
	\left|(\tau_j^{(2)},t]\cap S^{(2)}(t_j)\right|\left[\frac{|\Delta_j|}{2}+\text{sgn}(\Delta_j)\hat{D}_j\right]+\text{sgn}(\Delta_j)\underset{i\in (\tau_j^{(2)},t]\cap S^{(2)}(t_j)}{\sum}\varepsilon_i \, & \text{for }t>\tau_j\\
	0 &t=\tau_j\\
	\left| (t,\tau_j^{(2)}]\cap S^{(2)}(t_j) \right|\left[\frac{|\Delta_j|}{2}-\text{sgn}(\Delta_j)\hat{D}_j\right]-\text{sgn}(\Delta_j)\underset{i\in (t,\tau_j^{(2)}]\cap S^{(2)}(t_j) }{\sum}\varepsilon_{i}\qquad & \text{for }t>\tau_j
	\end{cases}\nonumber\\
	&=\begin{cases}
	\left(\pi_2(t)-\pi_2(\tau_j^{(2)})\right)\cdot\left[\frac{|\Delta_j|}{2}+\text{sgn}(\Delta_j)\hat{D}_j\right]+\text{sgn}(\Delta_j)\sum\limits_{i=\pi_2(\tau_j^{(2)})+1  }^{\pi_2(t)}\varepsilon_{\pi_2^{-1}(i)} \qquad & \text{for }t>\tau_j\\
	0 &t=\tau_j\\
	\left(\pi_2(\tau_j^{(2)})-\pi_2(t)\right)\cdot\left[\frac{|\Delta_j|}{2}-\text{sgn}(\Delta_j)\hat{D}_j\right]-\text{sgn}(\Delta_j)\sum\limits_{i=\pi_2(t)+1}^{\pi_2(\tau_j^{(2)})}\varepsilon_{\pi_2^{-1}(i)}\qquad & \text{for }t<\tau_j
	\end{cases}\nonumber\\
	&\text{since }(a,b]\cap S^{(2)}(t_j)=\pi_2^{-1}\left((\pi_2(a),\pi_2(b)]\right)\text{ for any }(a,b]\subset S^{(2)}(t_j)\nonumber\\
	\end{align}
	where 
	\begin{eqnarray}
	\hat{D}_j=\frac{\nu_j-v_j}{2}+\frac{\nu_{j-1}-v_{j-1}}{2}
	\end{eqnarray}
	which is less than $\rho_N$ in absolute value. From the equalities written above, we can deduce that for any integer $t$ such that $\pi_2^{-1}\left( \pi_2(\tau^{(2)}_j)+t \right)\in S^{(2)}(t_j)$, 
	\begin{eqnarray}
	&&\hat{X}_j^{(2)}\left( \pi_2^{-1}\left(\pi_2(\tau^{(2)}_j)+t  \right) \right)\nonumber\\
	&=&\begin{cases}
	t\cdot\left[\frac{|\Delta_j|}{2}+\text{sgn}(\Delta_j)\hat{D}_j\right]+\text{sgn}(\Delta_j)\sum\limits_{i=\pi_2(\tau_j^{(2)})+1  }^{\pi_2(\tau_j^{(2)})+t}\varepsilon_{\pi_2^{-1}(i)} \qquad & \text{for }t>\tau_j\\
	0 &t=\tau_j\\
	\left|t\right|\cdot\left[\frac{|\Delta_j|}{2}-\text{sgn}(\Delta_j)\hat{D}_j\right]-\text{sgn}(\Delta_j)\sum\limits_{i=\pi_2(\tau_j^{(2)})+t+1}^{\pi_2(\tau_j^{(2)})}\varepsilon_{\pi_2^{-1}(i)}\qquad & \text{for }t<\tau_j
	\end{cases}
	\end{eqnarray}
	This allows us to compare $\hat{X}_j^{(2)}$ with the random walk $X'_{|\Delta_j|-2|\hat{D}_j|}$:
	\begin{eqnarray}
	X'_{|\Delta_j|-2|\hat{D}_j|}(t):=\begin{cases}
	t\frac{|\Delta_j|-2|\hat{D}_j|}{2}+\text{sgn}(\Delta_j)\sum\limits_{i=1}^{t}\varepsilon_{ \pi_2^{-1}\left(\pi_2(\tau_j^{(2)})+i  \right) }\qquad & \text{for }t>0\\
	0 &t=0\\
	|t|\frac{|\Delta_j|-2|\hat{D}_j|}{2}-\text{sgn}(\Delta_j)\sum\limits_{i=t}^{-1}\varepsilon_{\pi_{2}^{-1}\left(\pi_2(\tau_j^{(2)})+i+1\right)}\qquad & \text{for }t<0
	\end{cases}.
	\end{eqnarray}
	Specifically, the random process $\hat{X}^{(2)}_j\left( \pi_2^{-1}\left(\pi_2(\tau^{(2)}_j)+t  \right) \right)$ either equals the sum $X'_{|\Delta_j|-2|\hat{D}_j|}(t)+2|\hat{D}_jt|1(t>0)$ or $X'_{|\Delta_j|-2|\hat{D}_j|}(t)+2|\hat{D}_jt|1(t<0)$. In either case, this begs the use of Lemma \ref{lem:randomwalkcompoppo} and Lemma \ref{lem:randomwalkcomp}, but first we must verify some the conditions of those results, namely that $|\Delta_j|>2|\hat{D}_j|$ (automatically true since $N$ was assumed to be large enough so that $|\Delta_j|\geq \underline{\Delta} \geq 8\rho_N$), and secondly, $\pi_2^{-1}\left( \pi_2(\tau^{(2)}_j)+t \right)\in S^{(2)}(t_j)$ for $|t|\leq Q_{(|\Delta_j|-2\rho_N)/\sigma}(\sqrt[J]{1-\alpha})$.
	To see that this is true, first use (\ref{eq:compdists}) to arrive at
	\begin{eqnarray}
	&&\left|\pi_2^{-1}\left(\pi_2(\tau^{(2)}_j)\pm Q_{(|\Delta_j|-2\rho_N)/\sigma}(\sqrt[J]{1-\alpha}) \right)-\tau_j^{(2)}\right|\nonumber\\
	&=&\left|\pi_2^{-1}\left(\pi_2(\tau^{(2)}_j)\pm Q_{(|\Delta_j|-2\rho_N)/\sigma}(\sqrt[J]{1-\alpha}) \right)-\pi_2^{-1}\left(\pi_2\left(\tau_j^{(2)}\right)\right)  \right|\nonumber\\
	&\leq&
	2Q_{(|\Delta_j|-2\rho_N)/\sigma}(\sqrt[J]{1-\alpha})
	\end{eqnarray}
	Therefore, given any $|t|\leq Q_{(|\Delta_j|-2\rho_N)/\sigma}(\sqrt[J]{1-\alpha})$, $\pi_2^{-1}\left( \pi_2\left( \tau_j^{(2)} \right)+t \right)\in\left[ \tau_j^{(2)}\pm 2Q_{(|\Delta_j|-2\rho_N)/\sigma}(\sqrt[J]{1-\alpha}) \right]-S^{(1)}$. Next, due to Lemma \ref{lem:quantbound}, there are positive expressions $C_1(\cdot)$, $C_2(\cdot)$, both decreasing, such that
	\begin{eqnarray}
	Q_{(|\Delta_j|-2\rho_N)/\sigma}(\sqrt[J]{1-\alpha})&\leq& C_1(|\Delta_j|-2\rho_N)\log\left(\frac{C_2(|\Delta_j|-2\rho_N)J}{\alpha} \right)\nonumber\\
	&\leq & C_1(\underline{\Delta}/2)\log\left(\frac{C_2(\underline{\Delta}/2)J}{\alpha} \right)\nonumber\\
	&\leq &C_1(\underline{\Delta}/2)\log\left(\frac{C_2(\underline{\Delta}/2)N}{\alpha} \right)
	\end{eqnarray}
	Since $w(N)$ is greater than order of $N^{1-\gamma}$, which in turn is greater in order than $\log(N)$, we see that for all large $N$:	
	\begin{eqnarray}
	&&t_j-Kw(N)\leq \tau_j^{(2)}-(K-1)w(N)+1 < \tau_j^{(2)}-2Q_{(|\Delta_j|-2\rho_N)/\sigma}(\sqrt[J]{1-\alpha})\nonumber\\
	&\leq & \pi_2^{-1}\left(\pi_2(\tau^{(2)}_j)- Q_{(|\Delta_j|-2\rho_N)/\sigma}(\sqrt[J]{1-\alpha}) \right)<\pi_2^{-1}\left(\pi_2(\tau^{(2)}_j)+ Q_{(|\Delta_j|-2\rho_N)/\sigma}(\sqrt[J]{1-\alpha}) \right)\nonumber\\
	&\leq &\tau_j^{(2)}+2Q_{(|\Delta_j|-2\rho_N)/\sigma}(\sqrt[J]{1-\alpha})\leq \tau_j^{(2)}+(K-1)w(N)-1\leq t_j+Kw(N)
	\end{eqnarray}
	Therefore given any $|t|\leq Q_{(|\Delta_j|-2\rho_N)/\sigma}(\sqrt[J]{1-\alpha})$,
	\begin{eqnarray}
	\pi_2^{-1}\left( \pi_2\left( \tau_j^{(2)} \right)+t \right)&\in&\left[ \tau_j^{(2)}\pm 2Q_{(|\Delta_j|-2\rho_N)/\sigma}(\sqrt[J]{1-\alpha}) \right]-S^{(1)}\nonumber\\
	&\subset& \left[ t_j\pm Kw(N) \right]-S^{(1)}\nonumber\\
	&=& S^{(2)}(t_j)
	\end{eqnarray}
	showing that the conditions of Lemma \ref{lem:randomwalkcompoppo} are satisfied.
	Before continuing, we make a small point before continuing. Note that for any integer $t^*\in S^{(2)}(t_j)$,  $\lambda_2\left(\tau_j^{(2)},\hat{\tau}^{(2)}_j\right)=t^*$ is an equivalent event to $\underset{t:\,\pi_2^{-1}\left( \pi_2(\tau_j^{(2)})+t \right)\in S^{(2)}}{\arg\min}\hat{X}^{(2)}_j\left( \pi_2^{-1}\left( \pi_2(\tau_j^{(2)})+t \right) \right)=t^*$. This is because
	\begin{eqnarray}\label{eq:redefinelambda2}
	&&\underset{t:\,\pi_2^{-1}\left( \pi_2(\tau_j^{(2)})+t \right)\in S^{(2)}}{\arg\min}\hat{X}^{(2)}_j\left( \pi_2^{-1}\left( \pi_2(\tau_j^{(2)})+t \right) \right)=t^*\nonumber\\
	&\longleftrightarrow & \underset{t\in S^{(2)}(t_j)}{\arg\min} \hat{X}^{(2)}_j(t)=\pi_2^{-1}\left( \pi_2(\tau_j^{(2)})+t^*\right)\nonumber\\
	&\longleftrightarrow & \pi_2\left( \underset{t\in S^{(2)}(t_j)}{\arg\min} \hat{X}^{(2)}_j(t) \right)-\pi_2(\tau_j^{(2)})=t^*\nonumber\\
	&\longleftrightarrow & \lambda_2\left(\tau_j^{(2)},\hat{\tau}^{(2)}_j\right)=t^*
	\end{eqnarray}
	We are now ready to apply Lemma \ref{lem:randomwalkcompoppo}:		
	\begin{eqnarray}\label{eq:underdiff1}
	&&\mathbb{P}\left[ \left|\lambda_2\left(\tau_j^{(2)},\hat{\tau}^{(2)}_j\right)\right|\leq Q_{(|\Delta_j|-2\rho_N)/\sigma}(\sqrt[J]{1-\alpha})\Bigg|\hat{J}=J,\boldsymbol{\hat{\tau}^{(1)}}=\boldsymbol{t},\, \boldsymbol{\hat{\nu}^{(1)} }=\boldsymbol{v}\right]\nonumber\\
	&&-\mathbb{P}\left[ \left|\underset{t:\,\pi_2^{-1}\left( \pi_2(\tau_j^{(2)})+t\right)\in S^{(2)}}{\arg\min}X'_{|\Delta_j|-2|\hat{D}_j|}(t)\right|\leq Q_{(|\Delta_j|-2\rho_N)/\sigma}(\sqrt[J]{1-\alpha})\Bigg|\hat{J}=J,\boldsymbol{\hat{\tau}^{(1)}}=\boldsymbol{t},\, \boldsymbol{\hat{\nu}^{(1)} }=\boldsymbol{v} \right]  \nonumber\\
	&=&\mathbb{P}\left[ \left|\underset{t:\,\pi_2^{-1}\left( \pi_2(\tau_j^{(2)})+t \right)\in S^{(2)}}{\arg\min}\sigma^{-1}\hat{X}^{(2)}_j\left( \pi_2^{-1}\left( \pi_2(\tau_j^{(2)})+t \right) \right)\right|\leq Q_{(|\Delta_j|-2\rho_N)/\sigma}(\sqrt[J]{1-\alpha})\Bigg|(\hat{J},\boldsymbol{\hat{\tau}^{(1)}},\boldsymbol{\hat{\nu}^{(1)} })=(J,\boldsymbol{t}, \boldsymbol{v})\right]\nonumber\\
	&&-\mathbb{P}\left[ \left|\underset{t:\,\pi_2^{-1}\left( \pi_2(\tau_j^{(2)})+t\right)\in S^{(2)}}{\arg\min}\sigma^{-1}X'_{|\Delta_j|-2|\hat{D}_j|}(t)\right|\leq Q_{(|\Delta_j|-2\rho_N)/\sigma}(\sqrt[J]{1-\alpha})\Bigg|\hat{J}=J,\boldsymbol{\hat{\tau}^{(1)}}=\boldsymbol{t},\, \boldsymbol{\hat{\nu}^{(1)} }=\boldsymbol{v} \right]\nonumber\\
	&&\text{indexing change by }(\ref{eq:redefinelambda2});\text{ and division by constant } \sigma\text{ does not change argmin value}\nonumber\\
	&\leq& 2\hat{D}_j\sigma^{-1}\Big[A_1^+\left((|\Delta_j|-2\hat{D}_j)/\sigma\right)\left(Q_{(|\Delta_j|-2\rho_N)/\sigma}(\sqrt[J]{1-\alpha})\right)^{3/2}\nonumber\\
	&&+B_1^+\left((|\Delta_j|-2\hat{D}_j)/\sigma\right)\sqrt{Q_{(|\Delta_j|-2\rho_N)/\sigma}(\sqrt[J]{1-\alpha})}\Big]\nonumber\\
	&&\times \exp\left[ -C_1^+\left( (|\Delta_j|-2\hat{D}_j)/\sigma \right)Q_{(|\Delta_j|-2\rho_N)/\sigma}(\sqrt[J]{1-\alpha}) \right]\nonumber
	\end{eqnarray}
	Because, as stated previously, $\hat{X}^{(2)}_j\left( \pi_2^{-1}\left(\pi_2(\tau^{(2)}_j)+t  \right) \right)$ either equals $X'_{|\Delta_j|-2|\hat{D}_j|}(t)+2|\hat{D}_jt|1(t>0)$ or $X'_{|\Delta_j|-2|\hat{D}_j|}(t)+2|\hat{D}_jt|1(t<0)$, Lemma \ref{lem:randomwalkcompoppo} leads to this inequality for some positive monotone expressions $A_1^+()$, $B_1^+()$, and $C_1^+()$, which are decreasing, decreasing, and increasing. This expression could be further bounded by
	\begin{align}\label{eq:firstcomp1}
	&\leq  2\sigma^{-1}\hat{D}_j\left[A_1^+\left(\underline{\Delta}/2\sigma\right)\left(Q_{(|\Delta_j|-2\rho_N)/\sigma}(\sqrt[J]{1-\alpha})\right)^{3/2}+B_1^+\left(\underline{\Delta}/2\sigma\right)\sqrt{Q_{(|\Delta_j|-2\rho_N)/\sigma}(\sqrt[J]{1-\alpha})}\right]\nonumber\\
	&\times\exp\left[ -C_1^+\left( \underline{\Delta}/2\sigma \right)Q_{(|\Delta_j|-2\rho_N)/\sigma}(\sqrt[J]{1-\alpha}) \right]\nonumber\\
	&\qquad \leq 2C^*\left(\underline{\Delta}/2\sigma\right)\rho_N\nonumber\\
	&\text{where }C_+^*(\cdot)=\sup_{x\in\mathbb{R}^+}\left([A(\cdot)x^{3/2}+B(\cdot)\sqrt{x}]\exp(-C(\cdot)x)\right)\text{, guaranteed to be finite}
	\end{align}
	In a similar manner, apply Lemma \ref{lem:randomwalkcomp} to obtain 
	\begin{align}\label{eq:firstcomp2}
	&\mathbb{P}\left[ \left|\lambda_2\left(\tau_j^{(2)},\hat{\tau}^{(2)}_j\right)\right|\leq Q_{(|\Delta_j|-2\rho_N)/\sigma}(\sqrt[J]{1-\alpha})\Bigg|\hat{J}=J,\boldsymbol{\hat{\tau}^{(1)}}=\boldsymbol{t},\, \boldsymbol{\hat{\nu}^{(1)} }=\boldsymbol{v}\right]\nonumber\\
	&-\mathbb{P}\left[ \left|\underset{t:\,\pi_2^{-1}\left( \pi_2(\tau_j^{(2)})+t\right)\in S^{(2)}}{\arg\min}X'_{|\Delta_j|-2|\hat{D}_j|}(t)\right|\leq Q_{(|\Delta_j|-2\rho_N)/\sigma}(\sqrt[J]{1-\alpha})\Bigg|\hat{J}=J,\boldsymbol{\hat{\tau}^{(1)}}=\boldsymbol{t},\, \boldsymbol{\hat{\nu}^{(1)} }=\boldsymbol{v} \right]  \nonumber\\
	&=\mathbb{P}\left[ \left|\underset{t:\,\pi_2^{-1}\left( \pi_2(\tau_j^{(2)})+t \right)\in S^{(2)}}{\arg\min}\frac{1}{\sigma}\hat{X}^{(2)}_j\left( \pi_2^{-1}\left( \pi_2(\tau_j^{(2)})+t \right) \right)\right|\leq Q_{(|\Delta_j|-2\rho_N)/\sigma}(\sqrt[J]{1-\alpha})\Bigg|\hat{J}=J,\boldsymbol{\hat{\tau}^{(1)}}=\boldsymbol{t},\, \boldsymbol{\hat{\nu}^{(1)} }=\boldsymbol{v}\right]\nonumber\\
	&-\mathbb{P}\left[ \left|\underset{t:\,\pi_2^{-1}\left( \pi_2(\tau_j^{(2)})+t\right)\in S^{(2)}}{\arg\min}\frac{1}{\sigma}X'_{|\Delta_j|-2|\hat{D}_j|}(t)\right|\leq Q_{(|\Delta_j|-2\rho_N)/\sigma}(\sqrt[J]{1-\alpha})\Bigg|\hat{J}=J,\boldsymbol{\hat{\tau}^{(1)}}=\boldsymbol{t},\, \boldsymbol{\hat{\nu}^{(1)} }=\boldsymbol{v} \right]\nonumber\\
	&\geq  -2A_1^-\left(\left( |\Delta_j|-6|\hat{D}_j|\right)/\sigma \right)\rho_N\sqrt{Q_{(|\Delta_j|-2\rho_N)/\sigma}(\sqrt[J]{1-\alpha})}\nonumber\\
	&\times\exp\left[ -B_1^- \left(\left( |\Delta_j|-6|\hat{D}_j|\right)/\sigma \right)  Q_{(|\Delta_j|-2\rho_N)/\sigma}(\sqrt[J]{1-\alpha}) \right]\nonumber\\
	&\text{for some positive decreasing expression }A_1^-\text{ and some positive increasing expression }B_1^-\nonumber\\
	&\geq -2A_1^-\left( \underline{\Delta}/2\sigma \right)\rho_N\sqrt{Q_{(|\Delta_j|-2\rho_N)/\sigma}(\sqrt[J]{1-\alpha})}\exp\left[ -B_1^- \left( \underline{\Delta}/2\sigma \right)  Q_{(|\Delta_j|-2\rho_N)/\sigma}(\sqrt[J]{1-\alpha}) \right]\nonumber\\
	&\geq -2C_1^-(\underline{\Delta}/2\sigma)\rho_N\quad\text{ where }C_1^-(\cdot)=\sup_{x\in\mathbb{R}^+}A_1^-(\cdot)\sqrt{x}\exp\left[ B_1^-(\cdot)x \right]\nonumber\\
	\end{align}
	Altogether, both (\ref{eq:firstcomp1}) and (\ref{eq:firstcomp2}) together imply
	\begin{eqnarray}\label{eq:firstcompabs}
	&&\Bigg|\mathbb{P}\left[ \left|\lambda_2\left(\tau_j^{(2)},\hat{\tau}^{(2)}_j\right)\right|\leq Q_{(|\Delta_j|-2\rho_N)/\sigma}(\sqrt[J]{1-\alpha})\Bigg|\hat{J}=J,\boldsymbol{\hat{\tau}^{(1)}}=\boldsymbol{t},\, \boldsymbol{\hat{\nu}^{(1)} }=\boldsymbol{v}\right]\nonumber\\
	&&-\mathbb{P}\left[ \left|\underset{t\in S^{(2)}(t_j)-\tau_j}{\arg\min}X'_{|\Delta_j|-2|\hat{D}_j|}(t)\right|\leq Q_{(|\Delta_j|-2\rho_N)/\sigma}(\sqrt[J]{1-\alpha})\Bigg|\hat{J}=J,\boldsymbol{\hat{\tau}^{(1)}}=\boldsymbol{t},\, \boldsymbol{\hat{\nu}^{(1)} }=\boldsymbol{v} \right]\Bigg|  \nonumber\\
	&\qquad &\leq 2\left( C_1^-(\underline{\Delta}/2\sigma)\vee C_1^+(\underline{\Delta}/2\sigma) \right)\rho_N
	\end{eqnarray}
	~\newline
	\newline
	\textbf{Second Component: }Now $X'_{|\Delta|-2|\hat{D}_j|}(t)/\sigma$ has the same exact distribution as $X_{(|\Delta|-2|\hat{D}_j|)/\sigma}(t)$, for all integers $t$ such that $\pi_2^{-1}\left( \pi_2(\tau_j^{(2)})+t\right)\in S^{(2)}(t_j)$. It was also shown in the previous section that the set $\left\{ t:\,\pi_2^{-1}\left( \pi_2(\tau_j^{(2)})+t\right)\in S^{(2)}(t_j) \right\}$ contains the interval of integers $\left[ \pm Q_{(|\Delta_j|-2\rho_N)/\sigma}(\sqrt[J]{1-\alpha}) \right]$. Therefore apply Lemma \ref{lem:twoXrandwalkcomp} to first obtain
	\begin{eqnarray}
	&&\mathbb{P}\left[ \left|\underset{ t:\,\pi_2^{-1}\left( \pi_2(\tau_j^{(2)})+t\right)\in S^{(2)}(t_j)}{\arg\min}X'_{|\Delta_j|-2|\hat{D}_j|}(t)\right|\leq Q_{(|\Delta_j|-2\rho_N)/\sigma}(\sqrt[J]{1-\alpha})\Bigg|\hat{J}=J,\boldsymbol{\hat{\tau}^{(1)}}=\boldsymbol{t},\, \boldsymbol{\hat{\nu}^{(1)} }=\boldsymbol{v} \right]\nonumber\\
	&=&\mathbb{P}\left[ \left|\underset{ t:\,\pi_2^{-1}\left( \pi_2(\tau_j^{(2)})+t\right)\in S^{(2)}(t_j)}{\arg\min}\frac{1}{\sigma}X'_{|\Delta_j|-2|\hat{D}_j|}(t)\right|\leq Q_{(|\Delta_j|-2\rho_N)/\sigma}(\sqrt[J]{1-\alpha})\Bigg|\hat{J}=J,\boldsymbol{\hat{\tau}^{(1)}}=\boldsymbol{t},\, \boldsymbol{\hat{\nu}^{(1)} }=\boldsymbol{v} \right]\nonumber
	\end{eqnarray}
	\begin{eqnarray}
	&=&\mathbb{P}\left[ \left|\underset{ t:\,\pi_2^{-1}\left( \pi_2(\tau_j^{(2)})+t\right)\in S^{(2)}(t_j)}{\arg\min}X_{(|\Delta_j|-2|\hat{D}_j|)/\sigma}(t)\right|\leq Q_{(|\Delta_j|-2\rho_N)/\sigma}(\sqrt[J]{1-\alpha}) \right]\nonumber\\
	&\geq &\mathbb{P}\left[ \left|\underset{ t:\,\pi_2^{-1}\left( \pi_2(\tau_j^{(2)})+t\right)\in S^{(2)}(t_j)}{\arg\min}X_{(|\Delta_j|-2\rho_N)/\sigma}(t)\right|\leq Q_{(|\Delta_j|-2\rho_N)/\sigma}(\sqrt[J]{1-\alpha}) \right]
	\end{eqnarray}
	and apply Lemma \ref{lem:twoXrandwalkcomp} to obtain another equality in the other direction
	\begin{eqnarray}\label{eq:underdiff4}
	\mathbb{P}\left[ \left|\underset{t:\,\pi_2^{-1}\left( \pi_2(\tau_j^{(2)})+t\right)\in S^{(2)}(t_j)}{\arg\min}X'_{|\Delta_j|-2|\hat{D}_j|}(t)\right|\leq Q_{(|\Delta_j|-2\rho_N)/\sigma}(\sqrt[J]{1-\alpha})\Bigg|\hat{J}=J,\boldsymbol{\hat{\tau}^{(1)}}=\boldsymbol{t},\, \boldsymbol{\hat{\nu}^{(1)} }=\boldsymbol{v} \right]\nonumber
	\end{eqnarray}
	\begin{eqnarray}
	&=&\mathbb{P}\left[ \left|\underset{t:\,\pi_2^{-1}\left( \pi_2(\tau_j^{(2)})+t\right)\in S^{(2)}(t_j)}{\arg\min}X_{(|\Delta_j|-2|\hat{D}_j|)/\sigma}(t)\right|\leq Q_{(|\Delta_j|-2\rho_N)/\sigma}(\sqrt[J]{1-\alpha}) \right]\nonumber\\
	&\leq & \mathbb{P}\left[ \left|\underset{t:\,\pi_2^{-1}\left( \pi_2(\tau_j^{(2)})+t\right)\in S^{(2)}(t_j)}{\arg\min}X_{(|\Delta_j|-2\rho_N)/\sigma}(t)\right|\leq Q_{(|\Delta_j|-2\rho_N)/\sigma}(\sqrt[J]{1-\alpha})\right]+
	\end{eqnarray}
	\begin{eqnarray}
	&&(2\rho_N-2|\hat{D}_j|)\Bigg[A_2\left((|\Delta_j|-2\rho_N)/\sigma\right)\left(Q_{(|\Delta_j|-2\rho_N)/\sigma}(\sqrt[J]{1-\alpha})\right)^{3/2}\nonumber\\
	&&+B_2\left((|\Delta_j|-2\rho_N)/\sigma\right)\sqrt{Q_{(|\Delta_j|-2\rho_N)/\sigma}(\sqrt[J]{1-\alpha})}\Bigg]\nonumber\\
	&&\times\exp\left[ -C_2\left( |\Delta_j|-2\rho_N \right)Q_{|\Delta_j|-2\rho_N}(\sqrt[J]{1-\alpha}) \right]\nonumber\\
	&\leq &\mathbb{P}\left[ \left|\underset{t\in S^{(2)}(t_j)-\tau_j}{\arg\min}X_{(|\Delta_j|-2\rho_N)/\sigma}(t)\right|\leq Q_{(|\Delta_j|-2\rho_N)/\sigma}(\sqrt[J]{1-\alpha})\right]+\nonumber\\
	&&2\rho_N\left[A_2\left(\underline{\Delta}/2\sigma\right)\left(Q_{(|\Delta_j|-2\rho_N)/\sigma}(\sqrt[J]{1-\alpha})\right)^{3/2}+B_2\left(\underline{\Delta}/2\sigma\right)\sqrt{Q_{(|\Delta_j|-2\rho_N)/\sigma}(\sqrt[J]{1-\alpha})}\right]\nonumber\\
	&&\times\exp\left[ -C_2\left( \underline{\Delta}/2\sigma \right)Q_{(|\Delta_j|-2\rho_N)/\sigma}(\sqrt[J]{1-\alpha}) \right]\nonumber\\
	&&\text{since }N\text{ is large enough such that } 2\rho_N<\underline{\Delta}/2\text{, and by the monotonicity of }A_2,B_2,\text{ and }C_2\nonumber\\
	&\leq &\mathbb{P}\left[ \left|\underset{t\in S^{(2)}(t_j)-\tau_j}{\arg\min}X_{(|\Delta_j|-2\rho_N)/\sigma}(t)\right|\leq Q_{(|\Delta_j|-2\rho_N)/\sigma}(\sqrt[J]{1-\alpha})\right]+2C^*_2\left(\frac{\underline{\Delta}}{2\sigma}\right)\rho_N\nonumber\\
	&&\text{where }C^*_2(\cdot)=\sup_{x\in\mathbb{R}^+}\left([A_2(\cdot)x^{3/2}+B_2(\cdot)\sqrt{x}]\exp(-C_2(\cdot)x)\right)
	\end{eqnarray}
	These two inequalities together imply a bound on the absolute difference:
	\begin{align}\label{eq:secondcompabs}
	\Biggr|\mathbb{P}&\left[ \left|\underset{ t:\,\pi_2^{-1}\left( \pi_2(\tau_j^{(2)})+t\right)\in S^{(2)}(t_j)}{\arg\min}X'_{|\Delta_j|-2|\hat{D}_j|}(t)\right|\leq Q_{(|\Delta_j|-2\rho_N)/\sigma}(\sqrt[J]{1-\alpha})\Bigg|\hat{J}=J,\boldsymbol{\hat{\tau}^{(1)}}=\boldsymbol{t},\, \boldsymbol{\hat{\nu}^{(1)} }=\boldsymbol{v} \right]\nonumber\\
	&-\mathbb{P}\left[ \left|\underset{t:\,\pi_2^{-1}\left( \pi_2(\tau_j^{(2)})+t\right)\in S^{(2)}(t_j)}{\arg\min}X_{(|\Delta_j|-2\rho_N)/\sigma}(t)\right|\leq Q_{(|\Delta_j|-2\rho_N)/\sigma}(\sqrt[J]{1-\alpha})\right]\Biggr|\nonumber\\
	&\qquad\leq 2C^*_2\left(\frac{\underline{\Delta}}{2\sigma}\right)\rho_N
	\end{align}
	~\newline
	\newline
	\textbf{Third Component: }We note that the set $\left\{ t:\,\pi_2^{-1}\left( \pi_2(\tau_j^{(2)})+t\right)\in S^{(2)}(t_j) \right\}$ contains the set $\left[\pm \left(\frac{(K-1)w(N)}{2}-1\right) \right]$. This is because by (\ref{eq:compdists}),
	\begin{eqnarray}
	\left|\pi_2^{-1}\left( \pi_2(\tau_j^{(2)})\pm \left(\frac{(K-1)w(N)}{2}-1\right)\right)-\tau_j^{(2)}\right|\leq 2\left(\frac{(K-1)w(N)}{2}-1\right)
	\end{eqnarray}
	secondly because
	\begin{eqnarray}
	&&t_j-Kw(N)\leq \tau_j^{(2)}-(K-1)w(N)+1 < \tau_j^{(2)}-2\left( \frac{(K-1)w(N)}{2}-1 \right)\nonumber\\
	&< & \tau_j^{(2)}+2\left( \frac{(K-1)w(N)}{2}-1 \right)\leq \tau_j^{(2)}+(K-1)w(N)\leq t_j+Kw(N)
	\end{eqnarray}
	Therefore, for any $t$ such that $|t|\leq \left(\frac{(K-1)w(N)}{2}-1\right)$, we have $\pi_2^{-1}(\pi_2(\tau_j^{(2)})+t)\in [t_j\pm Kw(N)]-S^{(1)}=S^{(2)}(t_j)$. This allows an application of Lemma \ref{lem:eqprob} to obtain
	\begin{eqnarray}\label{eq:underdiff2}
	&&\Biggr|\mathbb{P}\left[ \left|\underset{t:\,\pi_2^{-1}\left( \pi_2(\tau_j^{(2)})+t\right)\in S^{(2)}(t_j)}{\arg\min}X_{(|\Delta_j|-2\rho_N)/\sigma}(t)\right|\leq Q_{(|\Delta_j|-2\rho_N)/\sigma}(\sqrt[J]{1-\alpha})\right]\nonumber\\
	&&-\mathbb{P}\left[ \left|\underset{|t|\leq \frac{(K-1)W(N)}{2}-1}{\arg\min}X_{(|\Delta_j|-2\rho_N)/\sigma}(t)\right|\leq Q_{(|\Delta_j|-2\rho_N)/\sigma}(\sqrt[J]{1-\alpha})\right]\Biggr|\nonumber\\
	&\leq&  A_3((|\Delta_j|-2\rho_N)/\sigma)\exp\left(-B_3((|\Delta_j|-2\rho_N)/\sigma)\left(\frac{(K-1)w(N)}{2}-1\right)\right)\nonumber\\
	&&\text{for some decreasing expression }A_3()\text{ and increasing expression }B_3()\nonumber\\
	&\leq & A_3(\underline{\Delta}/2\sigma)\exp\left(-B_3(\underline{\Delta}/2\sigma)\left(\frac{(K-1)w(N)}{2}-1\right)\right)
	\end{eqnarray}
	The same application of the lemma can also yield
	\begin{eqnarray}\label{eq:underdiff3}
	&&\Bigg|\mathbb{P}\left[ \left|\underset{|t|\leq \frac{(K-1)W(N)}{2}-1}{\arg\min}X_{(|\Delta_j|-2\rho_N)/\sigma}(t)\right|\leq Q_{(|\Delta_j|-2\rho_N)/\sigma}(\sqrt[J]{1-\alpha})\right]
	\nonumber\\
	&&-\mathbb{P}\left[ \left|\underset{t\in\mathbb{Z}}{\arg\min}X_{(|\Delta_j|-2\rho_N)/\sigma}(t)\right|\leq Q_{(|\Delta_j|-2\rho_N)/\sigma}(\sqrt[J]{1-\alpha})\right]\Bigg|\nonumber\\
	&\leq &A_3(\underline{\Delta}/2\sigma)\exp\left(-B'(\underline{\Delta}/2\sigma)\left(\frac{(K-1)w(N)}{2}-1\right)\right)
	\end{eqnarray}
	Adding up these two upper bounds imply
	\begin{eqnarray}\label{eq:thirdcompabs}
	&&\Biggr|\mathbb{P}\left[ \left|\underset{t:\,\pi_2^{-1}\left( \pi_2(\tau_j^{(2)})+t\right)\in S^{(2)}(t_j)}{\arg\min}X_{(|\Delta_j|-2\rho_N)/\sigma}(t)\right|\leq Q_{(|\Delta_j|-2\rho_N)/\sigma}(\sqrt[J]{1-\alpha})\right]\nonumber\\
	&&-\mathbb{P}\left[ \left|\underset{t\in\mathbb{Z}}{\arg\min}X_{(|\Delta_j|-2\rho_N)/\sigma}(t)\right|\leq Q_{(|\Delta_j|-2\rho_N)/\sigma}(\sqrt[J]{1-\alpha})\right]\Bigg|\nonumber\\
	&\leq &2A_3(\underline{\Delta}/2\sigma)\exp\left(-B'(\underline{\Delta}/2\sigma)\left(\frac{(K-1)w(N)}{2}-1\right)\right)
	\end{eqnarray}
	~\newline
	\noindent\textbf{Sum of the Components:} Adding up the differences in (\ref{eq:firstcompabs}), (\ref{eq:secondcompabs}), and (\ref{eq:thirdcompabs}):
	\begin{eqnarray}\label{eq:individualbound}
	&&\Bigg|\mathbb{P}\left[ \left|\lambda_2\left(\tau_j^{(2)},\hat{\tau}^{(2)}_j\right)\right|\leq Q_{(|\Delta_j|-2\rho_N)/\sigma}(\sqrt[J]{1-\alpha})\Bigg|\hat{J}=J,\boldsymbol{\hat{\tau}^{(1)}}=\boldsymbol{t},\, \boldsymbol{\hat{\nu}^{(1)} }=\boldsymbol{v}\right]\nonumber\\
	&&-\mathbb{P}\left[ \left|\underset{t\in\mathbb{Z}}{\arg\min}X_{(|\Delta_j|-2\rho_N)/\sigma}(t)\right|\leq Q_{(|\Delta_j|-2\rho_N)/\sigma}(\sqrt[J]{1-\alpha})\right]\Bigg|\nonumber\\
	&\leq& C_4\rho_N+C_5\exp\left[-C_6(K-1)w(N)\right]
	\end{eqnarray}
	for some constants $C_4$, $C_5$, and $C_6$.
	\newline
	\newline
	\indent Finally, in order to bound (\ref{eq:condprodprobdev}), we note that given the two real valued triangular arrays $a_{N,1},\dots,a_{N,J}$ and $b_{N,1},\dots, b_{N,J}$, all of which contained in the continuous interval $[0,1]$, such that $\left|\frac{a_{N,i}-b_{N,i}}{b_{N,i}}\right|\leq C_N$ for $1\leq i\leq J$, where $JC_N\to 0$ as $N\to \infty$, then $\left|\prod a_{N,j}-\prod b_{N,j}\right|\to 0$. This is because
	\begin{eqnarray}
	\left|\prod_{j=1}^J a_{N,j}-\prod_{j=1}^J b_{N,j}\right|&=& \left|\prod_{j=1}^J b_{N,j}\right|\left| \prod_{j=1}^J \left(1-\frac{a_{N,j}-b_{N,j}}{b_{N,j}}\right)-1 \right|\nonumber\\
	\end{eqnarray}
	Since $\left|\prod b_{N,j}\right|\in [0,1]$, the above converges to 0 if $\prod \left(1-\frac{a_{N,j}-b_{N,j}}{b_{N,j}}\right)\to 1$, which is true since
	\begin{eqnarray}
	&&\prod_{j=1}^J \left(1-\frac{a_{N,j}-b_{N,j}}{b_{N,j}}\right)\geq (1-C_N)^J\geq 1-C_NJ\to 1\nonumber\\
	&&\prod_{j=1}^J \left(1-\frac{a_{N,j}-b_{N,j}}{b_{N,j}}\right)\leq (1+C_N)^J\leq (e^{C_N})^J\to e^0=1
	\end{eqnarray}
	This result is useful because for all $N$ large enough so that $2Kw(N)\leq \delta_N$, the stage two sets $S^{(2)}(t_j)$'s for $j=1,\dots,J$ are mutually exclusive sets, and hence by conditional independence, (\ref{eq:condprodprobdev}) equals
	\begin{eqnarray}\label{eq:condprodprobdev2}
	&&\Bigg| \prod_{j=1}^J\mathbb{P}\left[ \left|\lambda_2\left( \tau_j^{(2)},\hat{\tau}_j^{(2)} \right)\right|\leq Q_{(|\Delta_j|-2\rho_N)/\sigma}(\sqrt[J]{1-\alpha})\Big| \hat{J}=J;\,\boldsymbol{\hat{\tau}^{(1)} }=\boldsymbol{t};\, \boldsymbol{\hat{\nu}^{(1)} }=\boldsymbol{v} \right]\nonumber\\
	&&-\prod_{j=1}^J P_{(|\Delta_j|-2\rho_N)/\sigma}(\sqrt[J]{1-\alpha})\Bigg|
	%
	%
	\end{eqnarray}
	More-ever, using (\ref{eq:individualbound}),
	\begin{eqnarray}
	&&J\Bigg|\mathbb{P}\left[ \left|\lambda_2\left(\tau_j^{(2)},\hat{\tau}^{(2)}_j\right)\right|\leq Q_{(|\Delta_j|-2\rho_N)/\sigma}(\sqrt[J]{1-\alpha})\Bigg|\hat{J}=J,\boldsymbol{\hat{\tau}^{(1)}}=\boldsymbol{t},\, \boldsymbol{\hat{\nu}^{(1)} }=\boldsymbol{v}\right]\nonumber\\
	&&-\mathbb{P}_{|\Delta_j|-2\rho_N}(\sqrt[J]{1-\alpha})\Bigg|\left(\mathbb{P}_{|\Delta_j|-2\rho_N}(\sqrt[J]{1-\alpha})\right)^{-1}\nonumber\\
	&\leq&J\left(C_4\rho_N+C_5\exp\left[-C_6(K-1)w(N)\right]\right)(1-\alpha)^{-1/J}\nonumber\\
	&\leq&C_4(1-\alpha)^{-1}J\rho_N+C_5(1-\alpha)^{-1}N\exp\left[-C_6(K-1)w(N)\right]
	,\end{eqnarray}
	which goes to 0 since $J\rho_N\to 0$ and $w(N)\geq C N^{1-\gamma}$ for some constant $C$. This lets us conclude that (\ref{eq:condprodprobdev2}) converges to 0.
	%
\end{proof}

\subsection{Proof of Theorem \ref{thm:multidepend}}\label{sec:multidependthmproof}
\indent Notation-wise, we again utilize the $\pi_2$ function defined in (\ref{def:pi2func}), which is a bijection from the set $\{ 1,\dots,N \}-S^{(1)}$ to the set $\{ 1,\dots,N-|S^{(1)}| \}$.
\begin{proof} 
	Let $S^{(2)}\left( \hat{\tau}^{(1)}_k\right)$ for $k=1,...,\hat{J}$ be the second stage subsamples. As in previous sections, define
	\begin{eqnarray}
	\mathcal{R}_N:= \left\{ \hat{J}=J;\,\max_{j=1,\dots,J}\left| \hat{\tau}^{(1)}_j-\tau_j \right|\leq w(N);\, \max_{j=0,\dots,J}\left| \hat{\nu}^{(1)}_j-\nu_j \right|\leq \rho_N \right\}
	\end{eqnarray}
	We also define the following random functions on integers: for $k=1,...,J$, on the event $\mathcal{R}_N$ let
	\begin{align}
	\hat{X}^{(2)}_k(d)&:=
	\text{sgn}\left(\hat{\nu}^{(1)}_{k}-\hat{\nu}^{(1)}_{k-1}\right)\underset{j\in S^{(2)}(\hat{\tau}^{(1)}_k)}{\sum}\left(Y_j-\frac{\hat{\nu}^{(1)}_{k}+\hat{\nu}^{(1)}_{k-1}}{2}\right) \left(1\left(j\leq \pi_2^{-1}(\pi_2(\tau_k^{(2)})+d)\right)-1\left(j\leq \tau_k^{(2)}\right)\right) \nonumber\\ 
	&\qquad \text{for }\pi_2^{-1}(\pi_2(\tau_k^{(2)})+d)\in S^{(2)}_k\nonumber\\
	&\qquad\hat{X}^{(2)}_k(d):=\infty \qquad\text{otherwise}
	\end{align}
	and on the event $\mathcal{R}_N^C$ let 
	\begin{equation}
	\hat{X}^{(2)}_k(d):=d.
	\end{equation}
	so that the $\arg\min$ of $\hat{X}^{(2)}_k(d)$ is $d=-\infty$ on the event $\mathcal{R}_N^C$. Using this definition, for all sufficiently large $N$, the event $\{ \hat{J}=J,\, \lambda_2(\tau_k^{(2)},\hat{\tau}_k^{(2)})=j_k\text{ for }k=1,...,J \}\cap\mathcal{R}_N$ is equivalent to the event $$\left\{ \underset{d\in \mathbb{Z}}{\arg\min}\,\hat{X}^{(2)}_k(d)=j_k \text{ for }k=1,...,J\right\}.$$
	\newline
	\newline
	Further, for any integer $M>0$, we can easily obtain convergence properties by restricting the function $\hat{X}^{(2)}_k$ to the set $\{-M,-(M-1),...,M \}$. For sufficiently large $N$, when event $\mathcal{R}_N$ occurs, we have for any $d\in \{ -M,...,M \}$, and for all $k=1,...,J$,
	\begin{align}
	\hat{X}^{(2)}_k(d)= \text{sgn}(\Delta_k)\underset{j\in S^{(2)}(\hat{\tau}^{(1)}_k)}{\sum}\left(Y_j-\frac{\hat{\nu}^{(1)}_{k}+\hat{\nu}^{(1)}_{k-1}}{2}\right) \left(1\left(j\leq \pi_2^{-1}(\pi_2(\tau_k^{(2)})+d)\right)-1\left(j\leq \tau_k^{(2)}\right)  \right)\nonumber\\
	=\begin{cases} \text{sgn}(\Delta_k)\sum\limits_{j=\pi_2(\tau_k^{(2)})+1}^{\pi_2(\tau_k^{(2)})+d}\left( Y_{\pi^{-1}(j)}-\nu_k+\frac{\nu_k-\nu_{k-1}}{2}+\frac{1}{2}\left( \nu_k-\hat{\nu}_k^{(1)}+\nu_{k-1}-\hat{\nu}_{k-1}^{(1)}\right)\right)\quad &\text{for }d>0\\
	0 &\text{for }d=0\\
	-\text{sgn}(\Delta_k)\sum\limits_{j=\pi_2(\tau_k^{(2)})+d+1}^{\pi_2(\tau_k^{(2)})}\left( Y_{\pi_2^{-1}(j)}-\nu_{k-1}+\frac{\nu_{k-1}-\nu_{k}}{2}+\frac{1}{2}\left( \nu_k-\hat{\nu}_k^{(1)}+\nu_{k-1}-\hat{\nu}_{k-1}^{(1)}\right)\right) & \text{for }d<0
	\end{cases}\nonumber\\
	=\begin{cases}
	\frac{d|\Delta_k|}{2}+\text{sgn}(\Delta_k)\left(\sum\limits_{j=\pi_2(\tau_k^{(2)})+1}^{\pi_2(\tau_k^{(2)})+d}\varepsilon_{\pi_2^{-1}(j)}+\frac{d}{2}\left( \nu_k-\hat{\nu}_k^{(1)}+\nu_{k-1}-\hat{\nu}_{k-1}^{(1)}\right)\right)\quad &\text{for }d>0\\
	0 &\text{for }d=0\\
	-d\frac{|\Delta_k|}{2}-\text{sgn}(\Delta_k)\left(\sum\limits_{j=\pi_2(\tau_k^{(2)})+d+1}^{\pi_2(\tau_k^{(2)})}\varepsilon_{\pi_2^{-1}(j)}+\frac{d}{2}\left( \nu_k-\hat{\nu}_k^{(1)}+\nu_{k-1}-\hat{\nu}_{k-1}^{(1)}\right)\right) &\text{for }d<0
	\end{cases}\nonumber\\
	\end{align}
	Because $\max\limits_{i=1,...,J}|\hat{\nu}^{(1)}_i-\nu_o|\leq \rho_N$ under $\mathcal{R}_N$, this gives the uniform bound
	\begin{eqnarray}
	\left| \frac{d}{2\sigma}\left( \nu_k-\hat{\nu}_k^{(1)}+\nu_{k-1}-\hat{\nu}_{k-1}^{(1)}\right)\right| \leq \frac{1}{2}M\rho_N.
	\end{eqnarray}
	for all $k$ and $d$. The right side of the above inequality converges to 0 since $\rho_N\to 0$, and because all of this occurs with probability $\mathbb{P}[\mathcal{R}_N]\to 1$, this shows that the  $\hat{X}^{(2)}_k(d)$'s all jointly converge. Specifically, let $\varepsilon^*_{j,k}$ for $j=0,..,J$ and $k\in\mathbb{Z}$ be random variables such that $\{\varepsilon_{j,k}\}_{k\in\mathbb{Z}}$ are iid $\mathcal{E}_j$ variables for $j=0,\dots,J$. Next, we denote the random walks
	\begin{eqnarray}\label{eq:generalizedrandomwalk}
	X^*_{j}(d)=\begin{cases}
	d\frac{|\Delta_j|}{2}+\text{sgn}(\Delta_j)\sum\limits_{l=1}^d\varepsilon^*_{j,l}\quad &\text{for }d>0\\
	0 &\text{for }d=0\\
	-d\frac{|\Delta_j|}{2}-\text{sgn}(\Delta_j)\sum\limits_{l=d+1}^{0}\varepsilon^*_{j,l} &\text{for }d<0
	\end{cases}
	\end{eqnarray}
	(note that if the error terms are iid $N(0,\sigma^2)$, the random walk $X^*_{\Delta_k}(d)$ has precisely the same distribution as the random walk $\sigma X_{\Delta_k/\sigma}(d)$). We have the joint weak convergence 
	\begin{eqnarray}\label{eq:finitestochproc}
	&&\begin{pmatrix}
	\hat{X}^{(2)}_1(-M), & \hat{X}^{(2)}_1(-M+1), &\dots, &\hat{X}^{(2)}_1(M),\\
	\hat{X}^{(2)}_2(-M), & \hat{X}^{(2)}_2(-M+1), &\dots, & \hat{X}^{(2)}_2(M),\\
	\hdotsfor{4}, \\
	\hat{X}^{(2)}_J(-M), &\hat{X}^{(2)}_J(-M+1), &\dots,& \hat{X}^{(2)}_J(M) 
	\end{pmatrix} \quad\Rightarrow\quad\nonumber\\
	&&\begin{pmatrix}
	X_{1}^*(-M), & X_{1}^*(-M+1), &\dots, &X_{1}^*(M),\\
	X_{2}^*(-M), & X_{2}^*(-M+1), &\dots, & X_{2}^*(M),\\
	\hdotsfor{4}, \\
	X_{J}^*(-M), &X_{J}^*(-M+1), &\dots,& X_{J}^*(M)
	\end{pmatrix}
	\end{eqnarray}
	\indent Define $L^*_{j}:=\underset{k\in\mathbb{Z}}{\arg\min}X_{j}^*(k)$, and $L^{*(M)}_{j}:=\underset{|k|\leq M}{\arg\min}X_{j}(k)$ for $j=1,\dots,J$ (note that if the error terms of the data sequence is $N(0,1)$, $L^*_{\Delta_k}$ has the same distribution as $L_{\Delta_k}$). We have the joint weak convergence 
	\begin{eqnarray}
	\left(  \underset{|j|\leq M}{\arg\max}\hat{X}^{(2)}_1(d),\dots,\underset{|j|\leq M}{\arg\max}\hat{X}^{(2)}_J(d)\right) \quad\Rightarrow\quad \left(  L_{\Delta_1}^{*(M)},\dots, L^{*(M)}_{\Delta_J}\right) 
	\end{eqnarray}
	by the continuous mapping theorem, because $\arg\max$ is a continuous function on $\mathbb{R}^{2M+1}$ (except when at least two of the coordinates are equal, which has probability 0 if the error terms have densities).
	\newline
	\newline
	\indent Next, we establish that $\mathbb{P}\left[ \hat{J}=J,\, \lambda_2(\tau_k^{(2)},\hat{\tau}_k^{(2)})=j_k\text{ for }k=1,...,J \right]$ converges to the product of $\mathbb{P}[ L_{k}^*=j_k ]$ for $k=1,...,J$. We will do this by showing for any fixed $\epsilon>0$, the absolute difference between the two is smaller than $\epsilon$ for all large $N$. As in the proof of the single change point problem, this is accomplished through three main inequalities.
	\newline
	\newline
	\textbf{First Inequality}: From the result of Theorem \ref{thm:multiorder}, and by the fact that $\mathbb{P}[\mathcal{R}_N]\to1$, we can find an integer $K_0$ greater than $\max_k |j_k|$, such that for any $K_1\geq K_0$, we have for sufficiently large $N$
	\begin{eqnarray}
	\mathbb{P}\left[ \hat{J}=J,\, \max\limits_{k=1,...,J}\left| \hat{\tau}^{(2)}_k-\tau_k\right| \leq (2K_1+2),\;\mathcal{R}_N\right] \geq 1-\frac{\epsilon}{4}
	\end{eqnarray}
	For all sufficiently large $N$, $\left\{\hat{J}=J,\, \max\limits_{k=1,...,J}\left| \hat{\tau}^{(2)}_k-\tau_k\right| \leq (2K_1+2)\right\}\cap\mathcal{R}_N$ would mean $$\left\{\hat{J}=J,\,\max\limits_{k=1,..,J}\left|\lambda_2\left(\tau_k^{(2)},\hat{\tau}^{(2)}_k \right)\right|\leq K_1\right\}\cap\mathcal{R}_N$$
	and hence
	\begin{eqnarray}
	1-\frac{\epsilon}{3}&\leq& \mathbb{P}\left[\hat{J}=J,\,\min\limits_{k=1,..,J}\left|\lambda_2\left(\tau_k^{(2)},\hat{\tau}^{(2)}_k \right)\right|\leq K_1,\;\mathcal{R}_N\right]\nonumber\\
	&=& \mathbb{P}\left[\hat{J}=J;\; \max_{k=1,...,J}\left| \underset{d\in\mathbb{Z}}{\arg\min}\hat{X}^{(2)}_k(d)\right|\leq K_1\right]
	\end{eqnarray}
	Now  
	\begin{eqnarray}
	&&\underset{d\in\mathbb{Z}}{\arg\min}\hat{X}^{(2)}_k(d)=j_k \text{ for }k=1,...,J\nonumber\\
	&\longleftrightarrow& \underset{|d|\leq K_1}{\arg\min}\hat{X}^{(2)}_k(d)=j_k\text{ for } k=1,...,J,\text{ and }\max_{k=1,...,J}\left| \underset{d\in\mathbb{Z}}{\arg\min}\hat{X}^{(2)}_k(d)\right|\leq K_1
	\end{eqnarray}
	With steps very similar to those used in (\ref{eqrep1}), it can be shown that
	\begin{eqnarray}\label{firstineq}
	&&\left| \mathbb{P}\left[ \underset{d\in\mathbb{Z}}{\arg\min}\hat{X}^{(2)}_k(d)=j_k \text{ for }k=1,...,J\right] -\mathbb{P}\left[ \underset{|d|\leq K_1}{\arg\min}\hat{X}^{(2)}_k(d)=j_k \text{ for }k=1,...,J\right]\right|\nonumber\\
	&\leq & \mathbb{P}\left[   \max_{k=1,...,J}\left| \underset{d\in\mathbb{Z}}{\arg\min}\hat{X}^{(2)}_k(d)\right|> K_1 \right]\nonumber\\
	&\qquad&\leq \epsilon/4
	\end{eqnarray}
	\textbf{Second Inequality} We can find some integer $K_2>K_0$ such that
	\begin{eqnarray}
	\mathbb{P}\left[  \max_{k=1,..,J}|L_{k}^*|\leq K_2\right]\geq 1-\frac{\epsilon}{4}
	\end{eqnarray}
	Now $L_{k}^*=j_k$ for $k=1,...,J$ if and only if both $L_{k}^{*(K_2)}=j_k$ for $k=1,...,J$ and $\max_k|L_{k}|\leq K_2$. With steps very similar to those in (\ref{eqrep2}), we have
	\begin{eqnarray}\label{secondineq}
	&&\left| \mathbb{P}\left[  L_{k}^*=j_k\text{ for }k=1,...,J\right]-\mathbb{P}\left[  L^{*(K_2)}_{k}=j_k\text{ for }k=1,...,J\right]\right|\\
	&\leq & \mathbb{P}\left[  \max_{k=1,..,J}|L_{k}^*|> K_2\right]\nonumber\\
	&\leq &\epsilon/4
	\end{eqnarray}
	\textbf{Third Inequality} By weak convergence, we have
	\begin{eqnarray}\label{thirdineq}
	\left| \mathbb{P}\left[ L_{k}^{*(K_2)}=j_k\text{ for }k=1,...,J\right]-\mathbb{P}\left[ \underset{|d|\leq K_2}{\arg\min}\hat{X}^{(2)}_k(d)=j_k \text{ for }k=1,...,J\right]\right|\leq \frac{\epsilon}{4}
	\end{eqnarray}
	for all sufficiently large $N$.
	\newline
	\newline
	Combining the inequalities in (\ref{firstineq}), (\ref{secondineq}), and (\ref{thirdineq}) will give 
	\begin{eqnarray}
	&&\left| \mathbb{P}\left[  L_{k}^*=j_k\text{ for }k=1,...,J\right]-\mathbb{P}\left[ \hat{J}=J,\, \lambda_2\left( \tau_k^{(2)},\hat{\tau}_k^{(2)}\right)=j_k\text{ for }k=1,..,J,\;\mathcal{R}_N\right]\right|\nonumber\\
	&=& \left| \mathbb{P}\left[  L_{k}^*=j_k\text{ for }k=1,...,J\right]- \mathbb{P}\left[ \underset{d\in\mathbb{Z}}{\arg\min}\hat{X}^{(2)}_k(d)=j_k \text{ for }k=1,...,J\right]\right|\nonumber\\
	&\leq & \left| \mathbb{P}\left[ \underset{d\in\mathbb{Z}}{\arg\min}\hat{X}^{(2)}_k(d)=j_k \text{ for }k=1,...,J\right] -\mathbb{P}\left[ \underset{|d|\leq K_2}{\arg\min}\hat{X}^{(2)}_k(d)=j_k \text{ for }k=1,...,J\right]\right|\nonumber\\
	&+& \left| \mathbb{P}\left[ L_{k}^{*(K_2)}=j_k\text{ for }k=1,...,J\right]-\mathbb{P}\left[ \underset{|d|\leq K_2}{\arg\min}\hat{X}^{(2)}_k(d)=j_k \text{ for }k=1,...,J\right]\right|\nonumber\\
	&+& \left| \mathbb{P}\left[  L_{k}^*=j_k\text{ for }k=1,...,J\right]-\mathbb{P}\left[  L^{*(K_2)}_{k}=j_k\text{ for }k=1,...,J\right]\right|\nonumber\\
	&\qquad &\leq 3\epsilon/4
	\end{eqnarray}
	Additionally, for sufficiently large $N$ we have $\mathbb{P}[\text{not }\mathcal{R}_N]<\epsilon/4$ and hence
	\begin{eqnarray}
	&&\Bigg| \mathbb{P}\left[ \hat{J}=J,\, \lambda_2\left( \tau_k^{(2)},\hat{\tau}_k^{(2)}\right)=j_k\text{ for }k=1,..,J,\;\mathcal{R}_N\right]\nonumber\\
	&-&\mathbb{P}\left[ \hat{J}=J,\, \lambda_2\left( \tau_k^{(2)},\hat{\tau}_k^{(2)}\right)=j_k\text{ for }k=1,..,J\right] \Bigg|\nonumber\\&<&\epsilon/4
	\end{eqnarray}
	and hence
	\begin{eqnarray}
	&&\Bigg| \mathbb{P}\left[  L_{k}^*=j_k\text{ for }k=1,...,J\right]-\mathbb{P}\left[ \hat{J}=J,\, \lambda_2\left( \tau_k^{(2)},\hat{\tau}_k^{(2)}\right)=j_k\text{ for }k=1,..,J\right]\Bigg|<\epsilon
	\end{eqnarray}
	for all sufficiently large $N$.
\end{proof}

\subsection{Proof of Theorem \ref{thm:reconsistentineq}}\label{sec:reconsistentproof}
We will show that 
\begin{eqnarray}
\mathbb{P}\left[ \hat{J}=J;\,\left|\hat{\tau}^{re}_j -\tau^{**}_j\right|\leq Q_{(|\Delta_j|+2\rho_N)/\sigma}\left(1-\frac{\alpha_N}{J}\right)\;\forall j\right]\geq 1+\alpha_N+o(1)
\end{eqnarray}
\begin{proof}
	Letting $\mathcal{R}_N$ be the event that
	\begin{eqnarray}
	\left\{ \hat{J}=J;\, \max_{j=1,\dots,J}|\hat{\tau}^*_j-\tau^{**}_j|\leq w^*(N^*) ;\, \max_{j=0,\dots,J}|\hat{\nu}^{(1)}_j-\nu_j|\leq \rho_N  \right\}.
	\end{eqnarray}
	In a similar way to the proof for Theorem \ref{thm:increasingJasymprotics}, to prove
	\begin{eqnarray}
	&&\mathbb{P}\left[ \hat{J}=J;\, \max_{j=1,\dots,J}|\hat{\tau}^{re}_j-\tau^{**}_j|\leq Q_{(|\Delta_j|+2\rho_N)/\sigma}\left(1-\frac{\alpha_N}{J}\right)\text{ for all }j=1\dots,J \right]\nonumber\\
	&\geq&1-\alpha_N+o(1)
	\end{eqnarray}
	it is sufficient to demonstrate that 
	\begin{eqnarray}
	&&\mathbb{P}\left[ \max_{j=1,\dots,J}|\hat{\tau}^{re}_j-\tau^{**}_j|\leq Q_{(|\Delta_j|+2\rho_N)/\sigma}\left(1-\frac{\alpha_N}{J}\right)\; \forall j\Big| J=\hat{J}, \hat{\tau}^*_j=t_j,\hat{\nu}_j^{(1)}=v_j\;\forall j \right]\nonumber\\
	&=& 1-P_{(|\Delta_j|+2\rho_N)/\sigma}\left(1-\frac{\alpha_N}{J}\right)+o(1)
	\end{eqnarray}
	for all $t_j$'s and $v_j$'s permissible within $\mathcal{R}_N$, which we will assume when we write $t_j$'s and $v_j$'s from here on. 
	\newline
	\newline
	\indent We now try to bound the difference between 
	$$ \mathbb{P}\left[|\hat{\tau}^{re}_j-\tau^{**}_j|\leq Q_{(|\Delta_j|+2\rho_N)/\sigma}\left(1-\frac{\alpha_N}{J}\right)\Big| J=\hat{J}, \hat{\tau}^*_j=t_j,\hat{\nu}_j^{(1)}=v_j\text{ for all }j \right] $$
	and $ P_{(|\Delta_j|+2\rho_N)/\sigma}\left(1-\frac{\alpha_N}{J}\right)$ for all $j$. Each estimator equals the argmin of a random walk:
	\begin{eqnarray}
	\hat{\tau}^{re}_j&:=&\underset{t\in [t_j\pm \hat{d}_j ]}{\arg\min}(\text{sgn}(v_j-v_{j-1}))\sum_{i\in  [t_j\pm \hat{d}_j ]}\left( Y_i-\frac{v_j+v_{j-1}}{2} \right)\left[ 1(i\leq t)-1(i\leq \tau^{**}_j ) \right]\nonumber\\
	&=:&\underset{t\in [t_j\pm \hat{d}_j ]}{\arg\min}\hat{X}_j(t).
	\end{eqnarray}
	In comparison with the random walk
	\begin{eqnarray}
	X_j'(t):=\begin{cases}
	t\left(\frac{|\Delta_j|}{2}+\text{sgn}(\Delta_j)\hat{D}_j\right)+\text{sgn}(\Delta_j)\sum_{_i=1}^t\varepsilon_{\tau^{**}_j+i}\quad &t>0\\
	0 & t=0\\
	|t|\left(\frac{|\Delta_j|}{2}-\text{sgn}(\Delta_j)\hat{D}_j\right)-\text{sgn}(\Delta_j)\sum_{_i=1}^t\varepsilon_{\tau^{**}_j-i+1}
	\end{cases}
	\end{eqnarray}
	where $\hat{D}_j=\frac{v_j-\nu_j+v_{j-1}-\nu_{j-1}}{2}$. We have, for all sufficiently large $N$, $\hat{X}_j(t+\tau^{**}_j)$ equaling either $X_j'(t)-2|t|\hat{D}_j1(t>0)$ or $X'(t)-2|t|\hat{D}_j1(t<0)$ for all integers $t\in [t_j\pm \hat{d}_j]-\tau^{**}_j$. With regards to the set $[t_j\pm \hat{d}_j]-\tau^{**}_j$, for all large $N$ (such that $3w^*(N^*)\leq \delta^*_{N^*}/2$) this set contains the interval $[-\delta^*_{N^*}/2,\delta^*_{N^*}/2]$. This can be seen by the series of inequalities
	\begin{eqnarray}
	&&t_j-\hat{d}_j\leq \tau^{**}_j+w^*(N^*)-(\delta^*_{N^*}-2w^*(N^*))\leq \tau^{**}_j-\delta^*_{N^*}/2\nonumber\\
	&<&\tau^{**}_j-\delta^*_{N^*}/2\leq \tau^{**}_j-w^*(N^*)+(\delta^*_{N^*}-2w^*(N^*))\leq t_j+\hat{d}_j.
	\end{eqnarray}
	Furthermore, since $\log(J/\alpha_N)=o(\delta^*_{N^*})$, we can use Lemma \ref{lem:probbound} to obtain, for all large $N$,
	\begin{eqnarray}
	&&\tau^{**}_j-\delta^*_{N^*}/2\leq \tau^{**}_j-\frac{1}{B}\log\left( \frac{AJ}{\alpha_N} \right)\leq\tau^{**}_j-Q_{(|\Delta_j|+2\rho_N)/\sigma}\left(1-\frac{\alpha_N}{J}\right)\nonumber\\
	&<&\tau^{**}_j+Q_{(|\Delta_j|+2\rho_N)/\sigma}\left(1-\frac{\alpha_N}{J}\right)\leq \tau^{**}_j+\frac{1}{B}\log\left( \frac{AJ}{\alpha_N} \right)\leq \tau^{**}_j+\delta^*_{N^*}/2
	\end{eqnarray}
	where $A$ and $B$ are some constants. Therefore, we can apply Lemma \ref{lem:randomwalkcompoppo} to obtain the inequality
	\begin{eqnarray}
	&&\mathbb{P}\left[|\hat{\tau}^{re}_j-\tau^{**}_j|\leq Q_{(|\Delta_j|+2\rho_N)/\sigma}\left(1-\frac{\alpha_N}{J}\right)\Big| J=\hat{J}, \hat{\tau}^*_j=t_j,\hat{\nu}_j^{(1)}=v_j\text{ for all }j \right]\nonumber\\
	&=&\mathbb{P}\left[\left|\underset{t\in [t_j\pm \hat{d}_j]-\tau^{**}_j}{\arg\min}\hat{X}_j(t+\tau^{**}_j)\right|\leq Q_{(|\Delta_j|+2\rho_N)/\sigma}\left(1-\frac{\alpha_N}{J}\right)\Bigg| J=\hat{J}, \hat{\tau}^*_j=t_j,\hat{\nu}_j^{(1)}=v_j\text{ for all }j \right]\nonumber
	\end{eqnarray}
	\begin{eqnarray}
	&\geq &\mathbb{P}\left[\left|\underset{t\in [t_j\pm \hat{d}_j]-\tau^{**}_j}{\arg\min}\frac{1}{\sigma}X_j'(t)\right|\leq Q_{(|\Delta_j|+2\rho_N)/\sigma}\left(1-\frac{\alpha_N}{J}\right)\right]- C_1(|\underline{\Delta}|/2)|\hat{D}_j| \nonumber\\
	&\geq & \mathbb{P}\left[\left|\underset{t\in [t_j\pm \hat{d}_j]-\tau^{**}_j}{\arg\min}\frac{1}{\sigma}X_j'(t)\right|\leq Q_{(|\Delta_j|+2\rho_N)/\sigma}\left(1-\frac{\alpha_N}{J}\right)\right]- C_1(|\underline{\Delta}|/2)\rho_N
	\end{eqnarray}
	for all large $N$, where $C_1(|\underline{\Delta}|/2)$ is some expression independent of $N$ and $j$.
	\newline
	\newline
	Next, use Lemma \ref{lem:twoXrandwalkcomp} to obtain
	\begin{eqnarray}
	&&\mathbb{P}\left[\left|\underset{t\in [t_j\pm \hat{d}_j]-\tau^{**}_j}{\arg\min}\frac{1}{\sigma}X_j'(t)\right|\leq Q_{(|\Delta_j|+2\rho_N)/\sigma}\left(1-\frac{\alpha_N}{J}\right)\right]\nonumber\\
	&\geq& \mathbb{P}\left[\left|\underset{t\in [t_j\pm \hat{d}_j]-\tau^{**}_j}{\arg\min}X_{(|\Delta_j|+2\rho_N)/\sigma}(t)\right|\leq Q_{(|\Delta_j|+2\rho_N)/\sigma}\left(1-\frac{\alpha_N}{J}\right)\right]
	\end{eqnarray}
	We could also use Lemma \ref{lem:eqprob} twice to obtain
	\begin{eqnarray}
	&&\mathbb{P}\left[\left|\underset{t\in [t_j\pm \hat{d}_j]-\tau^{**}_j}{\arg\min}X_{(|\Delta_j|+2\rho_N)/\sigma}(t)\right|\leq Q_{(|\Delta_j|+2\rho_N)/\sigma}\left(1-\frac{\alpha_N}{J}\right)\right]\nonumber\\
	&\geq&P_{(|\Delta_j|+2\rho_N)/\sigma}\left(1-\frac{\alpha_N}{J}\right)-C_2(\underline{\Delta})\exp(-C_3(\underline{\Delta})\Delta^*_{N^*})
	\end{eqnarray}
	for some positive expressions $C_2(\underline{\Delta})$ and $C_3(\underline{\Delta})$ independent of $N$ and $j$. By combining these inequalities, we come to the conclusion that
	\begin{eqnarray}
	&&\mathbb{P}\left[|\hat{\tau}^{re}_j-\tau^{**}_j|\leq Q_{(|\Delta_j|+2\rho_N)/\sigma}\left(1-\frac{\alpha_N}{J}\right)\Big| J=\hat{J}, \hat{\tau}^*_j=t_j,\hat{\nu}_j^{(1)}=v_j\text{ for all }j \right]\nonumber\\
	&\geq &P_{(|\Delta_j|+2\rho_N)/\sigma}\left(1-\frac{\alpha_N}{J}\right)-C_1(|\underline{\Delta}|/2)\rho_N-C_2(\underline{\Delta})\exp(-C_3(\underline{\Delta})\delta^*_{N^*})
	\end{eqnarray}
	Therefore
	\begin{eqnarray}
	&&\mathbb{P}\left[ \max_{j=1,\dots,J}|\hat{\tau}^{re}_j-\tau^{**}_j|\leq Q_{(|\Delta_j|+2\rho_N)/\sigma}\left(1-\frac{\alpha_N}{J}\right)\; \forall j\Big| J=\hat{J}, \hat{\tau}^*_j=t_j,\hat{\nu}_j^{(1)}=v_j\;\forall j \right]\nonumber\\
	&=&\mathbb{P}\left[ \underset{j=1,\dots,J}{\bigcap}|\hat{\tau}^{re}_j-\tau^{**}_j|\leq Q_{(|\Delta_j|+2\rho_N)/\sigma}\left(1-\frac{\alpha_N}{J}\right)\; \forall j\Big| J=\hat{J}, \hat{\tau}^*_j=t_j,\hat{\nu}_j^{(1)}=v_j\;\forall j \right]\nonumber\\
	&\geq &1-\sum_{j=1}^J\mathbb{P}\left[|\hat{\tau}^{re}_j-\tau^{**}_j|> Q_{(|\Delta_j|+2\rho_N)/\sigma}\left(1-\frac{\alpha_N}{J}\right)\Big| J=\hat{J}, \hat{\tau}^*_j=t_j,\hat{\nu}_j^{(1)}=v_j\text{ for all }j \right]\nonumber\\
	&\geq &1-J\left[1-P_{(|\Delta_j|+2\rho_N)/\sigma}\left(1-\frac{\alpha_N}{J}\right)+\left(C_1(|\underline{\Delta}|/2)\rho_N+C_2(\underline{\Delta})\exp(-C_3(\underline{\Delta})\delta^*_{N^*})\right)\right]\nonumber\\
	&\qquad &\geq 1-\alpha_N-J\left(C_1(|\underline{\Delta}|/2)\rho_N+C_2(\underline{\Delta})\exp(-C_3(\underline{\Delta})\delta^*_{N^*})\right)
	\end{eqnarray}
	which converges to 1 since $\alpha_N\to 0$, $J\rho_N\to 0$, and $\delta^*_{N^*}>C(N^*)^\eta$ for some positive constant $C$ and $\eta$. 
\end{proof}

\subsection{Proof of Lemma \ref{lem:binsegsignalconsistent}}\label{sec:binsegsignalconsistentproof}
In order to prove Lemma \ref{lem:binsegsignalconsistent}, we will rely on the following result:
\begin{lemma} 
	\label{cor:signalconsistent}
	(i) Suppose that 
	we have an estimation scheme which when applied onto $Z_1,\dots, Z_{N^*}$ gives estimates $\hat{J}$ for $J$ and $\left( \hat{\tau}^*_1,\cdots,\hat{\tau}^*_{\hat{J}}\right)$ for $(\tau^*_1,\cdots,\tau^*_J)$ such that
	\begin{eqnarray}
	\mathbb{P}\left[ \hat{J}=J;\quad \max_{i=1,\cdots,J}|\hat{\tau}^*_j-\tau^*_j|\leq w^*(N^*)  \right]\geq 1-B^*_{N^*}
	\end{eqnarray} 
	for some sequences $w^*(N^*)$ and $B_{N^*}$, with $w^*(N^*)=o(\delta^*_{N^*})$ and $B_{N^*}\to 0$. Then, for any positive sequence $\{\rho^*_{N^*}\}$ such that $\frac{w^*(N^*)}{\delta^*_{N^*}}=o(\rho^*_{N^*})$, there exist constants $C_1$ and $C_2$, where
	\begin{eqnarray}\label{pbound}
	\mathbb{P}\left[ \hat{J}=J;\quad |\hat{\nu}^*_i-\nu^*_i|\geq \rho^*_{N^*} \right]\leq B_{N^*}+C_1w^*(N^*) \frac{\exp\left[ -C_2\delta^*_{N^*}\rho^{* \,2}_{N^*}\right]}{\sqrt{\delta^*_{N^*}}\rho^*_{N^*}}
	\end{eqnarray}
	for all $i=1,\cdots,J$, when $N^*$ is sufficiently large.
	\newline
	(ii) 	Moreover, as a consequence of part (i),
	\begin{eqnarray}\label{eq:sigconsistent}
	&&\mathbb{P}\left[ \hat{J}=J;\quad \max_{i=0,\cdots,J}|\hat{\nu}^*_i-\nu_i^*|<\rho^*_{N^*}\right] \geq 1 - \left(\frac{N^*}{\delta^*_{N^*}} + 2\right)B_{N^*} 
	\nonumber\\ &&-C_1 \left(\frac{N^*}{\delta^*_{N^*}} + 1 \right)\,w^*(N^*)\frac{\exp[-C_2\delta^*_{N^*} \rho^{*\,2}_{N^*}]}{\sqrt{\delta^*_{N^*}}\rho^*_{N^*}} \,.
	\end{eqnarray}
	It follows that in addition to the conditions in (i), if, furthermore, $N^* B_{N^*}/\delta^*_{N^*}\to 0$ and $(N^* w^*(N^*)/\delta^{*\,3/2}_{N^*}\rho^*_{N^*}) =$ $o(\exp[C_2\delta^*_{N^*} \rho^{*\,2}_{N^*}])$, then the probability in (\ref{eq:sigconsistent}) goes to 1. The $\hat{\nu}^*_i$'s are simultaneously consistent if $\rho^*_{N^*}$ also converges to 0.
\end{lemma}
\begin{proof}
	See Section \ref{sec:signalconsistentproof}
\end{proof}
In order to prove Lemma \ref{lem:binsegsignalconsistent}, it is sufficient to find a sequence $\rho_{N^*}^*\to 0$ such that $\rho_{N^*}^*$ satisfies all the conditions of Lemma \ref{cor:signalconsistent} and  $J\rho_{N^*}^*\to 0$ . The proof will proceed in such a fashion.

\begin{proof}
	\indent We start off by defining some notations. Since $\delta_N \geq CN^{1-\Xi}$ by (M3), we must have $J\leq C'N^\Lambda$ for some $C'>0$ and $\Lambda\in [0,\Xi]$. Using this notation, we will show that by setting $\rho^*_{N^*}=(N^*)^\theta$ where $\theta$ is chosen to be any value in $\Big( (3\Xi/\gamma-1)\vee (-3/8),-\frac{\Lambda}{\gamma} \Big)$, we will have a $\rho^*_{N^*}\to 0$ which satisfies the conditions of Lemma \ref{cor:signalconsistent} and $J\rho^*_{N^*}\to 0$. 
	\newline
	\newline
	\indent We must verify that $\Big( (3\Xi/\gamma-1)\vee (-3/8),-\frac{\Lambda}{\gamma} \Big)$ is a nonempty set by showing that both $(3\Xi/\gamma-1)$ and $-3/8$ are strictly smaller than $-\frac{\Lambda}{\gamma}$. First, due to condition (M7 (BinSeg)) we know that $\Xi/\gamma<1/4$, and therefore
	\begin{eqnarray}
	&&\frac{3\Xi}{\gamma}+\frac{\Lambda}{\gamma} \leq \frac{4\Xi}{\gamma}<1\nonumber\\
	&\rightarrow&  \frac{3\Xi}{\gamma}-1<-\frac{\Lambda}{\gamma},
	\end{eqnarray}
	and additionally, 
	\begin{eqnarray}
	-\frac{3}{8}<-\frac{1}{7}\leq -\frac{\Xi}{\gamma}\leq -\frac{\Lambda}{\gamma}.
	\end{eqnarray}
	Therefore, it is possible to choose some value of $\theta$ within the set $\Big( (3\Xi/\gamma-1)\vee (-3/8),-\frac{\Lambda}{\gamma} \Big)$.
	\newline
	\newline
	\indent To verify that $J\rho_{N^*}^*\to 0$, we first note 
	$$J\rho_{N^*}^*\lesssim N^\Lambda (N^*)^\theta\lesssim N^{\Lambda+\gamma\theta}.$$
	The rightmost term goes to 0 because $\theta<\Lambda/\gamma$.
	\newline
	\newline
	\indent To show that the BinSeg estimators satisfy the conditions of Lemma \ref{cor:signalconsistent}, we proceed as follows. First note that $\delta^*_{N^*}\geq C_1(N^*)^{1-\Xi/\gamma}$ for some $C_1>0$, and where $\Xi/\gamma < 1/4$.  For some positive constant $C_2$, set $w^*(N^*)= C_2 E_{N^*} = C_2\left(\frac{N^*}{\delta^*_{N^*}}\right)^2\log(N^*)$. Then, there is a positive constant $C_4$ 
	such that $w^*(N^*) = (C_4+o(1))(N^*)^{2\Xi/\gamma}\log(N^*)$. Set $B_{N^*}=C_5 /N^*$; this would mean 
	\begin{itemize}
		\item since $(N^*)^{2\Xi/\gamma}\log(N^*)=o((N^*)^{1-\Xi/\gamma})$ this does allow $w^*(N^*)=o(\delta^*_{N^*})$ to be satisfied;  
		\item $\frac{N^*}{\delta^*_{N^*}}B_{N^*}= \frac{C_5}{\delta^*_{N^*}} \to 0$; 
		\item since $3\Xi/\gamma-1<0$, and because $\rho^*_{N^*}=(N^*)^\theta$ for some $\theta$ satisfying $(3\Xi/\gamma-1)\vee (-3/8)<\theta<0$, this $\rho^*_{N^*}\to 0$ and satisfies:
		\begin{itemize}
			\item  $\frac{w^*(N^*)}{\delta^*_{N^*}}\leq C_6(N^*)^{3\Xi/\gamma-1}\log(N^*)$ for some $C_6>0$; latter expression is $o(\rho^*_{N^*})$; 
			\item $N^* w^*(N^*)/(\delta^{*\,3/2}_{N^*}\rho^*_{N^*}) \leq C_7(N^*)^{((7\Xi/\gamma - 1)/2)\,- \theta}$ and $\exp[C_2\delta^*_{N^*} (\rho^*_{N^*})^2] \geq \exp[C_8 (N^*)^{1-\Xi/\gamma + 2\theta}]$, for some positive $C_7$, $C_8$; as $1-\Xi/\gamma + 2\theta>3/4+2\theta > 0$ it follows that $N^* w^*(N^*)/\delta^{*\,3/2}_{N^*}\rho^*_{N^*} =o(\exp[C_2\delta^*_{N^*} (\rho^*_{N^*})^2])$. 
		\end{itemize}
	\end{itemize}
	Therefore, all conditions of Lemma \ref{cor:signalconsistent} for a sequence $\rho^*_{N^*}$ tending to 0 are satisfied. Next, combining the results of Theorem \ref{frythm} and Lemma  \ref{cor:signalconsistent}, we establish the simultaneous consistency of $\hat{J}$, the $\hat{\tau}_i$'s, and the $\hat{\nu}_i$'s. Specifically,  under conditions (M1) to (M6 (BinSeg)), we could combine the two limit results
	\begin{eqnarray}
	&&\mathbb{P}\left[\hat{J}=J;\quad \max_{i=1,...,J}|\hat{\tau}^*_i-\tau^*_i|\leq C E_{N^*}\right]\to 1\nonumber\\
	&&\mathbb{P}\left[\hat{J}=J;\quad \max_{i=0,...,J}|\hat{\nu}^*_i-\nu^*_i|\leq \rho^*_{N^*}\right]\to 1\nonumber\\
	&&\text{for any }\rho^*_{N^*}=(N^*)^\theta\text{, where }\theta\in\left( (3\Xi/\gamma-1)\vee\left( -\frac{3}{8} \right) \right)
	\end{eqnarray}
	to get the following through the Bonferroni inequality: 
	\begin{eqnarray}\label{eq:binseg1stconsistentproof}
	\mathbb{P}\left[\hat{J}=J;\quad \max_{i=1,...,J}|\hat{\tau}^*_i-\tau^*_i|\leq C E_{N^*};\quad \max_{i=0,...,J}|\hat{\nu}^*_i-\nu^*_i|\leq \rho^*_{N^*}\right]\to 1
	\end{eqnarray}
	as $N^*\to\infty$. 
\end{proof}

\subsection{Proof For Lemma \ref{cor:signalconsistent}}\label{sec:signalconsistentproof}
\begin{proof} First we focus on part (i). Consider the separate cases of $\tau_i$ for $i=0$ or $i=J$ (case 1), and $1\leq i<J$ (case 2). We know
	\begin{equation}
	\mathbb{P}\Big[ \hat{J}=J;\quad \max_{j=1,...,J} |\hat{\tau}_j^*-\tau_j^*|\leq w^*(N^*) \Big]\geq 1-B_{N^*}
	\end{equation}
	for some sequences $w^*(N^*)$ and $B_{N^*}$ which are  $o(\delta_{N^*}^*)$, and $o(1)$, respectively. Also as in the statement of the theorem, assume there is some sequence $\rho^*_{N^*}$ such that $w^*({N^*})/\delta_{N^*}^*=o(\rho_{N^*})$. For the rest of the proof assume ${N^*}$ is large enough so that 
	\begin{itemize}
		\item $1<w^*({N^*})< \frac{\delta^*_{N^*}}{6}$
		\item$6\frac{w^*({N^*})}{\delta^*_{N^*}}\bar{\theta}<\frac{\rho^*_{N^*}}{2}$
		\item $\delta_{N^*}^*>3$
	\end{itemize}
	\noindent \textbf{Case 1}: For $i=0$, $\hat{\nu}^*_0$ is the average of all $Z_t$'s where $t$ lies between 1 and $\hat{\tau}^*_1$, inclusive. We have 
	\begin{eqnarray}
	&&\mathbb{P}\left[ \hat{J}=J;\quad |\hat{\nu}^*_0-\nu^*_0|\geq \rho^*_{N^*} \right]\nonumber\\
	&\leq&\mathbb{P}\left[ \hat{J}=J;\quad |\hat{\tau}^*_1-\tau^*_1|>w^*({N^*})\right]+\mathbb{P}\left[ \hat{J}=J;\quad |\hat{\nu}^*_0-\nu^*_0|\geq \rho^*_{N^*};\quad  |\hat{\tau}^*_1-\tau^*_1|\leq w^*({N^*}) \right] \nonumber\\
	&\leq& B_{N^*}+\sum_{\tau:\,|\tau-\tau^*_1|\leq w^*({N^*})}\mathbb{P}\left[  \hat{J}=J;\quad |\hat{\nu}^*_0-\nu^*_0|\geq \rho^*_{N^*};\quad  \hat{\tau}^*_1=\tau  \right]\nonumber\\
	&\leq&  B_{N^*}+\sum_{\tau:\,|\tau-\tau^*_1|\leq w^*({N^*})}\mathbb{P}\left[ \hat{J}=J;\quad\hat{\tau}^*_1=\tau;\quad \left|\frac{1}{\tau}\sum_{j=1}^\tau Z_j-\nu^*_0\right|\geq \rho^*_{N^*} \right]\nonumber\\
	&\qquad&\leq B_{N^*}+\sum_{\tau:\,|\tau-\tau^*_1|\leq w^*({N^*})}\mathbb{P}\left[\left|\frac{1}{\tau}\sum_{j=1}^\tau (Z_j-\nu^*_0)\right|\geq \rho^*_{N^*} \right]
	\end{eqnarray}
	For all $\tau_1-w^*({N^*})\leq\tau\leq\tau^*_1$, we have $\frac{1}{\tau}\sum_{j=1}^\tau (Z_j-\nu^*_0)\sim N(0,\sigma^2/\tau)$, and hence
	\begin{eqnarray}\label{lowerpb1}
	\mathbb{P}\left[\left|\frac{1}{\tau}\sum_{j=1}^\tau (Z_j-\tau_0)\right|\geq \rho^*_{N^*} \right]
	&=& 2\left(1-\Phi(\sqrt{\tau}\rho^*_{N^*})\right)\nonumber\\
	&\leq& 2\frac{\phi(\sqrt{\tau}\rho^*_{N^*}/\sigma)}{\sqrt{\tau}\rho^*_{N^*}/\sigma}\nonumber\\
	&\leq& \sqrt{\frac{2}{\pi}}\cdot\frac{\exp\Big( -(\tau_1-w^*({N^*}))(\rho^*_{N^*})^2/(2\sigma^2)\Big)}{\sqrt{\tau_1-w^*({N^*})}\rho^*_{N^*}/\sigma}\nonumber\\
	&\leq & \frac{2\sigma}{\sqrt{\pi}}\cdot\frac{\exp\Big( -\delta^*_{N^*}(\rho^*_{N^*})^2/(4\sigma^2)\Big)}{\sqrt{\delta^*_{N^*}}\rho^*_{N^*}}\nonumber\\
	&&\left(\text{by }\tau_1w^*({N^*})>\delta^*_{N^*}/2\right)
	\end{eqnarray}
	For all $\tau^*_1<\tau\leq \tau^*_1+w^*({N^*})$ we have $\frac{1}{\tau}\sum_{j=1}^\tau (Z_j-\nu^*_0)\sim N\left(\frac{\tau-\tau^*_1}{\tau}(\nu^*_1-\nu^*_0),\frac{\sigma^2}{\tau}\right)$. Because
	\begin{eqnarray}
	\left|\frac{\tau-\tau^*_1}{\tau}(\nu^*_1-\nu^*_0)\right|\leq \frac{w^*({N^*})}{\delta^*_{N^*}}(2\bar{\theta})\leq \frac{\rho^*_{N^*}}{2},
	\end{eqnarray}
	the magnitude of the z-score of both $\pm\rho^*_{N^*}$ for the $N\left(\frac{\tau-\tau^*_1}{\tau}(\nu^*_1-\nu^*_0),\frac{\sigma^2}{\tau}\right)$ distribution is at least $\frac{\rho^*_{N^*}\sqrt{\tau}}{2\sigma}$, and hence
	\begin{eqnarray}
	\mathbb{P}\left[\left|\frac{1}{\tau}\sum_{j=1}^\tau (Z_j-\tau^*_0)\right|\geq \rho^*_{N^*} \right]&\leq&  2\left(1-\Phi\left(\frac{\rho^*_{N^*}\sqrt{\tau}}{2}\right)\right)\nonumber\\
	&\leq& 2\frac{\phi\left(\frac{\rho^*_{N^*}\sqrt{\tau}}{2\sigma}\right)}{\frac{\rho^*_{N^*}}{2\sigma}\sqrt{\tau}/\sigma}\nonumber\\
	&\leq& \frac{2\sigma\sqrt{2}}{\sqrt{\pi}}\cdot\frac{\exp\left(-\delta^*_{N^*}(\rho^*_{N^*})^2/(8\sigma^2)\right)}{\sqrt{\delta^*_{N^*}}\rho^*_{N^*} }
	\end{eqnarray}
	Therefore, the expression (\ref{lowerpb1}) can be bounded from above by
	\begin{align}\label{bb1}
	& B_{N^*}+(w^*({N^*})+1)\frac{2\sigma}{\sqrt{\pi}}\cdot\frac{\exp\Big( -\delta^*_{N^*}(\rho^*_{N^*})^2/(4\sigma^2)\Big)}{\sqrt{\delta^*_{N^*}}\gamma_{N^*}}+w^*({N^*})\frac{2\sigma\sqrt{2}}{\sqrt{\pi}}\cdot\frac{\exp\left(-\delta^*_{N^*}(\rho^*_{N^*})^2/(8\sigma^2)\right)}{\sqrt{\delta^*_{N^*}}\rho^*_{N^*}} \nonumber\\
	&\qquad \qquad\leq  B_{N^*}+\frac{6\sigma\sqrt{2}}{\sqrt{\pi}}\cdot\frac{w^*({N^*})\exp\left(-\delta^*_{N^*}(\rho^*_{N^*})^2/(8\sigma^2)\right)}{\sqrt{\delta^*_{N^*}}\rho^*_{N^*}}.\nonumber\\
	\end{align}
	For $i=J$, a very similar argument will bound $\mathbb{P}\left[ \hat{J}=J;\quad |\hat{\nu}^*_J-\nu^*_J|\geq \rho^*_{N^*} \right]$ by the same expression in (\ref{bb1}).
	\newline
	\newline
	\noindent \textbf{Case 2}: The procedure for this case will be similar to the steps for Case 1, but there are a few modifications. For $0<i<J$, $\hat{\nu}^*_i$ is the average of all $Z_t$'s for $\hat{\tau}^*_i<t\leq \hat{\tau}^*_{i+1}$. For the following part we re-write this average by considering the midpoint $\tau_i^{*(m)}:=\frac{\lceil \tau^*_i+\tau^*_{i+1}\rceil}{2}$ where $1<i<J$.\newline
	\newline
	\indent In the case where $\hat{\tau}^*_i$ and $\hat{\tau}^*_{i+1}$ are within $\delta^*_{N^*}/3$ (which is less than $|\tau^*_{i+1}-\tau^*_i|/3$) of $\tau^*_i$ and $\tau^*_{i+1}$ respectively, we have $\hat{\tau}^*_i < \tau_i^{*(m)} <\hat{\tau}^*_{i+1}$, and hence we can bound $|\hat{\nu}^*_i-\nu^*_i|$ by
	\begin{align}
	&\left|\frac{1}{\hat{\tau}^*_{i+1}-\hat{\tau}^*_i}\sum_{j=\hat{\tau}^*_i+1}^{\hat{\tau}^*_{i+1}}(Z_j-\nu^*_i)\right|\nonumber\\
	&= \left|\frac{\tau_i^{*(m)}-\hat{\tau}^*_{i}}{\hat{\tau}^*_{i+1}-\hat{\tau}^*_i}\left( \frac{1}{\tau_i^{*(m)}-\hat{\tau}^*_{i}}\sum_{j= \hat{\tau}^*_i+1}^{\tau_i^{*(m)}}(Z_j-\nu^*_i)\right)+\frac{\hat{\tau}^*_{i+1}-\tau_i^{*(m)}}{\hat{\tau}^*_{i+1}-\hat{\tau}^*_i}\left( \frac{1}{\hat{\tau}^*_{i+1}-\tau_i^{*(m)}}\sum_{j= \tau_i^{*(m)}+1}^{\hat{\tau}^*_{i+1}}(Z_j-\nu^*_i)\right)\right|
	\nonumber\\
	&\leq\frac{\tau_i^{*(m)}-\hat{\tau}^*_{i}}{\hat{\tau}^*_{i+1}-\hat{\tau}^*_i}\left| \frac{1}{\tau_i^{*(m)}-\hat{\tau}^*_{i}}\sum_{j= \hat{\tau}^*_i+1}^{\tau_i^{*(m)}}(Z_j-\nu^*_i)\right|+\frac{\hat{\tau}^*_{i+1}-\tau_i^{*(m)}}{\hat{\tau}^*_{i+1}-\hat{\tau}^*_i}\left| \frac{1}{\hat{\tau}^*_{i+1}-\tau_i^{(m)}}\sum_{j= \tau_i^{(m)}+1}^{\hat{\tau}^*_{i+1}}(Z_j-\nu^*_i)\right|\nonumber\\
	\end{align}
	In order for $|\hat{\nu}^*_i-\nu^*_i|$ to exceed $\rho^*_{N^*}$, at least one of $\left| \frac{1}{\tau_i^{*(m)}-\hat{\tau}^*_{i}}\sum_{j= \hat{\tau}^*_i+1}^{\tau_i^{*(m)}}(Z_j-\nu^*_i)\right|$ or \newline
	$\left| \frac{1}{\hat{\tau}^*_{i+1}-\tau_i^{*(m)}}\sum_{j= \tau_i^{*(m)}+1}^{\hat{\tau}^*_{i+1}}(Z_j-\nu^*_i)\right|$ must exceed $\rho^*_{N^*}$, or in other words,
	
	\begin{align}\label{bbp2}
	&\quad\mathbb{P}\left[ \hat{J}=J;\quad |\hat{\nu}^*_i-\nu^*_i|\geq \rho^*_{N^*}\right]\nonumber\\
	&\leq  \mathbb{P}\left[\hat{J}=J;\quad |\hat{\tau}^*_i-\tau^*_i|>w^*({N^*})\quad\text{ or }\quad|\hat{\tau}^*_{i+1}-\tau^*_{i+1}|>w(N)\right]+\nonumber\\
	&\mathbb{P}\left[ \hat{J}=J;\quad |\hat{\tau}^*_i-\tau^*_i|\leq w^*({N^*});\quad |\hat{\tau}^*_{i+1}-\tau^*_{i+1}|\leq w^*({N^*});\quad |\hat{\nu}^*_i-\nu^*_i|\geq\rho^*_{N^*}\right]\nonumber\\
	&\leq  B_{N^*}+\mathbb{P}\left[  \hat{J}=J;\quad |\hat{\tau}^*_i-\tau^*_i|\leq w^*({N^*});\quad |\hat{\tau}^*_{i+1}-\tau^*_{i+1}|\leq w^*({N^*});\quad \left| \frac{1}{\tau_i^{*(m)}-\hat{\tau}^*_{i}}\sum_{j= \hat{\tau}^*_i+1}^{\tau_i^{*(m)}}(Z_j-\nu^*_i)\right|\geq \rho^*_{N^*}\right]\nonumber\\
	&+\mathbb{P}\left[  \hat{J}=J;\quad |\hat{\tau}^*_i-\tau^*_i|\leq w^*({N^*});\quad |\hat{\tau}^*_{i+1}-\tau^*_{i+1}|\leq w^*({N^*});\quad \left| \frac{1}{\hat{\tau}^*_{i+1}-\tau_i^{*(m)}}\sum_{j= \tau_i^{*(m)}+1}^{\hat{\tau}^*_{i+1}}(Z_j-\nu^*_i)\right|\geq\rho^*_{N^*}\right]\nonumber\\
	&\leq  B_{N^*}+\mathbb{P}\left[  \hat{J}=J;\quad |\hat{\tau}^*_i-\tau^*_i|\leq w^*({N^*});\quad \left| \frac{1}{\tau_i^{*(m)}-\hat{\tau}^*_{i}}\sum_{j= \hat{\tau}^*_i+1}^{\tau_i^{*(m)}}(Z_j-\nu^*_i)\right|\geq \rho^*_{N^*}\right]\nonumber\\
	&+\mathbb{P}\left[  \hat{J}=J;\quad |\hat{\tau}^*_{i+1}-\tau^*_{i+1}|\leq w^*({N^*});\quad \left| \frac{1}{\hat{\tau}^*_{i+1}-\tau_i^{*(m)}}\sum_{j= \tau_i^{*(m)}+1}^{\hat{\tau}^*_{i+1}}(Z_j-\nu^*_i)\right|\geq\rho^*_{N^*}\right]\nonumber\\
	&\leq  B_{N^*}+\sum_{\tau:\, |\tau-\tau^*_i|\leq w^*({N^*})}\mathbb{P}\left[ \left| \frac{1}{\tau_i^{*(m)}-\tau}\sum_{j= \tau+1}^{\tau_i^{*(m)}}(Z_j-\nu^*_i)\right|\geq \rho^*_{N^*}\right]\nonumber\\
	&+\sum_{\tau:\,|\tau-\tau^*_{i+1}|\leq w^*({N^*})}\mathbb{P}\left[ \left| \frac{1}{\tau-\tau_i^{*(m)}}\sum_{j= \tau_i^{*(m)}+1}^{\tau}(Z_j-\nu^*_i)\right|\geq\rho^*_{N^*}\right]\nonumber\\
	\end{align}
	Next, we will bound $\mathbb{P}\left[ \left| \frac{1}{\tau_i^{*(m)}-\tau}\sum_{j= \tau+1}^{\tau_i^{*(m)}}(Z_j-\nu^*_i)\right|\geq \rho^*_{N^*}\right]$ for each $\tau$ such that $|\tau-\tau^*_i|\leq w^*({N^*})$. For $\tau^*_i\leq \tau\leq \tau^*_i+w^*({N^*})$ we have $\frac{1}{\tau_i^{*(m)}-\tau}\sum_{j= \tau+1}^{\tau_i^{*(m)}}(Z_j-\nu^*_i)\sim N\left(0,\frac{\sigma^2}{\tau_i^{*(m)}-\tau}\right)$, and hence
	\begin{eqnarray}\label{bbbp1}
	\mathbb{P}\left[ \left| \frac{1}{\tau_i^{*(m)}-\tau}\sum_{j= \tau+1}^{\tau_i^{*(m)}}(Z_j-\nu^*_i)\right|\geq \rho^*_{N^*}\right]&\leq& 2\left(1-\Phi\left( \rho^*_{N^*}\sigma^{-1}\sqrt{\tau_i^{*(m)}-\tau}\right)\right)\nonumber\\
	&\leq& \frac{2\sigma\sqrt{3}}{\sqrt{\pi}}\cdot \frac{\exp(-\delta^*_{N^*}(\rho^*_{N^*})^2/(12\sigma^2))}{\rho^*_{N^*}\sqrt{\delta^*_{N^*}}},
	\end{eqnarray}
	where we used the fact that $\tau_i^{*(m)}-\tau>\tau_i^{*(m)}-\tau^*_i-w^*({N^*})>\delta^*_{N^*}/3-\delta^*_{N^*}/6$. 
	\newline
	\newline
	For  $\tau^*_i-w^*({N^*})\leq\tau<\tau^*_i$, we have $\frac{1}{\tau_i^{*(m)}-\tau}\sum_{j= \tau+1}^{\tau_i^{*(m)}}(Z_j-\nu^*_i)\sim N\left(\frac{\tau^*_i-\tau}{\tau_i^{*(m)}-\tau}(\nu^*_i-\nu^*_{i-1}),\frac{\sigma^2}{\tau_i^{*(m)}-\tau}\right)$. The z-scores of $\pm\rho^*_{N^*}$ would have magnitudes greater than 
	\begin{align}
	&\quad\sigma^{-1}\sqrt{\tau_i^{*(m)}-\tau}\left(\rho^*_{N^*}-\left|\frac{\tau^*_i-\tau}{\tau_i^{*(m)}-\tau}(\nu^*_i-\nu^*_{i-1})\right|\right)\nonumber\\
	&\geq \sigma^{-1}\sqrt{\frac{\delta^*_{N^*}}{3}}\left(\rho^*_{N^*}-\frac{w^*({N^*})}{\delta^*_{N^*}/3}(2\bar{\theta})\right)\nonumber\\
	&\geq \sqrt{\frac{\delta^*_{N^*}}{3}}\cdot\frac{\rho^*_{N^*}}{2\sigma}
	\end{align}
	Hence, this gives the probability bound
	\begin{eqnarray}\label{bbbp2}
	\mathbb{P}\left[ \left| \frac{1}{\tau_i^{*(m)}-\tau}\sum_{j= \tau+1}^{\tau_i^{*(m)}}(Z_j-\nu^*_i)\right|\geq \rho^*_{N^*}\right]&\leq& 2\left(1-\Phi\left(  \sqrt{\frac{\delta^*_{N^*}}{3}}\cdot\frac{\rho^*_{N^*}}{2\sigma}\right)\right)\nonumber\\
	&\leq &\frac{2\sigma\sqrt{6}}{\sqrt{\pi}}\cdot\frac{\exp\left(- \delta(\gamma^*_{N^*})^2/(24\sigma^2)\right) }{\rho^*_{N^*}\sqrt{\delta^*_{N^*}}}.
	\end{eqnarray}
	Putting together the bounds in (\ref{bbbp1}) and (\ref{bbbp2}) will give 
	\begin{align}
	&\quad \sum_{\tau:\, |\tau-\tau_i|\leq w^*({N^*})}\mathbb{P}\left[ \left| \frac{1}{\tau_i^{*(m)}-\tau}\sum_{j= \tau+1}^{\tau_i^{*(m)}}(Z_j-\nu^*_i)\right|
	\geq \rho^*_{N^*}\right]\nonumber\\
	&\leq 3w^*({N^*})\cdot \frac{2\sigma\sqrt{6}}{\sqrt{\pi}}\cdot\frac{\exp\left(- \delta^*_{N^*}(\rho^*_{N^*})^2/(24\sigma^2)\right) }{\rho^*_{N^*}\sqrt{\delta^*_{N^*}}}
	\end{align}
	In an extremely similar manner, it can be argued that 
	\begin{align}
	&\quad\sum_{\tau:\,|\tau-\tau^*_{i+1}|\leq w^*({N^*})}\mathbb{P}\left[ \left| \frac{1}{\tau-\tau_i^{*(m)}}\sum_{j= \tau_i^{*(m)}+1}^{\tau}(Z_j-\nu^*_i)\right|\geq\rho^*_{N^*}\right]\nonumber\\
	&\leq 3w^*({N^*})\cdot \frac{2\sigma\sqrt{6}}{\sqrt{\pi}}\cdot\frac{\exp\left(- \delta^*_{N^*}(\rho^*_{N^*})^2/(24\sigma^2)\right) }{\rho^*_{N^*}\sqrt{\delta^*_{N^*}}}
	\end{align}
	Therefore, (\ref{bbp2}) can be bounded by 
	\begin{eqnarray}\label{bb2}
	B_{N^*}+\frac{12\sigma\sqrt{6}}{\sqrt{\pi}}\cdot \frac{w^*({N^*})\exp\left( -\delta^*_{N^*}(\rho^*_{N^*})^2/(24\sigma^2) \right)}{\rho^*_{N^*}\sqrt{\delta^*_{N^*}}}
	\end{eqnarray}
	\newline
	\newline
	By taking constants $C_1$ and $C_2$ to be the "worse" of the coefficients in (\ref{bb1}) and (\ref{bb2}), which are $\frac{12\sigma\sqrt{6}}{\sqrt{\pi}}$ and $1/(24\sigma^2)$ respectively, we can combine the result of both cases and establish
	\begin{eqnarray}
	\mathbb{P}\left[ \hat{J}=J;\; |\hat{\nu}^*_i-\nu^*_i|\geq\rho^*_{N^*} \right]\leq B_{N^*}+C_1w^*({N^*})\frac{\exp\left( -C_2\delta^*_{N^*}(\rho^*_{N^*})^2 \right)}{\rho^*_{N^*}\sqrt{\delta^*_{N^*}}}
	\end{eqnarray}
	for all $i=1,\dots,J$
\end{proof}
\noindent Using part (i), previously shown, it is straightforward to show part (ii):
\begin{proof}	 
	The complement of the event $\{\hat{J}=J;\quad \max_{i=0,...,J}|\hat{\nu}_i-\nu_i|<\rho_N\}$ is the event where either $\hat{J}\neq J$ or $\hat{J}=J$ and $|\hat{\nu}_i-\nu_i|\geq\rho_N$ for some $i$. For all sufficiently large $N$ and some positive constants $C_1$ and $C_2$ we have
	\begin{eqnarray}
	&& 1-\mathbb{P}\left[ \hat{J}=J;\quad \max_{i=0,...,J}|\hat{\nu}^*_i-\nu^*_i|<\rho^*_{N^*}\right] \nonumber\\
	&\leq& \mathbb{P}[\hat{J}\neq J]+\sum_{i=0}^J\mathbb{P}\left[ \hat{J}=J;\quad |\hat{\nu}^*_i-\nu^*_i|\geq \rho^*_{N^*}\right]\nonumber\\
	&\leq & B_{N^*}+(J+1)\left( B_{N^*}+C_1w^*({N^*}) \frac{\exp\left[ -C_2\delta^*_{N^*}(\rho^*_{N^*})^2\right]}{\sqrt{\delta^*_{N^*}}\rho^*_{N^*}} \right)\nonumber\\
	& \leq & B_{N^*}+\left(\frac{{N^*}}{\delta^*_{N^*}}+1\right)\left( B_{N^*}+C_1w^*({N^*}) \frac{\exp\left[ -C_2\delta^*_{N^*}(\rho^*_{N^*})^2\right]}{\sqrt{\delta^*_{N^*}}\rho^*_{N^*}} \right)\nonumber\\
	&\to & 0
	\end{eqnarray}
\end{proof}

\section{Supplement Part C (Probability Bounds on Argmin of Random Walks Absolute Value Drifts)}\label{sec:supplementpartc}
\subsection{Probability bound for Argmin of Random Walk}
Here we will derive a probability bound for random walks of the form
\begin{eqnarray}
X_\Delta(t):=\begin{cases}
t\left|\frac{\Delta}{2}\right|+\sum_{i=1}^t\varepsilon_i\qquad &t>0\\
0 &t=0\\
|t|\cdot\left|\frac{\Delta}{2}\right|-\sum_{i=-1}^{|t|}\varepsilon_i &t<0
\end{cases}
\end{eqnarray}
Specifically, the following exponential bound applies:
\begin{lemma}\label{lem:probbound}
	Suppose that $S$ is a set of integers and $m$ a positive integer such that $[-m,m]\subset S$, then 
	\begin{eqnarray}
	\mathbb{P}\left[ \left|\underset{t\in S}{\arg\min}X_\Delta(t)\right|>m \right]\leq A(\Delta)\exp(-B(\Delta)m)
	\end{eqnarray}
	where $A$ and $B$ are expressions dependent only on $|\Delta|$, with $A$ decreasing and $B$ increasing in $|\Delta|$.
\end{lemma}
\begin{proof}
	\begin{eqnarray}
	&&\mathbb{P}\left[ \left|\underset{t\in S}{\arg\min}X_\Delta(t)\right|>m \right]\nonumber\\
	&= & \sum_{j>m ,\,j\in S}\mathbb{P}\left[ \underset{t\in S}{\arg\min}X_\Delta(t)=j \right]+\sum_{j<-m ,\,j\in S}\mathbb{P}\left[ \underset{t\in S}{\arg\min}X_\Delta(t)=j \right]\nonumber\\
	&\leq & \sum_{t>m ,\,t\in S}\mathbb{P}\left[ X_\Delta(t)< 0 \right]+\sum_{t<-m ,\,t\in S}\mathbb{P}\left[ X_\Delta(t)<0 \right]\nonumber\\
	&\leq &\sum_{t>m ,\,t\in S}\mathbb{P}\left[ N\left(t\frac{|\Delta|}{2},t\right)< 0 \right]+\sum_{t<-m ,\,t\in S}\mathbb{P}\left[  N\left(|t|\frac{|\Delta|}{2},|t|\right)<0 \right]\nonumber\\
	&\leq &\sum_{t=m+1 }^\infty\mathbb{P}\left[ N\left(0,1\right)< -\frac{\sqrt{t}|\Delta|}{2} \right]+\sum_{t=-m-1}^{-\infty}\mathbb{P}\left[  N\left(0,1\right)<-\frac{\sqrt{|t|}|\Delta|}{2} \right]\nonumber\\
	&\leq & \sum_{t=m+1 }^\infty \exp\left( -\frac{t\Delta^2}{8} \right)+\sum_{t=-m-1}^{-\infty}\exp\left( -\frac{|t|\Delta^2}{8}\right)\nonumber\\
	&=& 2\left( \frac{\exp\left( -\frac{\Delta^2}{8} \right)}{1+\exp\left( -\frac{\Delta^2}{8} \right)} \right)\exp\left( -\frac{m\Delta^2}{8} \right)
	\end{eqnarray}
	
\end{proof}
This result has another implication. With probability approaching to 1 at an exponential pace, the argmin of $X_\Delta(t)$ equals the argmin over a smaller set:
\begin{lemma}\label{lem:eqprob}
	For any set of integers $S$ and positive integer $m$ such that $[-m,m]\in S$, 
	\begin{eqnarray}
	\mathbb{P}\left( \underset{t\in [-m,m]}{\arg\min}X_\Delta(t)=\underset{t\in S}{\arg\min}X_\Delta(t) \right)\geq 1-A(\Delta)\exp(-B(\Delta)m)
	\end{eqnarray}
	where $A$, $B$ are expressions in $\Delta$ that are, respectively, decreasing and increasing in $|\Delta|$ Lemma \ref{lem:probbound}.
\end{lemma}
\begin{proof}
	The two argmins of $X_\Delta(t)$ (over $S$ and over $[m,m]$) are different if and only if the argmin over $S$ is outisde of the interval $[-m,m]$. Therefore
	\begin{eqnarray}
	&&\mathbb{P}\left( \underset{t\in [-m,m]}{\arg\min}X_\Delta(t)\neq\underset{t\in S}{\arg\min}X_\Delta(t) \right)\nonumber\\
	&=&\mathbb{P}\left( \left|\underset{t\in S}{\arg\min}X_\Delta(t)\right|>m \right)\nonumber\\
	&\leq &A(\Delta)\exp(-B(\Delta)m)
	\end{eqnarray}
	%
	%
\end{proof}
This leads to
\begin{lemma}
	Suppose that $S$ is an integer set, $\ell$ and $m$ are positive integers, and $[-\ell,\ell]\subset [-m,m]\subset S$, then
	\begin{eqnarray}
	\left| \mathbb{P}\left[ \left|\underset{t\in[-m,m]}{\arg\min}X_\Delta(t)\right|\leq \ell \right]-\mathbb{P}\left[ \left|\underset{t\in S}{\arg\min}X_\Delta(t)\right|\leq \ell \right] \right|\leq A(\Delta)\exp(-B(\Delta)m)
	\end{eqnarray}
	for some expressions $A()$ and $B()$ that are respectively, decreasing and increasing with respect to $|\Delta|$.
\end{lemma}
\begin{proof}
	\begin{eqnarray}
	&&\mathbb{P}\left[ \left|\underset{t\in[-m,m]}{\arg\min}X_\Delta(t)\right|\leq \ell \right]-\mathbb{P}\left[ \left|\underset{t\in S}{\arg\min}X_\Delta(t)\right|\leq \ell \right]\nonumber\\
	&=&\mathbb{P}\left[ \left|\underset{t\in[-m,m]}{\arg\min}X_\Delta(t)\right|\leq\ell \text{ and } \left|\underset{t\in S}{\arg\min}X_\Delta(t)\right|> \ell \right]\nonumber\\
	&\leq &\mathbb{P}\left[ \left|\underset{t\in[-m,m]}{\arg\min}X_\Delta(t)\right|\neq\left|\underset{t\in S}{\arg\min}X_\Delta(t)\right|\right]\nonumber\\
	\end{eqnarray}
	The last line is greater than 0, and by Lemma \ref{lem:eqprob}, less than  $A(\Delta)\exp(-B(\Delta)m)$ for some appropriate expressions $A()$ and $B()$.
\end{proof}
\subsection{Quantiles}\label{sec:proofquantbound}
Due to Lemma \ref{lem:probbound}, the following statement can be made regarding the quantile: 
\begin{lemma}\label{lem:quantbound}
	Using the $A$ and $B$ from Lemma \ref{lem:probbound}, 
	\begin{eqnarray}
	Q_\Delta(\sqrt[J]{1-\alpha})\leq \frac{1}{B}\log\frac{AJ}{\alpha}
	\end{eqnarray}
\end{lemma}
\begin{proof}
	Using the inequality from Lemma \ref{lem:probbound}, 
	\begin{eqnarray}
	&&\mathbb{P}\left[ \left|\underset{t\in\mathbb{Z}}{\arg\min}X_\Delta(t)\right|\leq \frac{1}{B}\log\frac{AJ}{\alpha} \right]\nonumber\\
	&\geq & 1-A\exp\left( -B\frac{1}{B}\log\frac{AJ}{\alpha} \right)\nonumber\\
	&=&1-\frac{\alpha}{J}\nonumber\\
	&\geq &\sqrt[J]{1-\alpha}
	\end{eqnarray}
\end{proof}

\subsection{Comparison Between Random Walks, Part 1}
Here will show some probability inequalities between random walks with different drifts. These results are useful in proving Theorem \ref{thm:increasingJasymprotics}.
\begin{lemma}\label{lem:generalcompoppo}
	Suppose that $S$ is a set of integers, $m$ is a positive integer, and $[-m,m]\subset S$. Define, for any positive $\Delta_1$, $\Delta_2$, the random walks
	\begin{eqnarray}
	W_{\Delta_1,\Delta_2}(t)=\begin{cases}
	|t|\frac{|\Delta_1|}{2}+\sum_{i=-1}^t\varepsilon_i\qquad &t<0\\
	0 & t=0\\
	t\frac{|\Delta_2|}{2}+\sum_{i=1}^{t}\varepsilon_i\qquad& t>0
	\end{cases}
	\end{eqnarray}
	Then for any $\eta>0$,
	\begin{eqnarray}\label{eq:oppoineqfirst}
	&&\mathbb{P}\left[ \left|\underset{t\in S}{\arg\min} W_{\Delta_1,\Delta_2}(t)  \right|\leq m\right] \geq\mathbb{P}\left[ \left|\underset{t\in S}{\arg\min} W_{\Delta_1,\Delta_2+2\eta}(t)  \right|\leq m \right]\nonumber\\
	&&-\left(A(\Delta_2)m^{3/2}+B(\Delta_2)\sqrt{m}\right)\eta\exp(-C(\Delta_2)m)
	\end{eqnarray}
	for some expressions $A(\Delta_2)$, $B(\Delta_2)$ and $C(\Delta_2)$ that are, respectively, decreasing, decreasing, and increasing with respect to $|\Delta_2|$. Similarly, the following inequality holds:
	\begin{eqnarray}
	&&\mathbb{P}\left[ \left|\underset{t\in S}{\arg\min} W_{\Delta_1,\Delta_2}(t)  \right|\leq m\right] \geq\mathbb{P}\left[ \left|\underset{t\in S}{\arg\min} W_{\Delta_1+2\eta,\Delta_2}(t)  \right|\leq m \right]\nonumber\\
	&&-\left(A(\Delta_1)m^{3/2}+B(\Delta_1)\sqrt{m}\right)\eta\exp(-C(\Delta_1)m)
	\end{eqnarray}
	where the form of the expressions $A()$, $B()$, and $C()$ has identical forms as expressions used in (\ref{eq:oppoineqfirst}).
\end{lemma}
\begin{proof}
	It is only required to prove the inequality between $W_{\Delta_1, \Delta_2}$ and $W_{\Delta_1, \Delta_2+2\eta}$. This is because $W_{\Delta_1,\Delta_2}(t)$ has the same distribution as $W_{\Delta_2, \Delta_1}(-t)$ for all $t\in S$, and therefore 
	\begin{eqnarray}
	&&\mathbb{P}\left[ \left|\underset{t\in S}{\arg\min} W_{\Delta_1,\Delta_2}(t)  \right|\leq m \right]\nonumber\\
	&=&\mathbb{P}\left[ \left|\underset{t\in S}{\arg\min} W_{\Delta_2,\Delta_1}(-t)  \right|\leq m \right]\nonumber\\
	&=&\mathbb{P}\left[ \left|\underset{t\in -S}{\arg\min} W_{\Delta_2,\Delta_1}(t)  \right|\leq m \right]\nonumber\\
	&&\text{where }-S:=\{-t:t\in S\}
	\end{eqnarray}
	where the last inequality is due to the fact that the existence of an $|\ell|\leq m$ such that $W_{\Delta_2,\Delta_1}(-\ell)<W_{\Delta_2,\Delta_1}(-t)$ for all $t\in S$, $|t|>m$ could be true if and only if there exists an $|\ell|\leq m$ such that $W_{\Delta_2,\Delta_1}(\ell)<W_{\Delta_2,\Delta_1}(t)$ for all $t\in -S$, $|t|>m$. Similarly, we also have
	\begin{eqnarray}
	\mathbb{P}\left[ \left|\underset{t\in S}{\arg\min} W_{\Delta_1+2\eta,\Delta_2}(t)  \right|\leq m \right]=\mathbb{P}\left[ \left|\underset{t\in -S}{\arg\min} W_{\Delta_2,\Delta_1+2\eta}(t)  \right|\leq m \right],
	\end{eqnarray}
	and from here, an inequality can be derived by comparing 
	$$ \mathbb{P}\left[ \left|\underset{t\in -S}{\arg\min} W_{\Delta_2,\Delta_1}(t)  \right|\leq m \right]$$
	and 
	$$ \mathbb{P}\left[ \left|\underset{t\in -S}{\arg\min} W_{\Delta_2,\Delta_1+2\eta}(t)  \right|\leq m \right]$$
	using (\ref{eq:oppoineqfirst}). Therefore, the rest of the proof will only concern the random walks $W_{\Delta_1,\Delta_2}(\cdot)$ and $W_{\Delta_1,\Delta_2+2\eta}(\cdot)$.
	\newline
	\newline
	For the sake of brevity here, we will use the shorthand notations $W(t)$ for the random walk $W_{\Delta_1,\Delta_2}(t)$, and $W_+(t)$ for the random walk $W_{\Delta_1,\Delta_2+2\eta}(t)$. We are interested in the probability of the event when
	$\left|\underset{t\in S}{\arg\min}W(t)\right|> m$ and $\left|\underset{t\in S}{\arg\min}W_+(t)\right|\leq m$, so for now, assume that for some integer $k\in S$, such that $|k|>m$, $W(k)< W(t)$ for all $t\in S$, $|t|\leq m$. 
	\begin{itemize}
		\item If $k<-m$, then $W_+(k)=W(k)< W(t)\leq W(t)+t\eta 1(t>0)=W_+(t)$ for all $t\in S-\{k\}$; in other words $\left|\underset{t\in S}{\arg\min}W_+(t)\right|>m$, a contradiction. Therefore it is not possible for $k<-m$.
		\item This leaves the possibility that $k>m$. Additionally:
		\begin{itemize}
			\item as how '$k$' was defined, $W(k)< \min_{|t|\leq m}W(t)$
			\item because $W_+(k)=W(k)+k\eta$ is not the minimum among the $W_+(t)$'s for $t\in S$, we have 
			\begin{eqnarray}
			W(k)+k\eta&=&W_+(k)\nonumber\\
			&\geq& \min_{|t|\leq m} W_+(t) \nonumber\\
			&\geq&\min_{|t|\leq m}W(t)
			\end{eqnarray} 
		\end{itemize}
	\end{itemize}
	This breakdown of events shows that in order for the argmin of $W_+(t)$ to be within $[-m,m]$ but for the argmin of $W(t)$ to be outside this interval, there must be a $k>m$ where $ \left(\min_{|t|\leq m}W(t)\right)-\eta k< W(k) \leq \min_{|t|\leq m}W(t)$. Therefore
	\begin{eqnarray}\label{ref:probdcomposepos}
	&&\mathbb{P}\left[\left|\underset{t\in S}{\arg\min}W(t)\right|> m\text{ and }\left|\underset{t\in S}{\arg\min}W_+(t)\right|\leq m\right]\nonumber\\
	&\leq & \mathbb{P}\left[ \exists k:\, k>m\text{ and } \left(\min_{|t|\leq m}W(t)\right)-\eta k< W(k) \leq \min_{|t|\leq m}W(t)\right]\nonumber\\
	&\leq & \sum_{k\in S\cap (m,\infty)}\mathbb{P}\left[ \left(\min_{|t|\leq m}W(t)\right)-\eta k< W(k) \leq \min_{|t|\leq m}W(t)\right].
	\end{eqnarray}
	The random variable $\min_{|t|\leq m}W(t)$ can either equal $\min_{t\in [0,m]}W(t)$ or $\min_{t\in [-m,0]}W(t)$. Therefore, for any specific $k$, the event $\left(\min_{|t|\leq m}W(t)\right)-\eta k< W(k) \leq \min_{|t|\leq m}W(t)$ implies that either $\left(\min_{t\in [0,m]}W(t)\right)-\eta k< W(k) \leq \min_{t\in [0,m]}W(t)$ or $\left(\min_{t\in [-m,0]}W(t)\right)-\eta k< W(k) \leq \min_{t\in [-m,0]}W(t)$, yielding the inequality
	\begin{eqnarray}\label{ref:probdcomposepos2}
	&&\mathbb{P}\left[ \left(\min_{|t|\leq m}W(t)\right)-\eta k< W(k) \leq \min_{|t|\leq m}W(t)\right]\nonumber\\
	&\leq & \mathbb{P}\left[ \left(\min_{t\in [0,m]}W(t)\right)-\eta k< W(k) \leq \min_{t\in [0,m]}W(t)\right]\nonumber\\
	&&+\mathbb{P}\left[ \left(\min_{t\in [-m,0]}W(t)\right)-\eta k< W(k) \leq \min_{t\in [-m,0]}W(t)\right]
	\end{eqnarray}
	Both of the two probabilities in the last part can be bounded. First, because $k>m$, $W(k)$ is independent of $W(-1),\dots, W(-m)$, the distribution of $W(k)$ is still $N\left( k\frac{|\Delta_2|}{2},k \right)$ even after conditioning on the value of $\min_{t\in [-m,0]}W(t)$, hence
	\begin{eqnarray}\label{eq:probboundnegside}
	&&\mathbb{P}\left[ \left(\min_{t\in [-m,0]}W(t)\right)-\eta k< W(k) \leq \min_{t\in [-m,0]}W(t)\right]\nonumber\\
	&=&\mathbb{E}\left[ \mathbb{P}\left[ x-\eta k< W(k) \leq x\Big| \min_{t\in [-m,0]}W(t)=x \right] \right]\nonumber\\
	&= & \mathbb{E}\left[\mathbb{P}\left[ \frac{x}{\sqrt{k}}-\sqrt{k}\left(  \frac{|\Delta_2|}{2}+\eta  \right)< N(0,1)\leq \frac{x}{\sqrt{k}}-\sqrt{k} \frac{|\Delta_2|}{2}\Bigg| \min_{t\in [-m,0]}W(t)=x  \right]\right]\nonumber\\
	&&\text{where }x<0\text{ since }\min_{t\in [-m,0]}W(t)\leq W(0)=0\nonumber\\
	&=& \mathbb{E}\left[\int_{\frac{x}{\sqrt{k}}-\sqrt{k}\left(  \frac{|\Delta_2|}{2}+\eta  \right)}^{ \frac{x}{\sqrt{k}}-\sqrt{k} \frac{|\Delta_2|}{2}} \frac{\exp(-z^2/2)}{\sqrt{2\pi}} \,dz \Bigg| \min_{t\in [-m,0]}W(t)=x  \right]\nonumber\\
	&\leq &\int_{-\sqrt{k}\left(  \frac{|\Delta_2|}{2}+\eta  \right)}^{ -\sqrt{k} \frac{|\Delta_2|}{2}} \frac{\exp(-z^2/2)}{\sqrt{2\pi}} \,dz\nonumber\\
	&\leq & \frac{\eta\sqrt{k}}{\sqrt{2\pi}}\exp\left[ -\frac{\Delta_2^2}{8}k \right]
	\end{eqnarray}
	As for the other inequality, consider the event that $\left(\min_{t\in [0,m]}W(t)\right)-\eta k< W(k) \leq \min_{t\in [0,m]}W(t)$. Because $\min_{t\in [0,m]}W(t)\leq W(0)=0$, this event implies that for some $\ell\in [0,m]$, we have $W(\ell)-\eta k<W(k)<W(\ell)\leq 0$ (namely, letting $\ell=\underset{t\in [0,m]}{\arg\min}W(t)$ would work). Therefore
	\begin{eqnarray}\label{eq:decomposeprobdepend}
	&& \mathbb{P}\left[ \left(\min_{t\in [0,m]}W(t)\right)-\eta k< W(k) \leq \min_{t\in [0,m]}W(t)\right]\nonumber\\
	&\leq & \sum_{\ell=0}^{m}\mathbb{P}\left[W(\ell)-\eta k< W(k)\leq W(\ell)\leq 0 \right]\nonumber\\
	&= & \sum_{\ell=0}^{m}\mathbb{P}\left[-\eta k< W(k)-W(\ell)\leq 0\text{ and }W(\ell)\leq 0 \right]
	\end{eqnarray}
	Now $W(\ell)=\ell\frac{|\Delta_2|}{2}+\sum_{j=1}^\ell\varepsilon_j$ and $W(k)-W(\ell)=(k-\ell)\frac{|\Delta_2|}{2}+\sum_{j=\ell+1}^k\varepsilon_j$ are independent random variables, with distributions $N\left( \ell\frac{|\Delta_2|}{2},\ell \right)$ and $N\left( (k-\ell)\frac{|\Delta_2|}{2},(k-\ell) \right)$. Hence
	\begin{eqnarray}
	&&\mathbb{P}\left[-\eta k< W(k)-W(\ell)\leq 0\text{ and }W(\ell)\leq 0 \right]\nonumber\\
	&=& \mathbb{P}\left[-\eta k< W(k)-W(\ell)\leq 0\right]\cdot \mathbb{P}\left[ W(\ell)\leq 0 \right]\nonumber\\
	&=&\mathbb{P}\left[-\sqrt{k-\ell}\frac{|\Delta_2|}{2}-\frac{\eta k}{\sqrt{k-\ell}}< N(0,1)\leq-\sqrt{k-\ell}\frac{|\Delta_2|}{2}\right]\cdot \mathbb{P}\left[ N(0,1)\leq -\sqrt{\ell}\frac{|\Delta_2|}{2} \right]\nonumber\\
	&=&\int_{-\sqrt{k-\ell}\frac{|\Delta_2|}{2}-\frac{\eta k}{\sqrt{k-\ell}}}^{-\sqrt{k-\ell}\frac{|\Delta_2|}{2}}\frac{\exp(-z^2/2)}{\sqrt{2\pi}}\,dz\cdot \mathbb{P}\left[ N(0,1)\leq -\sqrt{\ell}\frac{|\Delta_2|}{2} \right]\nonumber\\
	&\leq & \frac{\eta k}{\sqrt{k-\ell}}\frac{\exp\left( -\frac{\Delta_2^2}{8}(k-\ell) \right)}{\sqrt{2\pi}}\cdot \frac{1}{2}\exp\left[ -\frac{\Delta_2^2}{8}\ell \right]\nonumber\\
	&=&\frac{1}{2\sqrt{2\pi}}\cdot \frac{\eta k}{\sqrt{k-\ell}}\exp\left[ -\frac{\Delta_2^2}{8} k \right]
	\end{eqnarray}
	Therefore, using (\ref{eq:decomposeprobdepend}),
	\begin{eqnarray}\label{eq:probboundposside}
	&& \mathbb{P}\left[ \left(\min_{t\in [0,m]}W(t)\right)-\eta k< W(k) \leq \min_{t\in [0,m]}W(t)\right]\nonumber\\
	&\leq & \sum_{\ell=0}^{m}\frac{1}{2\sqrt{2\pi}}\cdot \frac{\eta k}{\sqrt{k-\ell}}\exp\left[ -\frac{\Delta_2^2}{8} k \right]\nonumber\\
	&\leq &\frac{\exp\left[ -\frac{\Delta_2^2}{8} k \right]}{2\sqrt{2\pi}}\cdot\eta k\cdot \int_0^{m+1}\frac{1}{\sqrt{k-x}}\,dx\nonumber\\
	&=& \frac{\exp\left[ -\frac{\Delta_2^2}{8} k \right]}{\sqrt{2\pi}}\cdot\eta k\cdot(\sqrt{k}-\sqrt{k-m-1})\nonumber\\
	&\leq &\frac{\exp\left[ -\frac{\Delta_2^2}{8} k \right]}{\sqrt{2\pi}}\cdot\eta k\cdot \sqrt{m+1}
	\end{eqnarray}
	Therefore, 
	\begin{eqnarray}
	&&\mathbb{P}\left[\left|\underset{t\in S}{\arg\min}W(t)\right|> m\text{ and }\left|\underset{t\in S}{\arg\min}W_+(t)\right|\leq m\right]\nonumber\\
	&\leq & \sum_{k\in S\cap (m,\infty)} \mathbb{P}\left[ \left(\min_{t\in [0,m]}W(t)\right)-\eta k< W(k) \leq \min_{t\in [0,m]}W(t)\right]\nonumber\\
	&&+\sum_{k\in S\cap (m,\infty)}\mathbb{P}\left[ \left(\min_{t\in [-m,0]}W(t)\right)-\eta k< W(k) \leq \min_{t\in [-m,0]}W(t)\right]\nonumber\\
	&&\text{according to }(\ref{ref:probdcomposepos})\text{ and }(\ref{ref:probdcomposepos2})\nonumber\\
	&\leq &\sum_{k=m+1}^\infty \left(\frac{\exp\left[ -\frac{\Delta_2^2}{8} k \right]}{\sqrt{2\pi}}\cdot\eta k\cdot \sqrt{m+1}\right)+\sum_{k=m+1}^\infty \left(\frac{\eta\sqrt{k}}{\sqrt{2\pi}}\exp\left[ -\frac{\Delta_2^2}{8}k \right]\right)\nonumber\\
	&&\text{according to }(\ref{eq:probboundnegside})\text{ and }(\ref{eq:probboundposside})\nonumber\\
	&\leq &\frac{\eta(\sqrt{m+1}+1)}{\sqrt{2\pi}}\sum_{m+1}^\infty k \exp\left[ -\frac{\Delta_2^2}{8}k \right]\nonumber\\
	&\leq & \frac{3\sqrt{m}}{\sqrt{2\pi}} \left( \frac{m\exp\left[ -\frac{\Delta_2^2}{8} \right]}{1-\exp\left[ -\frac{\Delta_2^2}{8} \right]}+\frac{\exp\left[ -\frac{\Delta_2^2}{8} \right]}{\left( 1-\exp\left[ -\frac{\Delta_2^2}{8} \right] \right)^2} \right)\left( \eta \exp\left[ -\frac{\Delta_2^2}{8}m \right]\right)\nonumber\\
	&&\text{since }\sum_{k=a}^\infty kx^k=\left( \frac{a-1}{1-x}+\frac{1}{(1-x)^2} \right)x^a\text{ for any }|x|<1, a\in\mathbb{N}
	\end{eqnarray}
	From here, 
	\begin{eqnarray}
	&&\mathbb{P}\left[ \left|\underset{t\in S}{\arg\min} W(t)  \right|\leq m \right]\nonumber\\
	&\geq &\mathbb{P}\left[ \left|\underset{t\in S}{\arg\min} W_+(t)  \right|\leq m \right]-\mathbb{P}\left[ \left|\underset{t\in S}{\arg\min} W_+(t)  \right|\leq m\text{ and }\left|\underset{t\in S}{\arg\min} W(t)  \right|> m \right]\nonumber\\
	&\geq &\mathbb{P}\left[ \left|\underset{t\in S}{\arg\min} W_+(t)  \right|\leq m \right]-\left(A(\Delta_2)m^{3/2}+B(\Delta_2)\sqrt{m}\right)\eta\exp(-C(\Delta_2)m)
	\end{eqnarray}
	for some expressions $A(\Delta_2)$, $B(\Delta_2)$ and $C(\Delta_2)$ that are, respectively, decreasing, decreasing, and increasing with respect to $|\Delta_2|$.
\end{proof}
This result immediately leads to some results concerning the random walks $X_\Delta(\cdot)$, as they are a special case of the random walks $W_{\Delta_1,\Delta_2}(\cdot)$ where $\Delta_1=\Delta_2$.
\begin{lemma}\label{lem:randomwalkcompoppo}
	Suppose that for the random walk $Y_+(t)$ for $t\in S$ equals
	\begin{eqnarray}
	Y_+(t)=\begin{cases}
	X_\Delta(t)\qquad &\text{for }t\leq 0\\
	X_\Delta(t)+\eta t&\text{for }t>0
	\end{cases}
	\end{eqnarray}
	for some constant $\eta$ such that $0<\eta<\frac{|\Delta|}{2}$. Then for any $[-m,m]\subset S\subset\mathbb{Z}$, 
	\begin{align}\label{eq:lemrwalkcompoppo1}
	&\quad\mathbb{P}\left[\left|\underset{t\in S}{\arg\min}Y_+(t)\right|\leq m \right]\leq 	\mathbb{P}\left[\left|\underset{t\in S}{\arg\min}X_\Delta(t)\right|\leq m \right]\nonumber\\
	&+\eta\left[A(\Delta)m^{3/2}+B(\Delta)\sqrt{m}\right]\exp\left[-C(\Delta)m\right]
	\end{align}
	for some expressions $A()$, $B()$, and $C()$ which are, respectively, decreasing, decreasing, and increasing with respect to $|\Delta|$. The same probability inequality will hold if $Y_+(t)=X_\Delta(t)+\eta |t|1(t<0)$. 
	\newline
	\newline
	A similar set of inequalities hold for the random walk $Y_-(t)$ for $t\in S$, defined as
	\begin{eqnarray}
	Y_-(t)=\begin{cases}
	X_\Delta(t)\qquad &\text{for }t\leq 0\\
	X_\Delta(t)-\eta t&\text{for }t>0
	\end{cases}
	\end{eqnarray}
	For $\eta<\frac{|\Delta|}{2}$, and $[-m,m]\subset S$, we have
	\begin{eqnarray}\label{eq:lemrwalkcompoppo2}
	&&\mathbb{P}\left[\left|\underset{t\in S}{\arg\min}Y_-(t)\right|\leq m \right]\geq 	\mathbb{P}\left[\left|\underset{t\in S}{\arg\min}X_\Delta(t)\right|\leq m \right]-\nonumber\\
	&&\eta\left[A\left( \frac{|\Delta|}{2}-\eta \right)m^{3/2}+B\left( \frac{|\Delta|}{2}-\eta \right)\sqrt{m}\right]\exp\left[-C\left( \frac{|\Delta|}{2}-\eta \right)m\right]
	\end{eqnarray}
	for some expressions $A()$, $B()$, and $C()$ which are, respectively, decreasing, decreasing, and increasing with respect to $|\Delta|/2-\eta$. The same probability inequality will hold if $Y_-(t)=X_\Delta(t)-\eta |t|1(t<0)$.
\end{lemma}
\begin{proof}
	Apply Lemma \ref{lem:generalcompoppo} with $\Delta_1=\Delta_2=\Delta$ to prove (\ref{eq:lemrwalkcompoppo1}), and with $\Delta_1=\Delta$, $\Delta_2=\Delta-2\eta$ to prove (\ref{eq:lemrwalkcompoppo2}).
\end{proof}
Additionally, we can make probabilistic statements regarding the argmin of $X_\Delta(t)$ for two different values of $\Delta$:
\begin{lemma}\label{lem:twoXrandwalkcomp}
	For any $\Delta\neq 0$, $\eta>0$, and a set $S$ which contains the interval $[-m,m]$, 
	\begin{eqnarray}
	\mathbb{P}\left[ \left|\underset{t\in S}{\arg\min}X_\Delta(t)\right|\leq m \right]\leq \mathbb{P}\left[ \left|\underset{t\in S}{\arg\min}X_{|\Delta|+2\eta}(t)\right|\leq m \right]
	\end{eqnarray}
	and 
	\begin{eqnarray}
	&&\mathbb{P}\left[ \left|\underset{t\in S}{\arg\min}X_\Delta(t)\right|\geq m \right]\geq \mathbb{P}\left[ \left|\underset{t\in S}{\arg\min}X_{|\Delta|+2\eta}(t)\right|\leq m \right]\nonumber\\&&-2\eta\left[A(\Delta)m^{3/2}+B(\Delta)\sqrt{m}\right]\exp\left[-C(\Delta)m\right]
	\end{eqnarray}
	for some expressions $A()$, $B()$, and $C()$ which can take the same form and have the same monotonicity properties as the ones used in Lemma \ref{lem:generalcompoppo}. 
\end{lemma}
\begin{proof}
	The first inequality can be shown by noticing that the event $\left|\underset{t\in S}{\arg\min}X_\Delta(t)\right|\leq m$ implies that for some $|\ell|\leq m$, $X_\Delta(\ell)\leq X_\Delta(t)$ for all $t\in S$, $|t|>m$. This in turn implies that $X_{|\Delta|+\eta}(\ell)=X_{\Delta}(\ell)+|\ell|\eta<X_{\Delta}(t)+|t|\eta=X_{|\Delta|+\eta}(t)$ for all $t\in S$, $|t|>m$, which means $\left|\underset{t\in S}{\arg\min}X_{|\Delta|+\eta}(t)\right|\leq m$.
	\newline
	\newline
	The second inequality can be shown by applying Lemma \ref{lem:generalcompoppo} twice:
	\begin{eqnarray}
	&&\mathbb{P}\left[ \left|\underset{t\in S}{\arg\min}X_\Delta(t)\right|\leq m \right]\nonumber\\
	&\geq& \mathbb{P}\left[ \left|\underset{t\in S}{\arg\min}W_{|\Delta|,|\Delta|+2\eta}(t)\right|\leq m \right]-\eta\left[A(\Delta)m^{3/2}+B(\Delta)\sqrt{m}\right]\exp\left[-C(\Delta)m\right]\nonumber\\
	&\geq&\mathbb{P}\left[ \left|\underset{t\in S}{\arg\min}W_{|\Delta|+2\eta,|\Delta|+2\eta}(t)\right|\leq m \right]-2\eta\left[A(\Delta)m^{3/2}+B(\Delta)\sqrt{m}\right]\exp\left[-C(\Delta)m\right]\nonumber\\
	&=&\mathbb{P}\left[ \left|\underset{t\in S}{\arg\min}X_{|\Delta|+\eta}(t)\right|\leq m \right]-2\eta\left[A(\Delta)m^{3/2}+B(\Delta)\sqrt{m}\right]\exp\left[-C(\Delta)m\right]
	\end{eqnarray}
\end{proof}

\subsection{Comparison between Random Walks, Part 2}
Here we will prove inequalities similar to those presented in Lemma \ref{lem:generalcompoppo}, but in the other direction. These results are also useful in proving Theorem \ref{thm:increasingJasymprotics}.
\begin{lemma}
	Let the random walks $W_{\Delta_1,\Delta_2}$ be as they were defined in Lemma \ref{lem:generalcompoppo}. Then given any positive $\eta$, positive integer $m$, and set $S$ such that $|\eta|<|\Delta_1|/2$ and $[-m,m]\subsetneq{S}$,
	\begin{eqnarray}
	&&\mathbb{P}\left[ \left| \underset{t\in S}{\arg\min}W_{\Delta_1,\Delta_2+2\eta}(t) \right|\leq m \right]\geq \mathbb{P}\left[ \left| \underset{t\in S}{\arg\min}W_{\Delta_1,\Delta_2}(t) \right|\leq m \right]\nonumber\\
	&&-A\left( \frac{\Delta_1}{2}-\eta \right)\eta\sqrt{m}\exp\left( -B\left( \frac{\Delta_1}{2}-\eta \right)m \right)
	\end{eqnarray}
	for some positive expressions $A()$ and $B()$ which are, respectively, decreasing and increasing in $\frac{\Delta_1}{2}-\eta$. Similarly, between the random walks $W_{\Delta_1,\Delta_2}$ and $W_{\Delta_1+2\eta,\Delta_2}$ for $0<\eta<\frac{\Delta_2}{2}$ there is the inequality 
	\begin{eqnarray}
	&&\mathbb{P}\left[ \left| \underset{t\in S}{\arg\min}W_{\Delta_1+2\eta,\Delta_2}(t) \right|\leq m \right]\geq \mathbb{P}\left[ \left| \underset{t\in S}{\arg\min}W_{\Delta_1,\Delta_2}(t) \right|\leq m \right]\nonumber\\
	&&-A\left( \frac{\Delta_2}{2}-\eta \right)\eta\sqrt{m}\exp\left( -B\left( \frac{\Delta_2}{2}-\eta \right)m \right)
	\end{eqnarray}
	for some positive expressions $A()$ and $B()$ which are, respectively, decreasing and increasing in $\frac{\Delta_2}{2}-\eta$.
\end{lemma}
\begin{proof}
	We will show the inequality between $W_{\Delta_1,\Delta_2+2\eta}$ and $W_{\Delta_1,\Delta_2}$, and the result between $W_{\Delta_1+2\eta,\Delta_2}$ and $W_{\Delta_1,\Delta_2}$ can be shown in a similar fashion or in an argument similar to what was found in the proof for Lemma \ref{lem:generalcompoppo}. As in the proof of that lemma, we will use the shorthand notation $W$ for  $W_{\Delta_1,\Delta_2}$ and $W_+$ for $W_{\Delta_1,\Delta_2+2\eta}$.
	%
	\newline
	\newline
	We are interested in how $\left|\underset{t\in S}{\arg\min}W(t)\right|\leq m$ and $\left|\underset{t\in S}{\arg\min}W_+(t)\right|> m$ can simultaneously occur, which we will do by considering the possible values of the argmin of $W(t)$. If these two events are true, then first note that for some integer $k\in [-m,m]$, $W(k)\leq W(t)$ for all $t\in S$, $t\neq k$. 
	\begin{itemize}
		\item if $k\leq 0$, then $W_+(k)=W(k)\leq W(t)+|t|\eta 1(t>0)=W_+(t)$ for all $t\in S$ and $|t|>m$; in other words $\left|\underset{t\in S}{\arg\min}W_+(t)\right|\leq m$,
		\item thus $k >0$ and $W(k)=\underset{t\in [0,m]}{\arg\min}W(t)$. The only possible way for $\left|\underset{t\in S}{\arg\min}W_+(t)\right|> m$ is for the argmin to be less than $-m$: for any $t\in S$: 
		\begin{itemize}
			\item if $t>m$ then $W_+(t)=W(t)+t\eta>W(k)+k\eta=W_+(k)$, which means that the argmin of $W_+(t)$ cannot be greater than $m$ in absolute value
			\item therefore the only possible argmin for $W_+(t)$ is for some $\ell <-m$, and it must satisfy
			\begin{eqnarray}
			W(\ell)=W_+(\ell)<W_+(k)=W(k)+\eta k \leq W(k)+\eta m
			\end{eqnarray}
			but at the same time since $W(\ell)$ was not the minimum among the $W(t)$'s, we have $W(\ell)\geq W(k)$
		\end{itemize}
	\end{itemize}
	This breakdown of events shows that in order for the argmin of $W(t)$ to be within $[-m,m]$ but for the argmin of $W_+(t)$ to be outside this interval, there must be an $\ell<-m$ where $ \min_{t\in [0,m]}W(t)< W(\ell) \leq \min_{t\in [0,m]}W(t)+\eta m$. Hence:
	\begin{eqnarray}
	&&\mathbb{P}\left[ \left|\underset{t\in S}{\arg\min}W(t)\right|\leq m\text{ and } \left|\underset{t\in S}{\arg\min}W_+(t)\right|>m\right]\nonumber\\
	&\leq & \mathbb{P}\left[ \exists \ell<-m \text{ where } \min_{t\in [0,m]}W(t)< W(\ell) \leq \min_{t\in [0,m]}W(t)+\eta m \right]\nonumber\\
	&\leq & \sum_{\ell=-m-1}^{-\infty} \mathbb{P}\left[  \min_{t\in [0,m]}W(t)< W(\ell) \leq \min_{t\in [0,m]}W(t)+\eta m \right]\nonumber\\
	&= &\sum_{\ell=-m-1}^{-\infty} \mathbb{E}\left[\mathbb{P}\left[  x< W(\ell) \leq x+\eta m\Bigg| \min_{t\in [0,m]}W(t)=x \right]\right]\nonumber
	\end{eqnarray}
	\begin{align}
	&\leq  \sum_{\ell=-m-1}^{-\infty} \mathbb{E}\left[\mathbb{P}\left[  \frac{x}{\sqrt{|\ell|}}-\frac{\Delta_1}{2}\sqrt{|\ell|}< N\left( 0,1 \right) \leq \frac{x}{\sqrt{|\ell|}}-\frac{\Delta_1}{2}\sqrt{|\ell|}+\frac{\eta m}{\sqrt{|\ell|}}\Bigg| \min_{t\in [0,m]}W(t)=x \right]\right]\nonumber\\
	&= \sum_{\ell=-m-1}^{-\infty}\mathbb{E}\left[ \int_{\frac{x}{\sqrt{|\ell|}}-\frac{\Delta_1}{2}\sqrt{|\ell|}}^{\frac{x}{\sqrt{|\ell|}}-\frac{\Delta_1}{2}\sqrt{|\ell|}+\frac{\eta m}{\sqrt{|\ell|}}} \frac{1}{\sqrt{2\pi}}\exp(-z^2/2)\,dz \Bigg| \min_{t\in [0,m]}W(t)=x\right]\nonumber\\
	&\leq \sum_{\ell=-m-1}^{-\infty}  \int_{-\frac{\Delta_1}{2}\sqrt{|\ell|}}^{-\frac{\Delta_1}{2}\sqrt{|\ell|}+\frac{\eta m}{\sqrt{|\ell|}}} \frac{1}{\sqrt{2\pi}}\exp(-z^2/2)\,dz\nonumber\\
	&\text{since all possible values of }x\text{ are negative, and the integrated density is monotone}\nonumber\\
	&\leq \sum_{\ell=-m-1}^{-\infty}  \int_{-\frac{\Delta_1}{2}\sqrt{|\ell|}}^{-\frac{\Delta_1}{2}\sqrt{|\ell|}+\eta\sqrt{m}} \frac{1}{\sqrt{2\pi}}\exp(-z^2/2)\,dz\nonumber\\
	&\leq \sum_{\ell=-m-1}^{-\infty} \frac{\eta}{\sqrt{2\pi}}\sqrt{m}\exp\left[ -\frac{1}{2}\left( \frac{\Delta_1}{2}\sqrt{|\ell|}-\eta \sqrt{m}\right)^2 \right]\nonumber\\
	&\leq \frac{\eta\sqrt{m}}{\sqrt{2\pi}}\sum_{\ell=-m-1}^{-\infty} \exp\left[ -\frac{1}{2}\left(\frac{\Delta_1}{2}-\eta\right)^2|\ell| \right]\nonumber\\
	&\leq  \frac{1}{\sqrt{2\pi}}\left( \frac{\exp\left[- \frac{1}{2}\left( \frac{\Delta_1}{2}-\eta \right)^2 \right]}{1-\exp\left[- \frac{1}{2}\left( \frac{\Delta_1}{2}-\eta \right)^2\right]}\right)\eta\sqrt{m}\exp\left( -\frac{1}{2}\left( \frac{\Delta_1}{2}-\eta \right)^2m \right).\nonumber\\
	\end{align}
	Therefore, 
	\begin{eqnarray}
	&&\mathbb{P}\left[  \left|\underset{t\in S}{\arg\min}W_+(t)\right|\leq m\right]\nonumber\\
	&\geq & \mathbb{P}\left[ \left|\underset{t\in S}{\arg\min}W(t)\right|\leq m\right]-\mathbb{P}\left[ \left|\underset{t\in S}{\arg\min}W(t)\right|\leq m\text{ and } \left|\underset{t\in S}{\arg\min}Y_+(t)\right|>m\right]\nonumber\\
	& \geq & \mathbb{P}\left[ \left|\underset{t\in S}{\arg\min}W(t)\right|\leq m\right]-A'\eta \sqrt{m}\exp(-B'm)
	\end{eqnarray}
	for some constants $A'$ and $B'$ depending only on $\frac{\Delta_1}{2}-\eta$.
	
	%
	
\end{proof}
We can immediately apply this result to random walks of the form $X_\Delta$:

\begin{lemma}\label{lem:randomwalkcomp}
	Suppose that for the random walk $Y_+(t)$ for $t\in\mathbb{Z}$ equals
	\begin{eqnarray}
	Y_+(t)=\begin{cases}
	X_\Delta(t)\qquad &\text{for }t\leq 0\\
	X_\Delta(t)+\eta t&\text{for }t>0
	\end{cases}
	\end{eqnarray}
	for some constant $\eta$ such that $0<\eta<\frac{|\Delta|}{2}$. Then for any $[-m,m]\subset S\subset\mathbb{Z}$, 
	\begin{align}
	&\mathbb{P}\left[\left|\underset{t\in S}{\arg\min}Y_+(t)\right|\leq m \right]\geq 	\mathbb{P}\left[\left|\underset{t\in S}{\arg\min}X_\Delta(t)\right|\leq m \right]\nonumber\\
	&-A'\left( \frac{|\Delta|}{2}-\eta \right)\eta\sqrt{m}\exp\left( -B'\left( \frac{|\Delta|}{2}-\eta \right)m \right)
	\end{align}
	for some expressions $A'()$ and $B'()$ which are, respectively, decreasing and increasing with respect to $\left( \frac{|\Delta|}{2}-\eta \right)$. The same probability inequality will hold if $Y_+(t)=X_\Delta(t)+\eta |t|1(t<0)$. 
	\newline
	\newline
	In addition, if $Y_-(t)$ is defined differently as
	\begin{eqnarray}
	Y_-(t)=\begin{cases}
	X_\Delta(t)\qquad &\text{for }t\leq 0\\
	X_\Delta(t)-\eta t&\text{for }t>0
	\end{cases}
	\end{eqnarray}
	then we have the inequality 
	\begin{align}
	&\mathbb{P}\left[\left|\underset{t\in S}{\arg\min}Y_-(t)\right|\leq m \right]\leq 	\mathbb{P}\left[\left|\underset{t\in S}{\arg\min}X_\Delta(t)\right|\leq m \right]\nonumber\\
	&+A'\left( \frac{|\Delta|}{2}-\eta \right)\eta\sqrt{m}\exp\left( -B'\left( \frac{|\Delta|}{2}-\eta \right)m \right)
	\end{align}
	for some expressions $A'()$ and $B'()$ which are, respectively, decreasing and increasing with respect to $\left( \frac{|\Delta|}{2}-\eta \right)$.
\end{lemma} 
\section{Supplement Part D (Intelligent Sampling using Wild Binary Segmentation)}
\label{WBINSEG-SUPP} 
\subsection{Wild Binary Segmentation}\label{sec:wildbinseg}
We next discuss the Wild Binary Segmentation (WBinSeg) algorithm, introduced in \cite{fryzlewicz2014wild}. Similar to our treatment of the BinSeg procedure, we will explain the WBinSeg procedure in the context of applying it to the dataset $Z_1,\dots, Z_{N^*}$, a size $\sim N^\gamma$ size subsample of a larger dataset $Y_1,\dots,Y_N$ which satisfies the conditions (M1)-(M4). The steps of this algorithm are:
\begin{enumerate}
	\item Fix a threshold value $\zeta_{N^*}$ and initialize the segment set $SS=\{ (1,N) \}$, the change point estimate set $\underline{\hat{\tau}}=\emptyset$, and $M_{N^*}$ intervals $[s_1,e_1],\dots,[s_{M_{N^*}},e_{M_{N^*}}]$, where each $s_j$ and $e_j$ are uniformly picked from $\{1,\dots,{N^*}\}$.
	\item Pick any ordered pair $(s,e)\in SS$, remove it from $SS$ (update $SS$ by $SS\leftarrow SS-\{ (s,e) \}$). If $s\geq e$ then skip to step 6, otherwise continue to step 3.
	\item Define $\mathcal{M}_{s,e}:=\left\{ [s_i,e_i]:[s_i,e_i]\subseteq  [s,e] \right\}$.
	\begin{itemize}
		\item As an optional step, also take $\mathcal{M}_{s,e}\leftarrow \mathcal{M}_{s,e}\cup\left\{(s,e)\right\}$.
	\end{itemize}
	\item Find a $[s^*,e^*]\in\mathcal{M}_{s,e}$ such that
	$$\max_{b\in\{s^*,\dots,e^*-1\} }|\bar{Y}^b_{s^*,e^*}|=\max_{[s',e']\in\mathcal{M}_{s,e}}\left( \max_{b\in\{s',\dots,e'-1\} }|\bar{Y}^b_{s',e'}| \right)$$
	and let $b_0=\underset{b\in\{s^*,\dots,e^*-1\} }{\arg\max}|\bar{Z}^b_{s^*,e^*}|$.
	\item If $|\bar{Z}^{b_0}_{s^*,e^*}|\geq \zeta_{N^*}$, then add $b_0$ to the list of change point estimates (add $b_0$ to $\underline{\hat{\tau}}$), and add ordered pairs $(s,b_0)$ and $(b_0+1,e)$ to $SS$, otherwise skip to step 5.
	\item Repeat steps 2-4 until $SS$ contains no elements.
\end{enumerate}
Roughly speaking, WBinSeg performs very much like binary segmentation but with steps that maximize change point estimates over $M_{N^*}$ randomly chosen intervals. The consistency results in \cite{fryzlewicz2014wild} imply that in our setting, the following holds:
\begin{theorem}\label{thm:wbinsegresults}
	Suppose conditions (M1) to (M4) are satisfied and the tuning parameter $\zeta_{N^*}$ is chosen appropriately such that there exists positive constants $C_1$ and $C_2$ with $C_1\sqrt{\log({N^*})}\leq \zeta_{N^*}\leq C_2\sqrt{\delta}_{N^*}$. Denote $\hat{J}$, $\hat{\tau}_1,\dots,\hat{\tau}_{\hat{J}}$ as the estimates obtained from wild binary segmentation. Then, there exists positive constants $C_3,C_4$ where
	\begin{equation}\label{eq:wildbound}
	\mathbb{P}\Big[ \hat{J}=J;\quad \max_{i=1,...,J} |\hat{\tau}_i-\tau_i|\leq C_3\log({N^*}) \Big]\geq 1-C_4({N^*})^{-1}-\left( \frac{{N^*}}{\delta^*_{N^*}} \right)\left(1- \left(\frac{\delta^*_{N^*}}{3{N^*}}\right)^2\right)^{M_{N^*}}
	\end{equation}
\end{theorem}
We remark that the right side of (\ref{eq:wildbound}) does not necessarily converge to 1 unless $M_{N^*}\to\infty$ fast enough. Using some simple algebra, it was shown in the original paper, that for sufficiently large $N$, this expression can be bounded from below by $1-CN^{-1}$ for some $C>0$ if $M_{N^*}\geq  \left(\frac{3{N^*}}{\delta_{N^*}}\right)^2\log({N^*}^2/\delta_{N^*})$, a condition on $M_{N^*}$ which we we assume from here on in order to simplify some later analysis.
\newline
\newline
\indent Compared with the consistency result for binary segmentation given in Theorem \ref{frythm}, $\max_{j=1,\dots,J}|\hat{\tau}_j-\tau_j|$ can be bounded by some constant times $\log({N^*})$, which can grow much slower than $E_{N^*}=({N^*}/\delta^*_{N^*})^2\log({N^*})$ whenever ${N^*}/\delta^*_{N^*}\to \infty$. However, this comes at the cost of computational time. Suppose we perform wild binary segmentation with $M_{N^*}$ random intervals $[s_1,e_1],\dots,[s_{M_{N^*}},e_{M_{N^*}}]$, then for large ${N^*}$ the most time consuming part of the operation is to maximize the CUSUM statistic over each interval, with the other tasks in the WBinSeg procedure taking much less time. This takes an order of $\sum_{j=1}^{M_{N^*}}(e_j-s_j)$ time, and since the interval endpoints are drawn from $\{1,\dots,{N^*}\}$ with equal probability, we have $\mathbb{E}[e_j-s_j]=\frac{{N^*}}{3}(1+o(1))$ for all $j=1,\dots,M_{N^*}$. Hence the scaling of the average computational time for maximizing the CUSUM statistic in all random intervals, and WBinSeg as a whole, is $O({N^*}M_{N^*})=O\left( \frac{{N^*}^3}{(\delta^*_{N^*})^2}\log(({N^*})^2/\delta^*_{N^*}) \right)$ time, which is greater than $O({N^*}\log({N^*}))$ time for binary segmentation whenever ${N^*}/\delta^*_{N^*}\to\infty$. The trade-off between the increased accuracy of the estimates and the bigger computational time, as they pertain to intelligent sampling, will be analyzed later.
\newline
\newline
\indent As with the BinSeg method, we need to verify that the WBinSeg method also satisfies the consistency condition of (\ref{eq:firstconsistent}) at step (ISM3) of intelligent sampling, so that the results of Theorem \ref{thm:increasingJasymprotics} continue to hold with the latter as the first stage procedure. To this end, we need to demonstrate a set of signal estimators that satisfy the condition of Lemma \ref{cor:signalconsistent} with a $\rho^*_{N^*}$ such that $J\rho^*_{N^*}\to 0$. We do this by using the estimator proposed in (\ref{eq:signalestdef}), and also by imposing two further conditions:
\begin{enumerate}[label=(M\arabic* (WBinSeg)):]
	\setlength{\itemindent}{.5in}
	\setcounter{enumi}{7}
	\item $\Xi$ (from condition (M3)) is further restricted by $\Xi\in [0,1/3)\,,$
	\item $N_1$, from step (ISM2), is chosen so that $N_1=K_1N^\gamma$ for some $K_1>0$ and $\gamma>3\Xi$.
\end{enumerate}
\begin{lemma}\label{lem:wbinsegconsistent}
	Under conditions (M1) through (M4), (M8 (WBinSeg)), and (M9 (WBinSeg)), we have 
	\begin{eqnarray}\label{eq:wildsimulconsistent}
	\mathbb{P}\left[\hat{J}=J;\quad \max_{i=1,...,J}|\hat{\tau}_i-\tau_i|\leq w^*({N^*});\quad \max_{i=0,...,J}|\hat{\nu}_i-\nu_i|\leq \rho^*_{N^*}\right]\to 1
	\end{eqnarray}
	for some $\rho^*_{N^*}$ where $J\rho^*_{N^*}\to 0$.
\end{lemma}

\begin{proof}
	We need to show that wild binary segmentation does satisfy the requirements of Lemma \ref{cor:signalconsistent}. We have $\delta^*_{N^*}\gtrsim (N^{*})^{1-\frac{\Xi}{\gamma}}$ and $J\lesssim (N^{*})^{\frac{\Xi}{\gamma}}$. By the properties of the WBinSeg estimator shown in Theorem \ref{thm:wbinsegresults}, $w^*(N^*)\sim\log(N^*)$ and $B_{N^*} \sim{N^*}^{-1}$. We shall show that for $\rho=(N^*)^\theta$ where $\theta\in \left( \frac{\Xi}{2\gamma}-\frac{1}{2},-\frac{\Xi}{\gamma} \right)$, the conditions of Lemma \ref{cor:signalconsistent} are satisfied, and in addition, $J\rho^*_{N^*}\to 0$. 
	\newline
	\newline
	First, the set $ \left( \frac{\Xi}{2\gamma}-\frac{1}{2},-\frac{\Xi}{\gamma} \right)$ is a valid set since
	\begin{eqnarray}
	\frac{3\Xi}{2\gamma}-\frac{1}{2}=\frac{1}{2}\left( 3\frac{\Xi}{\gamma}-1 \right)<0.
	\end{eqnarray}
	by condition (M9 (WBinSeg)). Hence $\frac{\Xi}{2\gamma}-\frac{1}{2}<-\frac{\Xi}{\gamma}$. Second, $J\rho^*_{N^*}\to 0$ since
	\begin{eqnarray}
	J\rho^*_{N^*}\lesssim (N^*)^{ \frac{\Xi}{\gamma} }(N^*)^\theta
	\end{eqnarray}
	But 
	$$\theta+\frac{\Xi}{\gamma}<0.  $$
	Finally, as for the conditions of Lemma \ref{cor:signalconsistent}, we have
	\begin{itemize}
		\item $\frac{w^*(N^*)}{\delta^*_{N^*}}\lesssim(N^*)^{-(1-\frac{\Xi}{\gamma})}\log(N^*)\to 0$, hence $w^*(N^*)=o(\delta^*_{N^*})$
		\item $\frac{N^*}{\delta^*_{N^*}}B_{N^*}\sim\frac{1}{\delta^*_{N^*}}\to 0$
		\item because $\rho_{N^*}^*=(N^*)^\theta$ where $\theta\in \left( \frac{\Xi}{2\gamma}-\frac{1}{2},-\frac{\Xi}{\gamma} \right)$, we have $\rho_{N^*}^*=o(1)$
		\item $\frac{N^*w^*(N^*)}{(\delta^*_{N^*})^{3/2}\rho^*_{N^*}}\lesssim (N^*)^{1-\theta}\log(N^*)$, and $\delta^*_{N^*}(\rho^*_{N^*})^2\gtrsim (N^*)^{2\theta+1-\frac{\Xi}{\gamma}} $; since $2\theta+1-\frac{\Xi}{\gamma}>0$, this means that $\frac{N^*w^*(N^*)}{(\delta^*_{N^*})^{3/2}\rho^*_{N^*}}=o( \exp(C_2\delta^*_{N^*}(\rho^*_{N^*})^2) )$ for any positive constant $C_2$
	\end{itemize}
	Since all conditions of Lemma \ref{cor:signalconsistent} are satisfied, the signal estimators satisfy
	\begin{eqnarray}
	\mathbb{P}\left[ \hat{J}=J;\; \max_{j=0,\dots,J}|\hat{\nu}^*_j-\nu_j|\leq \rho^*_{N^*} \right]\to 1
	\end{eqnarray}
	which combines with the consistency of the change point estimators through a Bonferroni inequality to obtain the consistency result (\ref{eq:firstconsistent}).
\end{proof}

\begin{remark}
	As with the BinSeg algorithm, the WBinSeg procedure is asymptotically consistent but faces the same issues as BinSeg in a practical setting where the goal is to obtain confidence bounds $[\hat{\tau}_i\pm C_3\log({N^*})]$ for the change point $\tau_i$. Namely,  there are unspecified constants associated with the tuning parameter $\zeta_{N^*}$ and the confidence interval width $C_3\log({N^*})$ in (\ref{eq:wildbound}). The issue of choosing a confidence interval width will can be resolved by applying the procedure of Section \ref{sec:refitting}.
\end{remark}

%
%
%
%
\begin{table}
	\caption{Table of $\gamma_{min}$ and computational times for various values of $\Xi$, using WBinSeg at stage 1.}\label{tab:WBStime}
	\centering
	\begin{tabular}{|c|c|c|c|}
		\hline
		$\Xi$ & $[0,1/9)$ & $[1/9,1/7)$ & $[1/7,1/3)$ \\
		\hline
		$\gamma_{min}$ & $\frac{1-2\Xi+\Lambda}{2}$ & $\max\left\{\frac{1-2\Xi+\Lambda}{2},3\Xi+\eta \right\}$ & $\Xi+\eta$\\
		\hline
		Order of Time & $N^{(1+2\Xi+\Lambda)/2}\log(N)$ & $ N^{(1+2\Xi+\Lambda)/2}\log(N)$ or $N^{5\Xi+\eta}\log(N)$  & $N^{5\Xi+\eta}\log(N)$ \\
		\hline
		\hline
		$\gamma_{min}$ ($\Lambda=0$) & $\frac{1-2\Xi}{2}$ & $3\Xi+\eta$ & $3\Xi+\eta$ \\
		\hline
		Time ($\Lambda=0$) & $N^{(1+2\Xi)/2}\log(N)$ & $N^{5\Xi+\eta} \log(N)$  & $N^{5\Xi+\eta}\log(N)$  \\
		\hline
		\hline
		$\gamma_{min}$ ($\Lambda=\Xi$) & $\frac{1-\Xi}{2}$ & $\frac{1-\Xi}{2}$ & $3\Xi+\eta$\\
		\hline
		Time ($\Lambda=\Xi$) & $N^{(1+3\Xi)/2}\log(N)$ & $N^{(1+3\Xi)/2} \log(N)$  & $N^{5\Xi+\eta}\log(N)$ \\
		\hline
	\end{tabular}
\end{table}

\begin{figure}[H]
	\begin{center}
		\begin{overpic}[scale=0.42,tics=10]{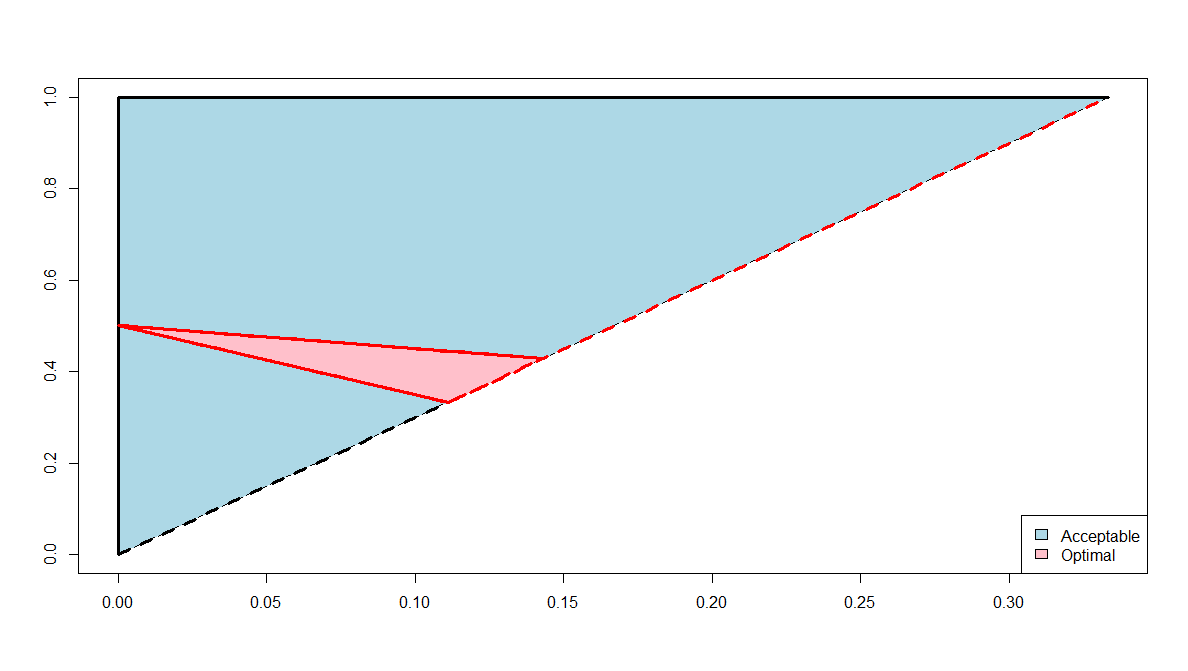}
			\put (55,2){\Large{$\Xi$}}
			\put (0,30){\Large$\gamma$}
			\put (12,21){$ \color{red} \gamma=\frac{1-2\Xi}{2} $}
			\put (20,31){$ \color{red} \gamma=\frac{1-\Xi}{2} $}
			\put (57,30){$\gamma=\Xi $}
			\put (44,56){$\gamma$ vs $\Xi$ for WBinSeg}
		\end{overpic}
	\end{center}
	\caption{Blue triangle encompasses all valid values of $\gamma$ vs $\Xi$ as set by (M8 (WBinSeg)). Pink region, solid red lines, and dotted red lines denotes $\gamma_{min}$ for each $\Xi$.}\label{fig:WBStime}
\end{figure}

\begin{remark} 
	\label{comparing-comuting-time}
	Although the $\gamma_{min}$ values are, across the board, smaller than those given in Table \ref{table-time-binseg} of the main paper (which records the $\gamma_{min}$ values and computational time of Binseg at Stage 1), the order of the actual computational time is also greater across the board. There is some advantage in using WBinSeg since consistency condition (\ref{eq:firstconsistent}) with $J\rho_N\to 0$ is satisfied in a greater regime of $\Xi$ ($\Xi<1/3$ as opposed to $\Xi<1/4$ for BinSeg), but in scenarios where the change points are placed far apart BinSeg would run in shorter time.
\end{remark} 
\subsection{Computational Time Order for Multiple Change Points}\label{sec:multipletime}

\indent To analyze the computational time when using WBinSeg at stage 1, we again assume that $\delta_N/N^\Xi\to K_1$ and $J(N)/N^\Lambda\to K_2$ for some constant $\Lambda\leq \Xi$ and positive constants $K_1$, $K_2$. To summarize the details, for $N_1\sim N^\gamma$, the average time for the first stage is $O(N^{\gamma+2\Xi}\log(N))$, while the second stage takes, on average,  $O(N^{1-\gamma+\Lambda}\log(N))$ time. Together with condition (M9 (WBinSeg)) and setting, $\gamma>3\Xi$, the order of average time for both stages combined is minimized by setting $\gamma_{min}=\max\left\{ \frac{1-2\Xi+\Lambda}{2},\Xi+\eta \right\}$ for any small constant $\eta$, with the average total computational time being $O(N^{\gamma_{min}+2\Xi}\log(N))$.
\newline
\newline
\noindent{\bf Detailed Analysis:} We have $\delta_{N}\sim N^{1-\Xi}$ for some $\Xi\in [0,1/3)$ and $J\sim N^\Lambda$ for some $\Lambda\in [0,\Xi]$. Given $n$ data points from (\ref{model}) with minimum separation $\delta_n$ between change points, it takes an order of $\frac{n^3}{\delta_n^2}\log(n)$ time to perform the procedure, due to having to use $M_n\sim \left(\frac{n}{\delta_n}\right)^2\log(n^2/\delta_n)$ random intervals, each requiring $O(n)$ time on average. The first stage works with a time series data of length order $N^\gamma$ with minimal separation $\delta_{N^*}^*$, and hence has an average computational time that is of the same order as$ (N^*/\delta_{N^*}^*)^2\cdot N^*\log(N^* )$, which is the same order as $N^{\gamma+2\Xi}\log(N)$. The second stage works with $\hat{J}$ intervals, each of width $CN^{1-\gamma}\log(N)$ for some constant $C$. Because we have
\begin{eqnarray}
\mathbb{P}[\hat{J}=J]\geq 1-C(N^*)^{-1}\qquad \text{ for some }C>0
\end{eqnarray}
by our earlier condition on $M_N$, we arrive at $\mathbb{E}[\hat{J}]=O(J)$ because
$$ \mathbb{E}[\hat{J}]\leq J(1-C(N^*)^{-1})+N(C(N^*)^{-1})\leq (C+1)J. $$
This in turn shows the expected computational time of the second stage is $O(JN^{1-\gamma}\log(N))$ which simplifies to $O(N^{1-\gamma+\Lambda}\log(N))$.
\newline
\newline
\indent Both stages combined are expected to take $O\left(N^{(\gamma+2\Xi)\vee (1-\gamma+\Lambda)}\log(N)\right)$ time. This fact combined with the requirement $\gamma>3\Xi$ lead to an optimal way to choose $\gamma$ to minimize the amount of computational time:
\begin{itemize}
	\item On the region $\Xi<1/9$ we can solve the equation $\gamma_{min}+2\Xi=1-\gamma_{min}+\Lambda$ to get the minimizing $\gamma$ as $\gamma_{min}=\frac{1-2\Xi+\Lambda}{2}$, which satisfies $\gamma_{min}>3\Xi$. This results in $O(N^{\frac{1+2\Xi+\Lambda}{2}}\log(N))$ computational time.
	\item On the region $\Xi\in [1/9,1/7)$:
	\begin{itemize}
		\item If $\frac{1-2\Xi+\Lambda}{2}> 3\Xi$, set $\gamma_{min}=\frac{1-2\Xi+\Lambda}{2}$ resulting in $O(N^{\frac{1+2\Xi+\Lambda}{2}}\log(N))$ computational time.
		\item Otherwise if $\frac{1-2\Xi+\Lambda}{2}\leq 3\Xi$, set $\gamma_{min}=3\Xi+\eta$, where $\eta>0$ is small, for $O(N^{5\Xi+\eta}\log(N))$ computational time.
	\end{itemize} 
	\item For $\Xi\in [1/7,1/3)$ also set $\gamma_{min}=\Xi+\eta$, where $\eta>0$ is small, for $O(N^{5\Xi+\eta}\log(N))$ computational time.
\end{itemize}

\subsection{Simulation Results for WBinSeg} We next looked at how effective intelligent sampling with WBinSeg would work in practice, by running a set of simulations with the same set of model parameters as in Setup 2 of Section \ref{sec:simulations}, and used the exact same method of estimation except with WBinSeg used in place of the Binseg algorithm. For the tuning parameters of WBinSeg, the same $\zeta_N$ was retained and the number of random intervals was taken as $M_n=20000$. Although we could have used the theoretically prescribed value of $M_N=9(N_1/\delta_{N_1}^*)^2\log(N_1^2/\delta_{N_1}^*)$, this turns out to be over 400,000 and is excessive as setting $M_N=20,000$ gave accurate estimates.
\begin{figure}[H]
	\begin{center}
		\includegraphics[scale=0.35]{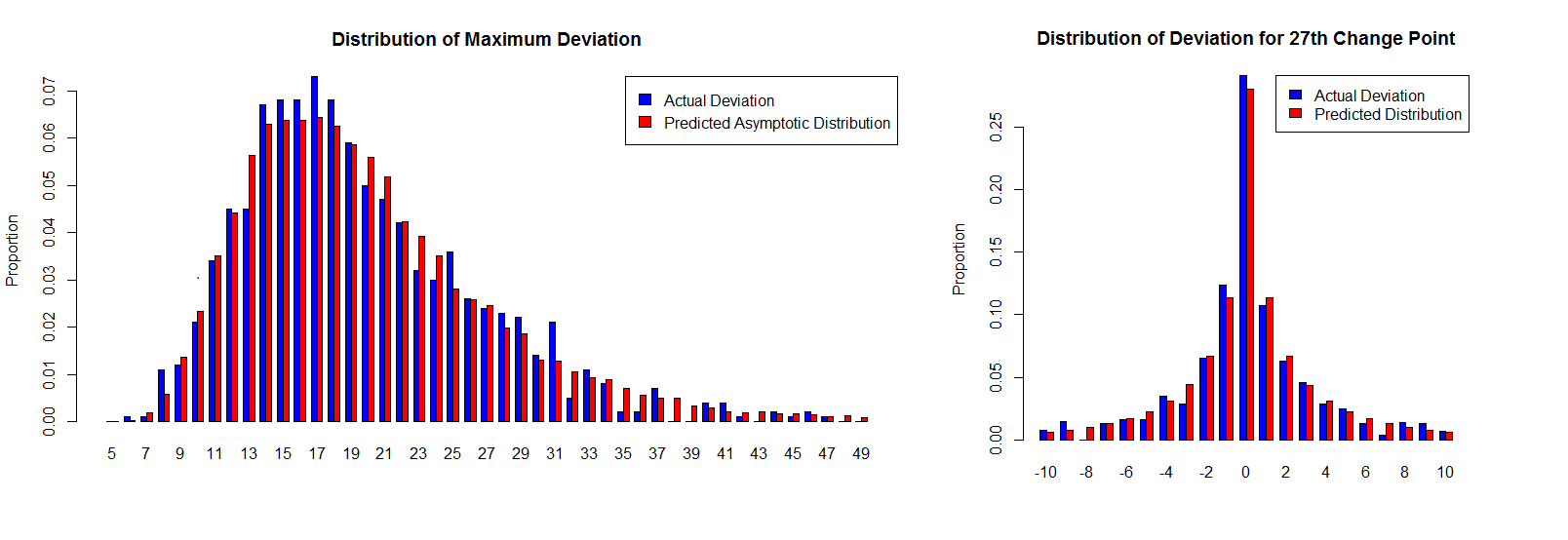}
	\end{center}
	\caption{Distributions of $\max_{1\leq j\leq 55}\lambda_2\left(\tau_j,\hat{\tau}_j^{(2)}\right)$ and $\lambda_2\left(\tau_{27},\hat{\tau}_{27}^{(2)}\right)$ from 1000 trials using the same parameters as setup 2 but employing WBinSeg instead of BinSeg. 
	}\label{fig:setup5}
\end{figure}
Using WBinSeg along with steps (D1) and (D2), the event $\{\hat{J}=J\}$ also occurred over 99\% of the time during simulations. One can also see from Figure \ref{fig:setup5} that the distribution of the $\hat{\tau}^{(2)}_j$'s again match with Theorem \ref{thm:multidepend}. Because the performance of intelligent sampling is near identical for this setup, regardless of whether BinSeg or WBinSeg was used, the reader may wonder why the latter isn't used for the previous few simulations. 
The reason is the following: when the re-fitting method from Section \ref{sec:refitting} is implemented, it results in the second stage intervals being of width $Q_1(1-\alpha/J)N^{1-\gamma}$, irrespective of which of BinSeg or WBinSeg was used at stage one [where $Q_1(1-\alpha/J)$ is the $1-\alpha/J$ quantile of $|L_1|$]. Hence, WBinSeg loses any possible advantage from the tighter confidence bound of width $O(\log(N))$ rather than $O(E_N)$ for BinSeg from stage one. So, in a sparse change point setting and with stage 1 refitting, WBinSeg provides no accuracy advantages but adds to the computational time, e.g., the 1000 iterations used to create Figure \ref{fig:setup5} averaged $\approx 293$ seconds, while the iterations used to create Figure \ref{fig:setup2} averaged $\approx 7$ seconds.

\bibliography{DCCP}
\bibliographystyle{abbrv}

\end{document}